\documentclass[11pt,a4paper]{book}

\newcommand{\MyTitle}{Derivation of Mean-field Dynamics for Fermions}
\newcommand{\MyAuthor}{S\"oren Petrat}
\newcommand{\MySubject}{Derivation of Mean-field Dynamics for Fermions}
\newcommand{\MyPdfKeywords}{Mean-field dynamics, Hartree-Fock equation}
\title{\MyTitle}
\author{\MyAuthor}

\usepackage[ngerman,english]{babel}

\usepackage{emptypage}

\usepackage{microtype}
\usepackage{hyperref}
\hypersetup{pdftitle={\MyTitle}, pdfauthor={\MyAuthor}, pdfsubject={\MySubject}, pdfkeywords={\MyPdfKeywords}}

\usepackage{amsmath}
\usepackage{amsfonts}
\usepackage{amssymb}
\usepackage{amsthm}
\usepackage{mathrsfs}

\usepackage{enumerate}
\usepackage{bm}
\usepackage{bbm}
\usepackage{verbatim}
\usepackage[usenames,dvipsnames]{color}

\theoremstyle{definition}\newtheorem{definition}{Definition}[chapter]
\theoremstyle{plain}\newtheorem{lemma}[definition]{Lemma}
\theoremstyle{plain}\newtheorem{theorem}[definition]{Theorem}
\theoremstyle{plain}\newtheorem{corollary}[definition]{Corollary}
\theoremstyle{plain}\newtheorem{assumption}[definition]{Assumption}
\theoremstyle{plain}\newtheorem{proposition}[definition]{Proposition}

\usepackage[twoside,
a4paper,
hmargin={34mm,24mm},
vmargin={30mm,34mm},
heightrounded,
footskip=10mm,
marginparsep={.1cm},
marginparwidth={1.7cm}
]{geometry}

\usepackage{fancyhdr}

\makeatletter
\renewcommand\part{%
  \if@openright
      \cleardoublepage
    \else
      \clearpage
   \fi
   \thispagestyle{empty}%
   \if@twocolumn
     \onecolumn
     \@tempswatrue
   \else
     \@tempswafalse
   \fi
   \null\vfil
   \secdef\@part\@spart}
\makeatother

\newcounter{remarks}
\newcommand{\be}{\begin{equation}}
\newcommand{\ee}{\end{equation}}

\newcommand{\scp}[2]{\left\langle #1 , #2 \right\rangle}
\newcommand{\SCP}[2]{\langle\!\langle #1 , #2 \rangle\!\rangle}
\newcommand{\bigSCP}[2]{\Big\langle\!\!\Big\langle #1 , #2 \Big\rangle\!\!\Big\rangle}
\newcommand{\bra}[1]{\langle #1 |}

\newcommand{\ket}[1]{| #1 \rangle}

\newcommand{\ketbra}[2]{| #1 \rangle \langle #2 |}

\newcommand{\ketbr}[1]{| #1 \rangle \langle #1 |}

\newcommand{\norm}[2][]{\left|\left| #2 \right|\right|_{#1}}
\newcommand{\Hilbert}{\mathscr{H}}
\newcommand{\Span}{\mathrm{span}}

\renewcommand{\Im}{\mathrm{Im}}
\newcommand{\AAA}{\mathcal{A}}
\newcommand{\QQQ}{\mathbb{Q}}
\newcommand{\RRR}{\mathbb{R}}
\newcommand{\CCC}{\mathbb{C}}
\newcommand{\NNN}{\mathbb{N}}
\newcommand{\ZZZ}{\mathbb{Z}}
\newcommand{\EEE}{\mathbb{E}}
\newcommand{\id}{\mathbbm{1}}
\newcommand{\floor}[1]{\left\lfloor #1 \right\rfloor}
\newcommand{\dir}{\mathrm{dir}}
\newcommand{\exch}{\mathrm{exch}}
\newcommand{\tr}{\mathrm{tr}}
\newcommand{\HS}{\mathrm{HS}}
\newcommand{\op}{\mathrm{op}}
\newcommand{\sym}{\mathrm{sym}}
\newcommand{\as}{\mathrm{as}}
\newcommand{\kin}{\mathrm{kin}}
\newcommand{\ia}{\mathrm{int}}

\newcommand{\supp}{\mathrm{supp}}
\newcommand{\vol}{\mathrm{vol}}
\newcommand{\Var}{\mathrm{Var}}
\newcommand{\mf}{\mathrm{mf}}
\newcommand{\Term}{\mathrm{Term}}
\newcommand{\loc}{\mathrm{loc}}
\newcommand{\ext}{\mathrm{ext}}
\newcommand{\semicl}{\mathrm{sc}}

\newcommand{\absatz}{\vspace{0.4cm}}
\newcommand{\eqexp}[1]{\text{\scriptsize [#1]~~~}}
\newcommand{\eqexpl}[1]{\text{\scriptsize [#1]\!\!\!\!\!\!\!\!\!\!}}

\begin{document}

\thispagestyle{empty}
\begin{center}
\mbox{}\\ [3\baselineskip]
{\large S\"oren Petrat}\\ [\baselineskip]
{\LARGE \textbf{Derivation of Mean-field Dynamics for Fermions}}\\ [3\baselineskip]
Dissertation an der Fakult\"at f\"ur\\
Mathematik, Informatik und Statistik der\\
Ludwig-Maximilians-Universit\"at M\"unchen\\[2\baselineskip]
\vfill
Eingereicht am 6.\ M\"arz 2014
\end{center}

\cleardoublepage

\pagenumbering{roman}

\thispagestyle{empty}
\begin{center}
\mbox{}\\ [3\baselineskip]
{\large S\"oren Petrat}\\ [\baselineskip]
{\LARGE \textbf{Derivation of Mean-field Dynamics for Fermions}}\\ [3\baselineskip]
Dissertation an der Fakult\"at f\"ur\\
Mathematik, Informatik und Statistik der\\
Ludwig-Maximilians-Universit\"at M\"unchen\\[2\baselineskip]
\vfill
\begin{minipage}[t]{.8\textwidth}
\mbox{}\\
{ \newlength\rowsep
\setlength\rowsep{\the\baselineskip}
Eingereicht am 6.\ M\"arz 2014 \\ [\rowsep]
Tag der m\"undlichen Pr\"ufung: 15.\ Mai 2014 \\ [\rowsep]
\begin{tabular}{@{} l l}
1.\ Gutachter: & {
\begin{minipage}[t][0pt][t]{\textwidth}
Prof.\ Dr.\ Peter Pickl \\
Mathematisches Institut der LMU \\
Theresienstr.\ 39 \\
D-80333 M\"unchen
\end{minipage}} \\[4\rowsep]
\makebox[0pt][l]{2.\ Gutachter:} & 
\makebox[0pt][l]{
\begin{minipage}[t][0pt][t]{\textwidth}
Prof.\ Dr.\ Benjamin Schlein \\
Institut f\"ur Mathematik, Universit\"at Z\"urich \\
Winterthurerstr.\ 190 \\
CH-8057 Z\"urich
\end{minipage}}
\end{tabular}
}
\end{minipage}
\end{center}

\cleardoublepage

\section*{Abstract}
In this work, we derive the time-dependent Hartree(-Fock) equations as an effective dynamics for fermionic many-particle systems. Our main results are the first for a quantum mechanical mean-field dynamics for fermions; in previous works, the mean-field limit is usually either coupled to a semiclassical limit, or the interaction is scaled down so much, that the system behaves freely for large particle number $N$. We mainly consider systems with total kinetic energy bounded by $const \cdot N$ and long-range interaction potentials, e.g., Coulomb interaction. Examples for such systems are large molecules or certain solid states. Our analysis also applies to attractive interactions, as, e.g., in fermionic stars. The fermionic Hartree(-Fock) equations are a standard tool to describe, e.g., excited states or chemical reactions of large molecules (like proteins). A deeper understanding of these equations as an approximation to the time evolution of a many body quantum system is thus highly relevant.

We consider the fermionic Hartree equations (i.e., the Hartree-Fock equations without exchange term) in this work, since the exchange term is subleading in our setting. The main result is that the fermionic Hartree dynamics approximates the Schr\"odinger dynamics well for large $N$. This statement becomes exact in the thermodynamic limit $N\to\infty$. We give explicit values for the rates of convergence. We prove two types of results. The first type is very general and concerns arbitrary free Hamiltonians (e.g., relativistic, non-relativistic, with external fields) and arbitrary interactions. The theorems give explicit conditions on the solutions to the fermonic Hartree equations under which a derivation of the mean-field dynamics succeeds. The second type of results scrutinizes situations where the conditions are fulfilled. These results are about non-relativistic free Hamiltonians with external fields, systems with total kinetic energy bounded by $const \cdot N$ and with long-range interactions of the form $|x|^{-s}$, with $0 < s < \frac{6}{5}$ (sometimes, for technical reasons, with a weaker or cut off singularity).

We prove our main results by using a new method for deriving mean-field dynamics developed by Pickl in [Lett.\ Math.\ Phys., 97(2):151--164, 2011]. This method has been applied successfully in quantum mechanics for deriving the bosonic Hartree and Gross-Pitaevskii equations. Its application to fermions in this work is new. The method is based on a functional that ``counts the number of particles outside the condensate'', i.e., in the case of fermions, it measures those parts of the Schr\"odinger wave function that are not in the antisymmetric product of the Hartree states. We show that convergence of the functional to zero (which means that the mean-field equations approximate the dynamics well) is equivalent to convergence of the corresponding reduced one-particle density matrices in trace norm and in Hilbert-Schmidt norm. Finally, we show how also the recently treated semiclassical mean-field limits can be derived with this method.

\clearpage

\selectlanguage{ngerman}

\section*{Zusammenfassung}
In dieser Arbeit werden die zeitabh\"angigen Hartree(-Fock) Gleichungen als effektive Dynamik f\"ur fermionische Vielteilchen-Systeme hergeleitet. Die Hauptresultate sind die ersten f\"ur eine quantenmechanische Mean-Field Dynamik ("`Mittlere-Feld Dynamik"') f\"ur Fermionen; in vorherigen Arbeiten ist der Mean-Field Limes \"ublicherweise entweder mit einem semiklassischen Limes gekoppelt oder die Wechselwirkung wird so stark runterskaliert, dass sich das System f\"ur gro{\ss}e Teilchenzahl $N$ frei verh\"alt. Wir betrachten haupts\"achlich Systeme, deren kinetische Energie durch $konst \cdot N$ beschr\"ankt ist, und langreichweitige Wechselwirkungen, wie z.B.\ Coulomb Wechselwirkung. Beispiele f\"ur solche Systeme sind gro{\ss}e Molek\"ule oder bestimmte Festk\"orper. Unsere Analyse gilt auch f\"ur anziehende Wechselwirkungen, wie z.B.\ in fermionischen Sternen. Die fermionischen Hartree(-Fock) Gleichungen sind ein Standardwerkzeug um z.B.\ angeregte Zust\"ande oder chemische Reaktionen in gro{\ss}en Molek\"ulen (wie Proteinen) zu beschreiben. Ein tieferes Verst\"andnis dieser Gleichungen als N\"aherung der Zeitentwicklung eines quantenmechanischen Vielteilchen-Systems ist daher \"au{\ss}erst relevant.

Wir betrachten in dieser Arbeit die fermionischen Hartree Gleichungen (d.h., die Hartree-Fock Gleichungen ohne Austauschterm), da der Austauschterm in unserem Fall von niedriger Ordnung ist. Das Hauptresultat ist, dass die fermionische Hartree Dynamik die Schr\"odinger Dynamik f\"ur gro{\ss}e $N$ gut ann\"ahert. Diese Aussage wird im thermodynamischen Limes $N\to\infty$ exakt. Wir geben explizite Konvergenzraten an. Es werden zwei Arten von Resultaten bewiesen. Die erste Art ist sehr allgemein und betrifft beliebige freie Hamiltonians (z.B.\ relativistisch, nicht-relativistisch, mit externen Feldern) und beliebige Wechselwirkungen. Die Theoreme geben explizite Bedingungen an die L\"osungen der fermionischen Hartree-Gleichungen an, unter denen eine Herleitung der Mean-Field Dynamik funktioniert. In der zweiten Art von Resultaten wird untersucht f\"ur welche Situationen diese Bedingungen erf\"ullt sind. Diese Resultate sind \"uber nicht-relativistische freie Hamiltonians mit externen Feldern, Systeme mit kinetischer Energie beschr\"ankt durch $konst \cdot N$ und mit langreichweitiger Wechselwirkung der Form $|x|^{-s}$, mit $0 < s < \frac{6}{5}$ (aus technischen Gr\"unden, manchmal mit abgeschnittener oder abgeschw\"achter Singularit\"at).

Die Hauptresultate werden mit einer neuen Methode zur Herleitung von Mean-Field Limiten bewiesen, die von Pickl in [Lett.\ Math.\ Phys., 97(2):151-164, 2011] entwickelt wurde. Diese Methode wurde in der Quantenmechanik erfolgreich zur Herleitung der bosonischen Hartree und Gross-Pitaevskii Gleichungen angewandt. Die Anwendung auf Fermionen in dieser Arbeit ist neu. Die Methode basiert auf einem Funktional, das die "`Anzahl der Teilchen au{\ss}erhalb des Kondensats z\"ahlt"', d.h. im Falle von Fermionen misst es die Anteile der Schr\"odinger Wellenfunktion, die nicht im antisymmetrisierten Produkt der Hartree-Zust\"ande sind. Wir zeigen, dass die Konvergenz des Funktionals gegen Null (was bedeutet, dass die Mean-Field Gleichungen die Dynamik gut ann\"ahern) \"aquivalent zur Konvergenz der zugeh\"origen Einteilchen-Dichtematrizen in Spur-Norm und Hilbert-Schmidt-Norm ist. Wir zeigen au{\ss}erdem wie die k\"urzlich behandelten semiklassischen Mean-Field Limiten mit dieser Methode hergeleitet werden k\"onnen.

\clearpage

\section*{Danksagung}

Ich m\"ochte zuallererst ganz herzlich meinem Doktorvater Peter Pickl f\"ur den Themenvorschlag und die wunderbare Betreuung dieser Doktorarbeit danken. Ich habe enorm viel von ihm gelernt (nicht nur $p$-$q$ Akrobatik) und es hat mich gefreut, dass ich an seinem Ideenreichtum und seiner Kreativit\"at teilhaben konnte. Weiterhin bin ich Detlef D\"urr \"uberaus dankbar f\"ur die langj\"ahrige Betreuung und Begleitung. Ich habe von ihm unsch\"atzbar viel \"uber Physik und Mathematik, insbesondere \"uber deren Bedeutung und Wert gelernt. Ich danke Roderich Tumulka f\"ur die tolle Zusammenarbeit an dem Projekt zu Viel-Zeit Wellenfunktionen und die vielf\"altige Unterst\"utzung w\"ahrend meiner Aufenthalte an der Rutgers University und auch w\"ahrend meiner Promotion. Benjamin Schlein danke ich f\"ur die Bereitschaft Zweitgutachter f\"ur diese Arbeit zu sein und f\"ur interessante Diskussionen w\"ahrend meines Besuchs in Bonn.

\absatz

Weiterhin ist es mir eine Freude der ganzen Arbeitsgruppe D\"urr/Pickl f\"ur das tolle wissenschaftliche Umfeld und die unvergleichlich wunderbare Atmosph\"are in der Gruppe zu danken. Ich habe die tiefgehenden Diskussionen \"uber Physik, Mathematik und Philosophie sehr genossen. Ganz besonders m\"ochte ich danken: Dirk-Andr\'{e} Deckert, der mich sozusagen f\"ur die Gruppe rekrutiert hat; Robert Grummt, Florian Hoffmann und Nicola Vona f\"ur die vielen gemeinsamen Aktivit\"aten; Florian Hoffmann ('s up?) insbesondere f\"ur die gute B\"uroatmosph\"are und eine tolle USA Reise; dem "`AK Relativity"' und "`DSS"' f\"ur spannende Diskussionen; Maximilian Jeblick und David Mitrouskas f\"ur die Zusammenarbeit an unserem Testteilchen-Projekt; dem Kinderzimmer f\"ur die gute Stimmung. Besonders m\"ochte ich Maximilian Jeblick, Nikolai Leopold und David Mitrouskas f\"ur das Korrekturlesen dieser Arbeit und f\"ur ihre Anmerkungen und Verbesserungsvorschl\"age danken.

\absatz

Neben meinem Studium an der LMU haben mich mehrere Aufenthalte an der Rutgers University sehr gepr\"agt. Daher m\"ochte ich an dieser Stelle herzlich allen danken, die mich w\"ahrend dieser Zeit unterst\"utzt haben - sei es durch ihre Gastfreundschaft oder "`Saturday morning waffles"' - und von denen ich so viel gelernt habe, durch Diskussionen und wissenschaftliche Zusammenarbeit. Insbesondere m\"ochte ich Shelly Goldstein, Michael Kiessling, August Krueger, Joel Lebowitz, Vishnya und Tim Maudlin, Ronald Ransome und Roderich Tumulka danken.

\absatz

Beim Cusanuswerk bedanke ich mich ganz herzlich f\"ur die Stipendien zur F\"orderung meines Studiums und meiner Doktorarbeit. Ich m\"ochte mich einerseits f\"ur die finanzielle Unterst\"utzung, gerade w\"ahrend meiner Promotion, aber auch f\"ur die "`ideelle F\"orderung"', z.B. durch Graduiertentagungen, und die pers\"onliche Betreuung bedanken. Der Fulbright-Kommission danke ich f\"ur das Stipendium, welches mein Auslandsjahr an der Rutgers University erm\"oglicht hat.

\absatz

Der gr\"o{\ss}te Dank geht letztendlich an meine Eltern und meinen Bruder. Meinen Eltern danke ich f\"ur alles, was sie f\"ur mich getan haben und ganz besonders daf\"ur, dass sie mir alle Freiheiten gelassen haben meinen eigenen Weg zu finden und zu gehen. Meinem Bruder bin ich f\"ur so viele Dinge dankbar: angefangen von einer tollen Kindheit, \"uber die vielf\"altigste Art und Weise der Unterst\"utzung im Studium und in der Promotion und ganz besonders f\"ur unsere jetzige starke Freundschaft.

\cleardoublepage

\selectlanguage{english}

\tableofcontents

\cleardoublepage

\pagenumbering{arabic}

\part{Background and Presentation of Main Results}\label{pt:part_one}

\chapter{Introduction}
This thesis is a contribution to the long-standing goal of statistical mechanics to derive \emph{effective dynamics} from microscopic laws of motion. By microscopic law of motion we mean for example Newton's equations in classical mechanics or the Schr\"odinger equation in quantum mechanics. We know that in many situations nature can be described very well by these theories. Nonetheless, very often the microscopic dynamics is not visible: air at room temperature for example obeys the laws of thermodynamics, which are not about the dynamics of the molecules that the gas is made of but rather about quantities like volume, pressure and temperature. So on a different \emph{scale} nature can appear very different. On a macroscopic scale we do not see the microscopic behavior that is described by the fundamental laws of motion, but we often see quite different behavior.

Such effective behavior arises in many different situations, usually when microscopic details can be neglected (e.g., when a system is described on a different scale or when forces can be replaced by their average value). Effective descriptions are much simpler, they involve much fewer degrees of freedom than the original microscopic description. Famous effective evolution equations are in classical mechanics the Boltzmann, Navier-Stokes or Vlasov equations, and in quantum mechanics the Hartree, Hartree-Fock or Gross-Pitaevskii equations. To derive an effective dynamics means to prove rigorously that the solutions to the effective equation approximate the solutions to the microscopic equation of motion well in certain situations. This is an ongoing project of mathematical physics. Only some cases are known where such a derivation can be conducted rigorously. In classical mechanics, this could for example be shown for the Vlasov equation \cite{hepp:1977}; however, for the Boltzmann equation it has been shown only for very short times \cite{lanford:1975}, and for the Navier-Stokes equation it is still an open problem (see the book by Spohn \cite{spohn:1991} for an excellent overview and introduction to this topic). In quantum mechanics, the derivation of effective dynamics for bosons near a condensate is well understood by now; see the works \cite{hepp:1974,spohn:1980,erdoes:2001,froehlich:2009,rodnianski:2009,pickl:2011method,pickl:2010hartree} for the case of the Hartree equation and \cite{erdoes:2006,erdoes:2007,erdoes:2007_2,erdoes:2009,erdoes:2010,pickl:2010gp_pos,pickl:2010gp_ext,benedikter:2012} for the Gross-Pitaevskii equation. However, only very little is known about derivations of mean-field dynamics for fermions. That is the topic of this thesis.

Thus far, only two situations have been considered in the literature: Either, the interaction is weakened so much that the particles behave freely for large particle number $N$, see \cite{bardos:2003,bardos:2004,bardos:2007,froehlich:2011}, or the mean-field limit is considered for wave functions with a semiclassical structure, such that this limit also leads to the classical Vlasov equation, see \cite{narnhofer:1981,spohn:1981,erdoes:2004,benedikter:2013,benedikter:2014}. (We give a more detailed discussion of the literature in Chapter~\ref{sec:mf_fermions_lit}.) The derivation of mean-field dynamics for fermions in a setting that leads to fully quantum mechanical behavior (in the sense that the dynamics is neither free nor close to a classical one) is what we are interested in.

We consider fermionic many-particle systems in quantum mechanics (mostly non-relativistic, although some of our theorems also apply to more general, e.g., relativistic settings). That is, the fundamental law of motion is the Schr\"odinger equation (we set $\hbar=1$)
\be\label{Schr_intro}
i \partial_t \psi^t = H \psi^t 
\ee
for antisymmetric complex-valued $N$-particle wave functions $\psi^t \in L^2(\RRR^{3N})$ (for simplicity, we neglect spin throughout this work). Antisymmetry means that $\psi^t(\ldots,x_j,\ldots,x_k,\ldots) = - \psi^t(\ldots,x_k,\ldots,x_j,\ldots) ~\forall j \neq k$. We consider Hamiltonians
\be\label{Schr_intro_H}
H = \sum_{j=1}^N H_j^0 + \sum_{i<j} v^{(N)}(x_i-x_j),
\ee
where $H_j^0$ acts only on $x_j$ and $v^{(N)}(x)=v^{(N)}(-x)$ is a real-valued pair-interaction potential (the superscript $(N)$ denotes a possible scaling and will be explained in Chapter~\ref{sec:mf_fermions}). According to \eqref{Schr_intro}, the unitary time evolution of an initial wave function $\psi^0$ is given by $\psi^t = e^{-iHt}\psi^0$ if $H$ is self-adjoint which we henceforth will assume. Note that for antisymmetric initial conditions $\psi^0$, the wave function $\psi^t$ remains antisymmetric under the Schr\"odinger evolution \eqref{Schr_intro} with Hamiltonian \eqref{Schr_intro_H} for all times. For the desired effective description, consider $N$ orthonormal one-particle wave functions (also called orbitals) $\varphi_1^t, \ldots, \varphi_N^t \in L^2(\RRR^3)$ which are solutions to the \emph{fermionic Hartree equations} (sometimes called \emph{reduced Hartree-Fock equations}). These are the coupled system of non-linear differential equations
\be\label{hartree_intro}
i \partial_t \varphi_j^t(x) = H^0 \varphi_j^t(x) + \bigg(v^{(N)} \star \sum_{k=1}^N |\varphi_k^t|^2 \bigg)(x) \, \varphi_j^t(x),
\ee
for $j=1,\ldots,N$, where $\star$ denotes convolution. Note that for orthonormal initial conditions $\varphi_1^0, \ldots, \varphi_N^0$, \eqref{hartree_intro} preserves the orthonormality for all times. The term
\be\label{intro_direct}
\bigg( v^{(N)} \star \sum_{k=1}^N |\varphi_k^t|^2 \bigg)(x) = \int_{\RRR^3 } v^{(N)}(x-y) \sum_{k=1}^N |\varphi_k^t(y)|^2 \, d^3y
\ee
is called the \emph{mean-field}. It can be viewed as the average value of the interaction potential at point $x$, created by particles distributed according to the density $\rho_N^t = \sum_{k=1}^N |\varphi_k^t|^2$. Note that closely related mean-field equations for fermions are the Hartree-Fock equations, where an additional exchange term
\be\label{exch_intro}
- \sum_{k=1}^N \Big(v \star ({\varphi_k^t}^*\varphi_j^t) \Big)(x) \, \varphi_k^t(x)
\ee
is present on the right-hand side of \eqref{hartree_intro}. In general, the Hartree-Fock equations are expected to be a better approximation than the fermionic Hartree equations; however, the exchange term is always smaller than the direct term \eqref{intro_direct}, and in our setting it is negligibly small (subleading compared to the direct term). Therefore, it is sufficient to consider only the fermionic Hartree equations here (see Chapter~\ref{sec:exch_term} for more details).

Now suppose that some initial $\varphi_1^0, \ldots, \varphi_N^0$ are given. Let the initial $N$-particle wave function be $\psi^0 \approx \bigwedge_{j=1}^N \varphi_j^0$, where $\bigwedge_{j=1}^N \varphi_j$ means the antisymmetrized product of $\varphi_1, \ldots, \varphi_N$ (see \eqref{antisymm_prod}). Then, under the Schr\"odinger evolution \eqref{Schr_intro}, this initial wave function evolves to $\psi^t = e^{-iHt} \psi^0$. We want to compare this $\psi^t$ to the wave function $\bigwedge_{j=1}^N \varphi_j^t$, where the $\varphi_j^t$ are the solutions to the fermionic Hartree equations \eqref{hartree_intro}. In other words, if still
\be
\psi^t \approx \bigwedge_{j=1}^N \varphi_j^t
\ee
at some time $t$, then the Schr\"odinger dynamics is approximated well by the Hartree dynamics and we say that we have derived the fermionic Hartree equations as an effective dynamics. That is the goal of this thesis.

Note, that in the presence of an interaction potential $v^{(N)}$ it is in general never true that $e^{-iHt} \bigwedge_{j=1}^N \varphi_j^0 = \bigwedge_{j=1}^N \varphi_j^t$, since the interaction leads to correlations between the particles. By correlations we mean those that are not due to the antisymmetry of the wave function, i.e., we mean that the wave function is in a superposition of more than one antisymmetric product state. We can therefore only expect the statement $\psi^t \approx \bigwedge_{j=1}^N \varphi_j^t$ to hold approximately. If a statement of the form $\psi^t \approx \bigwedge_{j=1}^N \varphi_j^t$ holds, then this means that only few correlations have developed. This can only be expected to happen in certain situations, for example, for short times (where the particles couldn't interact with each other long enough to produce correlations) or for weak interactions. The question is then: What exactly does ``few'', ``short'' or ``weak'' mean? This question is dealt with in Chapter~\ref{sec:mf_fermions}. There we identify interesting physical systems where we can expect the mean-field approximation to be valid.

After that, in the mathematical part of the thesis, we have to make precise what we mean by $\approx$ in $\psi^t \approx \bigwedge_{j=1}^N \varphi_j^t$. This is specified by a functional $\alpha(\psi^t,\varphi_1^t,\ldots,\varphi_N^t) =: \alpha(t)$ (first introduced by Pickl in \cite{pickl:2011method} for deriving mean-field limits for bosons), which measures ``how much'' of $\psi^t$ is not in the antisymmetric product of $\varphi_1^t,\ldots,\varphi_N^t$. In more detail, $\alpha(t)$ measures how many correlations have developed due to the interaction. Our main theorems give bounds on this $\alpha(t)$. Note again that the important question is if the antisymmetric product structure survives the time evolution. This is what the functional $\alpha(t)$ directly focuses on.

\absatz

\textbf{Structure of the thesis.} The thesis is organized into two parts. In Part~\ref{pt:part_one}, we introduce the subject of the thesis, give an overview of the underlying physics and present our main results. In Part~\ref{pt:part_two}, we give a proof of our main results.

In Chapter~\ref{sec:mf_fermions}, we provide a discussion of the mean-field description for fermions from a physical and mostly mathematically non-rigorous point of view. In Chapter~\ref{sec:mf_fermions_scaling_general}, we introduce and discuss the scaling we are later concerned with in some of our main results. This scaling is such that it leads to interesting quantum mechanical behavior. We discuss in some detail its meaning and possible applications of the scaled equations. In Chapter~\ref{sec:mf_fermions_semiclassical}, we give a brief overview of another interesting scaling, where the wave function naturally has a semiclassical structure, and, in fact, approximates the solutions to the classical Vlasov equation. We discuss the literature on the subject in more detail in Chapter~\ref{sec:mf_fermions_lit}. Furthermore, in Chapter~\ref{sec:mf_fermions_const_E_kin_fluc}, we make a remark about the connection between the correlations that develop due to the interaction and fluctuations around the mean-field. In Chapter~\ref{sec:exch_term}, we discuss the exchange term \eqref{exch_intro} that arises in the Hartree-Fock equations and argue why it is subleading in our setting.

In Chapter~\ref{sec:main_results}, we present the main results of this thesis and give an outline of their proofs. In Chapter~\ref{sec:counting_functional}, we first present in detail the definition of the counting functional $\alpha(t)$ and discuss its properties. In Chapter~\ref{sec:dens_mat_summary}, we explain how this $\alpha(t)$ is related to the difference of reduced density matrices in trace norm and Hilbert-Schmidt norm. The main result there are two lemmas showing that convergence of $\alpha(t)$ is equivalent to convergence of reduced density matrices in trace and Hilbert-Schmidt norm. The main results of this thesis are then stated and explained in detail in Chapter~\ref{sec:main_theorem_mf}. There are two kinds of results: those in Chapter~\ref{sec:main_theorem_mf_general_v} are concerned with general Hamiltonians as in \eqref{Schr_intro_H}, and they say that $\alpha(t)$ converges (i.e., the mean-field approximation is good) when certain assumptions on the solutions to the fermionic Hartree equations are fulfilled. Afterwards, in Chapter~\ref{sec:main_theorem_mf_x-s}, we present results that explicitly show that these assumptions are fulfilled for the scalings we discussed in Chapter~\ref{sec:mf_fermions_scaling_general} for many different interactions (and, in particular, for non-relativistic Hamiltonians, possibly with external fields). In Chapter~\ref{sec:outline_proof}, we give a detailed outline of the proofs of the so far presented results. In order to demonstrate that the $\alpha$-method is very versatile, we also give a proof of the convergence of $\alpha(t)$ for the semiclassical scaling, which was already achieved with other methods by Benedikter, Porta and Schlein \cite{benedikter:2013}. In Chapter~\ref{sec:estimates_semiclassical}, we present the main idea of our alternative proof, while we defer the full proof to Appendix~\ref{sec:proof_sc_scaling}. Finally, in Chapter~\ref{sec:outlook}, we give a brief outlook on remaining open problems related to mean-field descriptions for fermions.

Part~\ref{pt:part_two} of this thesis contains a proof of the main results. In Chapter~\ref{sec:notation}, we establish some notation, recall inequalities we often use and discuss in more detail the properties of certain projectors that are needed to define $\alpha(t)$. In Chapter~\ref{sec:density_matrices}, we prove the main results of Chapter~\ref{sec:dens_mat_summary} about the relation between $\alpha(t)$ and reduced density matrices. Then, in Chapter~\ref{sec:alpha_dot_and_general_lemmas}, we prove the results of Chapter~\ref{sec:main_theorem_mf_general_v}, and in Chapter~\ref{sec:mean-field_scalings_general}, we prove the results stated in Chapter~\ref{sec:main_theorem_mf_x-s}.

\chapter{The Physics: Mean-field Dynamics of Fermions}\label{sec:mf_fermions}
\section{A New Scaling for Fermionic Mean-field Limits}\label{sec:mf_fermions_scaling_general}
As we mentioned in the introduction, we can expect mean-field behavior only in certain situations. In this chapter we identify physical systems where there is interesting mean-field behavior. It will be convenient to consider scaled evolution equations. Here, we introduce the scaling, discuss its origin and physical relevance, give different formulations of the problem, and discuss applications of the presented scaled equations.

\subsection{The Scaled Equations and the Physics}\label{sec:mf_fermions_const_E_kin_physics}
There will be two types of theorems in this work: those in Chapter~\ref{sec:main_theorem_mf_general_v} are about general Hamiltonians of the form \eqref{Schr_intro_H}, and those in Chapter~\ref{sec:main_theorem_mf_x-s} concern non-relativistic Hamiltonians and certain long-range interactions. In order to explain what physical situations we have in mind, it is easier to formulate our equations explicitly here for the latter case. We first present and discuss the scaled equations and give more details on the origin of the scaling afterwards, in Chapter~\ref{sec:mf_fermions_scaling}. The following analysis is done for dimension $d=3$, but could also be conducted for other dimensions.

We consider the non-relativistic Schr\"odinger equation for an antisymmetric $N$ particle wave function $\psi^t \in L^2(\RRR^{3N})$ (we set $\hbar=1=2m$ throughout this work)
\be\label{scaling_Schr}
i \partial_t \psi^t(x_1,\ldots,x_N) = \left( \sum_{j=1}^N \left( -\Delta_{x_j} + w^{(N)}(x_j) \right) + N^{-\beta} \sum_{i<j} v(x_i-x_j) \right) \psi^t(x_1,\ldots,x_N),
\ee
where $\Delta_{x_j}$ is the Laplace operator, acting on $x_j$, $w^{(N)}$ is an external field (that can possibly depend on $N$), $\beta \in \RRR$ is the \emph{scaling exponent}, and $v(x)=v(-x)$ is a real-valued pair interaction potential. The corresponding fermionic Hartree equations are
\be\label{scaling_hartree}
i \partial_t \varphi_j^t(x) = \left( -\Delta + w^{(N)}(x) + N^{-\beta} \left(v \star \rho_N^t \right)(x) \right) \varphi_j^t(x),
\ee
for $j=1,\ldots,N$, where we denote the density by $\rho_N^t=\sum_{i=1}^N |\varphi_i^t|^2$. 

Let us now discuss for which physical systems Equation~\eqref{scaling_Schr} is applicable. For this discussion it is useful to consider the $N$-dependence of the expectation values of the kinetic energy,
\be\label{expc_E_kin}
E_{\kin,\psi^t} = \bigSCP{\psi^t}{\sum_{i=1}^N (- \Delta_{x_i}) \psi^t},
\ee
and of the interaction energy,
\be\label{expc_E_int}
E_{\ia,\psi^t} = \bigSCP{\psi^t}{\sum_{i<j} v(x_i-x_j) \psi^t}
\ee
(note that we did not include the factor $N^{-\beta}$ in our definition of the interaction energy), where $\SCP{\cdot}{\cdot}$ denotes the scalar product in $L^2(\RRR^{3N})$. (We will often refer to the expressions \eqref{expc_E_kin}, \eqref{expc_E_int} simply as kinetic and interaction energy, although they are only expectation values.) The situation we want to consider is one where the total kinetic energy is bounded from above by $CN$, where $C$ is some $N$-independent constant. We then say that $E_{\kin,\psi^t}$ is $O(N)$.\footnote{In the following we say that a function $f(N)$ is of order $N^p$, or simply $O(N^p)$, if there is a constant $C$ (independent of $N$) such that $f(N) \leq C N^p$. The interesting cases are usually when $f(N)$ is also bounded from below, i.e., when there is a constant $D$, such that $D N^p \leq f(N)$.} Now an interesting effect that holds only for fermions comes into play: Due to the antisymmetry of the wave function (or the Pauli principle or Fermi pressure), the particles have to occupy a volume that grows with $N$. Let us explain in more detail what this means. First, let us compare the situation to bosons. A very simple bosonic wave function is $\phi(x_1,\ldots,x_N) = \prod_{j=1}^N\varphi_V(x_j)$, where $\supp(\varphi_V) = V$ for some volume $V \subset \RRR^3$, and $E_{\kin,\phi} = \SCP{\phi}{ \sum_{j=1}^N (-\Delta_{x_j}) \phi} = N \scp{\varphi_V}{(-\Delta) \varphi_V}$ is $O(N)$ ($\scp{\cdot}{\cdot}$ denotes the scalar product in $L^2(\RRR^3)$). Here the particles occupy a \emph{constant} volume $V$. This, in contrast, is not possible for fermions. To illustrate this, let us give an example. Consider plane waves in a box, that is, the free ground state in $V_L = \big[-\frac{L}{2},\frac{L}{2}\big]^3$ with periodic boundary conditions. The general form of an antisymmetric product state is
\be\label{antisymm_prod}
\psi(x_1,\ldots,x_N) = \left( \bigwedge_{j=1}^N \varphi_j \right)(x_1,\ldots,x_N) := \frac{1}{\sqrt{N!}} \sum_{\sigma \in S_N} (-1)^{\sigma} \prod_{j=1}^N \varphi_{\sigma(j)}(x_j),
\ee
where $S_N$ is the set of all permutations of $1,\ldots,N$, $(-1)^{\sigma}$ is the sign of the permutation $\sigma$ and $\varphi_1,\ldots,\varphi_N$ are orthonormal. For free particles in a box,
\be\label{free_particles_box}
\varphi_j(x) = L^{-\frac{3}{2}} \, e^{i \frac{2\pi}{L} k_j \cdot x} \, \id_{V_L}(x),
\ee
where $k_j \in \ZZZ^3$. Since we want to consider the ground state, we choose the $k_j$ increasingly, such that $|k_N|$ is as small as possible (while of course $k_i\neq k_j \forall i\neq j$). For this wave function, we find that
\begin{align}\label{scaling_plane_wave_E_kin}
E_{\kin,\bigwedge \varphi_j} = \sum_{i=1}^N \scp{\varphi_i}{(-\Delta)\varphi_i} &= \sum_{i=1}^N \left( \frac{2\pi}{L} k_i \right)^2 \nonumber \\
&\leq C \left(\frac{2\pi}{L}\right)^2 \int_0^{N^{\frac{1}{3}}} r^2~ r^2 dr \nonumber \\
&\propto N^{\frac{5}{3}}L^{-2}.
\end{align}
Thus, if the kinetic energy is proportional to $N$, then $L \propto N^{\frac{1}{3}}$, i.e., the volume $L^3 \propto N$. In general, one can show that a similar statement holds for \emph{any} fermionic wave function. The precise statement is Lemma~\ref{lem:Omega}. It says in particular that, if the kinetic energy is of $O(N)$, then the average number of particles in a volume of $O(N)$ is of $O(N)$.

The fact that fermionic wave functions with kinetic energy of $O(N)$ naturally ``occupy a volume'' that grows in $N$ is now important for the interaction energy. We want to consider long-range interactions like the Coulomb interaction, so the ``size'' of the system is very important. What we want to use is that, with growing $N$, each particle ``feels'' more and more other particles. To illustrate this, let us consider the mean-field interaction term $|\cdot|^{-1} \star \rho_N^t$ from the Hartree equation (without the $N^{-\beta}$) with Coulomb interaction. For the example of plane waves from above, its maximum value can easily be evaluated, since $\rho_N^t = \frac{N}{L^3} \id_{V_L}$. Recall that $L\propto N^{\frac{1}{3}}$, i.e., the density $\rho = \frac{N}{L^3}$ is constant. We find
\be\label{coulomb_gs}
\big(|\cdot|^{-1} \star \rho_N^t\big)(x) \leq \big(|\cdot|^{-1} \star \rho_N^t\big)(0) = \rho \int_{\RRR^3} |y|^{-1} \id_{V_L}(y) \,d^3y \propto \int_0^{N^{\frac{1}{3}}} r^{-1}\, r^2 dr \propto N^{\frac{2}{3}}.
\ee
Thus, the interaction energy per particle is $O(N^{\frac{2}{3}})$, due to the long range of the Coulomb interaction; the total interaction energy is then $O(N^{\frac{5}{3}})$. If we now choose the scaling exponent $\beta=\frac{2}{3}$, then the kinetic term and the scaled interaction term in the Schr\"odinger equation are of the same order, $O(N)$. Thus, for times of $O(1)$, we would expect interesting mean-field behavior for large $N$: heuristically speaking, each particle feels $O(N^{\frac{2}{3}})$ other particles (due to the fact that the system size grows and the interaction has long range), but only with strength $O(N^{-\frac{2}{3}})$. Lemma~\ref{lem:scaling_x-s} makes this statement exact. It says that, under the condition that $E_{\kin}$ is $O(N)$, for interactions with long-range part like $|x|^{-s}$, with $0<s<\frac{6}{5}$, $v\star\rho_N$ is of order $N^{\beta}$, with scaling exponent $\beta=1-\frac{s}{3}$. Note again that it is only the long-range behavior of the interaction that dictates the scaling exponent $\beta$. (The interactions we consider in our main results sometimes have the singularity weakened or cut off.)

If an external field $w^{(N)}$ is present, then also the total external field energy
\be
E_{\ext,\psi^t} = \bigSCP{\psi^t}{\sum_{i=1}^N w^{(N)}(x_i) \psi^t}
\ee
should be of $O(N)$. In principle, the external field could be time-dependent, as long as it preserves the bound $E_{\kin,\psi^t} \leq CN$.

Let us summarize the orders in $N$ of the terms in the Equations~\eqref{scaling_Schr} and \eqref{scaling_hartree} (for ease of notation, without external field). In the following, note that the informal notation with the curly brackets refers to (the expectation values of) the energies associated with the terms in the equations. We consider long-range interaction potentials $v$ and the corresponding appropriate $\beta$; more exactly, for interactions with long-range behavior $|x|^{-s}$, $0<s<\frac{6}{5}$, the scaling exponent is $\beta=1-\frac{s}{3}$. For the Schr\"odinger equation we have
\be\label{Schr_orders}
i \partial_t \psi^t(x_1,\ldots,x_N) = \underbrace{ - \sum_{j=1}^N \Delta_{x_j} \psi^t(x_1,\ldots,x_N)}_{O(N)} + \underbrace{N^{-\beta} \sum_{i<j} v(x_i-x_j) \psi^t(x_1,\ldots,x_N)}_{O(N)},
\ee
and for the fermionic Hartree equations
\be\label{hartree_orders}
i \partial_t \varphi_j^t(x) = \underbrace{-\Delta \varphi_j^t(x)}_{O(1)} + \underbrace{N^{-\beta} \left(v \star \rho_N^t \right)(x) \varphi_j^t(x)}_{O(1)}.
\ee
Heuristically, one sees that we can indeed expect the mean-field approximation to be valid for arbitrary times $t$ of $O(1)$. From Equation~\eqref{hartree_orders} one can read off, that the limit of large $N$ leads to interesting mean-field behavior, since both the kinetic term and the interaction term are of $O(1)$. If, for example, the kinetic term would be of $O(1)$ but the interaction term were of $O(N^{-\delta})$ (for some $\delta>0$), then the interaction term would be negligibly small for large $N$ (for times of $O(1)$) and vanish in the limit $N\to\infty$. Our results in Chapter~\ref{sec:main_theorem_mf_x-s} are about the above Equations~\eqref{Schr_orders} and \eqref{hartree_orders}, possibly with external fields, and sometimes with weakened or cut off singularity of the interaction potential. The main result is, that, if the kinetic energy stays bounded by $CN$ for all times $t$, then the mean-field dynamics \eqref{hartree_orders} indeed approximates the Schr\"odinger dynamics \eqref{Schr_orders} well.

Finally, let us briefly compare the situation to the case of bosons near a condensate state $\phi(x_1,\ldots,x_N) = \prod_{j=1}^N\varphi(x_j)$. For such a state, if the associated kinetic energy $E_{\kin,\phi} = \SCP{\phi}{ \sum_{j=1}^N (-\Delta_{x_j}) \phi} = N \scp{\varphi}{(-\Delta) \varphi}$ is $O(N)$, then $\varphi$ naturally lives in some constant, $N$-independent volume. The density is therefore of $O(N)$; each particle ``feels'' the interaction of $O(N)$ other particles, such that the total interaction energy is of $O(N^2)$. For bosons near a condensate it is thus natural to choose the scaling exponent $\beta=1$, so that kinetic and interaction energy are of the same order (independent of the long-range part of the interaction potential). The mean-field description can be expected to hold due to high densities, and not due to the long range of the interaction. We will encounter a similar high-density situation in Chapter~\ref{sec:mf_fermions_semiclassical}, where we discuss the semiclassical scaling for fermions.

\subsection{Origin of the Scaling}\label{sec:mf_fermions_scaling}
Let us explain in this section how a factor $N^{-\beta}$ in front of the interaction arises from a rescaling of time and space coordinates.

Let us denote the ``microscopic'' or ``physical'' time and space coordinates by $\tilde{t} \in \RRR$ and $\tilde{x}_j \in \RRR^3$, $j=1,\ldots,N$. We denote the wave function in these coordinates by $\tilde{\psi}(\tilde{t},\tilde{x}_1,\ldots,\tilde{x}_N)$. We assume that it is normalized. In the following, let us consider non-relativistic fermions with Coulomb interaction, and, for ease of the presentation, without external field. The wave function $\tilde{\psi}$ is then a solution to the Schr\"odinger equation
\be\label{Schr_unscaled}
i \partial_{\tilde{t}} \tilde{\psi}(\tilde{t},\tilde{x}_1,\ldots,\tilde{x}_N) = \left( -\sum_{j=1}^N \Delta_{\tilde{x}_j} + \sum_{i<j} \frac{1}{\lvert \tilde{x}_i-\tilde{x}_j \rvert} \right) \tilde{\psi}(\tilde{t},\tilde{x}_1,\ldots,\tilde{x}_N),
\ee
where we set the coupling constant in front of the Coulomb potential $(4\pi\varepsilon_0)^{-1}=1$. As explained in the introduction, the mean-field approximation is expected to become better the larger the number $N$ of particles gets. We therefore consider $N$-dependent scalings. The scaling we are interested in is given by
\be\label{scaling1}
t=N^{\frac{4}{3}}\tilde{t},~ x=N^{\frac{2}{3}}\tilde{x}.
\ee
What is achieved by this scaling is a ``zoomed in'' description for ``very short'' times: $t$ and $x$ are very big compared to $\tilde{t}$ and $\tilde{x}$ for large $N$. Heuristically, one could say, that we want to look at small length scales where interesting quantum behavior happens on fast time scales. Let us now express the wave function $\tilde{\psi}$ in the new coordinates $t,x$. It is given by
\be\label{wf_rescaled}
\psi(t,x_1,\ldots,x_N) = N^{-N} \tilde{\psi}(N^{-\frac{4}{3}}t,N^{-\frac{2}{3}}x_1,\ldots,N^{-\frac{2}{3}}x_N),
\ee
where the prefactor $N^{-N}$ is introduced such that $\psi$ is normalized. The dynamics of the wave function $\psi$ is determined by the Schr\"odinger equation in the scaled coordinates, that is, $\psi$ is the solution to
\be\label{Schr_scaled}
i N^{\frac{4}{3}} \partial_t \psi(t,x_1,\ldots,x_N) = \left( - \sum_{j=1}^N N^{\frac{4}{3}} \Delta_{x_j} + N^{\frac{2}{3}} \sum_{i<j} \frac{1}{\lvert x_i-x_j \rvert} \right) \psi(t,x_1,\ldots,x_N),
\ee
which follows directly from \eqref{Schr_unscaled} (by applying the chain rule). Let us simplify this equation by dividing by $N^{\frac{4}{3}}$. Then the scaled Schr\"odinger equation is
\be\label{Schr_scaled2}
i \partial_t \psi(t,x_1,\ldots,x_N) = \left( - \sum_{j=1}^N \Delta_{x_j} + N^{-\frac{2}{3}} \sum_{i<j} \frac{1}{\lvert x_i-x_j \rvert} \right) \psi(t,x_1,\ldots,x_N),
\ee
which is exactly \eqref{scaling_Schr} for the case of Coulomb interaction. Thus, the effect of looking at the system on the new scales \eqref{scaling1} is a factor $N^{-\frac{2}{3}}$ in front of the interaction. 

Let us take a closer look at the wave functions $\tilde{\psi}$ and $\psi$. As explained in Chapter~\ref{sec:mf_fermions_const_E_kin_physics}, it is natural to consider wave functions $\psi$ with kinetic energy of $O(N)$; then the particles naturally ``occupy a volume'' that grows with $N$, such that also the interaction term is $O(N)$ in the scaled Schr\"odinger equation. Now suppose that $\psi$ lives in a volume proportional to $N$ (say, a ball with radius $N^{\frac{1}{3}}$, such that $\psi(t,x_1,\ldots,x_N)=0$ whenever any $|x_j| > N^{\frac{1}{3}}$). This means, that the wave function $\tilde{\psi}$ lives in a volume proportional to $N^{-1}$, as can be read off from \eqref{wf_rescaled}. Here, we see again that the effect of our coordinate rescaling is a ``zoomed in'' description, in this case, for a wave function with shrinking volume in $N$; this can be relevant for attractive interactions, e.g., gravitation, where the system becomes smaller the more particles are added.

An equation of the form \eqref{scaling_Schr}, i.e., with interaction $N^{-\beta}\sum_{i<j}v(x_i-x_j)$, can be derived from a scaling only for certain interactions, e.g., $v(x)=|x|^{-s}$. In this case, one can rescale $t=N^{2\delta}\tilde{t}$, $x=N^{\delta}\tilde{x}$. This leads to an interaction $N^{-\delta (2-s)} \sum_{i<j} |x_i-x_j|^{-s}$ in the scaled equation. As mentioned in Chapter~\ref{sec:mf_fermions_const_E_kin_physics}, for $\beta = 1-\frac{s}{3}$, which corresponds to $\delta=\frac{s-3}{3s-6}$, both kinetic and interaction terms are of the same order. For other interactions, the scaled equations look different; there, the effect of the scaling is that $v(\tilde{x})$ becomes $N^{\varepsilon}v(N^{-\delta}x)$ in the scaled equation, for some $\varepsilon\in\RRR$.

\subsection{Different Formulations of the Problem}
The goal of this thesis is to show that the mean-field equations for fermions approximate the Schr\"odinger dynamics well in certain situations. We saw above that systems with kinetic energy of $O(N)$ and long-range interactions are interesting systems where one can expect mean-field behavior on certain scales. There are now different ways of formulating what these scales are.

\begin{enumerate}[(a)]
\item In Chapter~\ref{sec:mf_fermions_const_E_kin_physics}, we saw that, if we put a factor $N^{-\beta}$ in front of the interaction, we can expect the mean-field approximation to hold for times of $O(1)$. Such a factor means that the interaction is weak. One possibility is that such weak interaction has a physical origin, e.g., it can be due to screening effects in a large molecule. There, excited states can be very delocalized: they interact with very many other electrons, but only weakly due to the screening from the nuclei.
\item For certain interactions (usually of the form $|x|^{-s}$), the factor $N^{-\beta}$ can also arise from a scaling in the sense of an $N$-dependent coordinate transformation, as described in Chapter~\ref{sec:mf_fermions_scaling}. The scaling factor arises ``out of convenience'', since one could as well work in the original coordinates. However, then one could not expect the mean-field approximation to hold for times of $O(1)$, but rather for very short times, in fact, times of $O(N^{-\frac{4}{3}})$ (for Coulomb interaction); see also Remark~\ref{itm:no_scaling} in Chapter~\ref{sec:main_theorem_mf_general_v}.
\item Alternatively, one could use no scaling at all. Then, in the situation of Chapter~\ref{sec:mf_fermions_const_E_kin_physics} and Coulomb interaction, the kinetic term is $O(N)$ and the interaction term is $O(N^{\frac{5}{3}})$. Then the mean-field approximation can only be expected to hold for short times, in fact, times of $O(N^{-\frac{2}{3}})$; see also Remark~\ref{itm:no_scaling} in Chapter~\ref{sec:main_theorem_mf_general_v}. This formulation is closer to the idea mentioned in the introduction, that times are so short that the particles could not develop severe correlations.
\item From a practical point of view, it would be useful to not use a scaling at all, and instead, given a fixed physical system, to leave all the constants $\hbar$, $m$ and coupling constants in the original equations. If one could calculate explicit error terms for how much the mean-field approximation deviates from the Schr\"odinger evolution (as we in fact do in some of our main results), then one can directly read off for how long the mean-field approximation can be expected to be good, depending on the constants in the Schr\"odinger equation and the parameters of a given system.
\end{enumerate}

When we write down the main results for the case where the kinetic energy is bounded by $CN$, we simply use interaction potentials with a prefactor $N^{-\beta}$. One could easily reformulate the results without this prefactor as we point out in Remark~\ref{itm:no_scaling} in Chapter~\ref{sec:main_theorem_mf_general_v}.

\subsection{Applications}
This work is mainly a theoretical work, showing that and how in principle the Hartree(-Fock) equations can be derived from the microscopic Schr\"odinger dynamics. We do not focus on practical applications here. However, we strongly want to emphasize that the time-dependent Hartree(-Fock) approximation has very high relevance throughout theoretical physics and chemistry. To illustrate this, let us mention a few applications here. (Since numerous references to the mentioned applications can easily be found, we refrain from explicitly providing them here.)

The Hartree(-Fock) equations are widely used in theoretical chemistry to describe chemical reactions or excited states of large molecules (e.g., large proteins). They are, for example, often used for numerical simulations of chemical reactions. In a large molecule, it is indeed the case that the total energy is proportional to $N$ (and the density is $O(1)$), in accordance with the scenario we discussed in Chapter~\ref{sec:mf_fermions_const_E_kin_physics}. One has to be a bit careful here: the equations for a real molecule do not have a scaling factor $N^{-\beta}$ in front of the interaction; in fact, as the stability of matter program of Lieb et al.\ has proven rigorously \cite{lieb:2010}, the interaction energy and external field energy from the nuclei together are of $O(N)$, which makes the system stable. However, the scaled equation \eqref{scaling_Schr} might model screening effects from the nuclei for very delocalized electrons, e.g., electrons in excited states or certain molecular bonds. Other applications are in solid state physics the description of electrons in metals (e.g., in conduction bands) or semiconductors. The time-dependent Hartree(-Fock) equations have also been used in nuclear physics to study collisions of large nuclei. With recent experimental advances in laser physics, it has become possible to study cold fermions in laser traps, and thus to directly check the validity of the mean-field approximation. Finally, the Hartree(-Fock) equations can be used to describe fermionic stars, e.g., neutron stars or white dwarfs. In this scenario, it is indeed the case that the systems size shrinks with the particle number, due to the attractive gravitational interaction. In particular, the scenario discussed in Chapter~\ref{sec:mf_fermions_scaling} can be applicable (see also the scaling discussed in Chapter~\ref{sec:mf_fermions_semiclassical}).

\section{Mean-field Limit Coupled to a Semiclassical Limit}\label{sec:mf_fermions_semiclassical}
Another situation where one can expect interesting mean-field behavior is when the mean-field limit is coupled to a semiclassical limit. In this case, the wave function $\psi^t(x_1,\ldots,x_N)$ is a solution to the Schr\"odinger equation
\be\label{Schr_scaled_sc}
i N^{-\frac{1}{3}} \partial_t \psi^t = \left( \sum_{j=1}^N \left( -N^{-\frac{2}{3}} \Delta_{x_j} + w^{(N)}(x_j) \right) + N^{-1} \sum_{i<j} v(x_i-x_j) \right) \psi^t,
\ee
and $\varphi_1^t,\ldots,\varphi_N^t$ are solutions to the corresponding Hartree equations
\be\label{hartree_scaled_sc}
i N^{-\frac{1}{3}} \partial_t \varphi_j^t(x) = \left( -N^{-\frac{2}{3}} \Delta + w^{(N)}(x) + N^{-1} \left(v \star \rho_N^t \right)(x) \right) \varphi_j^t(x),
\ee
for $j=1,\ldots,N$ (recall $\rho_N^t=\sum_{i=1}^N |\varphi_i^t|^2$). As in Chapter~\ref{sec:mf_fermions_scaling}, for Coulomb interaction, the scaling can arise from a coordinate transformation
\be\label{scaling_sc}
t=N\tilde{t},~ x=N^{\frac{1}{3}}\tilde{x},
\ee
i.e., similar to \eqref{scaling1}, one considers ``small'' time and length scales.

The physical situation one considers here is particles confined to some $N$-independent volume, e.g., particles in a box with fixed size, or in a nice external trapping potential. One then considers states close to the ground state of such a system. Then we already know from Chapter~\ref{sec:mf_fermions_const_E_kin_physics}, that the kinetic energy cannot be just $O(N)$, but it has to grow faster. If we consider the example from \eqref{free_particles_box} again, where the ground state has $E_\kin \propto N^{\frac{5}{3}}L^{-2}$, we see that for $N$-independent $L$, the kinetic energy is $O(N^{\frac{5}{3}})$. This is an effect that holds only for antisymmetric wave functions; for bosons, the ground state of free particles in a box (with appropriate boundary conditions) has kinetic energy $O(N)$. Thus, one considers a system with very high densities of $O(N)$. Then, naturally, the interaction energy per particle is $O(N)$, independent of the long-range properties of the interaction potential. One can also see this by considering the mean-field interaction term $\big(\,| \cdot |^{-1} \star \rho_N^t\big)$ from the fermionic Hartree equation for Coulomb interaction and the example of the ground state of free particles in a box, as in \eqref{coulomb_gs}. Here, $\rho_N^t = \frac{N}{L^3} \id_{V_L}$, that is, the density is proportional to $N$. Then
\be
\left(| \cdot |^{-1} \star \rho_N^t\right)(x) \leq \int_{\RRR^3} |y|^{-1} \frac{N}{L^3} \id_{V_L}(y) \,d^3y \propto N \int_0^L r^{-1} r^2 dr \propto N.
\ee
Thus, the total interaction energy is of $O(N^2)$. Together with the prefactors $N^{-\frac{2}{3}}$ and $N^{-1}$ from \eqref{Schr_scaled_sc}, both the kinetic term and the interaction term in the Schr\"odinger equation are of $O(N)$ (and in the Hartree equation \eqref{hartree_scaled_sc}, both terms are of $O(1)$). For large $N$, one can expect the mean-field approximation to be good, since each particle ``feels'' the interaction with $O(N)$ other particles, but only weakly, with ``strength'' $O(N^{-1})$.

A crucial difference to the case described in Chapter~\ref{sec:mf_fermions_const_E_kin_physics} is that in Equations~\eqref{Schr_scaled_sc} and \eqref{hartree_scaled_sc} there is an additional $N^{-\frac{1}{3}}$ in front of the time derivative. Heuristically, this factor should be there because the kinetic energy is $O(N^{\frac{5}{3}})$ and thus the average velocity per particle is $O(N^{\frac{1}{3}})$. That means, the particles are so fast that already after times of $O(N^{-\frac{1}{3}})$ they interacted with all other particles (recall that the size of the system is $N$-independent). For large $N$, the factor leads naturally to a semiclassical structure of the wave function. Formally, such a wave function is characterized by ``very small'' $\hbar$. If one sets $\varepsilon_N = N^{-\frac{1}{3}}$, then the Schr\"odinger equation \eqref{Schr_scaled_sc} is
\be\label{Schr_scaled_sc_hbar}
i \varepsilon_N \partial_t \psi^t = \left( \sum_{j=1}^N \left( -\varepsilon_N^2 \Delta_{x_j} + w^{(N)}(x_j) \right) + N^{-1} \sum_{i<j} v(x_i-x_j) \right) \psi^t,
\ee
i.e., the $\varepsilon_N$ appears exactly where the $\hbar$ would appear in the Schr\"odinger equation (in SI units). Considering very large $N$ means thus considering ``very small'' $\hbar$, hence the semiclassical structure. According to \cite{benedikter:2013}, the semiclassical structure can be characterized on the level of reduced one-particle density matrices $\mu^{\psi}_1$.\footnote{$\mu^{\psi}_1$ is defined by its integral kernel, \be
\mu^{\psi}_1(x;y) = \int d^3x_2\ldots d^3x_N \, \psi(x,x_2,\ldots,x_N) \psi^*(y,x_2,\ldots,x_N),
\ee
for more details see Chapter~\ref{sec:density_matrices}.} That is, the integral kernel of $\mu^{\psi}_1$ has roughly the form
\be\label{kernel_sc}
\mu^{\psi}_1(x;y) \approx \phi\left(N^{\frac{1}{3}}(x-y)\right) \chi(x+y).
\ee
This form expresses that the density profile $\mu^{\psi}_1(x;x)\approx \chi(2x)$ has a structure on an $N$-independent scale, thus it does not ``vary too much''. Furthermore, the ``velocity profile'', here approximately given by $\phi$, contains an extra $N^{\frac{1}{3}}$, expressing that the particles move very fast, in accordance with the kinetic energy per particle being $O(N^{\frac{2}{3}})$.

That the physics described by Equation~\eqref{Schr_scaled_sc} is, in a certain sense, close to classical physics can best be seen from the fact that the solutions to the Schr\"odinger equation \eqref{Schr_scaled_sc} are close to solutions to the classical Vlasov equation (as first discussed in \cite{narnhofer:1981}). The Vlasov equation is the classical mean-field equation
\be\label{vlasov}
\partial_t \rho^t(x,p) + p \cdot \nabla_x \rho^t(x,p) =  \Big( \nabla v \star \rho^t \Big)(x) \cdot \nabla_p \rho^t(x,p),
\ee
where $\rho^t(x,p)$ is the classical phase space density, and $v$ the classical interaction potential. In more detail, the Wigner transform of a solution to \eqref{Schr_scaled_sc} (which is a good quantity that can be compared to classical phase space densities),
\be
W_{\psi}(x,p) = (2\pi)^{-3} \int \mu_1^{\psi}(x+\varepsilon_N\frac{y}{2};x-\varepsilon_N\frac{y}{2}) e^{-ipy} \, d^3y,
\ee
is close to a solution $\rho^t(x,p)$ to the classical Vlasov equation. Still, the fermionic Hartree equations~\eqref{hartree_scaled_sc} are a better approximation to the Schr\"odinger dynamics \eqref{Schr_scaled_sc}, so the Hartree equations here describe quantum corrections to the Vlasov dynamics. (Note, in contrast, that the solutions to the equations \eqref{scaling_Schr} are in general not close to any classical dynamics.)

The main application of the discussed scaling and the mean-field equation \eqref{hartree_scaled_sc} are systems of gravitating fermions \cite{narnhofer:1981}. More generally, it should be applicable to certain ``high density'' situations.

Finally, let us discuss the connection between the semiclassical scaling and the one from Chapter~\ref{sec:mf_fermions_const_E_kin_physics}. The connection can best be seen for the case of Coulomb interaction. So let us compare the solutions $\psi^t$ to (for simplicity, we do not write out external fields)
\be\label{scaling_Schr_moi}
i \partial_t \psi^t(x_1,\ldots,x_N) = \left( -\sum_{j=1}^N \Delta_{x_j} + N^{-\frac{2}{3}} \sum_{i<j} \frac{1}{|x_i-x_j|} \right) \psi^t(x_1,\ldots,x_N),
\ee
with the solutions $\psi^t_{\semicl}$ to
\be\label{scaling_Schr_sc}
i N^{-\frac{1}{3}} \partial_t \psi^t_{\semicl}(x_1,\ldots,x_N) = \left( - \sum_{j=1}^N N^{-\frac{2}{3}} \Delta_{x_j} + N^{-1} \sum_{i<j} \frac{1}{|x_i-x_j|} \right) \psi^t_{\semicl}(x_1,\ldots,x_N).
\ee
As discussed, we assume that $\psi^t$ lives in a volume proportional to $N$, while $\psi^t_{\semicl}$ lives in an $N$-independent volume. For both $\psi^t$ and $\psi^t_{\semicl}$ one could expect mean-field behavior for times of $O(1)$. Let us now rescale $x \to N^{\frac{1}{3}}x$ and $t \to N^{\frac{1}{3}}t$, and consider the (from $\psi^t$) rescaled wave function $\tilde{\psi}^t$. This wave function then lives in an $N$-independent volume. It is a solution to the rescaled Schr\"odinger equation
\be\label{scaling_Schr_moi_scaled}
i N^{-\frac{1}{3}} \partial_t \tilde{\psi}^t(x_1,\ldots,x_N) = \left( -\sum_{j=1}^N N^{-\frac{2}{3}} \Delta_{x_j} + N^{-1} \sum_{i<j} \frac{1}{|x_i-x_j|} \right) \tilde{\psi}^t(x_1,\ldots,x_N).
\ee
This is exactly \eqref{scaling_Schr_sc}. However, recall that we expected mean-field behavior for $\psi^t$ for times of $O(1)$, i.e., we can expect mean-field behavior for $\tilde{\psi}^t$ only for times of $O(N^{-\frac{1}{3}})$. This is due to the fact that $\tilde{\psi}^t$, in contrast to $\psi^t_{\semicl}$, does not naturally have a semiclassical structure. Recall that the semiclassical structure meant that $\psi^t_{\semicl}$ has a density that ``varies on an $N$-independent scale'' (see \eqref{kernel_sc}). In contrast, it was natural to assume that the density of $\psi^t$ ``varies on an $N$-independent scale'', i.e., the density of $\tilde{\psi}^t$ ``varies on a scale $N^{-\frac{1}{3}}$''. Therefore, we can only expect mean-field behavior for short times of $O(N^{-\frac{1}{3}})$.

Although this semiclassical scaling is not the focus of this work (it has recently been treated in \cite{benedikter:2013}), we show in Chapter~\ref{sec:estimates_semiclassical} how the mean-field dynamics \eqref{hartree_scaled_sc} can be derived with the $\alpha$-method used in this work (the full proof can be found in Appendix~\ref{sec:proof_sc_scaling}).

\section{Literature}\label{sec:mf_fermions_lit}
The scaling from Chapter~\ref{sec:mf_fermions_const_E_kin_physics}, to the author's knowledge, has not been considered in the literature before for a derivation of mean-field dynamics. Thus far, only the case where the Schr\"odinger equation is
\be\label{Schr_scaled_beta1}
i \partial_t \psi^t(x_1,\ldots,x_N) = \left( \sum_{j=1}^N \left( -\Delta_{x_j} + w^{(N)}(x_j) \right) + N^{-1} \sum_{i<j} v(x_i-x_j) \right) \psi^t(x_1,\ldots,x_N),
\ee
has been considered, i.e., the case $\beta=1$. For Coulomb interaction and kinetic energy of $O(N)$, the interaction is thus scaled down by a factor $N^{\frac{1}{3}}$ too much; the interaction energy per particle is of $O(N^{-\frac{1}{3}})$, which leads to free evolution in the limit of large $N$. We show this explicitly in Proposition~\ref{pro:coulombN1}. The case $\beta=1$ could for example be interesting for systems with kinetic energy of $O(N)$, when the interaction does not go to zero at all for large $|x|$, e.g., $v(x) = \cos(|x|)$. The first result for $\beta=1$ was achieved in \cite{bardos:2003} where bounded $v$ are treated (see also the related works \cite{bardos:2004} and \cite{bardos:2007}). Later, in \cite{froehlich:2011}, the mean-field dynamics was derived for a class of potentials $v$ including Coulomb interaction. Note that it was a crucial improvement to consider Coulomb interaction; first, because it is physically very relevant, and second, because for Coulomb interaction an equation of the type \eqref{Schr_scaled_beta1} follows from a rescaling as discussed in Chapter~\ref{sec:mf_fermions_scaling}.

Another situation which has been considered in the literature is when the mean-field limit is coupled to a semiclassical limit, as discussed in Chapter~\ref{sec:mf_fermions_semiclassical}. A mean-field description of the dynamics \eqref{Schr_scaled_sc} has first been considered in \cite{narnhofer:1981}. There, it is shown that in the limit $N\to\infty$ and for a class of very regular interaction potentials the solutions to \eqref{Schr_scaled_sc} converge to solutions to the Vlasov equation \eqref{vlasov} (in a suitable sense). In \cite{spohn:1981}, a similar result is shown for a more general class of interactions, with fewer regularity assumptions. Later, in \cite{erdoes:2004}, the mean-field equations \eqref{hartree_scaled_sc} are derived from \eqref{Schr_scaled_sc} for bounded analytic potentials and for short times (times of $O(1)$ but smaller than a certain constant). Unlike in \cite{narnhofer:1981} and \cite{spohn:1981}, where only the limit $N\to\infty$ is considered, explicit error terms and a convergence rate of $N^{-1}$ are given in \cite{erdoes:2004}. Recently, in \cite{benedikter:2013}, this result was shown for all times, with fewer regularity assumptions on the interaction, and, depending on the exact formulation of the result, with different convergence rates. In the work \cite{benedikter:2013}, a new method is used for the proof, which uses a Gronwall-type estimate. Note that the extension of the result to all times is a crucial improvement. In Chapter~\ref{sec:estimates_semiclassical}, we show how the main results of \cite{benedikter:2013} can be reproduced with the $\alpha$-method used in this work; the full proof is given in Appendix~\ref{sec:proof_sc_scaling}. In \cite{benedikter:2014}, the results from \cite{benedikter:2013} are generalized to Hamiltonians with pseudo-relativistic kinetic part. However, to this date, a derivation of the semiclassical Hartree equations \eqref{hartree_scaled_sc} for the important case of Coulomb interaction is still missing. 

Finally, let us remark that the fermionic Hartree and Hartree-Fock equations are also widely used in the time-independent version, mostly to calculate ground states. The time-independent equations were actually originally considered by Hartree \cite{hartree:1928}, Fock \cite{fock:1930} and Slater \cite{slater:1930} (apparently Dirac \cite{dirac:1930} first wrote down the time-dependent version). Later, several properties of these equations were discussed and rigorously proven; e.g., in \cite{lieb:1974, lieb:1977, lions:1987, fefferman:1990, bach:1992, bach:1993, graf_solovej:1994, fefferman:1994}, about existence, uniqueness and properties of the solutions to the Hartree-Fock equations for atoms and molecules, and that the mean-field approximation indeed gives asymptotically correct ground state energies for large-$Z$ atoms and molecules.

\section{Fluctuations}\label{sec:mf_fermions_const_E_kin_fluc}
Let us come back here to the question whether the mean-field dynamics can be expected to be a good approximation to the Schr\"odinger dynamics. We already mentioned in the introduction that this can only hold if the particles develop ``few'' correlations due to the interaction. Let us suppose that initially the wave function has antisymmetric product structure, $\psi^0 = \bigwedge_{j=1}^N \varphi_j^0$, and that $N$ is very large. Then, in a more detailed physical picture, each particle ``feels'', on the one hand, a mean interaction coming from all the other particles, but, on the other hand, it also ``feels'' deviations from this mean interaction, i.e., fluctuations around the mean-field. These fluctuations can cause deviations from the mean-field dynamics. One can also think of this from the perspective of the law of large numbers. Suppose that $N$ particles $X_1,\ldots,X_N$ are distributed according to the density $\rho_N^t = \sum_{i=1}^N |\varphi_i^t|^2$, coming from the solutions to the fermionic Hartree equations. Then, typically, their contribution to the interaction at point $y\in \RRR^3$ is close to its mean value, i.e., 
\be
\sum_{k=1}^{N} v\left( y - X_k \right) \approx \int_{\RRR^3} v(y-x) \rho_N^t(x) \, d^3x = \left( v \star \rho_N^t \right)(y),
\ee
only if fluctuations are small. If they are not small, then we can not expect mean-field behavior. An example where the fluctuations are not small is Brownian motion. There, on the right scales, the mean-field prediction is wrong, and one sees instead a diffusive motion, coming from the fluctuations.

Let us now calculate the fluctuations around the mean-field, assuming an antisymmetric product state $\bigwedge_{j=1}^N \varphi_j$. First, the expectation value of the interaction potential at point $y \in \RRR^3$ is, of course, given by
\begin{align}\label{expectation}
\EEE\left( \sum_{k=1}^N v(x_k-y) \right) &:= \bigSCP{\bigwedge_{j=1}^N \varphi_j}{\left( \sum_{k=1}^N v(x_k-y) \right) \, \bigwedge_{j=1}^N \varphi_j} \nonumber \\
&= \sum_{j=1}^N \scp{\varphi_j}{v(\cdot - y) \varphi_j} \nonumber \\
&= (v \star \rho_N)(y),
\end{align}
i.e., the mean-field from the fermionic Hartree equations. The fluctuations around the mean-field at $y \in \RRR^3$ are given by the variance
\begin{align}\label{variance}
&\Var\left( \sum_{k=1}^N v(x_k-y) \right) \nonumber \\ 
&\qquad:= \EEE\left( \left(\sum_{k=1}^N v(x_k-y)\right)^2 \right) - \left(\EEE\left( \sum_{k=1}^N v(x_k-y) \right)\right)^2 \nonumber \\
&\qquad= N(N-1) \bigSCP{\bigwedge_{j=1}^N \varphi_j}{ v(x_1-y)v(x_2-y) \bigwedge_{j=1}^N \varphi_j} \nonumber \\
&\qquad\quad + N \bigSCP{\bigwedge_{j=1}^N \varphi_j}{ v(x_1-y)^2 \bigwedge_{j=1}^N \varphi_j} - N^2 \bigSCP{\bigwedge_{j=1}^N \varphi_j}{ v(x_1-y) \bigwedge_{j=1}^N \varphi_j}^2 \nonumber \\
&\qquad= \sum_{i,j=1}^N \bigg( \scp{\varphi_i}{v(\cdot-y)\varphi_i} \scp{\varphi_j}{v(\cdot-y)\varphi_j}  - \scp{\varphi_i}{v(\cdot-y)\varphi_j} \scp{\varphi_j}{v(\cdot-y)\varphi_i} \bigg) \nonumber \\
&\qquad\quad + (v^2\star\rho_N)(y) - (v\star\rho_N)^2(y) \nonumber \\
&\qquad= (v^2\star\rho_N)(y) - \sum_{i,j=1}^N \left\lvert \scp{\varphi_i}{v(\cdot-y)\varphi_j} \right\rvert^2 \nonumber \\
&\qquad\leq (v^2 \star \rho_N)(y).
\end{align}
Only if this variance is small enough, one can hope the mean-field approximation to hold. (One has to be a bit cautious here, since one additionally has to consider the time-scales on which the fluctuations happen.) We here supposed that the wave function is in an antisymmetric product state. If this is not the case, e.g., if after some time already severe correlations have developed, then typically even more correlations will develop, since many particles are not in the antisymmetric product structure anymore.

The estimates we present later capture the presented physical picture very nicely. One of our main results, Theorem~\ref{thm:estimates_terms_alpha_dot_beta_n}, holds exactly under the assumption that $(v^2 \star \rho_N)(y) \leq CN^{-1}$, i.e., when the fluctuations are very small, of $O(N^{-1})$. We come back to this point in Remark~\ref{itm:fluctuations} in Chapter~\ref{sec:main_theorem_mf_general_v}.

Finally, let us see what happens in the situation of Chapter~\ref{sec:mf_fermions_const_E_kin_physics} where there is a factor $N^{-\beta}$ in front of the interaction. We already saw, that one effect of this factor is that the kinetic energy becomes of the same order as the interaction energy. If we replace $v$ by $N^{-\beta}v$ in \eqref{expectation}, then the mean-field interaction is $O(1)$ (which, in the setting of Chapter~\ref{sec:mf_fermions_const_E_kin_physics}, is the same order as the kinetic term in the fermionic Hartree equations). However, what is crucial is that the $N^{-\beta}$ also makes the fluctuations small: from \eqref{variance} we can read off, that fluctuations are bounded by $N^{-2\beta}(v^2 \star \rho_N)(y)$ (we discuss the size of this term in Remarks~\ref{itm:Coulomb_Hardy} and \ref{itm:coulomb_nice} in Chapter~\ref{sec:main_theorem_mf_general_v}). Without the $N^{-\beta}$, the fluctuations would in general not be small.

\section{The Exchange Term}\label{sec:exch_term}
In this thesis we study the fermionic Hartree equations as mean-field dynamics for fermions. Another related dynamics is given by the Hartree-Fock equations. These are the coupled system of non-linear equations (here given without scaling)
\be\label{Hartree-Fock}
i \partial_t \varphi_j^t(x) = \Big( -\Delta + w(x) \Big) \varphi_j^t(x) + \sum_{k=1}^N \Big(v \star |\varphi_k^t|^2 \Big)(x) \, \varphi_j^t(x) - \sum_{k=1}^N \Big(v \star ({\varphi_k^t}^*\varphi_j^t) \Big)(x) \, \varphi_k^t(x),
\ee
for $j=1,\ldots,N$. In comparison to the fermionic Hartree equations, the Hartree-Fock equations contain an additional ``exchange term''. In general, the dynamics \eqref{Hartree-Fock} is expected to be a better approximation to the Schr\"odinger dynamics than the Hartree dynamics. However, for the situations considered in this work, the exchange term is of smaller order in $N$ than the direct term $v\star\rho_N^t$. Therefore, we consider only the fermionic Hartree dynamics. In the following, let us briefly discuss where the exchange term comes from and then argue why it is subleading in $N$ for the scaled equations considered in Chapter~\ref{sec:mf_fermions_const_E_kin_physics}.

Let us start by considering the Schr\"odinger dynamics with Hamiltonian 
\be
H = \sum_{j=1}^N \Big( -\Delta_{x_j} + w(x_j) \Big) + \sum_{i<j} v(x_i-x_j).
\ee
For fermions, the most simple structure of a wave function is an antisymmetrized product state $\bigwedge_{j=1}^N \varphi_j^t$. Let us now suppose that it is reasonable to approximate the Schr\"odinger wave function by \emph{some} antisymmetrized product state $\bigwedge_{j=1}^N \varphi_j^t$. How do we find a good evolution equation for $\varphi_1^t,\ldots,\varphi_N^t$? One way is to demand that the evolution should be such, that (the expectation value of) the total energy of the wave function $\bigwedge_{j=1}^N \varphi_j^t$ is preserved. It turns out that this job is done by the Hartree-Fock equations \eqref{Hartree-Fock}. Let us explain in more detail. The expectation value of the total energy for the wave function $\bigwedge_{j=1}^N \varphi_j^t$ is given by $\SCP{\bigwedge_{j=1}^N \varphi_j^t}{H\, \bigwedge_{j=1}^N \varphi_j^t}$. A straightforward calculation shows that 
\begin{align}\label{expec_H}
E^t := \bigSCP{\bigwedge_{j=1}^N \varphi_j^t}{H\, \bigwedge_{j=1}^N \varphi_j^t} &= \sum_{j=1}^N \int d^3x ~ {\varphi_j^t(x)}^* \big(-\Delta + w(t,x)\big)\,\varphi_j^t(x) \nonumber \\
&\quad+ \frac{1}{2} \sum_{j,k=1}^N \int d^3x\int d^3y ~ \lvert\varphi_j^t(x)\rvert^2 \,v(x-y)\, \lvert\varphi_k^t(y)\rvert^2 \nonumber \\
&\quad- \frac{1}{2} \sum_{j,k=1}^N \int d^3x\int d^3y ~ {\varphi_j^t(x)}^*\varphi_k^t(x) \,v(x-y)\, {\varphi_k^t(y)}^*\varphi_j^t(y).
\end{align}
It is then easy to check that $\partial_t E^t = 0$, if $\varphi_1^t,\ldots,\varphi_N^t$ are solutions to the Hartree-Fock equations \eqref{Hartree-Fock} (see, e.g., \cite{chadam:1975}). Thus, the Hartree-Fock evolution ensures that, if the approximation of the wave function by $\bigwedge_{j=1}^N \varphi_j^t$ is justified, then $\bigwedge_{j=1}^N \varphi_j^t$ is close to the energy of the solution to the Schr\"odinger equation.

Note that the variation of the energy functional $E^t = E(\varphi_1^t,\ldots,\varphi_N^t)$ leads to the time-independent Hartree-Fock equations, i.e., \eqref{Hartree-Fock} with $i\partial_t$ replace by a constant $e$.

Let us now discuss why the exchange term is subleading for the scaled equations considered in Chapter~\ref{sec:mf_fermions_const_E_kin_physics}. Of course, the best justification that the exchange term is subleading for the dynamics is given by the main results in Chapter~\ref{sec:main_theorem_mf}, which show that already the fermionic Hartree equations are a good approximation to the Schr\"odinger dynamics. Nevertheless, let us again regard the simple example of the non-interacting ground state in a box  $V_L = \big[-\frac{L}{2},\frac{L}{2}\big]^3$ from around Equation~\eqref{free_particles_box}. For such a wave function and Coulomb interaction, let us now estimate the order in $N$ of the total exchange energy $-\frac{1}{2}\sum_{j,k=1}^N \scp{\varphi_j}{\big(v \star ({\varphi_k}^*\varphi_j) \big)\, \varphi_k}$, i.e., the last term on the right-hand side of \eqref{expec_H}. (Keep in mind, that in the scaled equation there would be an additional $N^{-\beta}$ in front of the exchange term.) It turns out that already for this simple example an exact calculation is quite hard to perform for singular interaction potentials, like Coulomb interaction. Therefore, we give here a heuristic estimate (we use the $\lessapprox$ sign to indicate that an estimate is heuristic). In the following, recall that the Fourier transform of $|x|^{-1}$ is given by $const \cdot |k|^{-2}$ (in the sense of convolutions, see \cite{liebloss:2001} for more details). Also, recall that we number the $k_i \in \ZZZ^3$ such that $|k_N|$ is as small as possible, actually $|k_N|< const \cdot N^{\frac{1}{3}}$. In the following, $C$ denotes a constant (independent of $N$) that can be different from line to line. We find
\begin{align}
\sum_{i,j=1}^N \scp{\varphi_j}{\Big(v \star ({\varphi_i}^*\varphi_j) \Big)\, \varphi_i} &= L^{-6} \int_{V_L} d^3x~ \int_{V_L} d^3y~ \frac{1}{|x-y|} \sum_{i,j=1}^N e^{i\frac{2\pi}{L} (k_i-k_j)(x-y)} \nonumber \\
&\lessapprox L^{-6} \int_{V_L} d^3x~ \int_{\RRR^3} d^3y~ \frac{1}{|x-y|} \sum_{i,j=1}^N e^{i\frac{2\pi}{L} (k_i-k_j)(x-y)} \nonumber \\
&= L^{-6} \int_{V_L} d^3x~ \int_{\RRR^3} d^3z~ \frac{1}{|z|} \sum_{i,j=1}^N e^{i\frac{2\pi}{L} (k_i-k_j)z} \nonumber \\
&= L^{-3} \int_{\RRR^3} d^3z~ \frac{1}{|z|} \sum_{i,j=1}^N e^{i\frac{2\pi}{L} (k_i-k_j)z} \nonumber \\
&= L^{-3} \int_{\RRR^3} d^3z'~L^3 \frac{1}{L|z'|} \sum_{i,j=1}^N e^{i2\pi (k_i-k_j)z'} \nonumber \\
&\leq C L^{-1} \sum_{i,j=1}^N \frac{1}{|k_i-k_j|^2} \nonumber \\
&\approx C L^{-1} \int_{\left[0,N^{\frac{1}{3}}\right]^3} d^3k_1 \int_{\left[0,N^{\frac{1}{3}}\right]^3} d^3k_2 ~ \frac{1}{|k_1-k_2|^2} \nonumber \\
&\leq C L^{-1} N \int_0^{N^{\frac{1}{3}}} r^2 dr\, \frac{1}{r^2} \nonumber \\
&\leq C L^{-1} N^{\frac{4}{3}}.
\end{align}
Thus, recalling that $L\propto N^{\frac{1}{3}}$, we see that the total unscaled exchange energy is $O(N)$, so that the exchange term for each particle roughly gives a contribution of $O(1)$ to the dynamics. If we now take the scaling factor $\beta=\frac{2}{3}$ into account, we see that the contribution of the exchange term to the mean-field dynamics is $O(N^{-\frac{2}{3}})$, i.e., subleading in $N$. Note that we expect this behavior not just for plane waves, but for a much larger class of wave functions. (Concerning the exchange term the plane waves are not so special, since they have large overlap. Note that if all the orbitals would have disjoint support with each other, then the exchange term would vanish.) We come back to the role of the exchange term concerning the convergence rates between the Schr\"odinger and the mean-field time evolution in Remark~\ref{itm:exch}, following Theorem~\ref{thm:estimates_terms_alpha_dot_beta_n}.

Finally, note that for bounded interaction potentials $v$, it is easy to see on a heuristic level that the exchange term is small. Let us suppose that the orthonormal $\varphi_1,\ldots,\varphi_N$ are approximately a basis of $L^2(\RRR^3)$, i.e., $\sum_{i=1}^N \ketbr{\varphi_i} \to \id$, or $\sum_{j=1}^N \varphi_j(x)^*\varphi_j(y) \approx \delta(x-y)$ for large $N$. Then we have (recall that we are not mathematical precise at this point)
\begin{align}
\sum_{i,j=1}^N \scp{\varphi_j}{\Big(v \star ({\varphi_i}^*\varphi_j) \Big)\, \varphi_i} &= \int d^3x~ \int d^3y~ \sum_{j=1}^N \varphi_j(x)^*\varphi_j(y) \sum_{i=1}^N \varphi_i(x)^* v(x-y) \varphi_i(y) \nonumber \\
&\approx \int d^3x~ \int d^3y~ \delta(x-y) \sum_{i=1}^N \varphi_i(x)^* v(x-y) \varphi_i(y) \nonumber \\
&= v(0) \int d^3x~ \sum_{i=1}^N \varphi_i(x)^* \varphi_i(x) \nonumber \\
&= N \, v(0).
\end{align}
The exchange term in the unscaled Hartree-Fock equations is thus $O(1)$, i.e., with a scaling $N^{-\beta}$ in front of $v$, it is $O(N^{-\beta})$.

\chapter{Mathematical Results}\label{sec:main_results}
\section{The Counting Functional}\label{sec:counting_functional}
We first introduce the precise meaning of $\approx$ in $\psi \approx \bigwedge_{j=1}^N\varphi_j$. This is done via the functional $\alpha_f(\psi,\varphi_1,\ldots,\varphi_N)$. We want this $\alpha_f$ to be such that $\alpha_f=0$ for $\psi = \bigwedge_{j=1}^N\varphi_j$, and $\alpha_f=1$ for $\psi = \bigwedge_{j=1}^N\chi_j$, where $\chi_i$ is orthogonal to $\varphi_j$ for all $i,j$. In other words, $\alpha_f=0$ should mean that the approximation of $\psi$ by $\bigwedge_{j=1}^N\varphi_j$ is exact, while $\alpha_f=1$ should mean that this approximation is not valid at all. So $\alpha_f$ is supposed to measure the closeness of $\psi$ to the \emph{specific} antisymmetrized product of the $\varphi_1,\ldots,\varphi_N$. Furthermore, we want $\alpha_f$ to measure those parts of $\psi$ that ``do not contain'' $\varphi_1,\ldots,\varphi_N$. Loosely speaking, it should count how many particles are not in the antisymmetrized product structure (hence the name ``counting functional'').

We now define $\alpha_f$ and several projectors that are needed for its definition. In the following, we denote by $\SCP{\cdot}{\cdot}$ the scalar product on $L^2(\RRR^{3N})$ while $\scp{\cdot}{\cdot}$ denotes the scalar product on $L^2(\RRR^3)$.

\begin{definition}\label{def:projectors}
Let $\varphi_1, \ldots, \varphi_N \in L^2(\RRR^3)$ be orthonormal.
\begin{enumerate}[(a)]
\item For all $j,m = 1,\ldots,N$ we define the projector
\be
p_m^{\varphi_j} := \ketbr{\varphi_j}_m = \ketbr{\varphi_j(x_m)} = \underbrace{\id \otimes \ldots \otimes \id}_{m-1 ~ \text{times}} \otimes \,\ketbra{\varphi_j}{\varphi_j}\, \otimes \underbrace{\id \otimes \ldots \otimes \id}_{N-m ~ \text{times}},
\ee
i.e., its action on any $\psi \in L^2(\RRR^{3N})$ is given by
\be
\left(p_m^{\varphi_j}\psi\right)(x_1,\ldots,x_N) = \varphi_j(x_m)\int\varphi_j^*(x_m)\psi(x_1,\ldots,x_N) \, d^3x_m.
\ee
We define
\be
p_m := \sum_{j=1}^N p_m^{\varphi_j},
\ee
and
\be
q_m := 1-p_m.
\ee
\item For any $0 \leq k \leq N$ we define
\be
P_{N,k} = P_{N,k}^{ \, \varphi_1,\ldots,\varphi_N } := \left( \prod_{m=1}^k q_m \prod_{m=k+1}^N p_m \right)_{\sym} = \sum_{\vec{a} \in \AAA_k} \prod_{m=1}^N (p_m)^{1-a_m}(q_m)^{a_m}
\ee
with the set
\be
\AAA_k := \left\{ \vec{a}=(a_1,\ldots,a_N) \in \{0,1\}^N: \sum_{m=1}^N a_m=k \right\},
\ee
i.e., $P_{N,k}$ is the symmetrized tensor product of $q_1,\ldots,q_k,p_{k+1},\ldots,p_N$. We define $P_{N,k}=0$ for all $k<0$ and $k>N$.
\item We call any $f:\{0,\ldots,N\} \to [0,1]$ with $f(0)=0$, $f(N)=1$ a \emph{weight function}. For any weight function $f$ we define the operators
\be
\widehat{f} = \widehat{f}^{\,\varphi_1,\ldots,\varphi_N} := \sum_{k=0}^N f(k) P_{N,k}.
\ee
For any $d \in \ZZZ$, we define the shifted operators
\be
\widehat{f}_d := \sum_{k=-d}^{N-d} f(k+d) P_{N,k} = \sum_{k=0}^N f(k) P_{N,k-d} = \sum_{k=0}^N f(k+d) P_{N,k},
\ee
where for the last expression we defined $f(k)=0$ for all $k<0$ and $k>N$.
\item For any normalized $\psi \in L^2(\RRR^{3N})$ we define
\be\label{definition_alpha}
\alpha_f = \alpha_f(\psi,\varphi_1,\ldots,\varphi_N) := \SCP{\psi}{\widehat{f} \, \psi} = \sum_{k=0}^N f(k) \, \SCP{\psi}{P_{N,k} \psi}.
\ee
\end{enumerate}
\end{definition}

The functional $\alpha_f$ and the projectors from Definition~\ref{def:projectors} have first been introduced by Pickl \cite{pickl:2011method} for bosons, that is, with $p_m = \ketbr{\varphi}_m$. The functional was used in \cite{pickl:2011method,pickl:2010hartree} for the derivation of the bosonic Hartree equation, and in \cite{pickl:2010gp_ext,pickl:2010gp_pos} for the derivation of the Gross-Pitaevskii equation. Let us note here that for fermions $\alpha_f$ with the weight function $\frac{k}{N}$ has been used before by Graf and Solovej \cite{graf_solovej:1994} and Bach \cite{bach:1993} to measure deviation from the antisymmetrized product structure in the static setting; see also the remarks following \eqref{tr_mu_q_1_eq_alpha}.

Let us now explain these definitions a little further. (We give more details in Chapter~\ref{sec:properties_projectors}.) When regarded as operators on $L^2(\RRR^3)$, $p_1$ projects on the subspace spanned by $\varphi_1,\ldots,\varphi_N$, and $q_1$ projects on its complement. Therefore, $p_1q_1=0$. Note that $p_1$ and $q_1$ are indeed projectors, since the $\varphi_1,\ldots,\varphi_N$ are assumed to be orthonormal. One can then easily check that also the operators $P_{N,k}$ are projectors. Let us now consider the definition of $\alpha_f$ from \eqref{definition_alpha}. Heuristically, the scalar product $\SCP{\psi}{P_{N,k} \psi}$ gives a big contribution if $k$ of the orbitals $\varphi_1,\ldots,\varphi_N$ are \emph{not} contained in the wave function $\psi$. In other words, $P_{N,k}$ projects on those wave functions which are missing $k$ of the orbitals $\varphi_1,\ldots,\varphi_N$. Indeed, one finds for example for a wave function $\phi_\ell=\bigwedge_{j=1}^\ell \chi_j \bigwedge_{j=\ell+1}^N\varphi_j$ with $\chi_i \perp \varphi_j \forall i,j$, that $P_{N,k}\phi_\ell = \delta_{k \ell}\,\phi_k$. As we show in Lemma~\ref{lem:properties_P_Nk}, the $P_{N,k}$ have the property that $\sum_{k=0}^N P_{N,k} = 1$, i.e., we can define the decomposition $\psi = \sum_{k=0}^N \psi_k$, with $\psi_k=P_{N,k}\psi$. Then, loosely speaking, each $\psi_k$ has $(N-k)$ particles in one of the orbitals $\varphi_1,\ldots,\varphi_N$ and $k$ particles \emph{not} in the orbitals $\varphi_1,\ldots,\varphi_N$.

The function $f(k)$ determines how much weight is given to the contribution coming from each $P_{N,k}\psi$. By choosing $f(k)$ we can thus fine tune what is meant by closeness of $\psi$ to $\bigwedge_{j=1}^N\varphi_j$. One obvious and very simple weight is the relative number $\frac{k}{N}$. We always denote this weight function by
\be\label{weight_n}
n(k) = \frac{k}{N}
\ee
and the corresponding counting functional by $\alpha_n$. Loosely speaking, $\alpha_n$ measures the relative ``number of particles'' in $\psi$ that are not in the antisymmetrized product structure of the $\varphi_1,\ldots,\varphi_N$. It turns out that for this weight and due to the antisymmetry of $\psi$, the functional has the simple form
\be
\alpha_n := \sum_{k=0}^N \frac{k}{N} \SCP{\psi}{P_{N,k} \psi} = \SCP{\psi}{q_1 \psi},
\ee
see Lemma~\ref{lem:properties_P_Nk}. Recall here that $q_1$ projects on the complement of the subspace spanned by $\varphi_1,\ldots,\varphi_N$. Another important weight is given by
\be\label{weight_m_gamma}
m^{(\gamma)}(k) = \left\{\begin{array}{cl} \frac{k}{N^{\gamma}} &, \text{for } k \leq N^{\gamma}\\ 1 & , \text{otherwise,} \end{array}\right.
\ee
with some $0 < \gamma \leq 1$. The function $m^{(\gamma)}(k)$ gives a much larger weight to already very few particles outside the antisymmetrized product structure. On the other hand, for $k>N^{\gamma}$, i.e., very many particles outside the antisymmetrized product structure, $m^{(\gamma)}(k)$ gives the same weight $1$ for all $k>N^{\gamma}$. These properties enable us to derive mean-field approximations for a much wider range of physical situations. 

The goal of this work is to prove bounds on $\alpha_f\big(\psi^t,\varphi_1^t,\ldots,\varphi_N^t\big)$, where $\psi^t$ is a solution to the Schr\"odinger equation and $\varphi^t_1,\ldots,\varphi^t_N$ are solutions to the fermionic Hartree equations. In more detail, we first look for a bound of the type
\be\label{general_form_bound_alpha_dot}
\partial_t \alpha_f\big(\psi^t,\varphi_1^t,\ldots,\varphi_N^t\big) \leq C(t) \left( \alpha_f\big(\psi^t,\varphi_1^t,\ldots,\varphi_N^t\big) + N^{-\delta} \right),
\ee
which then, by Gronwall's Lemma (see Lemma~\ref{lem:gronwall}), implies the bound
\be\label{general_form_bound}
\alpha_f\big(\psi^t,\varphi_1^t,\ldots,\varphi_N^t\big) \leq \, e^{\int_0^t C(s) ds} \, \left( \alpha_f\big(\psi^0,\varphi_1^0,\ldots,\varphi_N^0\big) + N^{-\delta} \right),
\ee
where the function $C(t)$ is independent of $N$, and $\delta>0$ is called the \emph{convergence rate}. In the main theorems of Chapter~\ref{sec:main_theorem_mf}, the weight function $f$ is either $n$ from \eqref{weight_n} or $m^{(\gamma)}$ from \eqref{weight_m_gamma}. A bound of the form \eqref{general_form_bound} implies that if initially (at time $t=0$) $\alpha_f$ is small, then it stays small for times $t>0$ and $N$ large enough. In the thermodynamic limit, we arrive at the statement that $\lim_{N\to\infty}\alpha_f(t=0)=0$ implies $\lim_{N\to\infty}\alpha_f(t)=0$ for all $t>0$.

Let us summarize the advantages of using the functional $\alpha_f$ for the derivation of mean-field dynamics compared to other approaches:
\begin{itemize}
\item The idea to ``count the number of particles'' not in the antisymmetrized product seems very natural and has a clear physical interpretation. This is also reflected in the proof of a statement like \eqref{general_form_bound_alpha_dot}. As we show in Chapter~\ref{sec:outline_proof}, there are three contributions to $\partial_t \alpha_f(t)$, all of which have a clear physical meaning.
\item It seems that proofs which use BBGKY hierarchies (e.g., \cite{bardos:2003}) are hard to formulate for interactions with scalings weaker than $N^{-1}$, due to combinatorial reasons. Therefore, new methods like the $\alpha$-method or the one developed by Schlein et al.\ (applied to fermions by Benedikter, Porta and Schlein in \cite{benedikter:2013}) are useful.
\item The freedom in the choice of the weight function enables us to prove mean-field dynamics for many different setups, e.g., singular or weakly scaled interaction potentials.
\end{itemize}

\section{Connection to Density Matrices}\label{sec:dens_mat_summary}
It turns out that the functional $\alpha_f$ is closely related to the trace-norm of the difference between reduced one-particle density matrices. Let us here explain the relation and state the exact results; we give more technical details and the proofs in Chapter~\ref{sec:density_matrices}. For any normalized antisymmetric $\psi \in L^2(\RRR^{3N})$, the reduced one-particle density matrix is defined by its integral kernel
\be\label{definition_dens_mat_one_part}
\mu^{\psi}_1(x;y) = \int \psi(x,x_2,\ldots,x_N) \psi^*(y,x_2,\ldots,x_N) \,d^3x_2 \ldots d^3x_N.
\ee
For an antisymmetrized product state $\bigwedge_{j=1}^N \varphi_j$ we find
\be
\mu^{\bigwedge \varphi_j}_1 = \frac{1}{N} p_1.
\ee
Let us now consider $\alpha_n$, i.e., the $\alpha$-functional with the weight $n(k)=\frac{k}{N}$. First, let us mention that $\alpha_n = \tr(\mu_1^{\psi}q_1) = \tr(\mu_1^{\psi}(1-p_1))$, where $\tr(\cdot)$ denotes the trace. This can be seen by evaluating the trace in a basis that contains $\varphi_1,\ldots,\varphi_N$. If we denote the other basis vectors by $\{\varphi_{j}\}_{j>N}$, we find
\be\label{tr_mu_q_1_eq_alpha}
\tr(\mu_1^{\psi}q_1) = \sum_{j=1}^\infty \scp{\varphi_j}{\mu^{\psi}q_1 \varphi_j} = \sum_{j=N+1}^\infty \scp{\varphi_j}{\mu^{\psi} \varphi_j} = \SCP{\psi}{\sum_{j=N+1}^\infty \ketbr{\varphi_j}_1 \psi} = \SCP{\psi}{q_1 \psi}.
\ee
It is the expression $\tr(\mu_1^{\psi}q_1)$ that has been used before to measure deviation from the antisymmetrized product structure in the static setting, see \cite{bach:1993,graf_solovej:1994}. Now consider the difference of the reduced one-particle density matrices in trace norm,
\be\label{diff_dens_mat}
\norm[\tr]{\mu^{\bigwedge \varphi_j}_1 - \mu_1^\psi}.
\ee
Using $p_1+q_1=1$, let us decompose the reduced density of $\psi$ into four contributions,
\be
\mu_1^\psi = (p_1+q_1) \mu_1^\psi (p_1+q_1) = p_1 \mu_1^\psi p_1 + p_1 \mu_1^\psi q_1 + q_1 \mu_1^\psi p_1 + q_1 \mu_1^\psi q_1.
\ee
In the proof of Lemma~\ref{lem:density_conv}, recalling that $\alpha_n=\SCP{\psi}{q_1\psi}=||q_1\psi||^2$, we show that 
\be\label{dens_alpha_pp_qq}
\norm[\tr]{\mu^{\bigwedge \varphi_j}_1 - p_1 \mu_1^\psi p_1} = \norm[\tr]{q_1 \mu_1^\psi q_1} = \alpha_n.
\ee
Furthermore, one can show that
\be\label{dens_alpha_pq}
\norm[\tr]{p_1 \mu_1^\psi q_1} \leq \sqrt{\alpha_n} \quad \text{and} \quad \norm[\tr]{q_1 \mu_1^\psi p_1} \leq \sqrt{\alpha_n}.
\ee
Therefore,
\be
\norm[\tr]{\mu^{\bigwedge \varphi_j}_1 - \mu_1^\psi} \leq C \sqrt{\alpha_n}.
\ee
On the other hand, one can show that
\be
\alpha_n \leq 2 \norm[\tr]{\mu^{\bigwedge \varphi_j}_1 - \mu_1^\psi}.
\ee
As a consequence, convergence of $\mu_1^\psi$ to $\mu^{\bigwedge \varphi_j}_1$ in trace norm is equivalent to convergence of $\alpha_n$ to zero. However, there is a difference in the convergence rates. This comes from the fact that controlling the density matrix difference \eqref{diff_dens_mat} means to control ``more'' correlations than covered by $\alpha_n$. This can be seen from \eqref{dens_alpha_pp_qq}: with $\alpha_n$ one controls only certain ``diagonal'' parts of the density matrix difference, while the ``non-diagonal'' parts as in \eqref{dens_alpha_pq} are weighted more, with $\sqrt{\alpha_n}$. A similar analysis can be done for the Hilbert-Schmidt norm instead of the trace norm. There, due to the choice of normalization of the density matrix, an extra factor $\sqrt{N}$ appears. 

The relations between the different types of convergence are summarized in the following lemma. Recall that $\alpha_n$ is defined in \eqref{definition_alpha} with the weight $n(k)=\frac{k}{N}$, $\mu^{\psi}_1$ is defined in \eqref{definition_dens_mat_one_part}, and note that $\norm[\tr]{\cdot}$ and $\norm[\HS]{\cdot}$ denote the trace and Hilbert-Schmidt norms, respectively.

\begin{lemma}\label{lem:density_conv}
Let $\psi \in L^2(\RRR^{3N})$ be antisymmetric and normalized, and let $\varphi_1,\ldots,\varphi_N \in L^2(\RRR^3)$ be orthonormal. Then
\be\label{bound_tr_alpha_HS}
\norm[\tr]{\mu^{\psi}_1 - \mu^{\bigwedge \varphi_j}_1}^2 \leq 8 \, \alpha_n \leq 8 \sqrt{N} \norm[\HS]{\mu^{\psi}_1 - \mu^{\bigwedge \varphi_j}_1},
\ee
\be\label{bound_HS_alpha_tr}
N \norm[\HS]{\mu^{\psi}_1 - \mu^{\bigwedge \varphi_j}_1}^2 \leq 2 \, \alpha_n \leq \norm[\tr]{\mu^{\psi}_1 - \mu^{\bigwedge \varphi_j}_1}.
\ee
\end{lemma}

This lemma is the main result of this section. Its proof is given is Chapter~\ref{sec:density_matrices}. Note that it implies in particular that
\begin{align}\label{equivalence_alpha_tr_HS}
~& \lim_{N \to \infty} \alpha_n = 0 \nonumber \\
\Longleftrightarrow~ & \lim_{N \to \infty} \norm[\tr]{\mu^{\psi}_1 - \mu^{\bigwedge \varphi_j}_1} = 0 \nonumber \\
\Longleftrightarrow~ & \lim_{N \to \infty} \sqrt{N}\norm[\HS]{\mu^{\psi}_1 - \mu^{\bigwedge \varphi_j}_1} = 0.
\end{align}

Let us now consider more general weight functions $f(k)$ that dominate $n(k)$, i.e., $f(k) \geq n(k)$ for all $k$. This includes in particular the weight $m^{(\gamma)}(k)$ from \eqref{weight_m_gamma} which we use later. For those weights, the inequality $\alpha_f \geq \alpha_n$ holds, since
\be
\alpha_n = \sum_{k=0}^N \, n(k) \, \underbrace{\bigSCP{\psi}{P_{N,k} \psi}}_{\geq 0} \leq \sum_{k=0}^N \, f(k) \, \bigSCP{\psi}{P_{N,k} \psi} = \alpha_f.
\ee
Thus, Lemma~\ref{lem:density_conv} directly implies the following lemma.

\begin{lemma}\label{lem:density_conv_alpha_f}
Let $\psi \in L^2(\RRR^{3N})$ be antisymmetric and normalized, and let $\varphi_1,\ldots,\varphi_N \in L^2(\RRR^3)$ be orthonormal. Then, for all $f$ with $f(k)\geq \frac{k}{N} ~\forall k=1,\ldots,N$,
\be\label{bound_tr_alpha_f}
\norm[\tr]{\mu^{\psi}_1 - \mu^{\bigwedge \varphi_j}_1}^2 \leq 8 \, \alpha_f,
\ee
\be\label{bound_HS_alpha_f}
N \norm[\HS]{\mu^{\psi}_1 - \mu^{\bigwedge \varphi_j}_1}^2 \leq 2 \, \alpha_f.
\ee
\end{lemma}

We thus have that convergence of $\alpha_f$ to zero still implies convergence of $\mu_1^\psi$ to $\mu^{\bigwedge \varphi_j}_1$, but not the other way around (in general), i.e.,
\be\label{alpha_f_implies_tr}
\lim_{N \to \infty} \alpha_f = 0 ~\implies~ \lim_{N \to \infty} \norm[\tr]{\mu^{\psi}_1 - \mu^{\bigwedge \varphi_j}_1} = 0,
\ee
\be\label{alpha_f_implies_HS}
\lim_{N \to \infty} \alpha_f = 0 ~\implies~ \lim_{N \to \infty} \sqrt{N}\norm[\HS]{\mu^{\psi}_1 - \mu^{\bigwedge \varphi_j}_1} = 0.
\ee

Finally, let us make a remark about convergence in operator norm. Note that $||\mu^{\psi}_1||_{\op} \leq N^{-1}$ for antisymmetric $\psi$, so a possible indicator of convergence would be the operator norm times $N$. This is not a good type of convergence to consider, though, since the operator norm is given by the largest eigenvalue which at most can be $N^{-1}$ for fermionic density matrices. Thus, while convergence of $N$ times the operator norm does imply convergence of $\alpha_n$, the opposite is not true. One orbital not in the antisymmetrized product of the $\varphi_1,\ldots,\varphi_N$ is enough to let the operator norm of $N$ times the difference between the density matrices be equal to one, while $\alpha_n$ converges to zero. This is summarized in the following proposition which we also prove in Chapter~\ref{sec:density_matrices}.

\begin{proposition}\label{pro:density_matrix_op}
Let $\psi \in L^2(\RRR^{3N})$ be antisymmetric and normalized, and let $\varphi_1,\ldots,\varphi_N \in L^2(\RRR^3)$ be orthonormal. Then
\be\label{bound_alpha_op}
\alpha_n \leq N \norm[\op]{\mu^{\psi}_1 - \mu^{\bigwedge \varphi_j}_1},
\ee
i.e.,
\be\label{op_implies_alpha}
\lim_{N \to \infty} N \norm[\op]{\mu^{\psi}_1 - \mu^{\bigwedge \varphi_j}_1} = 0 ~\implies~ \lim_{N \to \infty} \alpha_n = 0.
\ee
The converse of \eqref{op_implies_alpha} is not true, i.e., $\alpha_n \to 0$ does not imply $N \norm[\op]{\mu^{\psi}_1 - \mu^{\bigwedge \varphi_j}_1} \to 0$.
\end{proposition}

\section{Main Results}\label{sec:main_theorem_mf}
We now state the main results of this work. The proofs of the results in Chapters \ref{sec:main_theorem_mf_general_v} and \ref{sec:main_theorem_mf_x-s} are given in Chapters \ref{sec:proofs_main_theorems_gen} and \ref{sec:proofs_main_theorems}. Note that in the rest of the Chapter we give the desired bounds only in terms of $\alpha_f$; the corresponding bounds for the convergence of density matrices can be read off from Lemmas \ref{lem:density_conv} and \ref{lem:density_conv_alpha_f}.

\subsection{Main Theorems for General $v^{(N)}$}\label{sec:main_theorem_mf_general_v}
The two theorems \ref{thm:estimates_terms_alpha_dot_beta_n} and \ref{thm:estimates_terms_alpha_dot_beta_general} in this subsection cover very general Hamiltonians. The theorems are of the form: Given certain properties of the solutions to the fermionic Hartree equations, the mean-field approximation for the dynamics is good, i.e., $\alpha(t)\leq C(t)\big(\alpha(0)+N^{-\delta}\big)$ for some $\delta>0$. We consider wave functions $\psi^t \in L^2(\RRR^{3N})$ that are solutions to
\be\label{schroedinger_eq_H0vN}
i \partial_t \psi^t = H^N \psi^t = \left( \sum_{j=1}^N H_j^0 + \sum_{1\leq i < j \leq N} v^{(N)}(x_i-x_j) \right) \psi^t,
\ee
where the Hamiltonian $H^N$ is a self-adjoint operator, $v^{(N)}(x)=v^{(N)}(-x)$ is a (possibly scaled) real interaction potential and $H_j^0$ acts only on the $j$-th particle. The most important example for $H_j^0$ is the non-relativistic free Hamiltonian with external field, $H_j^0=-\Delta_j+w^{(N)}(x_j)$, but we could also replace the Laplacian by relativistic operators like $\sqrt{-\Delta +m^2}-m$ ($m>0$) or $|\nabla|$. The fermionic mean-field equations for the one-particle wave functions $\varphi^t_1,\ldots,\varphi^t_N \in L^2(\RRR^3)$ are
\be\label{mean-field_H0vN}
i \partial_t \varphi_j^t(x)= \left( H^0 + \big(v^{(N)} \star \rho_N^t\big)(x) \right) \varphi_j^t(x),
\ee
for $j=1,\ldots,N$ and where $\rho_N^t = \sum_{i=1}^N |\varphi^t_i|^2$. Recall that antisymmetric initial wave functions stay antisymmetric under the evolution \eqref{schroedinger_eq_H0vN}, and that orthonormal initial one-particle wave functions stay orthonormal under the evolution \eqref{mean-field_H0vN}.

The first theorem gives a bound on $\alpha_n$ as defined in \eqref{definition_alpha} and \eqref{weight_n}. 

\begin{theorem}\label{thm:estimates_terms_alpha_dot_beta_n}
Let $t\in[0,T)$ for some $0<T\in \RRR\cup\infty$. Let $\psi^t \in L^2(\RRR^{3N})$ be a solution to the Schr\"odinger equation \eqref{schroedinger_eq_H0vN} with antisymmetric initial condition $\psi^0 \in L^2(\RRR^{3N})$. Let $\varphi_1^t,\ldots,\varphi_N^t \in L^2(\RRR^3)$ be solutions to the fermionic Hartree equations \eqref{mean-field_H0vN} with orthonormal initial conditions $\varphi_1^0,\ldots,\varphi_N^0 \in L^2(\RRR^3)$.

We assume that $v^{(N)}$ and $\rho_N^t := \sum_{i=1}^N |\varphi_i^t|^2$ for all $t\in[0,T)$ are such that there is a positive $D(t)$ (independent of $N$), such that
\be\label{alpha_dot_n_ass_2}
\sup_{y\in\RRR^3} \Big(\big(v^{(N)}\big)^2\star\rho_N^t\Big)(y) \leq D(t) \, N^{-1}.
\ee
Then there is a positive $C(t) = 24\sqrt{D(t)}$, such that
\be\label{main_alpha_ineq_n}
\alpha_n(t) \leq e^{\int_0^t C(s) ds} \, \alpha_n(0) + \left( e^{\int_0^t C(s) ds} - 1 \right) N^{-1}.
\ee
\end{theorem}

\noindent\textbf{Remarks.}
\begin{enumerate}
\setcounter{enumi}{\theremarks}
\item From Lemma~\ref{lem:density_conv} it follows that \eqref{main_alpha_ineq_n} implies for the reduced one-particle density matrices the bounds
\be\label{main_dens_mat_ineq_n}
\norm[\tr]{\mu^{\psi^t}_1 - \mu^{\bigwedge \varphi_j^t}_1} \leq C'(t) \, \left( \, \norm[\tr]{\mu^{\psi^0}_1 - \mu^{\bigwedge \varphi_j^0}_1}^{\frac{1}{2}} + N^{-\frac{1}{2}} \right),
\ee
and
\be\label{main_dens_mat_ineq_n_HS}
\sqrt{N} \norm[\HS]{\mu^{\psi^t}_1 - \mu^{\bigwedge \varphi_j^t}_1} \leq C'(t) \, \left( \left(\sqrt{N} \norm[\HS]{\mu^{\psi^0}_1 - \mu^{\bigwedge \varphi_j^0}_1} \,\right)^{\frac{1}{2}} + N^{-\frac{1}{2}} \right),
\ee
with $C'(t) = \sqrt{8} \, \exp\big(\frac{1}{2}\int_0^t C(s) ds\big)$.

\item The condition \eqref{alpha_dot_n_ass_2} is only a condition on the solutions to the fermionic Hartree equations \eqref{mean-field_H0vN}, and not on the solutions to the Schr\"odinger equation \eqref{schroedinger_eq_H0vN}.

\item Note that condition \eqref{alpha_dot_n_ass_2} implies that (by Cauchy-Schwarz and $\int \rho_N^t=N$)
\be
\sup_{y\in\RRR^3} \Big(\big\lvert v^{(N)} \big\rvert\star\rho_N^t\Big)(y) \leq \sqrt{D(t)}.
\ee
For purely positive or negative $v^{(N)}$, this inequality means that the scaled mean-field interaction is everywhere bounded. In particular, it means that the scaling of the interaction is chosen correctly; e.g., when $v^{(N)}=N^{-\beta}v$, the scaling exponent $\beta$ is chosen correctly (or too big), as discussed in Chapter~\ref{sec:mf_fermions}.

\item We show in Corollary~\ref{cor:estimates_terms_alpha_dot_Coulomb_rho_infty} and Theorem~\ref{thm:E_kin_only} for $H_j^0=-\Delta_j+w^{(N)}(x_j)$ more specifically for which situations condition \eqref{alpha_dot_n_ass_2} holds.

\item\label{itm:Coulomb_Hardy} Let us explain the condition \eqref{alpha_dot_n_ass_2}, for the case $H_j^0=-\Delta_j+w^{(N)}(x_j)$. First, note that we are interested in solutions with total kinetic energy of $O(N)$, that is $\sum_{i=1}^N ||\nabla \varphi_i^t||^2 < A N$. Then, in particular, $\varphi_j^t \in H^1(\RRR^3) ~\forall j=1,\ldots,N$\footnote{$H^1(\RRR^3)$ denotes the first Sobolev space, i.e.,
\be
H^1(\RRR^3) = \left\{ f \in L^2(\RRR^3):  \norm{\nabla f} < \infty \right\}.
\ee}. It turns out that for such solutions and for interesting potentials, that is, potentials with Coulomb singularity or less, the density $\rho_N^t$ is regular enough, so that the quantity on the left-hand side of \eqref{alpha_dot_n_ass_2} is always finite. This is due to the convolution which can smooth out the singularity. (For the Coulomb potential this can be seen by Hardy's inequality.) The condition \eqref{alpha_dot_n_ass_2} can be problematic for $v$ with a strong singularity. Consider the physically most relevant case of Coulomb interaction. There we have $v^{(N)}(x)=N^{-\beta}|x|^{-1}$ with $\beta=\frac{2}{3}$. In a scenario where $\sum_{i=1}^N ||\nabla \varphi_i^t||^2 < A N$, we get in the most general case, by Hardy's inequality, only
\be
\Big(\big(N^{-\beta}v\big)^2\star\rho_N^t\Big)(y) \leq N^{-2\beta}\, 4 A N \leq C N^{-\frac{1}{3}}.
\ee

\item\label{itm:coulomb_nice} Still, we expect that for Coulomb interaction for many scenarios condition \eqref{alpha_dot_n_ass_2} actually holds. For $\sum_{i=1}^N ||\nabla \varphi_i^t||^2 < A N$, the particles naturally occupy a volume proportional to $N$, so the density is of $O(1)$, as discussed in Chapter~\ref{sec:mf_fermions_const_E_kin_physics}. Then, heuristically,
\be\label{v2_heuristic}
\Big(\big(N^{-\beta}v\big)^2\star\rho_N^t\Big)(y) \approx N^{-\frac{4}{3}} \int_{O(N)} |x-y|^{-2} \, \rho_N^t(x) d^3x \lessapprox N^{-\frac{4}{3}} \int_0^{N^{\frac{1}{3}}} r^{-2} C r^2 dr \leq C N^{-1}.
\ee
Therefore, we expect that many solutions to the fermionic Hartree equations \eqref{hartree_scaled_xs_main} fulfill condition \eqref{alpha_dot_n_ass_2}. To show this would be a matter of solution theory for the equations \eqref{mean-field_H0vN}. If the initial conditions $\varphi_1^0,\ldots,\varphi_N^0$ are nice enough, then we expect that condition \eqref{alpha_dot_n_ass_2} holds for long or all times $t$. Note that the properties of the solutions can also depend on the external field $w^{(N)}$.

\item\label{itm:no_scaling} We could easily write down the theorem without any scaling, i.e., we could simply use $v$ instead of $v^{(N)}$. Then the theorem says that, if
\be\label{alpha_dot_n_ass_12_new}
\sup_{y\in\RRR^3} \Big(v^2\star\rho_N^t\Big)(y) \leq D_1^N(t),
\ee
then there is a positive $C^N(t) \propto \sqrt{N D_1^N(t)}$, such that
\be\label{main_alpha_ineq_n_new}
\alpha_n(t) \leq e^{\int_0^t C^N(s) ds} \, \alpha_n(0) + \left( e^{\int_0^t C^N(s) ds} - 1 \right) N^{-1}.
\ee
Suppose that $v$ is Coulomb interaction, $\sum_{i=1}^N ||\nabla \varphi_i^t||^2 < A N$, and that the solutions are nice enough, so that \eqref{v2_heuristic} holds. Then it follows, that the mean-field approximation is good for all times $t$ of $O(N^{-\frac{2}{3}})$, i.e.,
\be\label{main_alpha_ineq_n_new2}
\alpha_n(t) \leq e^{CN^{\frac{2}{3}}t} \, \alpha_n(0) + \left( e^{CN^{\frac{2}{3}}t} - 1 \right) N^{-1}.
\ee
For longer times, the mean-field approximation can not be expected to hold anymore; the dynamics becomes instead dominated by the fluctuations.

Let us also consider the situation from Chapter~\ref{sec:mf_fermions_scaling} before the rescaling. Suppose that $v$ is Coulomb interaction, and the ``system volume'' is proportional to $N^{-1}$, such that the density is $O(N^2)$. Suppose again that the density is nice enough, such that
\be\label{v2_heuristic_small_vol}
\Big(v^2\star\rho_N^t\Big)(y) \approx C N^2 \int_{O(N^{-1})} |x-y|^{-2} \,d^3x \lessapprox C N^2 \int_0^{N^{-\frac{1}{3}}}r^{-2}\, r^2 dr \leq C N^{\frac{5}{3}}.
\ee
Then, it follows from \eqref{main_alpha_ineq_n_new}, that the mean-field approximation is good for all times $t$ of $O(N^{-\frac{4}{3}})$, i.e.,
\be\label{main_alpha_ineq_n_new3}
\alpha_n(t) \leq e^{CN^{\frac{4}{3}}t} \, \alpha_n(0) + \left( e^{CN^{\frac{4}{3}}t} - 1 \right) N^{-1}.
\ee

\item\label{itm:fluctuations} Let us recall Chapter~\ref{sec:mf_fermions_const_E_kin_fluc}. There we showed that the fluctuations around the mean-field at point $y\in\RRR^3$ can be bounded by $\big(\big(v^{(N)}\big)^2\star\rho_N^t\big)(y)$. Thus, the condition \eqref{alpha_dot_n_ass_2} says that the fluctuations have to vanish for large $N$, with rate $N^{-1}$. Note that $N^{-1}$ is the typical size of fluctuations in the (weak) law of large numbers, for independently identically distributed random variables. It is therefore not surprising that under this condition the derivation of the mean-field dynamics succeeds. On the other hand, this condition seems much too restrictive: First, the fluctuations can actually be smaller than $\big(\big(v^{(N)}\big)^2\star\rho_N^t\big)(y)$, see the calculation \eqref{variance}; second, it is sufficient to make the fluctuations vanish in the limit $N\to\infty$ (e.g., they can be $O(N^{-\delta})$, for some $\delta>0$), which, as explained in Section~\ref{sec:mf_fermions_const_E_kin_fluc}, is a necessary condition for the mean-field description to be a good approximation. Indeed, it turns out that condition \eqref{alpha_dot_n_ass_2} can be weakened. This is shown in the next Theorem~\ref{thm:estimates_terms_alpha_dot_beta_general}.

\item\label{itm:exch} Let us consider the case of scaled Coulomb interaction $v^{(N)}(x)=N^{-\frac{2}{3}}|x|^{-1}$. We saw in Chapter~\ref{sec:exch_term} that for the example of plane waves the scaled exchange term is $O(N^{-\frac{2}{3}})$. Since the Hartree-Fock equations (the fermionic Hartree equations with exchange term) are a better approximation to the Schr\"odinger dynamics, it might seem surprising that for $\alpha_n(t)$ we find the convergence rate $N^{-1}$ instead of $N^{-\frac{2}{3}}$. However, looking at the proof of the theorem, we find that an exchange term of $O(N^{-\frac{2}{3}})$ gives an error term of $O(N^{-\frac{4}{3}})$ in the $\alpha_n$ estimate. We show this in Remark \ref{itm:exch_term_order}, following the proof in Chapter~\ref{sec:proofs_main_theorems_gen}. Only for the convergence in the sense of density matrices, see \eqref{main_dens_mat_ineq_n}, does the exchange term give an error term of $O(N^{-\frac{2}{3}})$, but there it is of smaller order than the convergence rate of $N^{-\frac{1}{2}}$ anyway. Note that it follows that for the fermionic Hartree equations with scaled Coulomb interaction the expected optimal convergence rate for the density matrices is between $N^{-\frac{1}{2}}$ and $N^{-\frac{2}{3}}$. If condition \eqref{alpha_dot_n_ass_2} holds, then the error term due to fluctuations is $O(N^{-1})$, see Remark \ref{itm:fluctuations}. It would then be interesting to see if one can improve the convergence rate for the density matrices to $N^{-1}$ by considering the Hartree-Fock equations instead of the Hartree equations.

Note that Theorem~\ref{thm:estimates_terms_alpha_dot_beta_n} shows that under the condition \eqref{alpha_dot_n_ass_2} the scaled exchange term is at most of $O(N^{-\frac{1}{2}})$. This is not so easy to see by directly estimating the exchange term, which is hard for singular potentials.
\end{enumerate}
\setcounter{remarks}{\theenumi}

The condition \eqref{alpha_dot_n_ass_2} can be relaxed and replaced by other conditions if we use the weight function $m^{(\gamma)}$ from \eqref{weight_m_gamma}. This allows to treat more singular interactions and smaller scaling exponents. Let us summarize the precise assumptions that we need on the scaled interaction $v^{(N)}$ and the density $\rho_N^t$.

\begin{assumption}\label{ass:for_main_thm}
For all $t\in[0,T)$, $\rho_N^t := \sum_{i=1}^N |\varphi_i^t|^2$ and $v^{(N)}$ are such that there are a (possibly $N$-dependent) volume $\Omega_N\subset\RRR^3$, positive $D_i(t)$ (independent of $N$) and, for some $0<\gamma\leq 1$, exponents $\delta_2<\gamma$, $\delta_3\geq 0$, $\delta_4\geq 0$ such that
\be\label{alpha_dot_m_ass_1}
\sup_{y\in\RRR^3} \Big(\big\lvert v^{(N)} \big\rvert\star\rho_N^t\Big)(y) \leq D_0(t),
\ee
\be\label{alpha_dot_m_ass_2}
\sup_{y\in\RRR^3} \Big(\big(v^{(N)}\big)^2\star\rho_N^t\Big)(y) \leq D_1(t) \, N^{-\gamma},
\ee
\be\label{alpha_dot_m_ass_3}
\int \bigg(\big(v^{(N)}\big)^2\star\rho^t_N\bigg)(y)\,\rho^t_N(y)\,d^3y \leq D_2(t) \, N^{\delta_2},
\ee
\be\label{alpha_dot_m_ass_4}
\sup_{y\in\RRR^3} \int_{\Omega_N+y} \Big(v^{(N)}(y-x)\Big)^2\rho^t_N(x)\,d^3x \leq D_3(t) \, N^{-1-\delta_3},
\ee
\be\label{alpha_dot_m_ass_5}
\sup_{y \in \RRR^3 \setminus \Omega_N} \big|v^{(N)}(y)\big| \leq D_4(t) \, N^{-\frac{1}{2}-\frac{\gamma}{2}-\delta_4}.
\ee
\end{assumption}

Under this assumption we can conclude convergence of $\alpha_{m^{(\gamma)}}(t)$. The following theorem is the most general version of our main result.

\begin{theorem}\label{thm:estimates_terms_alpha_dot_beta_general}
Let $t\in[0,T)$ for some $0<T\in \RRR\cup\infty$. Let $\psi^t \in L^2(\RRR^{3N})$ be a solution to the Schr\"odinger equation \eqref{schroedinger_eq_H0vN} with antisymmetric initial condition $\psi^0 \in L^2(\RRR^{3N})$. Let $\varphi_1^t,\ldots,\varphi_N^t \in L^2(\RRR^3)$ be solutions to the fermionic Hartree equations \eqref{mean-field_H0vN} with orthonormal initial conditions $\varphi_1^0,\ldots,\varphi_N^0 \in L^2(\RRR^3)$.

We assume that $v^{(N)}$ and $\rho_N^t := \sum_{i=1}^N |\varphi_i^t|^2$ for all $t\in[0,T)$ are such that Assumption~\ref{ass:for_main_thm} holds. Then there is a positive $C(t)$, such that
\be\label{main_alpha_ineq_m}
\alpha_{m^{(\gamma)}}(t) \leq e^{\int_0^t C(s) ds} \, \alpha_{m^{(\gamma)}}(0) + \left( e^{\int_0^t C(s) ds} - 1 \right) N^{-\delta},
\ee
where $0<\delta = \min\left\{ \gamma -\delta_2, \gamma + \frac{\delta_3}{2}, \gamma + \delta_4 \right\}$ and
\be
C(t) = 12 \max\bigg\{ 4 \sqrt{D_3(t)} N^{-\frac{\delta_3}{2}}, 4 \sqrt{2} D_4(t) N^{-\delta_4}, \sqrt{12} D_0(t), \sqrt{12} \frac{D_2(t)}{D_0(t)}, 8\sqrt{D_1(t)} \bigg\}.
\ee
\end{theorem}

\noindent\textbf{Remarks.}
\begin{enumerate}
\setcounter{enumi}{\theremarks}
\item Similar to \eqref{main_dens_mat_ineq_n}, the bound \eqref{main_alpha_ineq_m} implies
\be
\norm[\tr]{\mu^{\psi^t}_1 - \mu^{\bigwedge \varphi_j^t}_1} \leq \sqrt{8}\, e^{\frac{1}{2}\int_0^t C(s) ds} \, \Big(\alpha_{m^{(\gamma)}}(0)\Big)^{\frac{1}{2}} + \left( 8 \left(e^{\int_0^t C(s) ds} - 1\right) \right)^{\frac{1}{2}} \, N^{-\frac{\delta}{2}}
\ee
and a similar bound for the Hilbert-Schmidt norm. However, in general it is not true that $\alpha_{m^{(\gamma)}}(0) \leq C \Big|\Big|\mu^{\psi^0}_1 - \mu^{\bigwedge \varphi_j^0}_1\Big|\Big|_{\tr}$.
\item We show in Theorem~\ref{thm:E_kin_only} for which situations Assumption~\ref{ass:for_main_thm} holds.
\end{enumerate}
\setcounter{remarks}{\theenumi}

\subsection{Main Results for $-\Delta$ and Interactions $|x|^{-s}$}\label{sec:main_theorem_mf_x-s}
In this section we explicitly consider the non-relativistic Schr\"odinger equation
\be\label{Schr_scaled_xs_main}
i \partial_t \psi^t(x_1,\ldots,x_N) = \left( \sum_{j=1}^N \left( -\Delta_{x_j} + w^{(N)}(x_j) \right) + N^{-\beta} \sum_{i<j} v(x_i-x_j) \right) \psi^t(x_1,\ldots,x_N),
\ee
and the corresponding fermionic Hartree equations
\be\label{hartree_scaled_xs_main}
i \partial_t \varphi_j^t(x) = \left( -\Delta + w^{(N)}(x) + N^{-\beta} \left(v \star \rho_N^t \right)(x) \right) \varphi_j^t(x),
\ee
for $j=1,\ldots,N$, where $\rho_N^t=\sum_{i=1}^N |\varphi_i^t|^2$. The results in this subsection are concerned with potentials $|x|^{-s}$ with $0<s<\frac{6}{5}$, sometimes with singularity weakened or cutoff, and the corresponding $\beta=1-\frac{s}{3}$, as discussed in Chapter~\ref{sec:mf_fermions_scaling_general}. For the following results we assume that the mean kinetic energy per particle is bounded by a constant, independent of $N$, i.e., for the total kinetic energy we have $E_{\kin,\mf}(t) = \sum_{i=1}^N ||\nabla \varphi_i^t||^2 \leq AN$.

First, let us state a result about the Coulomb potential that replaces condition \eqref{alpha_dot_n_ass_2} by other conditions, which again depend on properties of the solutions to the fermionic Hartree equations. (Note that $\norm[\infty]{A} = \sup_x |A(x)|$ for all multiplication operators $A$.)

\begin{corollary}\label{cor:estimates_terms_alpha_dot_Coulomb_rho_infty}
Let $v(x) = \pm |x|^{-1}$ and $\beta=\frac{2}{3}$. We assume that $\varphi_1^t,\ldots,\varphi_N^t$ are such that
\be
E_{\kin,\mf}(t) = \sum_{i=1}^N \norm{\nabla \varphi_i^t}^2 \leq AN,
\ee
\be\label{cond_rho_infty}
\bigg\lvert\bigg\lvert \, \sum_{i=1}^N|\varphi_i^t|^2 \, \bigg\rvert\bigg\rvert_{\infty} \leq D,
\ee
for some $A,D>0$ (independent of $N$) and all $t \in [0,T)$. Then assumption \eqref{alpha_dot_n_ass_2} from Theorem~\ref{thm:estimates_terms_alpha_dot_beta_n} holds. Therefore, there is a positive constant $C$ (independent of $N$), such that
\be\label{main_alpha_ineq_n_applied2}
\alpha_n(t) \leq e^{Ct} \, \alpha_n(0) + \left( e^{Ct} - 1 \right) N^{-1}.
\ee
\end{corollary}

Let us now come to the main result of this section, where we state for which interactions the mean-field approximation holds under the only condition that $E_{\kin,\mf}(t) = \sum_{i=1}^N \norm{\nabla \varphi_i^t}^2 \leq AN$. Note that $E_{\kin,\mf}(t) \leq AN$ is basically just a condition on the initial states $\varphi_1^0, \ldots, \varphi_N^0$ and the external field. Recall that the fermionic Hartree time evolution conserves the total Hartree energy. For repulsive interactions, it is therefore expected that, for nice enough external fields, including, e.g., external Coulomb fields generated by nuclei with some $N$-independent distances to each other, the kinetic energy at all times $t$ is bounded by $AN$, if it is initially bounded by $CN$. A blowup of solutions is only expected to happen for strong attractive interactions (e.g., for gravitating fermions), see, e.g., \cite{froehlich:2007,hainzl:2009,hainzl:2010}. There are several works about solution theory to the Hartree(-Fock) equations \cite{bove:1974, chadam:1975, chadam:1976, bove:1976}; however, estimates of Sobolev norms with explicit $N$-dependence are rare.

We can treat interactions with weak enough singularities ($|x|^{-s}$, with $0<s<\frac{3}{5}$), with singularity cut off (and long-range behavior like $|x|^{-s}$, with $0<s<\frac{6}{5}$), and with long-range behavior like $|x|^{-1}$ but weaker singularity.

\begin{theorem}\label{thm:E_kin_only}
Let $\psi^t \in L^2(\RRR^{3N})$ be a solution to the Schr\"odinger equation \eqref{Schr_scaled_xs_main} with antisymmetric initial condition $\psi^0 \in L^2(\RRR^{3N})$. Let $\varphi_1^t,\ldots,\varphi_N^t \in L^2(\RRR^3)$ be solutions to the fermionic Hartree equations \eqref{hartree_scaled_xs_main} with orthonormal initial conditions $\varphi_1^0,\ldots,\varphi_N^0 \in L^2(\RRR^3)$, and with
\be
E_{\kin,\mf}(t) = \sum_{i=1}^N \norm{\nabla \varphi_i^t}^2 \leq AN
\ee
for some $A>0$ and all $t>0$. Then there is a positive constant $C$, such that
\begin{itemize}
\item \textbf{for interactions} 
\be
v(x) = \pm |x|^{-s}, \text{with } 0<s<\frac{3}{5} \text{ and } \beta=1-\frac{s}{3},
\ee
\textbf{we have}
\be\label{main_alpha_ineq_n_applied1}
\alpha_n(t) \leq e^{Ct} \, \alpha_n(0) + \left( e^{Ct} - 1 \right) N^{-1},
\ee
with $C \propto A^{\frac{s}{2}}$;
\item \textbf{for interactions}
\begin{align}
&v = \pm v_{s,\varepsilon} \in L^{\infty} \text{ with } 0 \leq v_{s,\varepsilon}(x) \left\{\begin{array}{cl} \leq |x|^{-s} &, \, \text{for } |x|\leq \varepsilon \\ =|x|^{-s} &, \, \text{for } |x|>\varepsilon , \end{array}\right., \text{with } \varepsilon>0, \nonumber \\
&\text{with } 0<s<\frac{6}{5} \text{ and } \beta = 1 - \frac{s}{3},
\end{align}
\textbf{we have}
\be\label{main_alpha_ineq_m_applied1}
\alpha_{m^{(\gamma)}}(t) \leq e^{Ct} \, \alpha_{m^{(\gamma)}}(0) + \left( e^{Ct} - 1 \right) N^{-\gamma},
\ee
for all $0 < \gamma \leq 1-\frac{2s}{3}$;
\item \textbf{for interactions}
\be
v(x) = \pm \left\{\begin{array}{cl} |x|^{-s} &, \, \text{for } |x|\leq 1 \\ |x|^{-1} &, \, \text{for } |x|>1 , \end{array}\right., \text{with } 0<s<\frac{1}{3} \text{ and } \beta = \frac{2}{3},
\ee
\textbf{we have}
\be\label{main_alpha_ineq_m_applied2}
\alpha_{m^{(\gamma)}}(t) \leq e^{Ct} \, \alpha_{m^{(\gamma)}}(0) + \left( e^{Ct} - 1 \right) N^{-\gamma},
\ee
for all $0<\gamma\leq \frac{1}{3} - \frac{4s}{9-15s}$.
\end{itemize}
\end{theorem}

Finally, let us show that, when the Coulomb interaction is scaled with $N^{-1}$, then, for systems with initial total kinetic energy bounded by $AN$, the dynamics is free.

\begin{proposition}\label{pro:coulombN1}
Let $\varphi_1^t,\ldots,\varphi_N^t \in L^2(\RRR^3)$ be solutions to the free equations
\be\label{free_equations}
i\partial_t \varphi_j^t(x) = -\Delta \varphi_j^t(x)
\ee
for $j=1,\ldots,N$, with orthonormal initial conditions $\varphi_1^0,\ldots,\varphi_N^0 \in L^2(\RRR^3)$ with
\be\label{condition_E_kin_0}
E_{\kin,\mf}(0) = \sum_{i=1}^N \norm{\nabla \varphi_i^0}^2 \leq AN
\ee
for some $A>0$. Let $\psi^t \in L^2(\RRR^{3N})$ be a solution to the Schr\"odinger equation \eqref{Schr_scaled_xs_main} with $\beta=1$, $v(x)=\pm |x|^{-1}$ and $w^{(N)}(x)=0$, and with antisymmetric initial condition $\psi^0 \in L^2(\RRR^{3N})$. Then, for $\frac{1}{3} < \gamma < 1$, there is a positive constant $C$ such that for all $t\geq 0$,
\be\label{main_alpha_ineq_coulombN1_gamma}
\alpha_{m^{(\gamma)}}(t) \leq \alpha_{m^{(\gamma)}}(0) +Ct \, N^{-\delta},
\ee
where $0<\delta = \min\left\{ \frac{\gamma}{2}-\frac{1}{6}, -\frac{\gamma}{2}+\frac{1}{2} \right\}$. In particular, for $\gamma=\frac{2}{3}$ we have the maximal convergence rate $\delta=\frac{1}{6}$.
\end{proposition}

\noindent\textbf{Remarks.}
\begin{enumerate}
\setcounter{enumi}{\theremarks}
\item The proposition also holds with external field $w^{(N)}(x)$ that is such that it preserves the bound $\sum_{i=1}^N ||\nabla \varphi_i^t||^2 \leq AN$ for all times.
\item Note that for $\beta=1$, $v(x)=|x|^{-1}$ and $w^{(N)}(x)$ that preserve the bound $\sum_{i=1}^N ||\nabla \varphi_i^t||^2 \leq AN$ for all times, the condition \eqref{alpha_dot_n_ass_2} from Theorem~\ref{thm:estimates_terms_alpha_dot_beta_n} holds due to Hardy's inequality and energy conservation. Therefore, in this case we can deduce the bound
\be
\alpha_n(t) \leq e^{Ct} \, \alpha_n(0) + \left( e^{Ct} - 1 \right) N^{-1},
\ee
which gives a better convergence rate than in Proposition~\ref{pro:coulombN1}, but exponential growth in time. However, as explained in Remark~\ref{itm:coulomb_nice}, for nice initial data, one would expect that
\be
\Big(v^2\star\rho_N^t\Big)(y) \leq D \, N^{\frac{1}{3}},
\ee
which implies the bound
\be
\alpha_n(t) \leq \exp\left( C N^{-\frac{1}{3}} t \right) \, \big( \alpha_n(0) + N^{-1}\big),
\ee
and therefore in particular
\be
\alpha_n(t) \leq \alpha_n(0) + Ct \, N^{-\frac{1}{3}}.
\ee
\end{enumerate}
\setcounter{remarks}{\theenumi}

\section{Outline of the Proof}\label{sec:outline_proof}
The proofs of the main results from Chapters \ref{sec:dens_mat_summary} and \ref{sec:main_theorem_mf} are given in Chapters~\ref{sec:notation}--\ref{sec:mean-field_scalings_general}. In Chapter~\ref{sec:notation}, we establish some notation, state inequalities we often use during the proofs, and explain in more detail properties of the projectors from Definition~\ref{def:projectors}. In Chapter~\ref{sec:density_matrices}, the proofs of the lemmas about the convergence of reduced density matrices are given. These proofs were already outlined in Chapter~\ref{sec:dens_mat_summary}. Then, in Chapter~\ref{sec:alpha_dot_and_general_lemmas}, we prove the theorems for general free Hamiltonians $H^0$ and interactions $v^{(N)}$ from Chapter~\ref{sec:main_theorem_mf_general_v}, and in Chapter~\ref{sec:mean-field_scalings_general}, we prove the results for $H^0=-\Delta$ and interaction potentials $|x|^{-s}$ from Chapter~\ref{sec:main_theorem_mf_x-s}. Let us now outline these proofs of our main results.

The general strategy of the proof is the following. First, we calculate the time derivative of $\alpha_f(t) = \alpha_f(\psi^t,\varphi_1^t,\ldots,\varphi_N^t)$, where $\psi^t$ is a solution to the Schr\"odinger equation and $\varphi_1^t,\ldots,\varphi_N^t$ are solutions to the fermionic Hartree equations. (This is a simple, straightforward calculation.) Second, we bound the time derivative by terms proportional to $\alpha_f(t)$ or $N^{-\delta}$ for some $\delta>0$, i.e.,
\be
\partial_t \alpha_f(t) \leq C(t) \left(\alpha_f(t) +N^{-\delta} \right).
\ee
Then we use the Gronwall Lemma (which we state as Lemma~\ref{lem:gronwall}) to conclude that
\be
\alpha_f(t) \leq e^{\int_0^t C(s) ds} \, \alpha_f(0) + \left( e^{\int_0^t C(s) ds} - 1 \right) N^{-\delta},
\ee
which is the desired bound.

\textbf{Outline for Theorem~\ref{thm:estimates_terms_alpha_dot_beta_n}.} Let us start with the most simple case where we use $\alpha_n(t)$, i.e., the $\alpha$-functional with the weight function $n(k)=\frac{k}{N}$, as considered in Theorem~\ref{thm:estimates_terms_alpha_dot_beta_n}. The first step is to calculate $\partial_t \alpha_n(t)$. This is done in Chapter~\ref{sec:alpha_dot}, for general weight functions $f(k)$. In the case of $\alpha_n(t)$ this is a very simple calculation, due to the identity $\alpha_n(t) = \SCP{\psi^t}{q_1\psi^t}$. Recall here our loose notation: the projectors $p_1$ and $q_1$ are always time dependent, since they are build with the time dependent solutions to the fermionic Hartree equations $\varphi_1^t,\ldots,\varphi_N^t$. In fact, $q_1$ solves the Heisenberg equation of motion $i\partial_t q_1 = [H_1^{\mf},q_1]$, where $H_1^{\mf}=H_1^0 + V_1^{(N)}$ is the ``mean-field Hamiltonian'', acting on the first variable (the variable that $q_1$ depends on), and $[a,b]=ab-ba$ is the commutator. The wave function $\psi^t$ solves the Schr\"odinger equation $i\partial_t \psi^t = H \psi^t$ with $H = \sum_{j=1}^N H_j^0 + \sum_{i<j} v^{(N)}_{ij}$. We thus find
\begin{align}\label{summary_dt_alpha}
\partial_t \alpha_n(t) &= \partial_t \SCP{\psi^t}{q_1\psi^t} \nonumber \\
&= \SCP{(\partial_t \psi^t)}{q_1\psi^t} + \SCP{\psi^t}{(\partial_t q_1)\psi^t} + \SCP{\psi^t}{q_1(\partial_t \psi^t)} \nonumber \\
&= i \SCP{H\psi^t}{q_1\psi^t} - i \SCP{\psi^t}{[H_1^{\mf}, q_1]\psi^t} - i \SCP{\psi^t}{q_1 H\psi^t} \nonumber \\
&= i \SCP{\psi^t}{\big[H-H_1^{\mf},q_1\big]\psi^t} \nonumber \\
&= i \SCP{\psi^t}{\left[\sum_{j} H_j^0 + \sum_{i<j} v^{(N)}_{ij} - H_1^0 - V_1^{(N)},q_1\right]\psi^t} \nonumber \\
&= i \SCP{\psi^t}{\left[\sum_{j \geq 2} v^{(N)}_{1j} - V_1^{(N)},q_1\right]\psi^t} \nonumber \\
&= i \SCP{\psi^t}{\left[ (N-1) v^{(N)}_{12} - V_1^{(N)},q_1\right]\psi^t},
\end{align}
where we used that $[h_j,q_1]=0 \,\forall j\geq 2$ for all operators $h_j$ that act only on the $j$-th variable, and in the last step we used the antisymmetry of $\psi^t$. Note that the kinetic and external field terms coming from the Schr\"odinger and the fermionic Hartree equations cancel. This is the reason why Theorem~\ref{thm:estimates_terms_alpha_dot_beta_n} (and also \ref{thm:estimates_terms_alpha_dot_beta_general}) holds for any $H_j^0$. We can simplify \eqref{summary_dt_alpha} by inserting two identities $1=p_1+q_1$ and $1=p_2+q_2$ in front of each $\psi^t$. Due to the commutator structure, $v_{12}=v_{21}$ (i.e., $v(x_1-x_2) = v(x_2-x_1)$) and $p_1q_1=0=p_2q_2$, only three summands remain, such that we find
\begin{align}\label{alpha_derivative_n_gen}
\partial_t \alpha_n(t) &= 2 \, \Im\, \bigSCP{\psi^t}{q_1\bigg( (N-1)p_2v^{(N)}_{12}p_2 - V^{(N)}_1 \bigg) p_1 \psi^t} \nonumber \\
&\quad + 2 \, \Im\, \bigSCP{\psi^t}{q_1q_2 \, (N-1)v^{(N)}_{12} \, p_1p_2 \psi^t} \nonumber \\
&\quad + 2 \, \Im\, \bigSCP{\psi^t}{q_1q_2 (N-1)v^{(N)}_{12} p_1q_2 \psi^t}.
\end{align}
Note that the second and third term on the right-hand side of \eqref{alpha_derivative_n_gen} do not depend on the mean-field $V_1$ at all; it is only the first term where the mean-field makes the contribution coming from the Schr\"odinger interaction small.

The three terms on the right-hand side of \eqref{alpha_derivative_n_gen} have a nice intuitive explanation. Let us call the contributions coming from any projector $p$ ``particle in the Sea'' and those coming from any projector $q$ ``particles outside the Sea''. With Sea we thus mean the antisymmetrized product of $\varphi_1^t,\ldots,\varphi_N^t$, the ``Fermi Sea'' or ``condensate'', and with ``outside the Sea'' we mean those parts of $\psi^t$ that do not contain $\varphi_1^t,\ldots,\varphi_N^t$. The $v_{12}^{(N)}$ summand from the first term on the right-hand side of \eqref{alpha_derivative_n_gen} gives a contribution only if two particles in the Sea (the $p_1,p_2$ on the right side of the scalar product) transition into one particle outside the Sea and one in the Sea (the $q_1,p_2$ on the left side of the scalar product). Furthermore, the contribution from this term is ``big'' only if $\psi^t$ contains many parts of $\varphi_1^t,\ldots,\varphi_N^t$, since there are three $p$'s in the scalar product. The second term on the right-hand side of \eqref{alpha_derivative_n_gen} gives a contribution only if two particles in the Sea (the $p_1,p_2$ on the right side of the scalar product) transition into two particles outside the Sea (the $q_1,q_2$ on the left side of the scalar product). The third term contributes only if one particle in the Sea and one outside the Sea (the $p_1,q_2$ on the right side of the scalar product) transition into two particles outside the Sea (the $q_1,q_2$ on the left side of the scalar product). Due to the three $q$'s, this third term is ``big'' if $\psi^t$ already contains many parts orthogonal to $\varphi_1^t,\ldots,\varphi_N^t$. Finally, note that these three contributions are exactly what one would intuitively expect. The change in the ``number of particles in the Sea'' ($\partial_t \alpha_n(t)$) can be caused by those three transitions (and their reverse processes): two particles in the Sea interact and one gets kicked out (becomes correlated), two particles in the Sea interact and both get kicked out, one particle in the Sea interacts with one outside the Sea and gets kicked out.

Now let us discuss how these three terms can be bounded rigorously. At this point, if $v^{(N)}$ has both negative and positive parts, we split up $v^{(N)} = v^{(N)}_+ - v^{(N)}_-$, with $v^{(N)}_{\pm} \geq 0$, and then split up each of the three terms into two contributions, coming from $v^{(N)}_+$ and $v^{(N)}_-$. Since each contribution is estimated separately, we only deal with positive $v^{(N)}$ in the following.

\textbf{The} $\boldsymbol{qq}$-$\boldsymbol{pq}$ \textbf{term.} Let us begin with the third term on the right-hand side of \eqref{alpha_derivative_n_gen}, which is the easiest to bound. Note that $\norm{q_1q_2\psi^t} \leq \norm{q_1\psi^t} = \sqrt{\SCP{\psi^t}{q_1\psi^t}} = \sqrt{\alpha_n(t)}$. Using Cauchy-Schwarz, we find that
\begin{align}\label{summary_term3_v}
N \bigSCP{\psi^t}{q_1q_2 v_{12}^{(N)} p_1q_2 \psi^t} &\leq N \norm{q_1q_2\psi^t} \norm{v_{12}^{(N)} p_1q_2 \psi^t} \nonumber \\
&\leq N \sqrt{\alpha_n(t)} \sqrt{\bigSCP{\psi^t}{q_2p_1 \Big(v_{12}^{(N)}\Big)^2 p_1q_2 \psi^t}} \nonumber \\
&\leq N \sqrt{\alpha_n(t)} \sqrt{\norm{q_2 \psi^t}^2 \sup_{\phi} \bigSCP{\phi}{p_1 \Big(v_{12}^{(N)}\Big)^2 p_1 \phi}} \nonumber \\
&\leq \alpha_n(t) \sqrt{N^2 \sup_{\phi} \bigSCP{\phi}{p_1 \Big(v_{12}^{(N)}\Big)^2 p_1 \phi}},
\end{align}
where the supremum is taken over all $\phi \in L^2(\RRR^{3N})$ which are antisymmetric in all but the second variable (because of the $q_2$). As we show in Chapter~\ref{sec:estimates_projectors}, it turns out that
\be\label{summary_p_1v2_p_1}
\bigSCP{\phi}{p_1 \Big(v_{12}^{(N)}\Big)^2 p_1 \phi} \leq N^{-1} \sup_{y\in\RRR^3}\Big( \Big(v^{(N)}\Big)^2 \star \rho_N^t \Big)(y)
\ee
(where $\rho_N^t=\sum_{i=1}^N |\varphi_i^t|^2$), which follows from diagonalizing $p_1 \Big(v_{12}^{(N)}\Big)^2 p_1$; the extra factor $N^{-1}$ comes from the antisymmetry of $\phi$ (it does not matter that $\phi$ is not antisymmetric in the second variable). Note that the inequality \eqref{summary_p_1v2_p_1} is very similar to the inequality
\be\label{trace_A_1}
\SCP{\phi}{A_1 \phi} \leq N^{-1} \norm[\tr]{A_1}
\ee
for any self-adjoint $A_1$ (that acts only on $x_1$) and antisymmetric $\phi$, where $\norm[\tr]{\cdot}$ denotes the trace norm. This inequality can be proven by diagonalizing $A_1 = \sum_j \lambda_j \, \ketbr{\varphi_j}_1$ and calculating
\be
\SCP{\phi}{A_1 \phi} = \SCP{\phi}{\sum_j \lambda_j \ketbr{\varphi_j}_1 \phi} \leq \sum_j |\lambda_j| \, \sup_j  \big| \SCP{\phi}{\ketbr{\varphi_j}_1 \phi}\big|.
\ee
Since $\sum_j |\lambda_j| = \norm[\tr]{A_1}$ and $\SCP{\phi}{\ketbr{\varphi_j}_1 \phi} = N^{-1} \sum_{m=1}^N \SCP{\phi}{\ketbr{\varphi_j}_m \phi} \leq N^{-1}$ (since $\sum_{m=1}^N \ketbr{\varphi_j}_m$ is a projector due to the orthonormality of the $\varphi_j$'s), \eqref{trace_A_1} follows. To summarize, the $qq$-$pq$ term from the right-hand side of \eqref{alpha_derivative_n_gen} is bounded by $C \alpha_n(t)$, under the assumptions \eqref{alpha_dot_n_ass_2} of Theorem~\ref{thm:estimates_terms_alpha_dot_beta_n}.

\textbf{The} $\boldsymbol{qq}$-$\boldsymbol{pp}$ \textbf{term.} Let us now estimate the second term on the right-hand side of \eqref{alpha_derivative_n_gen}. Here, we have again two $q$'s available, so the term should be proportional to $\alpha_n(t)$. However, both $q$'s are on the same side of the scalar product, so we cannot directly apply Cauchy-Schwarz. But by a trick using the antisymmetry of $\psi^t$, we can shift the $q_2$ to the right side of the scalar product, on the expense of a boundary term of $O(N^{-1})$. In more detail, we estimate, using Cauchy-Schwarz again,
\begin{align}\label{summary_term2}
\bigSCP{\psi^t}{q_1q_2 (N-1) v^{(N)}_{12} p_1p_2 \psi^t} &= \bigSCP{\psi^t}{q_1 \sum_{m=2}^N q_m v^{(N)}_{1m} p_1p_m \psi^t} \nonumber \\
&\leq \norm{q_1\psi^t} \norm{\sum_{m=2}^N q_m v^{(N)}_{1m} p_1p_m \psi^t} \nonumber \\
&\leq \sqrt{\alpha_n(t)} \sqrt{\sum_{m,n=2}^N \bigSCP{\psi^t}{p_m p_1 v^{(N)}_{1m} q_m q_n v^{(N)}_{1n} p_1p_n \psi^t}} \nonumber \\
&\leq \sqrt{\alpha_n(t)} \sqrt{\sum_{m \neq n=2}^N \bigSCP{\psi^t}{p_m p_1 v^{(N)}_{1m} q_m q_n v^{(N)}_{1n} p_1p_n \psi^t}} \nonumber \\
& \quad + \sqrt{\alpha_n(t)} \sqrt{\sum_{m=2}^N \bigSCP{\psi^t}{p_m p_1 v^{(N)}_{1m} q_m v^{(N)}_{1m} p_1p_m \psi^t}} \nonumber \\
&\leq \sqrt{\alpha_n(t)} \sqrt{N^2 \bigSCP{\psi^t}{q_3 p_2 p_1 v^{(N)}_{12} v^{(N)}_{13} p_1p_3 q_2\psi^t}} \nonumber \\
& \quad + \sqrt{\alpha_n(t)} \sqrt{N \bigSCP{\psi^t}{p_2 p_1 \Big(v^{(N)}_{12}\Big)^2 p_1p_2 \psi^t}} \nonumber \\
&\leq \sqrt{\alpha_n(t)} \sqrt{\big(N^2 \alpha_n(t) + N\big) \sup_{\phi}\bigSCP{\phi}{p_2 p_1 \Big(v^{(N)}_{12}\Big)^2 p_1p_2 \phi}}.
\end{align}
Similar to \eqref{summary_p_1v2_p_1}, one can show that 
\begin{align}
\bigSCP{\phi}{p_2 p_1 \Big(v^{(N)}_{12}\Big)^2 p_1p_2 \phi} &\leq N^{-2} \int \left( \Big(v^{(N)}\Big)^2 \star \rho_N^t \right)(y) \rho_N^t(y) \, d^3y, \nonumber \\
&\leq N^{-1} \sup_{y\in\RRR^3} \left( \Big(v^{(N)}\Big)^2 \star \rho_N^t \right)(y),
\end{align}
which is also shown in Chapter~\ref{sec:estimates_projectors}, and where the second inequality follows from H\"older's inequality and $\int \rho_N^t = N$. Thus, if condition \eqref{alpha_dot_n_ass_2} of Theorem~\ref{thm:estimates_terms_alpha_dot_beta_n} holds, we find
\begin{align}
\bigSCP{\psi^t}{q_1q_2 (N-1) v^{(N)}_{12} p_1p_2 \psi^t} &\leq C \sqrt{\alpha_n(t)} \sqrt{\alpha_n(t) + N^{-1}} \nonumber \\
&\leq C \sqrt{\alpha_n(t)^2 + 2 \alpha_n(t)N^{-1} + N^{-2}} \nonumber \\
&= C (\alpha_n(t) + N^{-1}).
\end{align}

\textbf{The} $\boldsymbol{qp}$-$\boldsymbol{pp}$ \textbf{term.} We now turn to the first term on the right-hand side of \eqref{alpha_derivative_n_gen}. This term is the hardest to control, since here it is crucial to use the fact that the interaction from the Schr\"odinger equation and the mean-field from the fermionic Hartree equations cancel in a certain sense. Note that we did not use any such cancellation so far. Before we outline the estimate, let us stress that here lies the crucial difference compared to the bosonic Hartree equation, as treated in \cite{pickl:2011method, pickl:2010hartree}. In this case there is just one orbital $\varphi$, so for example $p_2=\ketbra{\varphi(x_2)}{\varphi(x_2)}$. Furthermore, $v^{(N)}(x)=N^{-1}v(x)$ and $V^{(N)}_1(x)=\left( v\star|\varphi|^2 \right)(x)$. Then (recall that $q_2=1-p_2$)
\begin{align}
V^{(N)}_1 - p_2 (N-1)v^{(N)}_{12} p_2 &= \left(v\star|\varphi|^2\right)(x_1) - \frac{(N-1)}{N} \ketbra{\varphi(x_2)}{\varphi(x_2)} v_{12} \ketbra{\varphi(x_2)}{\varphi(x_2)} \nonumber \\
&= \left(v\star|\varphi|^2\right)(x_1) - \frac{(N-1)}{N} \left(v\star|\varphi|^2\right)(x_1) \, p_2 \nonumber \\
&= \left(v\star|\varphi|^2\right)(x_1) \, (q_2 + N^{-1}p_2),
\end{align}
such that the first term can be bounded by $\alpha_n(t)$ (due to the two available $q$'s) and a term of order $N^{-1}$.\footnote{Note that in \cite{pickl:2011method, pickl:2010hartree}, the term $\left(v\star|\varphi|^2\right)(x_1) \, q_2$ is usually regarded as being part of the $qq$-$pq$ term.} In the fermionic case we cannot use the same argument since $p_2$ is a sum of projectors, each projecting on one of the $N$ orbitals. However, as we have used before and show in Chapter~\ref{sec:estimates_projectors}, one can diagonalize the operator $p_2v_{12}^{(N)}p_2$ in the sense that $p_2v_{12}^{(N)}p_2 = \sum_{i=1}^N \lambda_i(x_1) \ketbr{\chi_i^{x_1}(x_2)}$, where it turns out that $\sum_{i=1}^N \lambda_i(x) = V_1^{(N)}(x)$. Then 
\begin{align}\label{v_v_mf}
V_1^{(N)} - p_2 (N-1)v_{12}^{(N)} p_2 &= \sum_{i=1}^N \lambda_i(x_1) - (N-1) \sum_{i=1}^N \lambda_i(x_1) \, \ketbra{\chi_i^{x_1}(x_2)}{\chi_i^{x_1}(x_2)} \nonumber \\
&= \sum_{i=1}^N \lambda_i(x_1) \Big( 1 - (N-1)\ketbra{\chi_i^{x_1}(x_2)}{\chi_i^{x_1}(x_2)} \Big).
\end{align}
Due to the antisymmetry of $\psi^t$, the term $\big( 1 - (N-1)\ketbra{\chi_i^{x_1}(x_2)}{\chi_i^{x_1}(x_2)} \big)$ corresponds to a projector called $q^{\chi_i^{x_1}}$. It can be shown that this additional $q^{\chi_i^{x_1}}$ gives us an additional $q_2$ on the right side of the scalar product of the first term; again at the expense of a small boundary term of $O(N^{-1})$. We do not present the full estimate here in the outline, since it is technical and lengthy; the explicit estimate can be found in Chapter~\ref{sec:alpha_m_dot_rigorous}. To summarize, also the first term can be bounded by $C \alpha_n(t)$. This concludes the proof of Theorem~\ref{thm:estimates_terms_alpha_dot_beta_n}.

Note, that the $qp$-$pp$ term is not only technically the hardest to control, but also the reason why we cannot prove Theorem~\ref{thm:E_kin_only} for Coulomb interaction (including the singularity). The conditions from the estimate make further properties of the solutions to the fermionic Hartree equations necessary.

\absatz

\textbf{Outline for Theorem~\ref{thm:estimates_terms_alpha_dot_beta_general}.} Let us now describe what is gained by using a different weight function than $n(k)$. We saw that in the third term (and, as it turns out, also in the first term) on the right-hand side of \eqref{alpha_derivative_n_gen}, there is only one projector $p$ available which lead to the condition $\big(v^{(N)}\big)^2\star\rho_N^t \leq CN^{-1}$. This condition seems to be too strong and we would like to relax it. Let us first consider the time derivative of $\alpha_f(t)$ for a general weight function $f(k)$. This is, as in the case $f(k)=n(k)$, a straightforward, but more lengthy calculation, which is the content of Chapter~\ref{sec:alpha_dot}. Let us present the result here:
\begin{align}\label{alpha_derivative_summary}
\partial_t \alpha_f(t) =& 2 \, \Im\, \bigSCP{\psi^t}{N\left(\widehat{f}-\widehat{f}_{-1}\right)q_1p_2\Big( (N-1)v^{(N)}_{12} - V^{(N)}_1 \Big)p_1p_2 \psi^t} \nonumber \\
& + \Im\, \bigSCP{\psi^t}{N\left(\widehat{f}-\widehat{f}_{-2}\right)q_1q_2 \, (N-1)v^{(N)}_{12} \, p_1p_2 \psi^t} \nonumber \\
& + 2 \, \Im\, \bigSCP{\psi^t}{N\left(\widehat{f}-\widehat{f}_{-1}\right)q_1q_2\Big( (N-1)v^{(N)}_{12} - V^{(N)}_1 \Big)p_1q_2 \psi^t}.
\end{align}
The terms $\widehat{f}-\widehat{f}_{-d}$ can be interpreted as a ``derivative'', since
\begin{align}
\widehat{f}-\widehat{f}_{-1} &= \sum_{k=0}^N \,\Big( f(k) - f(k-1) \Big) \, P_{N,k}.
\end{align}
Let us now see what we gain by choosing the weight function
\be
m^{(\gamma)}(k) = \left\{\begin{array}{cl} \frac{k}{N^{\gamma}} &, \text{for } k \leq N^{\gamma}\\ 1 & , \text{otherwise} \end{array}\right.
\ee
for some $0<\gamma \leq 1$. (Note that for $\gamma=1$ we recover the case $m^{(1)}(k)=n(k)=\frac{k}{N}$ from above.) For the ``derivative'' we find
\begin{align}\label{derivative_m}
\widehat{m^{(\gamma)}}-\widehat{m^{(\gamma)}}_{-1} &= \sum_{k=0}^N \,\Big( m(k) - m(k-1) \Big) \, P_{N,k} \nonumber \\
&\approx \sum_{k=0}^{N^{\gamma}} \, N^{-\gamma} \, P_{N,k}.
\end{align}
Now consider the splitting $\phi = \sum_{k=0}^N P_{N,k} \phi := \sum_{k=0}^N \phi_k$ for some antisymmetric $\phi \in L^2(\RRR^{3N})$. Then we find that
\be
\left( \widehat{m^{(\gamma)}}-\widehat{m^{(\gamma)}}_{-1} \right) \phi_k = 0 \quad \forall k > N^{\gamma}.
\ee
We use this fact to improve the estimate for the third term. Heuristically, $q_1q_2\psi^t$ has only contributions coming from large $k$ (at least in the first and second variable, which is enough to make the argument work). But for these contributions the ``derivative'' is zero, which makes the third term on the right-hand side of \eqref{alpha_derivative_summary} very small. This can also be used for the first term, where we gain additional projectors $q$ due to cancellations between Schr\"odinger and mean-field interactions, as explained after \eqref{v_v_mf}. The downside of using the weight function $m^{(\gamma)}(k)$ is that we gain only a prefactor $N^{-\gamma}$ from the ``derivative'' \eqref{derivative_m}, instead of $N^{-1}$ when we use $n(k)$. However, as it turns out, this effect can be controlled, and only leads to a worse convergence rate. This heuristic reasoning is made precise in Chapter~\ref{sec:alpha_m_dot_rigorous}, where we bound $\partial_t \alpha_{m^{(\gamma)}}(t)$ rigorously. For that, several lemmas are necessary, which we establish in Chapter~\ref{sec:general_lemmas}. Finally, in Chapter~\ref{sec:proofs_main_theorems_gen}, we use the conditions from Theorem~\ref{thm:estimates_terms_alpha_dot_beta_general} to calculate the convergence rate and to put the estimates together.

\absatz

\textbf{Outline for results of Section~\ref{sec:main_theorem_mf_x-s}.} Let us now discuss how to prove the results for $H_j^0=-\Delta_j+w^{(N)}(x_j)$ and interactions $v^{(N)}(x) = N^{-\beta} |x|^{-s}$. In order to apply Theorems \ref{thm:estimates_terms_alpha_dot_beta_n} and \ref{thm:estimates_terms_alpha_dot_beta_general}, we have to evaluate expressions like $v\star\rho_N^t$ and $v^2\star\rho_N^t$ (where $\rho_N^t=\sum_{i=1}^N|\varphi_i^t|^2$) for interactions $|x|^{-s}$ with weak or cut off singularity. It turns out that naturally these expressions can be bounded in terms of the total kinetic energy. A key role is played by the kinetic energy inequality due to Lieb and Thirring \cite{lieb:1975,lieb:2010},
\be\label{summary_kin_en_ineq}
\int_{\RRR^3} \Big( \rho_N^t(x) \Big)^{\frac{5}{3}} \, d^3x \leq C \, \sum_{i=1}^N \norm{\nabla \varphi_i^t}^2.
\ee
This inequality crucially depends on the fermionic nature of the wave function; for bosonic wave functions, it only holds with an extra factor $N^{\frac{2}{3}}$ on the right-hand side (think of the example of plane waves in a box from Chapter~\ref{sec:mf_fermions_const_E_kin_physics}). Let us now consider the mean-field interaction $|\cdot|^{-s} \star \rho_N^t$. We assume that the total kinetic energy $\sum_{i=1}^N ||\nabla \varphi_i^t||^2 \leq AN$. We first split the integration into two parts, over a ball with radius $R_N \propto N^{\frac{1}{3}}$ and its complement, then apply H\"older's inequality, and use \eqref{summary_kin_en_ineq} and $\int \rho_N^t = N$:
\begin{align}\label{summary_rho_x-1_calc}
\int_{\RRR^3} \frac{\rho_N^t(x)}{|x-y|^s} \,d^3x &= \int_{B_{R_N}(y)} \frac{\rho_N^t(x)}{|x-y|^s} \,d^3x + \int_{\overline{B_{R_N}(y)}} \frac{\rho_N^t(x)}{|x-y|^s} \,d^3x \nonumber \\
\eqexp{by H\"older} &\leq \left( \int_{B_{R_N}(y)} \rho_N^t(x)^{\frac{5}{3}} \,d^3x \right)^{\frac{3}{5}} \left( \int_{B_{R_N}(y)} |x-y|^{-\frac{5}{2}s} \,d^3x \right)^{\frac{2}{5}} + \nonumber \\
&\quad\quad \left( \int_{\overline{B_{R_N}(y)}} \rho_N^t(x) \,d^3x \right) \left( \sup_{x \in \overline{B_{R_N}(y)}} |x-y|^{-s} \right) \nonumber \\
&\leq \left( \int_{\RRR^3} \rho_N^t(x)^{\frac{5}{3}} \,d^3x \right)^{\frac{3}{5}} \left( \int_{B_{R_N}(0)} |x|^{-\frac{5}{2}s} \,d^3x \right)^{\frac{2}{5}} + \nonumber \\
&\quad\quad \left( \int_{\RRR^3} \rho_N^t(x) \,d^3x \right) \left( \sup_{x \in \overline{B_{R_N}(0)}} |x|^{-s} \right) \nonumber \\
\eqexp{by \eqref{summary_kin_en_ineq}} &\leq C N^{\frac{3}{5}} R_N^{\frac{6}{5}-s} + N R_N^{-s} \nonumber \\
&\leq C N^{1-\frac{s}{3}}.
\end{align}
Since $|x|^{-\frac{5}{2}s}$ is integrable over a ball only for $0 < s <\frac{6}{5}$, we restrict ourselves to those $s$. We thus showed that $\beta=1-\frac{s}{3}$ is the correct scaling exponent for interactions $v_s(x)=|x|^{-s}$. Note that this remains so when we cut off the singularity, i.e., \eqref{summary_rho_x-1_calc} can in general not be improved.

The condition $v^2\star\rho_N^t$ can be evaluated by similar methods (which we do in Chapter~\ref{sec:mean-field_scalings}). However, here we have to deal with a much stronger singularity. Thus, we either need stronger conditions on $\rho_N^t$, for example $||\rho_N^t||_{\infty} \leq C$ as we consider in Corollary~\ref{cor:estimates_terms_alpha_dot_Coulomb_rho_infty}, or we need to restrict ourselves to weaker or cut off singularities, as considered in Theorem~\ref{thm:E_kin_only}. Finally, Proposition~\ref{pro:coulombN1} can be proven by using that for $v^{(N)}(x)=N^{-1}|x|^{-1}$, the mean-field interaction is only of $O(N^{-\frac{1}{3}})$, as \eqref{summary_rho_x-1_calc} shows. The explicit proofs of the results are given in Chapter~\ref{sec:proofs_main_theorems}.

\section{Theorem and Sketch of Proof for Semiclassical Scaling}\label{sec:estimates_semiclassical}
The proof we outlined in Chapter~\ref{sec:outline_proof} was, in a sense, tailor-made for particles with average velocities of $O(1)$. In order to demonstrate that the $\alpha$-method also works for situations where this is not the case, we here give a derivation of the mean-field dynamics for the semiclassical case discussed in Chapter~\ref{sec:mf_fermions_semiclassical}. There, one has to use the fact that the average velocities are $O(N^{\frac{1}{3}})$. A derivation of the mean-field dynamics in this case has recently been given in \cite{benedikter:2013}. Here, we mostly reproduce the result obtained there; we actually use estimates about the propagation of properties of the initial data from \cite{benedikter:2013}. A slight improvement is that our conditions on the initial data are more transparent and general. Let us state our main theorem for the semiclassical case here, and give an outline of the proof, in particular, of what steps are different compared to Chapter~\ref{sec:outline_proof}. The full proof can be found in Appendix~\ref{sec:proof_sc_scaling}.

We consider, as discussed in Chapter~\ref{sec:mf_fermions_semiclassical}, the non-relativistic Schr\"odinger equation with semiclassical scaling (for simplicity, without external fields),
\be\label{outline_Schr_scaled_sc_app}
i N^{-\frac{1}{3}} \partial_t \psi^t = - N^{-\frac{2}{3}} \sum_{j=1}^N \Delta_{x_j} \psi^t + N^{-1} \sum_{1\leq i<j \leq N} v(x_i-x_j) \psi^t.
\ee
The corresponding semiclassical fermionic Hartree equations are
\be\label{outline_hartree_scaled_sc_app}
i N^{-\frac{1}{3}} \partial_t \varphi_j^t = \left( -N^{-\frac{2}{3}} \Delta + N^{-1} \left(v \star \rho_N^t \right) \right) \varphi_j^t,
\ee
and the semiclassical Hartree-Fock equations are
\be\label{outline_hartree_fock_scaled_sc_app}
i N^{-\frac{1}{3}} \partial_t \varphi_j^t = \left( -N^{-\frac{2}{3}} \Delta + N^{-1} \left(v \star \rho_N^t \right) \right) \varphi_j^t - N^{-1} \sum_{k=1}^N \left(v \star (\varphi_k^{t*}\varphi_j^t)\right) \varphi_k^t,
\ee
for $j=1,\ldots,N$, and recall $\rho_N^t = \sum_{i=1}^N |\varphi_i^t|^2$. Note that here we do not have to use the long-range behavior of the interaction, since we are interested in solutions in some constant, $N$-independent volume, i.e., very high densities. For technical reasons, we consider basically bounded interactions (the more exact conditions are stated in Theorem~\ref{thm:sc_main_thm} below). Note that for these interactions the exchange term is always subleading (taking the $N^{-1}$ from the scaling into account). Since the exchange term is easier to handle for bounded interactions, we can prove the theorem for both the fermionic Hartree and Hartree(-Fock) equations. The following theorem is analogous to \cite[Thm.\ 2.1]{benedikter:2013}. We write $p_1(0)=\sum_{j=1}^N \ketbr{\varphi_j^0}_1$ for the projector $p_1$ at time $t=0$. Recall that we denote the trace norm by $\norm[\tr]{\cdot}$ (see also Chapter~\ref{sec:trace_and_HS}).

\begin{theorem}\label{thm:sc_main_thm}
Let $\psi^t \in L^2(\RRR^{3N})$ be a solution to the Schr\"odinger equation \eqref{outline_Schr_scaled_sc_app} with antisymmetric initial condition $\psi^0 \in L^2(\RRR^{3N})$. Let $\varphi_1^t,\ldots,\varphi_N^t \in L^2(\RRR^3)$ be either solutions to the fermionic Hartree equations \eqref{outline_hartree_scaled_sc_app} or to the Hartree-Fock equations \eqref{outline_hartree_fock_scaled_sc_app}, with orthonormal initial conditions $\varphi_1^0,\ldots,\varphi_N^0 \in L^2(\RRR^3)$.

We assume that $v \in L^1(\RRR^3)$ and
\be\label{sc_thm_cond0}
\int d^3k \, (1+|k|^2)\, |\hat{v}(k)| < \infty,
\ee
where $\hat{v}$ is the Fourier transform of $v$. We also assume that the initial conditions $\varphi_1^0,\ldots,\varphi_N^0 \in L^2(\RRR^3)$ are such, that
\be\label{sc_thm_cond1}
\sup_{k\in\RRR^3} \frac{1}{1+|k|} \norm[\tr]{\, \left[ p_1(0), e^{ik\cdot x} \right] \,} \leq cN^{\frac{2}{3}},
\ee
\be\label{sc_thm_cond2}
\norm[\tr]{\, \left[ p_1(0), \nabla \right] \,} \leq c N,
\ee
for some constant $c>0$, where $p_1(0) = \sum_{j=1}^N \ketbr{\varphi_j^0}_1$.

Then, there are positive $C_1,C_2$, such that for all $t>0$,
\be\label{main_alpha_ineq_sc}
\alpha_n(t) \leq \exp\big(C_1 \exp(C_2 \, t)\big) \Big( \alpha_n(0) + N^{-1} \Big).
\ee
\end{theorem}

\noindent\textbf{Remarks.}
\begin{enumerate}
\setcounter{enumi}{\theremarks}
\item It follows from Lemma~\ref{lem:density_conv} that \eqref{main_alpha_ineq_sc} implies for density matrices the estimate
\be\label{main_dens_mat_ineq_sc}
\norm[\tr]{\mu^{\psi^t}_1 - \mu^{\bigwedge \varphi_j^t}_1} \leq \tilde{C}(t) \left( \norm[\tr]{\mu^{\psi^0}_1 - \mu^{\bigwedge \varphi_j^0}_1}^{\frac{1}{2}} + \frac{1}{\sqrt{N}} \right),
\ee
for some constant $\tilde{C}(t)$; a similar estimate holds in Hilbert-Schmidt norm.
\item The theorem also holds with external fields that are such that they preserve the bounds \eqref{sc_thm_cond1} and \eqref{sc_thm_cond2} for all $t$.
\end{enumerate}
\setcounter{remarks}{\theenumi}

Let us now outline the proof of Theorem~\ref{thm:sc_main_thm}. Recall that we are looking for a bound
\be
\partial_t \alpha_n(t) \leq C(t) \left(\alpha_n(t) +N^{-1} \right),
\ee
and then use the Gronwall Lemma to deduce \eqref{main_alpha_ineq_sc}. It turns out that for this proof it is sufficient to use $\alpha_n(t)$, i.e., the weight function $n(k)=\frac{k}{N}$. For $\psi^t$ a solution to the Schr\"odinger equation \eqref{outline_Schr_scaled_sc_app} and $\varphi_1^t,\ldots,\varphi_N^t$ solutions to the fermionic Hartree equations \eqref{outline_hartree_scaled_sc_app} (or the Hartree-Fock equations \eqref{outline_hartree_fock_scaled_sc_app}), we find (as in \eqref{alpha_derivative_n_gen}, using explicitly the scaling)
\begin{align}\label{outline_alpha_derivative_sc}
\partial_t \alpha_n(t) &= 2 N^{-\frac{2}{3}} \, \Im\, \bigSCP{\psi^t}{q_1\Big( (N-1)p_2v_{12}p_2 - V_1 \Big) p_1 \psi^t} \nonumber \\
&\quad + 2 N^{-\frac{2}{3}} \, \Im\, \bigSCP{\psi^t}{q_1q_2 (N-1)v_{12} p_1p_2 \psi^t} \nonumber \\
&\quad + 2 N^{-\frac{2}{3}} \, \Im\, \bigSCP{\psi^t}{q_1q_2 (N-1)v_{12} p_1q_2 \psi^t}
\end{align}
(where $V_1$ is either the direct or direct plus exchange term). As mentioned above, we cannot use here that we gain an additional $N^{-\frac{1}{3}}$ from the long-range behavior of the interaction (as we did, e.g., in Lemma~\ref{lem:scaling_x-s}); the mean-field term $V_1$ is of $O(N)$, such that it seems that $\partial_t \alpha_n(t)$ is of $O(N^{\frac{1}{3}})$. However, what we use now is that the average velocity of the particles is $O(N^{\frac{1}{3}})$ (due to the high density and the fact that we consider fermions). In Theorem~\ref{thm:sc_main_thm} we phrased this in the form that
\be\label{vel_O_1}
\norm[\tr]{\, \left[ p_1(0), e^{ik\cdot x} \right] \,} \leq c N^{\frac{2}{3}}
\ee
(where $c$ can depend on $k$). An argument why \eqref{vel_O_1} (together with \eqref{sc_thm_cond2}) expresses the fact that the average velocities are $O(N^{\frac{1}{3}})$ can be found in \cite{benedikter:2013}; at this point, let us just recall from Chapter~\ref{sec:mf_fermions_semiclassical} that a semiclassical density matrix has roughly the form
\be
\phi\left(N^{\frac{1}{3}}(x-y)\right) \chi(x+y).
\ee
It turns out, that \eqref{vel_O_1} captures that there is an additional factor $N^{\frac{1}{3}}$ in the ``velocity profile'' $\phi$ of the density matrix. One part of the proof of Theorem~\ref{thm:sc_main_thm} is to propagate the conditions \eqref{sc_thm_cond1} and \eqref{sc_thm_cond2} in time, i.e., to show that 
\be\label{sc_thm_cond1_t}
\sup_{k\in\RRR^3} \frac{1}{1+|k|} \norm[\tr]{\, \left[ p_1(t), e^{ik\cdot x} \right] \,} \leq c(t) N^{\frac{2}{3}},
\ee
\be\label{sc_thm_cond2_t}
\norm[\tr]{\, \left[ p_1(t), \nabla \right] \,} \leq c(t) N,
\ee
where $p_1(t) = \sum_{j=1}^N \ketbr{\varphi_j^t}_1$, and $\varphi_1^t,\ldots,\varphi_N^t$ are solutions to the Hartree(-Fock) equations. Note that \eqref{sc_thm_cond2_t} is necessary to show that \eqref{sc_thm_cond1_t} holds, but will not be used in the estimates for $\partial_t \alpha_n(t)$. The exact statement is Lemma~\ref{lem:sc_prop_sc} which has been proven in \cite{benedikter:2013}. The constants $c(t)$ are of the form $c(t) = c_1 \exp(c_2 t)$ (which is the reason why we get the double exponential in the estimate \eqref{main_alpha_ineq_sc}; the other exponential comes from the Gronwall argument for $\partial_t\alpha_n(t)$). Let us from now on take for granted that \eqref{sc_thm_cond1_t} holds for all times and note that (due to $p_1q_1=0$)
\be\label{vel_O_1_t}
\norm[\tr]{\, p_1 e^{ik\cdot x} q_1 \,} = \norm[\tr]{\, \left[ p_1, e^{ik\cdot x} \right] q_1 \,} \leq  \norm[\tr]{\, \left[ p_1, e^{ik\cdot x} \right] \,} \leq C N^{\frac{2}{3}},
\ee
where for ease of notation we do not write out the $t$-dependence of $p_1,q_1$ anymore (also constants $C$ can be time-dependent).

The estimate \eqref{vel_O_1_t} can now be used to gain an additional factor $N^{-\frac{1}{3}}$ in the time derivative of $\alpha_n(t)$. Let us here only regard the third term from \eqref{outline_alpha_derivative_sc}, the $qq$-$pq$ term, which is again the most simple to estimate. Using the Fourier decomposition of the interaction potential, $v(x) = \int d^3k \, \hat{v}(k) e^{ikx}$, and Cauchy-Schwarz, we find
\begin{align}\label{outline_term3}
N^{\frac{1}{3}} \, \Big| \bigSCP{\psi^t}{q_1q_2 v_{12} p_1q_2 \psi^t} \Big| &= N^{\frac{1}{3}} \, \Big\lvert \int d^3k \, \hat{v}(k) \, \bigSCP{\psi^t}{q_1q_2 e^{ik(x_1-x_2)} p_1q_2 \psi^t} \Big\rvert \nonumber \\
&= N^{\frac{1}{3}} \, \Big\lvert \int d^3k \, \hat{v}(k) \, \bigSCP{\psi^t}{q_1e^{ikx_1}p_1 q_2 e^{-ikx_2} q_2 \psi^t} \Big\rvert.
\end{align}
It is now convenient to use the singular value decomposition of the (compact) operator $q_1e^{ikx_1}p_1$, i.e., we use $q_1e^{ikx_1}p_1 = \sum_{\ell} \mu_{\ell} \ketbra{\phi_{\ell}}{\tilde{\phi}_{\ell}}_1$ for some orthonormal $\{ \phi_{\ell} \}_{\ell\in\NNN}$ and $\{ \tilde{\phi}_{\ell} \}_{\ell\in\NNN}$ and $\mu_{\ell}>0$, where $\sum_{\ell} \mu_{\ell} = \norm[\tr]{q_1e^{ikx_1}p_1}$. Then, by Cauchy-Schwarz, \eqref{trace_A_1} and \eqref{vel_O_1_t},
\begin{align}
N^{\frac{1}{3}} \, \Big\lvert \bigSCP{\psi}{q_1q_2 v_{12} p_1q_2 \psi} \Big\rvert
&= N^{\frac{1}{3}} \Big\lvert \int d^3k \, \hat{v}(k) \sum_{\ell} \mu_{\ell} \bigSCP{\psi}{q_2\ketbra{\phi_{\ell}}{\tilde{\phi}_{\ell}}_1 e^{-ikx_2}q_2 \psi} \Big\rvert \nonumber \\
\eqexp{by C.-S.} &\leq N^{\frac{1}{3}} \int d^3k \, |\hat{v}(k)| \sum_{\ell} \mu_{\ell} \, \Big|\Big|\bra{\phi_{\ell}}_1 q_2\psi\Big|\Big| \norm{\bra{\tilde{\phi}_{\ell}}_1 q_2 \psi} \nonumber \\
\eqexp{by \eqref{trace_A_1}} &\leq N^{\frac{1}{3}} \int d^3k \, |\hat{v}(k)| \norm[\tr]{q_1e^{ikx_1}p_1} N^{-1} \norm{q_2\psi}^2.
\end{align}

Therefore, if $\int d^3k \, |\hat{v}(k)| < \infty$ (or, more exactly, when \eqref{sc_thm_cond0} holds), we find the bound
\be
N^{\frac{1}{3}} \, \Big| \bigSCP{\psi^t}{q_1q_2 v_{12} p_1q_2 \psi^t} \Big| \leq C \alpha_n(t).
\ee
For the first and second term in \eqref{outline_alpha_derivative_sc}, we proceed similar to Chapter~\ref{sec:outline_proof}, and use the singular value decomposition of $q_1e^{ikx_1}p_1$ again. As in Chapter~\ref{sec:outline_proof}, an additional boundary term of $O(N^{-1})$ arises, such that in total we find
\be
\partial_t \alpha_n(t) \leq C(t) \Big( \alpha_n(t) + N^{-1} \Big).
\ee

\section{Outlook}\label{sec:outlook}
In this thesis we have seen how the fermionic Hartree(-Fock) equations can be derived from the microscopic Schr\"odinger dynamics in a many particle limit. An understanding of how and why this is the case is very important with respect to the more general goal to understand how macroscopic (or effective) behavior arises from microscopic physics. With the results outlined in Chapter~\ref{sec:mf_fermions_lit} and this work, we now have a good understanding of how this works for mean-field descriptions for (non-relativistic) fermions. However, several more detailed and more advanced questions are still open.

\begin{itemize}
\item Concerning the results in this work, it would be interesting to show that the conditions we formulated on the solutions to the fermionic Hartree equations (in particular, in Theorem~\ref{thm:estimates_terms_alpha_dot_beta_n}, Theorem~\ref{thm:estimates_terms_alpha_dot_beta_general} or Corollary~\ref{cor:estimates_terms_alpha_dot_Coulomb_rho_infty}) hold for the physically very relevant case of Coulomb interaction under suitable smoothness conditions on the initial data. It might furthermore be interesting to investigate other physically relevant interactions.
\item A derivation for the semiclassical case with Coulomb interaction is still an open problem.
\item One could try a derivation of mean-field limits for fermions for relativistic equations, e.g., the Dirac-Fock equations (see, e.g., \cite{hainzl:2005} and references therein). A first step in this direction is \cite{benedikter:2014} where a pseudo-relativistic Hamiltonian is considered.
\item It would be interesting to identify relevant situations where the exchange term is not subleading. Alternatively, one could look for scenarios where the exchange term is subleading, but gives a larger contribution to the dynamics than the error terms in a derivation of the fermionic Hartree equations. Then one could try to show that the Hartree-Fock equations give a better approximation to the Schr\"odinger dynamics than the Hartree equations.
\item It might be possible to give technically better estimates for $\partial_t \alpha_f(t)$, which reflect even better that correlations are caused by fluctuations, as discussed in Chapter~\ref{sec:mf_fermions_const_E_kin_fluc} (see, in particular, Equation~\ref{variance}). A first step in this direction is \cite{mitrouskas:2013} where the fluctuations around the mean-field are analyzed much more carefully.
\item Another topic is to analyze and derive other effective evolution equations for fermions. For example, it could be interesting to consider scalings for fermions that are similar to Gross-Pitaevskii scalings for bosons. Another very interesting scaling is the so-called kinetic limit for fermions (see, e.g., \cite{lukkarinen:2009}) where the long-time behavior is investigated and the dynamics is approximated by a quantum Boltzmann equation.
\end{itemize}

\part{Proof of Main Results}\label{pt:part_two}

\chapter{Notation and Preliminaries}\label{sec:notation}
\section{Notation and Basic Inequalities}
Let us first establish some notation that we use throughout the following chapters. We denote by $\Hilbert$ a Hilbert space and  we always assume it is separable. Its inner product is denoted by $\scp{\cdot}{\cdot}$ or $\SCP{\cdot}{\cdot}$ and the norm of any $f\in \Hilbert$ by $\norm{f} = \sqrt{\scp{f}{f}}$. For any $z\in \CCC^d$ we write $|z|^2:=\sum_{i=1}^d |z_i|^2$, where for $z_i\in \CCC$, $|z_i|^2=z_i^*z_i$, with ${}^*$ denoting complex conjugation. The Hilbert space of complex square integrable functions on $\RRR^d$ is denoted by $L^2(\RRR^d) = L^2(\RRR^d, \CCC)$ and $H^1(\RRR^d)$ denotes the first Sobolev space, i.e.,
\be
H^1(\RRR^d) = \left\{ f \in L^2(\RRR^d):  \norm{\nabla f} < \infty \right\}.
\ee
For $f\in L^2(\RRR^d)$ we sometimes write
\be
\norm{f}^2 := \int_{\RRR^d} |f(x)|^2 \, d^dx = \int |f|^2.
\ee
In order to differentiate between the scalar product on $L^2(\RRR^{3N})$ and scalar products on another $L^2(\RRR^d)$ (usually $L^2(\RRR^3)$) we always write $\SCP{\cdot}{\cdot}$ for the scalar product on $L^2(\RRR^{3N})$. We denote by $\SCP{\cdot}{\cdot}_{a+1,\ldots,N}$ the scalar product only in the variables $x_{a+1},\ldots,x_N$, i.e., it is a ``partial trace'' or ``partial scalar product'', formally defined for any $\chi,\psi\in L^2(\RRR^{3N})$ by
\be\label{partial_scp}
\SCP{\psi}{\chi}_{a+1,\ldots,N}(x_1,\ldots,x_a) := \int d^3x_{a+1} \ldots \int d^3x_N \, \psi^*(x_1,\ldots,x_N) \chi(x_1,\ldots,x_N),
\ee
which should be regarded as a vector in $L^1(\RRR^{3a})$ (for $\chi=\psi$, it is the diagonal of the reduced $a$-particle density matrix, see Chapter~\ref{sec:reduced_density_matrices}). As mentioned in Definition~\ref{def:projectors}, for any $\varphi\in L^2(\RRR^3)$, we use the bra-ket notation
\be
p_m^{\varphi} = \ketbra{\varphi}{\varphi}_m = \ketbra{\varphi(x_m)}{\varphi(x_m)}
\ee
for the projector defined by
\be
\left(p_m^{\varphi}\psi\right)(x_1,\ldots,x_N) = \varphi(x_m)\int\varphi^*(x_m)\psi(x_1,\ldots,x_N)d^3x_m
\ee
for any $\psi \in L^2(\RRR^{3N})$. In other words,
\be
p_m^{\varphi} = \underbrace{\id \otimes \ldots \otimes \id}_{m-1 ~ \text{times}} \otimes \,\ketbra{\varphi}{\varphi}\, \otimes \underbrace{\id \otimes \ldots \otimes \id}_{N-m ~ \text{times}}.
\ee
In the same style we denote by $\ket{\cdot}_m$ a vector in $L^2(\RRR^3)$ acting only on the $m$-th variable of $L^2\big((\RRR^3)^N\big)$, by $\bra{\cdot}_m$ its dual, and by $\scp{\cdot}{\cdot}_m$ the scalar product only in the $m$-th variable.

For any operator $A:L^2(\RRR^d)\to L^2(\RRR^d)$ we denote the operator norm by
\begin{align}
\norm[\op]{A} &:= \sup_{0 \neq f \in L^2(\RRR^d)} \frac{\norm{Af}}{\norm{f}} \nonumber \\
&\;= \sup_{f \in L^2(\RRR^d), \norm{f} = 1} \norm{Af}.
\end{align}
For self-adjoint $A$, the operator norm can be expressed as
\be
\norm[\op]{A} = \sup_{f \in L^2, \norm{f}=1} |\scp{f}{Af}| \, ,
\ee
see, e.g., \cite{reedsimon1:1980}. We denote the trace norm by $\norm[\tr]{\cdot}$ and the Hilbert-Schmidt norm by $\norm[\HS]{\cdot}$, see Section~\ref{sec:trace_and_HS} for more details. We denote the commutator of two operators $A,B$ by
\be
\left[ A, B \right] := AB - BA.
\ee

Given $\varphi_1,\ldots,\varphi_N \in L^2(\RRR^3)$ we denote the density by $\rho_N(x) = \sum_{i=1}^N |\varphi_i(x)|^2$ and the total kinetic energy by
\be
E_{\kin,\mf} = \sum_{i=1}^N \scp{\varphi_i}{(-\Delta)\varphi_i} = \sum_{i=1}^N \norm{\nabla \varphi_i}^2.
\ee
We use the notation $\prod \varphi$ (for any $\varphi \in L^2(\RRR^3)$) as abbreviation for the simple product
\be
\left(\prod_{i=1}^N \varphi\right)(x_1,\ldots,x_N) = \prod_{i=1}^N \varphi(x_i)
\ee
and $\bigwedge \varphi_j$ as abbreviation for the antisymmetrized product
\be
\left(\bigwedge_{j=1}^N \varphi_j\right)(x_1,\ldots,x_N) = \frac{1}{\sqrt{N}} \sum_{\sigma \in S_N} (-1)^{\sigma} \prod_{j=1}^N \varphi_{\sigma(j)}(x_i),
\ee
where $S_N$ is the symmetric group and $(-1)^\sigma$ the sign of the permutation $\sigma$.

Given a function $h:\RRR^d\to\RRR$ we introduce $h_{12}:\RRR^d \times \RRR^d \to \RRR, h_{12}(x_1,x_2)=h(x_1-x_2)$. We use the Landau notation $O(\cdot)$, i.e., $f(N) \in O(g(N))$ means that $\lim_{N\to\infty}\frac{f(N)}{g(N)} < \infty$. We always denote by $B_R(x)$ the open ball with radius $R$ around $x$, i.e.,
\be
B_R(x) = \left\{ y \in \RRR^d: |x-y| < R \right\}.
\ee
For any set $\Omega \subset \RRR^d$ we write $\overline{\Omega} = \RRR^d \setminus \Omega$.

Let us also list some well-known inequalities that we frequently use (for proofs, see, e.g., \cite{liebloss:2001}). We denote by $(\Omega,\Sigma,\mu)$ a general measure space.
\begin{itemize}
\item \emph{Cauchy-Schwarz inequality for $\CCC^N$.} Let $a,b \in \CCC^N$. Then
\be
\left\lvert \, \sum_{k=1}^N a_k b_k \, \right\rvert \leq \sqrt{\sum_{k=1}^N \lvert a_k \rvert^2} \sqrt{\sum_{k=1}^N \lvert b_k \rvert^2}.
\ee
\item \emph{Cauchy-Schwarz inequality for $L^2(\Omega)$.} Let $\psi,\chi \in L^2(\Omega)$. Then
\be
\left\lvert \scp{\psi}{\chi} \right\rvert \leq \norm{\psi} \norm{\chi}.
\ee
\item \emph{H\"older's inequality.} Let $1 \leq p,q \leq \infty$ with $\frac{1}{p} + \frac{1}{q} = 1$. Let $f \in L^p(\Omega)$, $g \in L^q(\Omega)$. Then $fg \in L^1(\Omega)$ and
\be
\norm[1]{fg} \leq \norm[p]{f} \norm[q]{g}.
\ee
\end{itemize}

\section{More about the Projectors}\label{sec:properties_projectors}
In this section we summarize some properties of the projectors from Definition~\ref{def:projectors} and define more projectors that we need in the course of the proofs in the following chapters.

\begin{definition}\label{def:projectors2}
Let $\varphi_1, \ldots, \varphi_N \in L^2(\RRR^3)$ be orthonormal.
\begin{enumerate}[(a)]
\item  We define
\be
p^{\varphi_j} := \sum_{m=1}^N p_m^{\varphi_j}
\ee
and
\be
q^{\varphi_j} = 1- p^{\varphi_j}.
\ee
\item For $a \leq n < N$ we define
\be
P_a^{\{i_1,\ldots,i_n\}} := \left(\prod_{m=1}^a q_{i_m} \prod_{m=a+1}^n p_{i_m}\right)_{\sym}.
\ee
\end{enumerate}
\end{definition}

The operator
\be
p_m^{\varphi} = \ket{\varphi}\bra{\varphi}_m = \ket{\varphi(x_m)}\bra{\varphi(x_m)}
\ee
is indeed a projector on $L^2(\RRR^{3N})$. For $\varphi_i \perp \varphi_j$ and all $m,n=1,\ldots,N$ we have
\be
p_m^{\varphi_i} p_m^{\varphi_j} = 0 \quad\text{and}\quad \left[ p_m^{\varphi_i}, p_n^{\varphi_j} \right] = 0.
\ee
From that we conclude the following properties of the projectors $p_m$ and $q_m=1-p_m$:
\be
p_m q_m = 0,
\ee
\be
[p_m,p_n] = [p_m,q_n] = [q_m,q_n] = 0,
\ee
for all $m,n=1,\ldots,N$. For antisymmetric $\psi_{as} \in L^2(\RRR^{3N})$, we have
\be
p_m^{\varphi}p_n^{\varphi} \psi_{as} = 0
\ee
for all $m \neq n$, since
\begin{align}
\left(p_m^{\varphi}p_n^{\varphi} \psi_{as}\right)(x_1,\ldots,x_N) &= \varphi(x_m)\varphi(x_n) \int dx_m dx_n \varphi(x_m)^*\varphi(x_n)^* \psi_{as}(\ldots,x_m,\ldots,x_n,\ldots) \nonumber \\
&= - \varphi(x_m)\varphi(x_n) \int dx_m dx_n \varphi(x_m)^*\varphi(x_n)^* \psi_{as}(\ldots,x_n,\ldots,x_m,\ldots) \nonumber \\
&= - \varphi(x_m)\varphi(x_n) \int dx_n dx_m \varphi(x_n)^*\varphi(x_m)^* \psi_{as}(\ldots,x_m,\ldots,x_n,\ldots).
\end{align}
Therefore, on antisymmetric functions in $L^2(\RRR^{3N})$, the operators $p^{\varphi_i}$ and $q^{\varphi_i} = 1 - p^{\varphi_i}$ are projectors, and for any antisymmetric $\psi_{as} \in L^2(\RRR^{3N})$ we have
\be
p^{\varphi_i} q^{\varphi_i} \psi_{as} = 0,
\ee
\be
[p^{\varphi_i},p^{\varphi_j}] \psi_{as} = [p^{\varphi_i},q^{\varphi_j}] \psi_{as} = [q^{\varphi_i},q^{\varphi_j}] \psi_{as} = 0,
\ee
for all $i,j=1,\ldots,N$.

When one considers the operator norm of the projectors $p_m^{\varphi}$ it makes an important difference if it is calculated on all $L^2$ functions or only on antisymmetric functions in $L^2$, as the following lemma shows. Recall that $\SCP{\cdot}{\cdot}_{a+1,\ldots,N}$ denotes the scalar product only in the variables $x_{a+1},\ldots,x_N$.
\begin{lemma}\label{lem:projector_norms}
\begin{enumerate}[(a)]
\item Let $\psi_{as}\in L^2(\RRR^{3N})$ be antisymmetric and normalized. Then, for all $m=1,\ldots,N$,
\be
\SCP{\psi_{as}}{p_m^{\varphi}\psi_{as}} \leq \frac{1}{N}.
\ee
\item Let $\psi_{as}^{1,\ldots,a}\in L^2(\RRR^{3N})$ be antisymmetric in all variables except $x_1,\ldots,x_a$. Let $m \in \{ a+1,\ldots,N \}$. Then
\be
\SCP{\psi_{as}^{1,\ldots,a}}{p_m^{\varphi}\psi_{as}^{1,\ldots,a}} \leq \frac{1}{N-a} \SCP{\psi_{as}^{1,\ldots,a}}{\psi_{as}^{1,\ldots,a}}.
\ee
Furthermore, also
\be\label{scp_a-N}
\SCP{\psi_{as}^{1,\ldots,a}}{p_m^{\varphi}\psi_{as}^{1,\ldots,a}}_{a+1,\ldots,N}(x_1,\ldots,x_a) \leq \frac{1}{N-a} \SCP{\psi_{as}^{1,\ldots,a}}{\psi_{as}^{1,\ldots,a}}_{a+1,\ldots,N}(x_1,\ldots,x_a),
\ee
for almost all $x_1,\ldots,x_a$ (with the definition of $\SCP{\cdot}{\cdot}_{a+1,\ldots,N}$ from \eqref{partial_scp}).
\end{enumerate}
\end{lemma}

\begin{proof}
\begin{enumerate}[(a)]
\item First, note that for all antisymmetric $\psi_{as}$, $\norm{p^{\varphi}\psi_{as}}\leq \norm{\psi_{as}}$, since $p^{\varphi}$ is a projector on antisymmetric $\psi_{as}$. Using this and the antisymmetry of $\psi_{as}$ we find
\be
\SCP{\psi_{as}}{p_m^{\varphi}\psi_{as}} = \frac{1}{N} \SCP{\psi_{as}}{\sum_{n=1}^N p_n^{\varphi}\psi_{as}} = \frac{1}{N} \SCP{\psi_{as}}{p^{\varphi}\psi_{as}} = \frac{1}{N} \norm{p^{\varphi}\psi_{as}}^2 \leq \frac{1}{N} \norm{\psi_{as}}^2.
\ee

\item Now suppose that $\psi_{as}^{1,\ldots,a}$ is antisymmetric in all variables except $x_1,\ldots,x_a$ and let $m \in \{ a+1,\ldots,N \}$. Then $\sum_{n=a+1}^N p_n^{\varphi}$ is still a projector on those functions, and
\begin{align}\label{antisymm1a_norm}
\norm{\sum_{n=a+1}^N p_n^{\varphi} \psi_{as}^{1,\ldots,a}}^2 &= \bigSCP{\sum_{\ell=a+1}^N p_{\ell}^{\varphi}\psi_{as}^{1,\ldots,a}}{\sum_{n=a+1}^N p_n^{\varphi} \psi_{as}^{1,\ldots,a}} \nonumber \\
&= \bigSCP{\psi_{as}^{1,\ldots,a}}{\sum_{n=a+1}^N p_n^{\varphi} \psi_{as}^{1,\ldots,a}} \nonumber \\
&\leq \norm{\sum_{n=a+1}^N p_n^{\varphi} \psi_{as}^{1,\ldots,a}}.
\end{align}
Therefore, using Cauchy-Schwarz,
\begin{align}\label{antisymm1a_estimate}
\SCP{\psi_{as}^{1,\ldots,a}}{p_m^{\varphi}\psi_{as}^{1,\ldots,a}} &= \frac{1}{N-a} \bigSCP{\psi_{as}^{1,\ldots,a}}{\sum_{n=a+1}^N p_n^{\varphi} \psi_{as}^{1,\ldots,a}} \nonumber \\
&\leq \frac{1}{N-a} \norm{ \psi_{as}^{1,\ldots,a}} \norm{\sum_{n=a+1}^N p_n^{\varphi} \psi_{as}^{1,\ldots,a}} \nonumber \\
&\leq \frac{1}{N-a} \SCP{\psi_{as}^{1,\ldots,a}}{\psi_{as}^{1,\ldots,a}}.
\end{align}
Both \eqref{antisymm1a_norm} and \eqref{antisymm1a_estimate} remain true if the scalar products (and the corresponding norms) are only partial, i.e., when it is integrated only in the variables $x_{a+1},\ldots,x_N$. \hfill $\qedhere$
\end{enumerate}
\end{proof}

By using the example of the antisymmetrized product state $\bigwedge_{j=1}^{N} \varphi_j$, we see that the operator norm on antisymmetric functions $\norm[\op,\as]{\cdot}$ of $p_m^{\varphi_k}$ is indeed
\be
\norm[\op,\as]{p_m^{\varphi_k}} = \frac{1}{\sqrt{N}},
\ee
while in general 
\be
\norm[\op]{p_m^{\varphi_k}} = 1,
\ee
which can be seen by using the product state $\prod_{j=1}^N\varphi_k(x_j)$. Let us make one more remark (which is not necessary for the proofs later). One could as well define
\be
P_{N,k}' := \left( \prod_{j=1}^k q^{\varphi_j} \prod_{j=k+1}^N p^{\varphi_j} \right)_{\sym} = \sum_{\vec{a} \in \AAA_k} \prod_{j=1}^N (p^{\varphi_j})^{1-a_j}(q^{\varphi_j})^{a_j},
\ee
with the same notation as in Definition~\ref{def:projectors}, and
\be
\widehat{f'} = \sum_{k=0}^N f(k) P'_{N,k},
\ee
i.e., one could define $P_{N,k}$, $\widehat{f}$ and $\alpha_f$ with the projectors $p^{\varphi_j}$ instead of $p_m$. However, on antisymmetric functions both definitions coincide, i.e., for all antisymmetric $\psi_{as} \in L^2(\RRR^{3N})$,
\be
P_{N,k}\psi_{as} = P'_{N,k}\psi_{as}.
\ee
This can be seen by multiplying out $P_{N,k}\psi_{as}$ and $P'_{N,k}\psi_{as}$, and using $p_m^{\varphi_i}p_m^{\varphi_j} = 0$ for $i \neq j$ and $p_m^{\varphi}p_n^{\varphi}\psi_{as} = 0$ for $n \neq m$.

\chapter{Density Matrices}\label{sec:density_matrices}
In this chapter, we prove the results from Chapter~\ref{sec:dens_mat_summary} about the relation of $\alpha_f(t)$ to the reduced density matrices of $\psi^t$ and of the antisymmetrized product state $\bigwedge \varphi_j^t$.

\section{Trace Norm and Hilbert-Schmidt Norm}\label{sec:trace_and_HS}
We first give a brief overview of the definition and some properties of trace class and Hilbert-Schmidt operators. The following well-known statements and their proofs can be found, for example, in \cite[chapter~VI]{reedsimon1:1980}.

Let $\Hilbert$ be a Hilbert space and $A,B:\Hilbert\to\Hilbert$. Let $\{ \varphi_i \}_{i\in\NNN}$ be an orthonormal basis of $\Hilbert$. For any bounded positive operator $A:\Hilbert\to\Hilbert$ we define the trace of $A$ as $\tr(A)=\sum_{i\in\NNN} \scp{\varphi_i}{A \varphi_i}$. It is independent of the chosen orthonormal basis and
\begin{itemize}
\item $\tr(A+B) = \tr(A) + \tr(B)$,
\item $\tr(\lambda A) = \lambda \, \tr(A)$ for all $\lambda>0$,
\item $\tr(UAU^{-1}) = \tr(A)$ for any unitary operator $U$.
\end{itemize}
A bounded operator $A$ is called trace class if and only if $\tr|A| < \infty$. The trace class operators are a Banach space with norm $\norm[\tr]{\cdot}=\tr|\cdot|$. We have $|\tr(A)| \leq \norm[\tr]{A}$. If $A$ is trace class then so is $A^*$, the adjoint of $A$. If $A$ is trace class and $B$ is bounded then $AB$ and $BA$ are trace class. Every trace class operator is compact.

A bounded operator $A$ is called a Hilbert-Schmidt operator if and only if $\tr(A^*A) < \infty$. The Hilbert-Schmidt operators are a Hilbert space with the scalar product defined as $(A,B) = \sum_{i\in\NNN} \scp{\varphi_i}{A^*B \varphi_i}$ and norm $\norm[\HS]{A}=\sqrt{\tr(A^*A)}$. Every Hilbert-Schmidt operator is compact.  On $L^2$ spaces, Hilbert-Schmidt operators have a simple form, as the following statement shows. Let $(M,\mu)$ be a measure space. A bounded operator $A$ on $L^2(M,d\mu)$ is Hilbert-Schmidt if and only if there is a $K \in L^2(M \times M, d\mu \otimes d\mu)$ with
\be
(Af)(x) = \int K(x,y)f(y)d\mu(y).
\ee
Then
\be
\norm[\HS]{A}^2 = \int |K(x,y)|^2 d\mu(x) d\mu(y).
\ee
For a positive trace class operator $A: L^2(\RRR^d) \to L^2(\RRR^d)$ ($d\in\NNN$) with \emph{continuous integral kernel} $K(x,y)$, we have that
\be
\norm[\tr]{A} = \tr(A) = \int K(x,x) \, d^dx.
\ee
If $A$ is trace class then it is also a Hilbert-Schmidt operator; in fact
\be
\norm[\op]{A} \leq \norm[\HS]{A} \leq \norm[\tr]{A}.
\ee
Self-adjoint trace class operators $A$ can be diagonalized with real eigenvalues $\lambda_i$ ($i \in \NNN$) and we have
\be
\norm[\op]{A} = \sup_{i \in \NNN} |\lambda_i|, ~\quad \norm[\HS]{A}^2 = \sum_{i \in \NNN} |\lambda_i|^2, ~\quad \norm[\tr]{A} = \sum_{i \in \NNN} |\lambda_i|^.
\ee

Finally we collect some inequalities (which hold whenever the respective norms exist), which we frequently use in the proofs of the lemmas in Chapter~\ref{sec:reduced_density_matrices}:
\be\label{tr_op_tr}
\norm[\tr]{AB} \leq \norm[\op]{A} \norm[\tr]{B},
\ee
\be\label{HS_op_HS}
\norm[\HS]{AB} \leq \norm[\op]{A} \norm[\HS]{B},
\ee
\be\label{tr_HS_HS}
\norm[\tr]{AB} \leq \norm[\HS]{A} \norm[\HS]{B}.
\ee

\section{Convergence of Reduced Density Matrices}\label{sec:reduced_density_matrices}
For any normalized symmetric or antisymmetric $\psi \in L^2(\RRR^{3N})$ we define the \emph{reduced k-particle density matrix} $\mu^{\psi}_k: L^2(\RRR^{3k}) \to L^2(\RRR^{3k})$ by its integral kernel
\begin{align}\label{definition_dens_mat}
&\mu^{\psi}_k(x_1,\ldots,x_k;y_1,\ldots,y_k) \nonumber \\
&= \int \psi(x_1,\ldots,x_k,x_{k+1},\ldots,x_N) \psi^*(y_1,\ldots,y_k,x_{k+1},\ldots,x_N) \, d^3x_{k+1} \ldots d^3x_N.
\end{align}
Note that $\mu^{\psi}_k$ has an eigenfunction expansion.\footnote{One can write
\be
\mu^{\psi}_k(x_1,\ldots,x_k;y_1,\ldots,y_k) = \sum_{i=1}^{\infty} \mu_i \phi_i(x_1,\ldots,x_k) \phi_i^*(y_1,\ldots,y_k),
\ee
with $\phi_i \in L^2(\RRR^{3k})$ and $\mu_i > 0$. Then, in particular, the diagonal is
\be
\mu^{\psi}_k(x_1,\ldots,x_k;x_1,\ldots,x_k) = \sum_{i=1}^{\infty} \mu_i \left|\phi_i(x_1,\ldots,x_k)\right|^2.
\ee
}
Reduced density matrices have the following well-known properties.

\begin{lemma}\label{lem:properties_density_matrix}
\begin{enumerate}[(a)]
\item $\mu^{\psi}_k$ is non-negative, i.e., $\scp{f}{\mu^{\psi}_k f} \geq 0 \quad \forall f \in L^2(\RRR^{3k})$,
\item $\norm[\tr]{\mu^{\psi}_k}=\tr\big(\mu^{\psi}_k\big) = 1$,
\item For antisymmetric $\psi_{as}$, $\norm[\op]{\mu^{\psi_{as}}_1} \leq \frac{1}{N}$.
\end{enumerate}
\end{lemma}
In order to get used to notation we provide a proof for Lemma~\ref{lem:properties_density_matrix}.

\begin{proof}
\begin{enumerate}[(a)]
\item For all $f \in L^2(\RRR^{3k})$ we find
\begin{align}
\scp{f}{\mu^{\psi}_k f} &= \SCP{\psi}{ \ketbra{f}{f} \, \psi} \nonumber \\
&= \int dx_{k+1} \dots dx_N \left| \int dx_1 \dots dx_k f^*(x_1,\ldots,x_k) \psi(x_1,\ldots,x_N) \right|^2 \nonumber \\
&\geq 0.
\end{align}
\item Recall that $\psi \in L^2(\RRR^{3N})$ is normalized. Since $\mu^{\psi}_k$ is non-negative, $\norm[\tr]{\mu^{\psi}_k}=\tr\big(\mu^{\psi}_k\big)$. Also,
\be
\tr\big(\mu^{\psi}_k\big) = \SCP{\psi}{\psi} = 1.
\ee
\item For all $f \in L^2(\RRR^3)$ with $\norm{f} = 1$ and antisymmetric $\psi_{as} \in L^2(\RRR^{3N})$ we find
\begin{align}
\scp{f}{\mu^{\psi_{as}}_1 f}  &= \SCP{\psi_{as}}{ \ketbra{f(x_1)}{f(x_1)} \, \psi_{as}} \nonumber \\
&= \SCP{\psi_{as}}{ p_1^f \psi_{as}} \nonumber \\
&= \frac{1}{N} \bigSCP{\psi_{as}}{ \sum_{i=1}^N p_i^f \psi_{as}}.
\end{align}
Since $p^f = \sum_{i=1}^N p_i^f$ is a projector on antisymmetric $\psi_{as}$ we have
\be
\SCP{\psi_{as}}{ p^f \psi_{as}} \leq 1.
\ee
Since $\norm[\op]{\mu^{\psi}_1} = \sup_{f \in L^2, \norm{f}=1} \scp{f}{\mu^{\psi}_1 f}$ the statement follows. \hfill $\qedhere$
\end{enumerate}
\end{proof}

Let us give some examples of one- and two-particle density matrices. For a simple product $\prod_{i=1}^N \varphi$ (a bosonic condensed state) we find
\begin{align}
&\mu^{\prod \varphi}_1 = p_1^{\varphi}, \\
&\mu^{\prod \varphi}_2 = p_1^{\varphi}p_2^{\varphi}.
\end{align}
For an antisymmetrized product state $\bigwedge_{j=1}^{N} \varphi_j$ we find
\begin{align}
\mu^{\bigwedge \varphi_j}_1 &= \frac{1}{N} p_1, \\
\mu^{\bigwedge \varphi_j}_2 &= \frac{1}{N(N-1)} \left(p_1p_2 - \sum_{i,j=1}^N \ketbra{\varphi_i}{\varphi_j}_1 \, \ketbra{\varphi_j}{\varphi_i}_2 \right)
\end{align}

We now give the proof of Lemma~\ref{lem:density_conv}, i.e., of the relation between convergence in the $\alpha_n$ sense and convergence of the reduced density matrices in trace norm and Hilbert-Schmidt norm. Recall that Lemma~\ref{lem:density_conv} concerns $\alpha_n$, i.e., the $\alpha$-functional with the weight $n(k)=\frac{k}{N}$.

\begin{proof}[Proof of Lemma \ref{lem:density_conv}]
Recall that $\mu^{\bigwedge \varphi_j}_1 = \frac{1}{N} p_1$. We first show that $\frac{1}{N}p_1 - p_1 \mu^{\psi}_1 p_1$ is a non-negative operator with trace norm $\alpha_n$. Note that the operator $p_1\mu^{\psi}_1p_1$ maps the $N$-dimensional subspace $\Span(\varphi_1,\ldots,\varphi_N)$ to itself. Also, $p_1\mu^{\psi}_1p_1$ is non-negative and self-adjoint. We can therefore diagonalize it, i.e., there is an orthonormal basis $\{ \chi_1,\ldots,\chi_N \}$, such that
\be
p_1\mu^{\psi}_1p_1 = \sum_{i=1}^N \lambda_i \, \ketbra{\chi_i}{\chi_i}_1 = \sum_{i=1}^N \lambda_i \, p_1^{\chi_i},
\ee
with $\lambda_i \geq 0 ~\forall i=1,\ldots,N$. Note that, since $\Span(\chi_1,\ldots,\chi_N ) = \Span(\varphi_1,\ldots,\varphi_N)$,
\be
p_1 = \sum_{i=1}^N p_1^{\varphi_i} = \sum_{i=1}^N p_1^{\chi_i}.
\ee
We also have that all $\lambda_i \leq \frac{1}{N}$, since
\be
\lambda_i = \scp{\chi_i}{p_1\mu^{\psi}_1p_1 \chi_i} = \bigSCP{\psi}{\ketbra{\chi_i}{\chi_i}_1\psi} \leq \frac{1}{N} \bigSCP{\psi}{\psi} \leq \frac{1}{N}.
\ee
Also note that
\be
\sum_{i=1}^N \lambda_i = \sum_{i=1}^N \scp{\chi_i}{p_1\mu^{\psi}_1p_1 \chi_i} = \bigSCP{\psi}{\sum_{i=1}^N \ketbra{\chi_i}{\chi_i}_1\psi} = \bigSCP{\psi}{p_1 \psi} = 1- \alpha_n,
\ee
since $1=p_1+q_1$. Since $0 \leq \lambda_i \leq \frac{1}{N}$,
\be
\frac{1}{N} p_1 - p_1\mu^{\psi}_1p_1 = \sum_{i=1}^N \left(\frac{1}{N} - \lambda_i\right) p_1^{\chi_i}
\ee
is non-negative and
\begin{align}\label{trace_mu_phi-mu_psi}
\norm[\tr]{\frac{1}{N} p_1 - p_1\mu^{\psi}_1p_1} &= \tr\left(\frac{1}{N} p_1 - p_1\mu^{\psi}_1p_1\right) = \tr\left(\sum_{i=1}^N \left(\frac{1}{N} - \lambda_i\right) p_1^{\chi_i}\right) \nonumber \\
&= \sum_{i=1}^N \left(\frac{1}{N} - \lambda_i\right) \scp{\chi_i}{\chi_i} \nonumber \\
&= 1 - \sum_{i=1}^N \lambda_i \nonumber \\
&= \alpha_n.
\end{align}

We now show $\norm[\tr]{\mu^{\psi}_1 - \mu^{\bigwedge \varphi_j}_1} \leq \sqrt{8 \alpha_n}$. Note that the operators $\mu_1^{\psi}$, $p_1\mu_1^{\psi}p_1$ and $q_1\mu_1^{\psi}q_1$ are non-negative, and that
\be\label{trace_pmup_qmuq}
\norm[\tr]{p_1\mu^{\psi}_1p_1} = \bigSCP{\psi}{p_1 \psi} = 1-\alpha_n \quad\text{and}\quad \norm[\tr]{q_1\mu^{\psi}_1q_1} = \bigSCP{\psi}{q_1 \psi} = \alpha_n.
\ee
By inserting two identities $1=p_1+q_1$ we find, using \eqref{trace_mu_phi-mu_psi}, \eqref{trace_pmup_qmuq}, the triangle inequality (abbreviated $\Delta$ ineq.), and $\norm[\tr]{AB} \leq \norm[\HS]{A}\norm[\HS]{B}$,
\begin{align}\label{estimate_trace_norm_pq}
\norm[\tr]{\mu^{\bigwedge \varphi_j}_1 - \mu^{\psi}_1} &= \norm[\tr]{\frac{1}{N} p_1 - p_1\mu^{\psi}_1p_1 - p_1\mu^{\psi}_1q_1 - q_1\mu^{\psi}_1p_1 - q_1\mu^{\psi}_1q_1} \nonumber \\
\eqexp{by $\Delta$ ineq.} &\leq \norm[\tr]{\frac{1}{N} p_1 - p_1\mu^{\psi}_1p_1} + \norm[\tr]{p_1\mu^{\psi}_1q_1} + \norm[\tr]{q_1\mu^{\psi}_1p_1} + \norm[\tr]{q_1\mu^{\psi}_1q_1} \nonumber \\
\eqexp{by \eqref{trace_mu_phi-mu_psi}, \eqref{trace_pmup_qmuq}} &\leq \alpha_n + \norm[\tr]{p_1\sqrt{\mu^{\psi}_1}\sqrt{\mu^{\psi}_1}q_1} + \norm[\tr]{q_1\sqrt{\mu^{\psi}_1}\sqrt{\mu^{\psi}_1}p_1} + \alpha_n \nonumber \\
\eqexp{by \eqref{tr_HS_HS}} &\leq 2 \alpha_n + 2\norm[\HS]{p_1\sqrt{\mu^{\psi}_1}} \norm[\HS]{q_1\sqrt{\mu^{\psi}_1}} \nonumber \\
&= 2 \alpha_n + 2 \sqrt{\norm[\tr]{p_1\mu^{\psi}_1p_1} \norm[\tr]{q_1\mu^{\psi}_1q_1}} \nonumber \\
\eqexp{by \eqref{trace_pmup_qmuq}} &= 2 \alpha_n + 2 \sqrt{\alpha_n (1-\alpha_n)}.
\end{align}
Since $0 \leq \alpha_n \leq 1$, it is indeed true that
\be\label{alpha_sqrt_ineq}
2 \alpha_n + 2 \sqrt{\alpha_n (1-\alpha_n)} \leq \sqrt{8\alpha_n},
\ee
since the continuous function $f(\alpha)= \sqrt{8\alpha} - 2 \alpha - 2 \sqrt{\alpha (1-\alpha)}$ has its only minimum at $\alpha=\frac{1}{2}$ with $f\left(\frac{1}{2}\right)=0$, and also $f(0)=f(1) \geq 0$, thus $f(\alpha)\geq 0$ for all $\alpha\in [0,1]$, showing \eqref{alpha_sqrt_ineq}.

We now show $2 \alpha_n \leq \norm[\tr]{\mu^{\psi}_1 - \mu^{\bigwedge \varphi_j}_1}$. We find, using \eqref{trace_mu_phi-mu_psi}, \eqref{trace_pmup_qmuq}, $\tr(q_1\mu^{\psi}_1p_1) = \tr(p_1\mu^{\psi}_1q_1) = 0$ and $|\tr(A)|\leq \norm[\tr]{A}$, that
\begin{align}\label{alpha_leq_tr}
2 \alpha_n &= \tr\left( \mu^{\bigwedge \varphi_j}_1 - p_1\mu^{\psi}_1p_1 \right) + \tr\left( q_1\mu^{\psi}_1q_1 \right) \nonumber \\
&= \tr\left( \mu^{\bigwedge \varphi_j}_1 - p_1\mu^{\psi}_1p_1 + q_1\mu^{\psi}_1q_1 \right) \nonumber \\
&= \tr\left( \left(\mu^{\bigwedge \varphi_j}_1 - \mu^{\psi}_1\right) (p_1-q_1) \right) \nonumber \\
&\leq \norm[\tr]{ \left(\mu^{\bigwedge \varphi_j}_1 - \mu^{\psi}_1\right) (p_1-q_1) } \nonumber \\
\eqexp{by \eqref{tr_op_tr}} &\leq \norm[\op]{p_1-q_1} \norm[\tr]{ \mu^{\bigwedge \varphi_j}_1 - \mu^{\psi}_1 } \nonumber \\
&= \norm[\tr]{ \mu^{\bigwedge \varphi_j}_1 - \mu^{\psi}_1 }.
\end{align}
Note that indeed $\norm[\op]{p_1-q_1}=1$, e.g., since for all $f \in L^2(\RRR^3)$,
\be
\norm{(p_1-q_1)f}^2 = \scp{f}{(p_1-q_1)^2f}= \scp{f}{(p_1+q_1)f}=\norm{f}^2.
\ee

We now show $\norm[\HS]{\mu^{\psi}_1 - \mu^{\bigwedge \varphi_j}_1}^2 \leq \frac{2}{N} \alpha_n$. Recall that $\norm[\op]{\mu^{\psi}_1}=\frac{1}{N}$ and $\norm[\tr]{\mu^{\psi}_1}=1$. We find, using $\norm[\tr]{AB} \leq \norm[\op]{A} \norm[\tr]{B}$,
\begin{align}\label{HS_leq_alpha}
\norm[\HS]{\mu^{\psi}_1 - \mu^{\bigwedge \varphi_j}_1}^2 &= \tr\left( \left(\mu^{\psi}_1 - \mu^{\bigwedge \varphi_j}_1\right)^2 \right) \nonumber \\
&= \tr\left( \frac{1}{N^2} p_1 - \frac{1}{N}p_1\mu^{\psi}_1 - \frac{1}{N}\mu^{\psi}_1 p_1 + \left(\mu^{\psi}_1\right)^2 \right) \nonumber \\
&= \frac{1}{N} - \frac{1}{N}\bigSCP{\psi}{p_1 \psi} - \frac{1}{N}\bigSCP{\psi}{p_1 \psi} + \tr\left( \left(\mu^{\psi}_1\right)^2 \right) \nonumber \\
\eqexp{by \eqref{tr_op_tr}} &\leq \frac{1}{N} \left( 1 - \bigSCP{\psi}{p_1 \psi} \right) - \frac{1}{N}\bigSCP{\psi}{p_1 \psi} + \norm[\op]{\mu^{\psi}_1}\norm[\tr]{\mu^{\psi}_1} \nonumber \\
\eqexp{by Lem.~\ref{lem:properties_density_matrix}} &= \frac{1}{N} \alpha_n - \frac{1}{N}\bigSCP{\psi}{p_1 \psi} + \frac{1}{N} \nonumber \\
&= \frac{2}{N} \alpha_n.
\end{align}

We now show $\alpha_n \leq \sqrt{N} \norm[\HS]{\mu^{\psi}_1 - \mu^{\bigwedge \varphi_j}_1}$. Using $\norm[\tr]{AB} \leq \norm[\HS]{A} \norm[\HS]{B}$, $\norm[\tr]{AB} \leq \norm[\op]{A} \norm[\tr]{B}$, $\norm[\HS]{p_1}=\sqrt{N}$, $\norm[\op]{p_1}=1$ and \eqref{trace_mu_phi-mu_psi} we find
\begin{align}\label{alpha_leq_HS}
\norm[\HS]{\mu^{\bigwedge \varphi_j}_1 - \mu^{\psi}_1} &\geq \frac{\norm[\tr]{\left(\mu^{\bigwedge \varphi_j}_1 - \mu^{\psi}_1\right)p_1}}{\norm[\HS]{p_1}} \nonumber \\
\eqexp{by \eqref{tr_op_tr}} &\geq \frac{1}{\sqrt{N}} \frac{\norm[\tr]{p_1\left(\mu^{\bigwedge \varphi_j}_1 - \mu^{\psi}_1\right)p_1}}{\norm[\op]{p_1}} \nonumber \\
&= \frac{1}{\sqrt{N}} \norm[\tr]{\frac{1}{N}p_1 - p_1 \mu^{\psi}_1 p_1} \nonumber \\
\eqexp{by \eqref{trace_mu_phi-mu_psi}} &= \frac{1}{\sqrt{N}} \, \alpha_n.
\end{align}

Together, the inequalities \eqref{estimate_trace_norm_pq} with \eqref{alpha_sqrt_ineq}, \eqref{alpha_leq_tr}, \eqref{HS_leq_alpha} and \eqref{alpha_leq_HS} prove \eqref{bound_tr_alpha_HS} and \eqref{bound_HS_alpha_tr}. \hfill $\qedhere$
\end{proof}

Finally, let us prove Proposition~\ref{pro:density_matrix_op} which shows that the operator norm is not useful to measure the desired convergence.

\begin{proof}[Proof of Proposition~\ref{pro:density_matrix_op}]
Starting from \eqref{trace_mu_phi-mu_psi} from the proof of Lemma~\ref{lem:density_conv} and using $\norm[\tr]{AB} \leq \norm[\HS]{A} \norm[\HS]{B}$, $\norm[\HS]{AB} \leq \norm[\HS]{A} \norm[\op]{B}$ and $\norm[\HS]{p_1}^2 = N$, we find
\begin{align}
\alpha_n = \norm[\tr]{\frac{1}{N} p_1 - p_1\mu^{\psi}_1p_1} &= \norm[\tr]{p_1 \left(\mu^{\bigwedge \varphi_j}_1 - \mu^{\psi}_1\right)p_1} \nonumber \\
\eqexp{by \eqref{tr_HS_HS}} &\leq \norm[\HS]{p_1} \norm[\HS]{\left(\mu^{\bigwedge \varphi_j}_1 - \mu^{\psi}_1\right)p_1} \nonumber \\
\eqexp{by \eqref{HS_op_HS}} &\leq \norm[\HS]{p_1}^2 \norm[\op]{\mu^{\bigwedge \varphi_j}_1 - \mu^{\psi}_1} \nonumber \\
&= N \norm[\op]{\mu^{\bigwedge \varphi_j}_1 - \mu^{\psi}_1},
\end{align}
which proves \eqref{bound_alpha_op} and thus \eqref{op_implies_alpha}.

We now construct an example for which $\alpha_n \to 0$, but $N \norm[\op]{\mu^{\psi}_1 - \mu^{\bigwedge \varphi_j}_1} = 1$. Consider the density matrix
\be
\mu^{\tilde{\psi}}_1 = \frac{1}{N} \left( p_1^{\chi} + \sum_{i=2}^N p_1^{\varphi_i} \right) \quad\text{with}\quad \scp{\chi}{\varphi_i}=0 ~\forall i=1,\ldots,N,
\ee
arising from the wave function $\tilde{\psi} = \bigwedge_{j=2}^{N} \varphi_j \wedge \chi$. We find $\alpha_n=\frac{1}{N}$ and
\be
N \norm[\op]{\mu^{\tilde{\psi}}_1 - \mu^{\bigwedge \varphi_j}_1} = N \norm[\op]{\frac{1}{N}\left( p_1^{\chi} - p_1^{\varphi_1} \right)} = 1,
\ee
thus providing the desired example. \hfill $\qedhere$
\end{proof}

\chapter{Proof of Theorems for General $v^{(N)}$}\label{sec:alpha_dot_and_general_lemmas}
\section{The Time Derivative of $\alpha_f(t)$}\label{sec:alpha_dot}
The expression for the time derivative of $\alpha_f(t)$ for arbitrary weight functions $f(k)$ follows from direct calculation. We calculate it here for the general setting where the wave function $\psi^t \in L^2(\RRR^{3N})$ is a solution to
\be\label{schroedinger_eq_general}
i \partial_t \psi^t = H^N \psi^t = \left( \sum_{j=1}^N H_j^0 + \sum_{1\leq i < j \leq N} v^{(N)}(x_i-x_j) \right) \psi^t,
\ee
where the Hamiltonian $H^N$ is a self-adjoint operator, $v^{(N)}(x)=v^{(N)}(-x)$ is a (possibly scaled) real interaction potential and $H_j^0$ acts only on the $j$-th variable.

The general form of the fermionic mean-field equations for the one-particle wave functions $\varphi^t_1,\ldots,\varphi^t_N \in L^2(\RRR^3)$ is
\be\label{mean-field_eq}
i \partial_t \varphi_j^t(x) = H^\mf \varphi_j^t(x) = H^0\varphi_j^t(x) + \left(V^{N,j,\varphi_1^t,\ldots,\varphi_N^t}\varphi_j^t\right)(x),
\ee
where $V^{N,j,\varphi_1^t,\ldots,\varphi_N^t} =: V^{(N)}$ is the mean-field interaction. The two interesting cases are when there is only the \emph{direct interaction},
\be\label{mean-field_dir_int}
V^{N,j,\varphi_1^t,\ldots,\varphi_N^t}(x) = V^{\dir,(N)}(x) = (v^{(N)} \star \rho_N^t)(x),
\ee
where $\rho_N^t = \sum_{i=1}^N |\varphi^t_i|^2$, and when there is direct and \emph{exchange interaction},
\begin{align}\label{mean-field_direxch_int}
V^{N,j,\varphi_1^t,\ldots,\varphi_N^t}\varphi_j^t(x) &= \left( V^{\dir,(N)} + V^{\exch,(N)} \right)\varphi_j^t(x) \nonumber \\
&= (v^{(N)} \star \rho_N^t)(x)\varphi_j^t(x) - \sum_{\ell=1}^N \left(v^{(N)} \star (\varphi_{\ell}^{t*}\varphi_j^t)\right)(x) \, \varphi_{\ell}^t(x).
\end{align}
In the following chapters we often add a subscript to some of the operators above to denote the particle index on which the operator acts, for example, as an operator on $L^2(\RRR^{3N})$,
\be
V^{(N)}_k = \underbrace{\id \otimes \ldots \otimes \id}_{k-1 ~ \text{times}} \otimes V^{(N)} \otimes \underbrace{\id \otimes \ldots \otimes \id}_{N-k ~ \text{times}}.
\ee

\begin{lemma}\label{lem:alpha_derivative_pre}
Let $\psi^t\in L^2(\RRR^{3N})$ be an antisymmetric solution to the Schr\"odinger equation \eqref{schroedinger_eq_general} and let $\varphi_1^t,\ldots,\varphi_N^t\in L^2(\RRR^3)$ be orthonormal solutions to the mean-field equations \eqref{mean-field_eq}. We define $W_{12} := N(N-1)v^{(N)}_{12} - NV^{(N)}_1 - NV^{(N)}_2$. Then, for $\widehat{f}$ and $\alpha_f(t)$ from Definition~\ref{def:projectors},
\be\label{f_dot}
i \partial_t \widehat{f} = \left[ \, \sum_{m=1}^N H_m^{\mf} , \widehat{f} \, \right],
\ee
\be\label{alpha_dot_commutator}
\partial_t \alpha_f(t) = \frac{i}{2} \bigSCP{\psi^t}{\left[ W_{12}, \widehat{f} \,\right] \psi^t}.
\ee
\end{lemma}

\noindent\textbf{Remarks.}
\begin{enumerate}
\setcounter{enumi}{\theremarks}
\item At this point, let us emphasize again that the kinetic and external field terms coming from the Schr\"odinger and the fermionic mean-field equations cancel, which is why the main theorems hold for any $H_j^0$. Furthermore, only $W_{12}$, the difference between Schr\"odinger interaction and mean-field, enters, i.e., the method focuses directly on the relevant point.
\end{enumerate}
\setcounter{remarks}{\theenumi}

\begin{proof}
Note that the operators $p_m,q_m,P_{N,k}$ all depend on $t$ through the orbitals $\varphi_1^t,\ldots,\varphi_N^t$. For ease of notation,we do not explicitly write out this $t$-dependence. In order to prove \eqref{f_dot} we first calculate the time derivatives of $p_m$ and $q_m$, and then of $P_{N,k}$. We find
\begin{align}
i \partial_t p_m &= \sum_{j=1}^N \, i \partial_t \bigg( \ketbr{\varphi_j^t(x_m)} \bigg) \nonumber \\
&= \sum_{j=1}^N \left( \left( i \partial_t\ket{\varphi_j^t(x_m)} \right) \bra{\varphi_j^t(x_m)} + \ket{\varphi_j^t(x_m)} \left( i \partial_t \bra{\varphi_j^t(x_m)}\right) \right) \nonumber \\
&= \sum_{j=1}^N \bigg( H_m^{\mf}\ket{\varphi_j^t(x_m)} \bra{\varphi_j^t(x_m)} - \ket{\varphi_j^t(x_m)} \bra{H_m^{\mf}\varphi_j^t(x_m)} \bigg) \nonumber \\
&= \left[ H_m^{\mf}, p_m \right],
\end{align}
and, using $p_m+q_m=1$,
\be
i \partial_t q_m = -i \partial_t p_m = - \left[ H_m^{\mf}, p_m \right] = \left[ H_m^{\mf}, q_m \right].
\ee
Now recall from Definition~\ref{def:projectors} that
\be
P_{N,k} = \sum_{\vec{a} \in \AAA_k} \prod_{m=1}^N (p_m)^{1-a_m}(q_m)^{a_m}
\ee
with the set
\be
\AAA_k = \left\{ \vec{a}=(a_1,\ldots,a_N) \in \{0,1\}^N: \sum_{m=1}^N a_m=k \right\}.
\ee
For the following calculation we abbreviate $R_m = (p_m)^{1-a_m}(q_m)^{a_m}$. Then $i \partial_t R_m = [ H_m^{\mf}, R_m ]$ holds for $a_m=0$ and $a_m=1$. It follows that
\begin{align}
i \partial_t P_{N,k} &= i \partial_t \sum_{\vec{a} \in \AAA_k} \prod_{\ell=1}^N R_{\ell} \nonumber \\
&= \sum_{\vec{a} \in \AAA_k} \sum_{m=1}^N \left( \prod_{\ell=1}^{m-1} R_{\ell} \right) i \partial_t R_m \left( \prod_{\ell=m+1}^N R_{\ell}\right) \nonumber \\
&= \sum_{\vec{a} \in \AAA_k} \sum_{m=1}^N \left( \prod_{\ell=1}^{m-1} R_{\ell} \right) \left[ H_m^{\mf}, R_m \right] \left( \prod_{\ell=m+1}^N R_{\ell}\right) \nonumber \\
&= \sum_{\vec{a} \in \AAA_k} \sum_{m=1}^N \left[ H_m^{\mf}, \prod_{\ell=1}^N R_{\ell} \right] \nonumber \\
&= \left[ \sum_{m=1}^N H_m^{\mf}, P_{N,k} \right],
\end{align}
and thus
\be
i \partial_t \widehat{f} = i \partial_t \sum_{k=0}^N f(k) P_{N,k} = \sum_{k=0}^N f(k) i \partial_t P_{N,k} = \sum_{k=0}^N f(k) \left[ \sum_{m=1}^N H_m^{\mf}, P_{N,k} \right] = \left[ \sum_{m=1}^N H_m^{\mf}, \widehat{f} \right].
\ee
Using this and the antisymmetry of $\psi^t$, we calculate the time derivative of $\alpha_f(t)$. We find
\begin{align}
\partial_t \alpha_f(t) &= \partial_t \bigSCP{\psi^t}{\widehat{f}\,\psi^t} \nonumber \\
&= \bigSCP{\left(\partial_t\psi^t\right)}{\widehat{f}\,\psi^t} + \bigSCP{\psi^t}{\widehat{f}\left(\partial_t\psi^t\right)} + \bigSCP{\psi^t}{\left(\partial_t\widehat{f}\right)\psi^t} \nonumber \\
&= i\, \bigSCP{H^N\psi^t}{\widehat{f}\,\psi^t} - i \, \bigSCP{\psi^t}{\widehat{f} \, H^N\psi^t} - i \, \bigSCP{\psi^t}{\left[\sum_{m=1}^N H_m^{\mf},\widehat{f}\right]\psi^t} \nonumber \\
&= i \, \bigSCP{\psi^t}{\left[H^N-\sum_{m=1}^N H_m^{\mf},\widehat{f}\right]\psi^t} \nonumber \\
&= i \, \bigSCP{\psi^t}{\left[\sum_{j=1}^N H_j^0 + \sum_{1\leq i < j \leq N} v^{(N)}(x_i-x_j) - \sum_{m=1}^N \left( H_m^0 + V_m^{(N)} \right),\widehat{f}\right]\psi^t} \nonumber \\
&= i \, \bigSCP{\psi^t}{\left[\sum_{1\leq i < j \leq N} v^{(N)}(x_i-x_j) - \sum_{m=1}^N V_m^{(N)},\widehat{f}\right]\psi^t} \nonumber \\
&= i \, \bigSCP{\psi^t}{\left[ \frac{N(N-1)}{2} v^{(N)}(x_1-x_2) - \frac{N}{2} V_1^{(N)} - \frac{N}{2} V_2^{(N)},\widehat{f}\right]\psi^t} \nonumber \\
&= \frac{i}{2} \bigSCP{\psi^t}{\left[W_{12},\widehat{f}\right]\psi^t}.
\end{align}
\end{proof}

In order to simplify the expression \eqref{alpha_dot_commutator} we need the following auxiliary lemma.

\begin{lemma}\label{lem:shift_Ps}
As in Definition~\ref{def:projectors2}, we abbreviate $P_0^{\{1,2\}} = p_1p_2$, $P_1^{\{1,2\}} = p_1q_2 + q_1p_2$ and $P_2^{\{1,2\}} = q_1q_2$. Let $h_{12}$ be an operator that acts only on the first and second particle index. Then
\begin{enumerate}[(a)]
\item for all $a,b = 0,1,2$, and for all $k=1,\ldots,N$,
\be\label{shift_Ps_right_to_left}
\left( P_a^{\{1,2\}}h_{12}P_b^{\{1,2\}} \right) P_{N,k} = P_{N,k+a-b} \left(P_a^{\{1,2\}}h_{12}P_b^{\{1,2\}}\right),
\ee
\be\label{shift_Ps_left_to_right}
P_{N,k} \left( P_a^{\{1,2\}}h_{12}P_b^{\{1,2\}} \right) = \left( P_a^{\{1,2\}}h_{12}P_b^{\{1,2\}} \right) P_{N,k+b-a},
\ee
\item for all $a,b = 0,1,2$,
\be\label{shift_f_right_to_left}
\left( P_a^{\{1,2\}}h_{12}P_b^{\{1,2\}} \right) \widehat{f} = \widehat{f}_{b-a} \left(P_a^{\{1,2\}}h_{12}P_b^{\{1,2\}}\right),
\ee
\be\label{shift_f_left_to_right}
\widehat{f} \left( P_a^{\{1,2\}}h_{12}P_b^{\{1,2\}} \right) = \left( P_a^{\{1,2\}}h_{12}P_b^{\{1,2\}} \right) \widehat{f}_{a-b}.
\ee
\end{enumerate}
\end{lemma}

\begin{proof}
\begin{enumerate}[(a)]
\item We first split up the $P_{N,k}$,
\be
P_{N,k} = \sum_{d=0}^2 P_d^{\{1,2\}} P_{k-d}^{\{3,\ldots,N\}},
\ee
where $P_{k-d}^{\{3,\ldots,N\}}$, as in Definition~\ref{def:projectors2}, contains $k-d$ $q$'s, and acts only on the variables $3,\ldots,N$. Now note that ${P_a}^{\{1,2\}}{P_b}^{\{1,2\}}=\delta_{ab}{P_a}^{\{1,2\}}$. Then we find
\begin{align}\label{P_k_shift}
\left( P_a^{\{1,2\}}h_{12}P_b^{\{1,2\}} \right) P_{N,k} &= \left( P_a^{\{1,2\}}h_{12}P_b^{\{1,2\}} \right) \sum_{d=0}^2 P_d^{\{1,2\}} P_{k-d}^{\{3,\ldots,N\}} \nonumber \\
&= P_a^{\{1,2\}}h_{12}P_b^{\{1,2\}} P_{k-b}^{\{3,\ldots,N\}} \nonumber \\
&= P_{k-b}^{\{3,\ldots,N\}} P_a^{\{1,2\}}h_{12}P_b^{\{1,2\}} \nonumber \\
&= \sum_{d=0}^2 P_d^{\{1,2\}} P_{k+a-b-d}^{\{3,\ldots,N\}} \left( P_a^{\{1,2\}}h_{12}P_b^{\{1,2\}} \right) \nonumber \\
&= P_{N,k+a-b} \left( P_a^{\{1,2\}}h_{12}P_b^{\{1,2\}} P_b^{\{1,2\}} \right).
\end{align}
In the same way (or by just renaming $k'=k+a-b$, such that $k=k'+b-a$) we find that \eqref{shift_Ps_left_to_right} holds.
\item From \eqref{P_k_shift} it follows directly that
\begin{align}
\left( P_a^{\{1,2\}}h_{12}P_b^{\{1,2\}} \right) \widehat{f} &= P_a^{\{1,2\}}h_{12}P_b^{\{1,2\}} \sum_{k=0}^N f(k) P_{N,k} \nonumber \\
&= \sum_{k=0}^N f(k) P_{N,k+a-b} P_a^{\{1,2\}}h_{12}P_b^{\{1,2\}} \nonumber \\
&= \widehat{f}_{b-a} \left( P_a^{\{1,2\}}h_{12}P_b^{\{1,2\}} \right),
\end{align}
and in the same way \eqref{shift_f_left_to_right} follows directly from \eqref{shift_Ps_left_to_right}. \hfill $\qedhere$
\end{enumerate}
\end{proof}

With Lemma~\ref{lem:shift_Ps} we can simplify the expression \eqref{alpha_dot_commutator} for the time derivative of $\alpha_f(t)$ by splitting it into three parts, each of which will be estimated separately later.

\begin{lemma}\label{lem:alpha_derivative}
Let $\psi^t\in L^2(\RRR^{3N})$ be an antisymmetric solution to the Schr\"odinger equation \eqref{schroedinger_eq_general} and let $\varphi_1^t,\ldots,\varphi_N^t\in L^2(\RRR^3)$ be orthonormal solutions to the mean-field equations \eqref{mean-field_eq}. Then
\begin{align}\label{alpha_derivative}
\partial_t \alpha_f(t) &= 2 \, \Im\, \bigSCP{\psi^t}{N\left(\widehat{f}-\widehat{f}_{-1}\right)q_1 \left( (N-1)p_2v^{(N)}_{12}p_2 - V^{(N)}_1 \right) p_1 \psi^t} \nonumber \\
&\quad + \Im\, \bigSCP{\psi^t}{N\left(\widehat{f}-\widehat{f}_{-2}\right)q_1q_2 (N-1)v^{(N)}_{12} p_1p_2 \psi^t} \nonumber \\
&\quad + 2 \, \Im\, \bigSCP{\psi^t}{N\left(\widehat{f}-\widehat{f}_{-1}\right)q_1q_2 (N-1)v^{(N)}_{12} p_1q_2 \psi^t}.
\end{align}
\end{lemma}

\noindent\textbf{Remarks.}
\begin{enumerate}
\setcounter{enumi}{\theremarks}
\item Note that the time derivative is formally the same as for bosons, where $p_1:=\ketbr{\varphi}_1$, see \cite{pickl:2010gp_pos,pickl:2011method}. Note that in \cite{pickl:2010gp_pos,pickl:2011method} the splitting into three summands is done slightly differently: compared to \eqref{alpha_derivative}, an additional identity $1=p_2+q_2$ is added in front of $V^{(N)}_1$.
\item For the case $f(k) = n(k) = \frac{k}{N}$ we find a simple expression for the time derivative of $\alpha_n(t) = \SCP{\psi^t}{q_1 \psi^t}$. Note that, in view of Definition~\ref{def:projectors} and the identity $\sum_{k=0}^N P_{N,k} = 1$ (which we prove later in Lemma~\ref{lem:properties_P_Nk}),
\be
\widehat{n} - \widehat{n}_{-1} = \sum_{k=1}^N \left( \frac{k}{N} - \frac{(k-1)}{N} \right) P_{N,k} = \frac{1}{N} \sum_{k=1}^N P_{N,k} = \frac{1}{N} - \frac{1}{N} P_{N,0},
\ee
and
\begin{align}
\widehat{n} - \widehat{n}_{-2} &= \frac{1}{N} P_{N,1} + \sum_{k=2}^N \left( \frac{k}{N} - \frac{(k-2)}{N} \right) P_{N,k} \nonumber \\
&= \frac{1}{N} P_{N,1} + \frac{2}{N} \sum_{k=2}^N P_{N,k} \nonumber \\
&= \frac{2}{N} - \frac{1}{N} P_{N,1} - \frac{2}{N} P_{N,0}.
\end{align}
Then, using $P_{N,0}q_1 = 0 = P_{N,1} q_1q_2$, the expression \eqref{alpha_derivative} simplifies to
\begin{align}\label{alpha_derivative_n_remark}
\partial_t \alpha_n(t) &= 2 \, \Im\, \bigSCP{\psi^t}{q_1 \left( (N-1)p_2v^{(N)}_{12}p_2 - V^{(N)}_1 \right) p_1 \psi^t} \nonumber \\
&\quad + \Im\, \bigSCP{\psi^t}{q_1q_2 (N-1)v^{(N)}_{12} p_1p_2 \psi^t} \nonumber \\
&\quad + 2 \, \Im\, \bigSCP{\psi^t}{q_1q_2 (N-1)v^{(N)}_{12} p_1q_2 \psi^t}.
\end{align}
\end{enumerate}
\setcounter{remarks}{\theenumi}

\begin{proof}[Proof of Lemma \ref{lem:alpha_derivative}]
We calculate the time derivative of $\alpha_f(t)$ using the expression \eqref{alpha_dot_commutator} from Lemma~\ref{lem:alpha_derivative_pre}. The idea of the proof is to insert two identities $1=p_1+q_1$ and $1=p_2+q_2$ in front of each $\psi^t$ (which leads to $16$ summands) and then to use Lemma~\ref{lem:shift_Ps} in order to shift $\widehat{f}$. It turns out that a lot of terms drop out due to the commutator structure. We again use the notation $P_0^{\{1,2\}} = p_1p_2$, $P_1^{\{1,2\}} = p_1q_2 + q_1p_2$ and $P_2^{\{1,2\}} = q_1q_2$. Note that $\sum_{a=0}^2 P_a^{\{1,2\}} = 1$, since $p_1p_2+p_1q_2+q_1p_2+q_1q_2= (p_1+q_1)(p_2+q_2) = 1$. We abbreviate $P_a=P_a^{\{1,2\}}$ for the following calculation. Inserting the two identities and using Lemma~\ref{lem:shift_Ps} we find ($a \leftrightarrow b$ means that we interchanged the indices $a$, $b$)
\begin{align}
\partial_t \alpha_f(t) &= \frac{i}{2} \bigSCP{\psi^t}{\left[W_{12},\widehat{f}\right]\psi^t} \nonumber \\
&= \frac{i}{2} \sum_{a,b=0}^2 \bigSCP{\psi^t}{P_a\left[W_{12},\widehat{f}\right]P_b\,\psi^t} \nonumber \\
&= \frac{i}{2} \sum_{a} \bigSCP{\psi^t}{\left( P_aW_{12}P_a\widehat{f} - \widehat{f}P_aW_{12}P_a \right)\psi^t} \nonumber \\
&\quad + \frac{i}{2} \sum_{a>b} \bigSCP{\psi^t}{\left( P_aW_{12}P_b\widehat{f} - \widehat{f}P_aW_{12}P_b \right)\psi^t} \nonumber \\
&\quad + \frac{i}{2} \sum_{a<b} \bigSCP{\psi^t}{\left( P_aW_{12}P_b\widehat{f} - \widehat{f}P_aW_{12}P_b \right)\psi^t} \nonumber \\
\eqexp{by Lem.~\ref{lem:shift_Ps}} &= \frac{i}{2} \sum_{a>b} \bigSCP{\psi^t}{\left( P_aW_{12}P_b\widehat{f} - \widehat{f}P_aW_{12}P_b \right)\psi^t} \nonumber \\
\eqexp{by $a\leftrightarrow b$} &\quad + \frac{i}{2} \sum_{a>b} \bigSCP{\psi^t}{\left( P_bW_{12}P_a\widehat{f} - \widehat{f}P_bW_{12}P_a \right)\psi^t} \nonumber \\
\eqexp{by Lem.~\ref{lem:shift_Ps}} &= \frac{i}{2} \sum_{a>b} \bigSCP{\psi^t}{\left( \widehat{f}_{b-a} - \widehat{f} \right) P_aW_{12}P_b \,\psi^t} \nonumber \\
&\quad - \frac{i}{2} \sum_{a>b} \bigSCP{\psi^t}{ P_bW_{12}P_a \left( \widehat{f}_{b-a} - \widehat{f} \right) \psi^t} \nonumber \\
&= - \Im \sum_{a>b} \bigSCP{\psi^t}{\left( \widehat{f}_{b-a} - \widehat{f} \right) P_aW_{12}P_b \,\psi^t} \nonumber \\
&= - \Im  \,\bigSCP{\psi^t}{\left(\widehat{f}_{-1}-\widehat{f}\right)(p_1q_2+q_1p_2)W_{12}p_1p_2 \psi^t} \nonumber \\
&\quad - \Im \, \bigSCP{\psi^t}{\left(\widehat{f}_{-2}-\widehat{f}\right)q_1q_2W_{12}p_1p_2 \psi^t} \nonumber \\
&\quad - \Im \, \bigSCP{\psi^t}{\left(\widehat{f}_{-1}-\widehat{f}\right)q_1q_2W_{12}(p_1q_2+q_1p_2) \psi^t} \nonumber \\
\eqexp{by $W_{12}=W_{21}$} &= 2 \, \Im \, \bigSCP{\psi^t}{\left( \widehat{f} - \widehat{f}_{-1} \right)q_1p_2W_{12}p_1p_2 \psi^t} \nonumber \\
&\quad + \Im \, \bigSCP{\psi^t}{\left( \widehat{f} - \widehat{f}_{-2} \right)q_1q_2W_{12}p_1p_2 \psi^t} \nonumber \\
&\quad + 2 \, \Im \, \bigSCP{\psi^t}{\left( \widehat{f} - \widehat{f}_{-1} \right)q_1q_2W_{12}p_1q_2 \psi^t}.
\end{align}
Now recall that $W_{12} := N(N-1)v^{(N)}_{12} - NV^{(N)}_1 - NV^{(N)}_2$. Inserting this definition, and using $p_1q_1 = 0 = p_2q_2$ and $p_2+q_2=1$, we find
\begin{align}
\partial_t \alpha_f(t) &= 2 \, \Im \, \bigSCP{\psi^t}{\left( \widehat{f} - \widehat{f}_{-1} \right)q_1p_2 \left( N(N-1)v^{(N)}_{12} - NV^{(N)}_1 - NV^{(N)}_2 \right) p_1p_2 \psi^t} \nonumber \\
&\quad + \Im \, \bigSCP{\psi^t}{\left( \widehat{f} - \widehat{f}_{-2} \right)q_1q_2 \left( N(N-1)v^{(N)}_{12} - NV^{(N)}_1 - NV^{(N)}_2 \right) p_1p_2 \psi^t} \nonumber \\
&\quad + 2 \, \Im \, \bigSCP{\psi^t}{\left( \widehat{f} - \widehat{f}_{-1} \right)q_1q_2 \left( N(N-1)v^{(N)}_{12} - NV^{(N)}_1 - NV^{(N)}_2 \right) p_1q_2 \psi^t} \nonumber \\
&= 2 \, \Im \, \bigSCP{\psi^t}{\left( \widehat{f} - \widehat{f}_{-1} \right)q_1p_2 \left( N(N-1)v^{(N)}_{12} - NV^{(N)}_1 \right) p_1p_2 \psi^t} \nonumber \\
&\quad + \Im \, \bigSCP{\psi^t}{\left( \widehat{f} - \widehat{f}_{-2} \right)q_1q_2 \left( N(N-1)v^{(N)}_{12} \right) p_1p_2 \psi^t} \nonumber \\
&\quad + 2 \, \Im \, \bigSCP{\psi^t}{\left( \widehat{f} - \widehat{f}_{-1} \right)q_1q_2 \left( N(N-1)v^{(N)}_{12} - NV^{(N)}_1 \right) p_1q_2 \psi^t} \nonumber \\
&= 2 \, \Im \, \bigSCP{\psi^t}{N\left( \widehat{f} - \widehat{f}_{-1} \right)q_1 \left( (N-1)p_2v^{(N)}_{12}p_2 - V^{(N)}_1 \right) p_1 \psi^t} \nonumber \\
&\quad + \Im \, \bigSCP{\psi^t}{N\left( \widehat{f} - \widehat{f}_{-2} \right)q_1q_2 \left( (N-1)v^{(N)}_{12} \right) p_1p_2 \psi^t} \nonumber \\
&\quad + 2 \, \Im \, \bigSCP{\psi^t}{N\left( \widehat{f} - \widehat{f}_{-1} \right)q_1q_2 \left( (N-1)v^{(N)}_{12} \right) p_1q_2 \psi^t}.
\end{align}
\end{proof}

In order to control the time derivative of $\alpha_f(t)$, each of the three terms in \eqref{alpha_derivative} is bounded separately. Before we estimate these terms in Chapter~\ref{sec:alpha_m_dot_rigorous}, we need to establish several techniques and properties of these terms in Chapters \ref{sec:general_lemmas} and \ref{sec:estimates_projectors}.

\section{General Lemmas about the Projectors, $\widehat{f}$ and $\alpha_f(t)$}\label{sec:general_lemmas}
The following lemma collects some properties of the projectors $P_{N,k}$ and the operator $\widehat{n}$ from Definition~\ref{def:projectors} and \eqref{weight_n} that we use in subsequent lemmas.

\begin{lemma}\label{lem:properties_P_Nk}
For the objects from Definition~\ref{def:projectors} and all antisymmetric $\psi_{as} \in L^2(\RRR^{3N})$,
\be
P_{N,k} P_{N,\ell} = \delta_{k\ell} P_{N,k} \quad\forall k,\ell = 1,\ldots,N ,
\ee
\be
\sum_{k=0}^N P_{N,k} = 1,
\ee
\be
\widehat{n} := \sum_{k=0}^N \frac{k}{N} P_{N,k} = \frac{1}{N} \sum_{m=1}^N q_m,
\ee
\be
\alpha_n := \SCP{\psi_{as}}{\widehat{n} \, \psi_{as}} = \SCP{\psi_{as}}{q_1\psi_{as}},
\ee
\be
\SCP{\psi_{as}}{q_1q_2\psi_{as}} \leq \frac{N}{N-1} \SCP{\psi_{as}}{\left(\widehat{n}\right)^2\psi_{as}}.
\ee
\end{lemma}

\begin{proof}
Recall from Definition~\ref{def:projectors} that the projectors $P_{N,k}$ are defined by
\be
P_{N,k} := \sum_{\vec{a} \in \AAA_k} \prod_{m=1}^N (p_m)^{1-a_m}(q_m)^{a_m}
\ee
with the set
\be
\AAA_k := \left\{ \vec{a}=(a_1,\ldots,a_N) \in \{0,1\}^N: \sum_{m=1}^N a_m=k \right\}.
\ee
From $p_mq_m=0$ it follows that for all $k,\ell = 1,\ldots,N$,
\begin{align}
P_{N,k} P_{N,\ell} &= \sum_{\vec{a} \in \AAA_k}\sum_{\vec{b} \in \AAA_{\ell}} \prod_{m=1}^N (p_m)^{1-a_m}(q_m)^{a_m}(p_m)^{1-b_m}(q_m)^{b_m} \nonumber \\
&= \sum_{\vec{a} \in \AAA_k}\sum_{\vec{b} \in \AAA_{\ell}} \prod_{m=1}^N \delta_{a_mb_m} (p_m)^{1-a_m}(q_m)^{a_m} \nonumber \\
&= \delta_{k\ell} \sum_{\vec{a} \in \AAA_k} \prod_{m=1}^N (p_m)^{1-a_m}(q_m)^{a_m} \nonumber \\
&= \delta_{k\ell} P_{N,k}.
\end{align}
From $\bigcup_{k=0}^N \AAA_k = \{0,1\}^N$ and $p_m+q_m=1$ it follows that
\begin{align}\label{calc_sum_P_Nk_id}
\sum_{k=0}^N P_{N,k} &= \sum_{k=0}^N \sum_{\vec{a} \in \AAA_k} \prod_{m=1}^N (p_m)^{1-a_m}(q_m)^{a_m} \nonumber \\
&= \sum_{\vec{a} \in \{0,1\}^N} \prod_{m=1}^N (p_m)^{1-a_m}(q_m)^{a_m} \nonumber \\
&= \prod_{m=1}^N \left(p_m + q_m\right) \nonumber \\
&= 1.
\end{align}
Now recall that $P_{N,k}$ contains $k$ $q$'s in each summand and therefore
\be
\left( \sum_{m=1}^{N} q_m \right) P_{N,k} = k \, P_{N,k}.
\ee
Using this and \eqref{calc_sum_P_Nk_id} we find
\be\label{n_hat_sum_q}
\widehat{n} = \sum_{k=0}^N \frac{k}{N} P_{N,k} = \sum_{k=0}^N \frac{1}{N} \left( \sum_{m=1}^{N} q_m \right) P_{N,k} =  \frac{1}{N} \left( \sum_{m=1}^{N} q_m \right) \sum_{k=0}^N P_{N,k} = \frac{1}{N} \sum_{m=1}^N q_m.
\ee
From the antisymmetry of $\psi_{as}$ and \eqref{n_hat_sum_q} it then follows directly that
\be\label{n_hat_q}
\SCP{\psi_{as}}{\widehat{n}\,\psi_{as}} = \frac{1}{N} \sum_{m=1}^N \SCP{\psi_{as}}{q_m\psi_{as}} = \frac{1}{N} \sum_{m=1}^N \SCP{\psi_{as}}{q_1\psi_{as}} = \SCP{\psi_{as}}{q_1\psi_{as}} .
\ee
Using again the antisymmetry of $\psi_{as}$, as well as \eqref{n_hat_sum_q} and \eqref{n_hat_q}, we find
\begin{align}
\SCP{\psi_{as}}{q_1q_2\psi_{as}} &= \frac{1}{N-1} \sum_{m=2}^N \SCP{\psi_{as}}{q_1q_m\psi_{as}} \nonumber \\
&= \frac{1}{N-1} \left( \sum_{m=1}^N \SCP{\psi_{as}}{q_1q_m\psi_{as}} - \SCP{\psi_{as}}{q_1\psi_{as}} \right) \nonumber \\
&= \frac{1}{N-1} \left( \frac{1}{N} \sum_{m=1}^N \sum_{n=1}^N \SCP{\psi_{as}}{q_mq_n\psi_{as}} - \SCP{\psi_{as}}{q_1\psi_{as}} \right) \nonumber \\
&= \frac{N}{N-1} \left( \sum_{m=1}^N \sum_{n=1}^N \frac{1}{N^2} \SCP{\psi_{as}}{q_mq_n\psi_{as}} - \frac{1}{N} \SCP{\psi_{as}}{q_1\psi_{as}} \right) \nonumber \\
&= \frac{N}{N-1} \left( \SCP{\psi_{as}}{\left( \frac{1}{N}\sum_{m=1}^N q_m \right)^2 \psi_{as}} - \frac{1}{N} \SCP{\psi_{as}}{q_1\psi_{as}} \right) \nonumber \\
\eqexp{by \eqref{n_hat_sum_q}} &= \frac{N}{N-1} \left( \SCP{\psi_{as}}{\left( \widehat{n} \right)^2 \psi_{as}} - \frac{1}{N} \underbrace{\SCP{\psi_{as}}{q_1\psi_{as}}}_{\geq 0} \right) \nonumber \\
&\leq \frac{N}{N-1} \SCP{\psi_{as}}{\left( \widehat{n} \right)^2 \psi_{as}}.
\end{align}
\end{proof}

We now turn to the operators $\widehat{f}$ from Definition~\ref{def:projectors}. The following lemma gives a simple expression for powers of $\widehat{f}$.

\begin{lemma}\label{lem:powers_f}
For all $0 < s \in \QQQ$ and $f\geq 0$,
\be
\left(\widehat{f}\,\right)^s = \widehat{f^s}.
\ee
\end{lemma}

\begin{proof}
Recall that according to Lemma~\ref{lem:properties_P_Nk},
\be
P_{N,k} P_{N,\ell} = \delta_{k\ell} P_{N,k}
\ee
for all $k,\ell = 1,\ldots,N$. Let $0 \neq n \in \NNN$. Then
\begin{align}\label{powers_f_n}
\left(\widehat{f}\right)^{n} &= \sum_{k_1=0}^N f(k_1)P_{N,k_1} \cdot \ldots \cdot \sum_{k_n=0}^N f(k_n)P_{N,k_n} \nonumber \\
&= \sum_{k_1,\ldots,k_n=0}^N \left(\prod_{j=1}^n f(k_j)\right) \left(\prod_{j=1}^n P_{N,k_j}\right) \nonumber \\
&= \sum_{k=0}^N f(k)^n P_{N,k} \nonumber \\
&= \widehat{f^n}.
\end{align}
Recall that $f \geq 0$. It follows that for all $0 \neq m \in \NNN$,
\be
\left(\widehat{f^{\frac{1}{m}}}\right)^{m} = \widehat{f^{\frac{m}{m}}} = \widehat{f},
\ee
so
\be\label{powers_f_1_m}
\left(\widehat{f}\right)^{\frac{1}{m}} = \widehat{f^{\frac{1}{m}}}.
\ee
Together, Equations \eqref{powers_f_n} and \eqref{powers_f_1_m} prove the lemma.
\end{proof}

For completeness, let us also show how $\widehat{f}$ can be inverted. (The next lemma is not necessary for proving the rest of the statements in this thesis.)

\begin{lemma}
Let $f(0)=0$ and $f(k)>0$ for all $0 < k \leq N$. Then, for all $0 < s \in \QQQ$,
\be
\widehat{f}^{\,-s} \left(\id-P_{N,0}\right) = \widehat{f^{-s}},
\ee
with the definition $\widehat{f^{-s}}=\sum_{k=1}^N f(k)^{-s} P_{N,k}$.
\end{lemma}

\begin{proof}
Recall that $f(0)=0$. First, note that
\begin{align}
\widehat{f^{-1}} \widehat{f} &= \sum_{k=1}^N \frac{1}{f(k)} P_{N,k} \sum_{\ell=1}^N f(\ell) P_{N,\ell} \nonumber \\
&= \sum_{k=1}^N P_{N,k} \nonumber \\
&= \id - P_{N,0}.
\end{align}
It follows that
\be
\widehat{f^{-1}} = \widehat{f}^{-1} \left( \id-P_{N,0} \right).
\ee
Therefore, for all $0<s\in \QQQ$,
\be
\widehat{f^{-s}} = \widehat{f^{-1}}^s = \widehat{f}^{-s} \left( \id-P_{N,0} \right)^s.
\ee
We now show that $\left( \id-P_{N,0} \right)^s = \id-P_{N,0}$. First, let $0 \neq n \in \NNN$. By induction ($\left( \id-P_{N,0} \right)^1 = \id-P_{N,0}$, and now assuming $\left( \id-P_{N,0} \right)^n = \id-P_{N,0}$),
\begin{align}
\left( \id-P_{N,0} \right)^{n+1} &= \left( \id-P_{N,0} \right)^n \left( \id-P_{N,0} \right) \nonumber \\
&= \left( \id-P_{N,0} \right) \left( \id-P_{N,0} \right) \nonumber \\
&= \id - P_{N,0} - P_{N,0} + P_{N,0} \nonumber \\
&= \id-P_{N,0},
\end{align}
so $\left( \id-P_{N,0} \right)^n = \id-P_{N,0}$ for all $0 \neq n \in \NNN$. Then also $\left( \id-P_{N,0} \right)^{\frac{1}{m}} = \id-P_{N,0}$, since $\left( \id-P_{N,0} \right)^m = \id-P_{N,0}$, which proves the lemma.
\end{proof}

In order to estimate the three terms in \eqref{alpha_derivative} of the time derivative of $\alpha_f(t)$, we need to use the Cauchy-Schwarz inequality on both sides of the scalar product. It turns out that we often need a $\sqrt{\widehat{f}-\widehat{f}_{-d}}$ together with a $\psi$, so we need to shift one $\sqrt{\widehat{f}-\widehat{f}_{-d}}$ to the other side of the scalar product. The following lemma shows how this can be done.

\begin{lemma}\label{lem:root_f_hat}
Let $h_{12}$ be an operator acting only on $x_1$ and $x_2$, $a,b \in \{0,1,2\}$ and $d \in \{1,2\}$. (Recall that according to Definition~\ref{def:projectors2}, ${P_0}^{\{1,2\}} = p_1p_2$, ${P_1}^{\{1,2\}} = p_1q_2 + q_1p_2$ and ${P_2}^{\{1,2\}} = q_1q_2$.) Then, for all monotone increasing $f$ and $\psi\in L^2(\RRR^{3N})$,
\begin{align}\label{root_f_hat_eq}
&\bigSCP{\psi}{\left( \widehat{f}-\widehat{f}_{-d} \right) P^{\{1,2\}}_a h_{12} P^{\{1,2\}}_b \psi} \nonumber \\
&\qquad= \Big\langle\!\!\Big\langle \psi , \left( \widehat{f}-\widehat{f}_{-d} \right)^{\frac{1}{2}} P^{\{1,2\}}_a h_{12} P^{\{1,2\}}_b \times \nonumber \\
&\qquad\quad \times \left( \widehat{f}_{a-b}-\widehat{f}_{a-b-d} + \sum_{\ell=1}^{a-b} f(N-d+\ell)P_{N,N-(a-b)+\ell}\right)^{\frac{1}{2}} \psi \Big\rangle\!\!\Big\rangle \nonumber \\
&\qquad= \bigSCP{\psi}{\left( \widehat{f}-\widehat{f}_{-d} \right)^{\frac{1}{2}} P^{\{1,2\}}_a h_{12} P^{\{1,2\}}_b \left( \widehat{f}_{a-b}-\widehat{f}_{a-b-d}\right)^{\frac{1}{2}} \psi}.
\end{align}
\end{lemma}

\begin{proof}
Recall that for all $k<0$ and $k>N$ we define $f(k)=0$ and $P_{N,k}=0$. Also recall the definition of the shifted $\widehat{f}$,
\be
\widehat{f}_d = \sum_{k=0}^N f(k+d) P_{N,k}.
\ee
We find, using Lemma~\ref{lem:shift_Ps} and $P_{N,k}P_{N,\ell}=\delta_{k\ell}P_{N,k}$ from Lemma~\ref{lem:properties_P_Nk}, that
\begin{align}\label{f-f_-d}
&\bigSCP{\psi}{\big( \widehat{f}-\widehat{f}_{-d} \big) P^{\{1,2\}}_a h_{12} P^{\{1,2\}}_b \psi} \nonumber \\
&\qquad= \bigSCP{\psi}{\sum_{k=0}^N \big( f(k)-f(k-d) \big) P_{N,k} P^{\{1,2\}}_a h_{12} P^{\{1,2\}}_b \psi} \nonumber \\
\eqexpl{by Lem.~\ref{lem:shift_Ps}} &\qquad = \bigSCP{\psi}{\sum_{k=0}^N \big( f(k)-f(k-d) \big)^{\frac{1}{2}} P_{N,k} P^{\{1,2\}}_a h_{12} P^{\{1,2\}}_b \times \nonumber \\
&\qquad\quad \times \big( f(k)-f(k-d) \big)^{\frac{1}{2}} P_{N,k+b-a} \psi} \nonumber \\
&\qquad= \bigSCP{\psi}{\sum_{k=0}^N \big( f(k)-f(k-d) \big)^{\frac{1}{2}} P_{N,k} P^{\{1,2\}}_a h_{12} P^{\{1,2\}}_b \times \nonumber \\
&\qquad\quad \times \sum_{\ell=b-a}^{N+b-a} \big( f(\ell-b+a)-f(\ell-b+a-d) \big)^{\frac{1}{2}} P_{N,\ell} \, \psi}.
\end{align}
Note that $f(k)-f(k-d)\geq 0$ since $f$ is monotone increasing. Furthermore, with Lemma~\ref{lem:powers_f},
\begin{align}\label{root_f-f_a-b}
& \sum_{\ell=b-a}^{N+b-a} \Big( f(\ell-b+a)-f(\ell-b+a-d) \Big)^{\frac{1}{2}} P_{N,\ell} \, \psi \nonumber \\
&\qquad= \left( \sum_{\ell=b-a}^{N+b-a} \Big( f(\ell-b+a)-f(\ell-b+a-d) \Big) P_{N,\ell} \right)^{\frac{1}{2}} \psi \nonumber \\
&\qquad= \left( \widehat{f}_{a-b} - \widehat{f}_{a-b-d} + \sum_{\ell=N+b-a+1}^N f(\ell-b+a-d)P_{N,\ell} \right)^{\frac{1}{2}} \psi \nonumber \\
&\qquad= \left( \widehat{f}_{a-b} - \widehat{f}_{a-b-d} + \sum_{\ell=1}^{a-b} f(N-d+\ell)P_{N,N-(a-b)+\ell} \right)^{\frac{1}{2}} \psi.
\end{align}
In order to show the second equality in \eqref{root_f_hat_eq} we use that in \eqref{f-f_-d} there is a $P^{\{1,2\}}_b$ in front of \eqref{root_f-f_a-b}. The operator $P^{\{1,2\}}_b$ contains $b$ $q$'s or $(2-b)$ $p$'s. Therefore $P_{N,\ell}P^{\{1,2\}}_b = 0$ for all $N+b-a < \ell \leq N$ for all combinations of $a,b = 0,1,2$.\footnote{The possible cases are:
\begin{itemize}
\item $b-a = -1$: Then $\ell = N$ and either $b=0$ or $b=1$. Then $P^{\{1,2\}}_b$ contains at lest one $p$ and since $P_{N,N}p_1=0$, also $P_{N,\ell}P^{\{1,2\}}_b = 0$.
\item $b-a = -2$: Then either $\ell = N-1$ or $\ell = N$ and $b=0$. Then $P^{\{1,2\}}_b$ contains two $p$'s and since $P_{N,N-1}p_1p_2=0=P_{N,N}p_1p_2$, also $P_{N,\ell}P^{\{1,2\}}_b = 0$.
\end{itemize}} Then, in \eqref{root_f-f_a-b} multiplied with $P^{\{1,2\}}_b$, one can replace $\sum_{\ell=b-a}^{N+b-a}$ by $\sum_{\ell=b-a}^N$, and
\begin{align}
& \sum_{\ell=b-a}^{N+b-a} \Big( f(\ell-b+a)-f(\ell-b+a-d) \Big)^{\frac{1}{2}} P_{N,\ell} P^{\{1,2\}}_b \psi \nonumber \\
&\qquad= \sum_{\ell=b-a}^N \Big( f(\ell-b+a)-f(\ell-b+a-d) \Big)^{\frac{1}{2}} P_{N,\ell} P^{\{1,2\}}_b \psi \nonumber \\
&\qquad= \left( \sum_{\ell=b-a}^N f(\ell-b+a)P_{N,\ell} - \sum_{\ell=b-a}^N f(\ell-b+a-d)P_{N,\ell} \right)^{\frac{1}{2}} P^{\{1,2\}}_b \psi \nonumber \\
&\qquad= \big( \widehat{f}_{a-b} - \widehat{f}_{a-b-d} \big)^{\frac{1}{2}} P^{\{1,2\}}_b \psi.
\end{align}
\end{proof}

The next lemma shows how expressions involving $\sqrt{\widehat{f}-\widehat{f}_{-d}}\,\psi$ and possibly one or two $q$'s can be estimated explicitly for the weight function $m^{(\gamma)}(k)$ that we use later.

\begin{lemma}\label{lem:q_root_f}
Let $\widehat{m} := \widehat{m^{(\gamma)}}=\sum_{k=0}^N m^{(\gamma)}(k) P_{N,k}$ with $m^{(\gamma)}(k)$ as in \eqref{weight_m_gamma}. For any antisymmetric $\psi_{as}\in L^2(\RRR^{3N})$, we abbreviate
\be
\widetilde{\psi}_1 := \left( \widehat{m}-\widehat{m}_{-d} \right)^{\frac{1}{2}}\psi_{as},
\ee
and
\be
\widetilde{\psi}_0 := \left( \widehat{m}_d-\widehat{m} + \sum_{\ell=1}^d m(N-d+\ell)P_{N,N-d+\ell} \right)^{\frac{1}{2}}\psi_{as},
\ee
for any $d =1,2$. Then, for all $c=0,1$ and all normalized antisymmetric $\psi_{as}$,
\begin{align}
\norm{\widetilde{\psi}_c}^2 \leq d \, N^{-\gamma},
\end{align}
\begin{align}
\norm{q_1 \widetilde{\psi}_c}^2 \leq d(d+1)^c \, N^{-1} \, \alpha_m,
\end{align}
\begin{align}
\norm{q_1q_2 \widetilde{\psi}_c}^2 \leq d(d+1)^{2c} \, N^{\gamma-2} \, \alpha_m.
\end{align}
\end{lemma}

\begin{proof}
The proof of this lemma is not hard but the many different cases and the appearance of certain boundary terms make it a bit lengthy. Let us therefore give a short version of the proof first and afterwards present more details. First, recall that according to Lemma~\ref{lem:properties_P_Nk}, for all $\psi_{as}$,
\be\label{q_1P_N_K}
\SCP{\psi_{as}}{q_1\psi_{as}} = \sum_{k=0}^N \frac{k}{N} \SCP{\psi_{as}}{P_{N,k}\psi_{as}},
\ee
and
\be\label{q_1q_2P_N_K}
\SCP{\psi_{as}}{q_1q_2\psi_{as}} \leq 2 \sum_{k=0}^N \left(\frac{k}{N}\right)^2 \SCP{\psi_{as}}{P_{N,k}\psi_{as}}.
\ee
We denote by $\floor{\cdot}$ the floor function, i.e., for any $x \in \RRR$, $\floor{x} = \max\{m \in \ZZZ: m \leq x \}$. Note that
\be
\alpha_m = \sum_{k=0}^N m(k) \bigSCP{\psi_{as}}{P_{N,k}\psi_{as}} = \sum_{k=1}^{\floor{N^{\gamma}}} \frac{k}{N^{\gamma}} \bigSCP{\psi_{as}}{P_{N,k}\psi_{as}} + \sum_{k=\floor{N^{\gamma}}+1}^N  \bigSCP{\psi_{as}}{P_{N,k}\psi_{as}}.
\ee

\textbf{Short version.} 
Note that
\be
\widehat{m} - \widehat{m}_{-d} = \sum_{k=0}^N \, \Big( m(k) - m(k-d) \Big) \, P_{N,k} \approx N^{-\gamma} \sum_{k=0}^{N^{\gamma}} P_{N,k},
\ee
and similarly for $\widehat{m}_d - \widehat{m}$. (It is this point where we neglect boundary terms at $k\approx N^{\gamma}$; later the ``$\approx \dots$'' is replaced by ``$\leq C \dots$''.) Therefore (note that $q_1\widehat{m} = \widehat{m}q_1$ and $\sum_{k=0}^N P_{N,k}=1$),
\be
\norm{ \left( \widehat{m}-\widehat{m}_{-d} \right)^{\frac{1}{2}}\psi_{as}}^2 = \bigSCP{\psi_{as}}{ \Big( \widehat{m} - \widehat{m}_{-d} \Big) \psi_{as}} \approx \, N^{-\gamma} \,\sum_{k=0}^{N^{\gamma}} \bigSCP{\psi_{as}}{P_{N,k} \psi_{as}} \leq N^{-\gamma},
\ee
\begin{align}
\norm{ q_1 \left( \widehat{m}-\widehat{m}_{-d} \right)^{\frac{1}{2}}\psi_{as}}^2 &= \bigSCP{\psi_{as}}{ q_1 \Big( \widehat{m} - \widehat{m}_{-d} \Big) \psi_{as}} \approx \, N^{-\gamma} \,\sum_{k=0}^{N^{\gamma}} \frac{k}{N} \bigSCP{\psi_{as}}{P_{N,k} \psi_{as}} \nonumber \\
&= N^{-1} \,\sum_{k=0}^{N^{\gamma}} \frac{k}{N^\gamma} \bigSCP{\psi_{as}}{P_{N,k} \psi_{as}} \leq N^{-1} \alpha_m,
\end{align}
and
\begin{align}
\norm{ q_1q_2 \left( \widehat{m}-\widehat{m}_{-d} \right)^{\frac{1}{2}}\psi_{as}}^2 &= \bigSCP{\psi_{as}}{ q_1q_2 \Big( \widehat{m} - \widehat{m}_{-d} \Big) \psi_{as}} \lessapprox \, N^{-\gamma} \,\sum_{k=0}^{N^{\gamma}} \frac{k^2}{N^2} \bigSCP{\psi_{as}}{P_{N,k} \psi_{as}} \nonumber \\
&= N^{\gamma-2} \,\sum_{k=0}^{N^{\gamma}} \frac{k^2}{N^{2\gamma}} \bigSCP{\psi_{as}}{P_{N,k} \psi_{as}} \leq N^{\gamma-2} \alpha_m.
\end{align}
What we now do in the more detailed proof is to keep track of the boundary terms at $k\approx N^{\gamma}$ and therewith keep track of the exact constants that appear in the estimates.

\textbf{Detailed proof.} Let us first consider
\be
\widetilde{\psi} = \left( \widehat{m}-\widehat{m}_{-d} \right)^{\frac{1}{2}}\psi_{as}.
\ee
We need the following estimates:
\be
\Big( m(k) - m(k-d) \Big) \leq \left\{\begin{array}{cl} \frac{k}{N^{\gamma}} - \frac{(k-d)}{N^{\gamma}} & , k\leq N^{\gamma}+d \\ 0 & , k > N^{\gamma}+d \end{array}\right\} = \left\{\begin{array}{cl} \frac{d}{N^{\gamma}} & , k\leq N^{\gamma}+d \\ 0 & , k > N^{\gamma}+d \end{array}\right. ,
\ee
\begin{align}
\frac{k}{N}\Big( m(k) - m(k-d) \Big) &\leq \left\{\begin{array}{cl} \frac{dk}{NN^{\gamma}} & , k\leq N^{\gamma}+d \\ 0 & , k > N^{\gamma}+d \end{array}\right. \nonumber \\
&\leq \left\{\begin{array}{cl} \frac{dk}{NN^{\gamma}} & , k\leq N^{\gamma} \\ \frac{d\left(N^{\gamma}+d\right)}{NN^{\gamma}} & , N^{\gamma}<k\leq N^{\gamma}+d \\ 0 & , k > N^{\gamma}+d \end{array}\right. \nonumber \\
&\leq \frac{d(d+1)}{N} \, \left\{\begin{array}{cl} \frac{k}{N^{\gamma}} & , k\leq N^{\gamma} \\ 1 & , N^{\gamma}<k\leq N^{\gamma}+d \\ 0 & , k > N^{\gamma}+d \end{array}\right. ,
\end{align}
\begin{align}
\left(\frac{k}{N}\right)^2 \Big( m(k) - m(k-d) \Big) &\leq \, \left\{\begin{array}{cl} \frac{dk^2}{N^2N^{\gamma}} & , k\leq N^{\gamma}+d \\ 0 & , k > N^{\gamma}+d \end{array}\right. \nonumber \\
&\leq \left\{\begin{array}{cl} \frac{dN^{\gamma}k}{N^2N^{\gamma}} & , k\leq N^{\gamma} \\ \frac{d\left(N^{\gamma}+d\right)^2}{N^2N^{\gamma}} & , N^{\gamma}<k\leq N^{\gamma}+d \\ 0 & , k > N^{\gamma}+d \end{array}\right. \nonumber \\
&\leq \frac{d(d+1)^2N^{\gamma}}{N^2} \, \left\{\begin{array}{cl} \frac{k}{N^{\gamma}} & , k\leq N^{\gamma} \\ 1 & , N^{\gamma}<k\leq N^{\gamma}+d \\ 0 & , k > N^{\gamma}+d \end{array}\right. .
\end{align}
With theses estimates we have
\begin{align}
\norm{ \left( \widehat{m}-\widehat{m}_{-d} \right)^{\frac{1}{2}}\psi_{as}}^2 &= \bigSCP{\psi_{as}}{ \Big( \widehat{m} - \widehat{m}_{-d} \Big) \psi_{as}} \nonumber \\
&= \sum_{k=0}^N \, \underbrace{\Big( m(k) - m(k-d) \Big)}_{\geq 0} \, \underbrace{\bigSCP{\psi_{as}}{P_{N,k} \psi_{as}}}_{\geq 0} \nonumber \\
&\leq \sum_{k=0}^{\floor{N^{\gamma}}+d} \, \frac{d}{N^{\gamma}} \, \bigSCP{\psi_{as}}{P_{N,k} \psi_{as}} \nonumber \\
&\leq \frac{d}{N^{\gamma}} \bigSCP{\psi_{as}}{\sum_{k=0}^N P_{N,k} \psi_{as}} \nonumber \\
&= \frac{d}{N^{\gamma}},
\end{align}
and, with \eqref{q_1P_N_K},
\begin{align}
&\norm{q_1 \left( \widehat{m}-\widehat{m}_{-d} \right)^{\frac{1}{2}}\psi_{as}}^2 \nonumber \\
&\qquad= \sum_{k=0}^N \Big( m(k) - m(k-d) \Big) \bigSCP{\psi_{as}}{q_1 P_{N,k} \psi_{as}} \nonumber \\
&\qquad= \sum_{k=0}^N \, \frac{k}{N}\Big( m(k) - m(k-d) \Big) \, \bigSCP{\psi_{as}}{P_{N,k} \psi_{as}} \nonumber \\
&\qquad\leq \frac{d(d+1)}{N} \left( \sum_{k=0}^{\floor{N^{\gamma}}} \frac{k}{N^{\gamma}} \bigSCP{\psi_{as}}{P_{N,k} \psi_{as}} + \sum_{k=\floor{N^{\gamma}}+1}^N \bigSCP{\psi_{as}}{P_{N,k} \psi_{as}} \right) \nonumber \\
&\qquad\leq \frac{d(d+1)}{N} \, \alpha_m ,
\end{align}
and, with \eqref{q_1q_2P_N_K},
\begin{align}
&\norm{q_1q_2 \left( \widehat{m}-\widehat{m}_{-d} \right)^{\frac{1}{2}}\psi_{as}}^2 \nonumber  \\
&\qquad= \sum_{k=0}^N \Big( m(k) - m(k-d) \Big) \bigSCP{\psi_{as}}{q_1q_2 P_{N,k} \psi_{as}} \nonumber \\
&\qquad\leq 2 \sum_{k=0}^N \, \left(\frac{k}{N}\right)^2 \Big( m(k) - m(k-d) \Big) \, \bigSCP{\psi_{as}}{P_{N,k} \psi_{as}} \nonumber \\
&\qquad\leq \frac{d(d+1)^2N^{\gamma}}{N^2} \left( \sum_{k=0}^{\floor{N^{\gamma}}} \frac{k}{N^{\gamma}} \bigSCP{\psi_{as}}{P_{N,k} \psi_{as}} + \sum_{k=\floor{N^{\gamma}}+1}^N \bigSCP{\psi_{as}}{P_{N,k} \psi_{as}} \right) \nonumber \\
&\qquad\leq \frac{d(d+1)^2N^{\gamma}}{N^2} \, \alpha_m .
\end{align}

We now consider
\be
\widetilde{\psi} = \left( \widehat{m}_d-\widehat{m} + \sum_{\ell=1}^d m(N-d+\ell)P_{N,N-d+\ell} \right)^{\frac{1}{2}}\psi_{as}.
\ee
We need the estimates
\be
\Big( m(k+d) - m(k) \Big) \leq \left\{\begin{array}{cl} \frac{d}{N^{\gamma}} & , k\leq N^{\gamma} \\ 0 & ,N^{\gamma}<k \leq N-d \end{array}\right. ,
\ee
\be
\frac{k}{N}\Big( m(k+d) - m(k) \Big) \leq \frac{d}{N} \, \left\{\begin{array}{cl} \frac{k}{N^{\gamma}} & , k\leq N^{\gamma} \\ 0 & , N^{\gamma}<k \leq N-d \end{array}\right. ,
\ee
\begin{align}
\left(\frac{k}{N}\right)^2 \Big( m(k+d) - m(k) \Big) &\leq \, \left\{\begin{array}{cl} \frac{dk^2}{N^2N^{\gamma}} & , k\leq N^{\gamma} \\ 0 & , N^{\gamma}<k \leq N-d \end{array}\right. \nonumber \\
&\leq \frac{dN^{\gamma}}{N^2} \, \left\{\begin{array}{cl} \frac{k}{N^{\gamma}} & , k\leq N^{\gamma} \\ 0 & , N^{\gamma}<k \leq N-d \end{array}\right. .
\end{align}
With these estimates we find
\begin{align}
& \norm{\left( \widehat{m}_d-\widehat{m} + \sum_{\ell=1}^d m(N-d+\ell)P_{N,N-d+\ell} \right)^{\frac{1}{2}}\psi_{as}}^2 \nonumber \\
&\qquad= \bigSCP{\psi_{as}}{ \Big( \widehat{m}_d-\widehat{m} + \sum_{\ell=1}^d m(N-d+\ell)P_{N,N-d+\ell} \Big) \psi_{as}} \nonumber \\
&\qquad= \sum_{k=0}^N \, \bigSCP{\psi_{as}}{\left( \Big(m(k+d) - m(k)\Big) P_{N,k} + \sum_{\ell=1}^d m(N-d+\ell)P_{N,N-d+\ell} \right) \psi_{as}} \nonumber \\
&\qquad= \sum_{k=0}^{N-d} \, \Big(m(k+d) - m(k)\Big) \, \bigSCP{\psi_{as}}{ P_{N,k} \psi_{as}} \nonumber \\
&\qquad\leq \sum_{k=0}^{\floor{N^{\gamma}}} \, \frac{d}{N^{\gamma}} \, \bigSCP{\psi_{as}}{P_{N,k} \psi_{as}} \nonumber \\
&\qquad\leq \frac{d}{N^{\gamma}},
\end{align}
and, with \eqref{q_1P_N_K},
\begin{align}
& \norm{q_1 \left( \widehat{m}_d-\widehat{m} + \sum_{\ell=1}^d m(N-d+\ell)P_{N,N-d+\ell} \right)^{\frac{1}{2}}\psi_{as}}^2 \nonumber \\
&\qquad= \sum_{k=0}^{N-d} \, \frac{k}{N} \Big(m(k+d) - m(k)\Big) \, \bigSCP{\psi_{as}}{ P_{N,k} \psi_{as}} \nonumber \\
&\qquad\leq \frac{d}{N} \left( \sum_{k=0}^{\floor{N^{\gamma}}} \frac{k}{N^{\gamma}} \bigSCP{\psi_{as}}{P_{N,k} \psi_{as}} + \sum_{k=\floor{N^{\gamma}}+1}^N \bigSCP{\psi_{as}}{P_{N,k} \psi_{as}} \right) \nonumber \\
&\qquad\leq \frac{d}{N} \, \alpha_m ,
\end{align}
and, with \eqref{q_1q_2P_N_K},
\begin{align}
& \norm{q_1q_2 \left( \widehat{m}_d-\widehat{m} + \sum_{\ell=1}^d m(N-d+\ell)P_{N,N-d+\ell} \right)^{\frac{1}{2}}\psi_{as}}^2 \nonumber \\
&\qquad\leq 2 \sum_{k=0}^{N-d} \, \left(\frac{k}{N}\right)^2 \Big(m(k+d) - m(k)\Big) \, \bigSCP{\psi_{as}}{ P_{N,k} \psi_{as}} \nonumber \\
&\qquad\leq \frac{dN^{\gamma}}{N^2} \left( \sum_{k=0}^{\floor{N^{\gamma}}} \frac{k}{N^{\gamma}} \bigSCP{\psi_{as}}{P_{N,k} \psi_{as}} + \sum_{k=\floor{N^{\gamma}}+1}^N \bigSCP{\psi_{as}}{P_{N,k} \psi_{as}} \right) \nonumber \\
&\qquad\leq \frac{dN^{\gamma}}{N^2} \, \alpha_m .
\end{align}
\end{proof}

\section{Diagonalization of $p_2h_{12}p_2$ and Related Lemmas}\label{sec:estimates_projectors}
In the time derivative of $\alpha_f$ from equation \eqref{alpha_derivative} there appears the operator $p_2v_{12}p_2$. Later, when we use the Cauchy-Schwarz inequality on the terms from \eqref{alpha_derivative}, we also have to deal with related operators like $p_2v_{12}^2p_2$. Generally, for any function $h(x)$ (recall that we write $h_{12}=h(x_1-x_2)$), an operator of the type $p_2h_{12}p_2$ is a multiplication operator in $x_1$ and a projector onto the $N$-dimensional subspace $\Span(\varphi_1,\ldots,\varphi_N)$ in the second variable. Therefore one can write it as an $x_1$ dependent $(N \times N)$-matrix, acting on the second variable. This matrix is self-adjoint and non-negative for $h\geq 0$. Therefore, for fixed $x_1$, one can diagonalize it, as is shown in the following lemma. (Since we later split $v = v_{+} - v_{-}$, with $v_{+},v_{-} \geq 0$ we state the lemma only for non-negative $h$.) Recall that we denote by $\ket{\cdot}_m$ a vector acting on the $m$-th variable of $L^2(\RRR^{3N})$, and by $\scp{\cdot}{\cdot}_m$ the scalar product only in the $m$-th variable.

\begin{lemma}\label{lem:diag_pvp}
Let $h(x)$ be a non-negative function. Let $h$ and $\varphi_1,\ldots,\varphi_N$ be such that
\be
(h \star \rho_N)(x) < \infty
\ee
for all $x\in \RRR^3$, where $\rho_N(x):=\sum_{i=1}^N |\varphi_i(x)|^2$. Then, for fixed $x_1$, there are orthonormal functions $\chi_1^{x_1},\ldots,\chi_N^{x_1} \in \Span(\varphi_1,\ldots,\varphi_N)$ and non-negative eigenvalues $\lambda_1(x_1),\ldots,\lambda_N(x_1)$, such that
\be
p_2h_{12}p_2 = \sum_{i=1}^N \lambda_i(x_1) \, \ket{\chi_i^{x_1}} \bra{\chi_i^{x_1}}_2 = \sum_{i=1}^N \lambda_i(x_1) \, p_2^{\chi_i^{x_1}}
\ee
with
\be
\lambda_i(x_1) = \scp{\chi_i^{x_1}}{h_{12} \, \chi_i^{x_1}}_2(x_1) = \int d^3x~ h(x_1-x) \left|\chi_i^{x_1}(x) \right|^2 < \infty \quad \forall\, i=1,\ldots,N ~\forall\, x_1\in \RRR^3,
\ee
and, for $i \neq j$,
\be
\scp{\chi_i^{x_1}}{h_{12} \, \chi_j^{x_1}}_2(x_1) = 0 \quad \forall~ x_1 \in \RRR^3.
\ee
Furthermore,
\be
\sum_{i=1}^N \lambda_i(x_1) = \sum_{i=1}^N \scp{\varphi_i}{h_{12} \, \varphi_i}_2(x_1) = (h \star \rho_N)(x_1).
\ee
\end{lemma}

\begin{proof}
In the following we always keep $x_1$ fixed. First, note that
\begin{align}
p_2h_{12}p_2 &= \sum_{i,j=1}^N \ket{\varphi_i}\bra{\varphi_i}_2 \, h(x_1-x_2) \, \ket{\varphi_j}\bra{\varphi_j}_2 \nonumber \\
&= \sum_{i,j=1}^N \scp{\varphi_i}{h_{12}\,\varphi_j}_2(x_1) \, \ket{\varphi_i}\bra{\varphi_j}_2.
\end{align}
In the second variable this is a self-adjoint $(N\times N)$-matrix. For $h\geq 0$ it is non-negative and can therefore be diagonalized with non-negative eigenvalues. That means, there is a unitary $(N\times N)$-matrix $U(x_1)$, such that, for all $i=1,\ldots,N$, the functions
\be\label{chi_unitary_phi}
\ket{\chi_i^{x_1}}_2 = \sum_{k=1}^N U_{ik}(x_1) \, \ket{\varphi_k}_2
\ee
are orthonormal, and such that
\be
p_2h_{12}p_2 = \sum_{i=1}^N \lambda_i(x_1) \, \ket{\chi_i^{x_1}} \bra{\chi_i^{x_1}}_2 = \sum_{i=1}^N \lambda_i(x_1) \, p_2^{\chi_i^{x_1}}.
\ee
Note that \eqref{chi_unitary_phi} can be inverted, i.e.,
\be
\ket{\varphi_k}_2 = \sum_{\ell=1}^N U^*_{\ell k}(x_1) \, \ket{\chi_{\ell}^{x_1}}_2,
\ee
where ${}^*$ denotes complex conjugation. The projector $p_2$ is independent of the choice of basis, therefore
\be
p_2 = \sum_{i=1}^N \ket{\varphi_i} \bra{\varphi_i}_2 = \sum_{i=1}^N \ket{\chi_i^{x_1}} \bra{\chi_i^{x_1}}_2.
\ee
Then, since $\Span(\chi_1^{x_1},\ldots,\chi_N^{x_1}) = \Span(\varphi_1,\ldots,\varphi_N)$,
\begin{align}
\lambda_i(x_1) &= \scp{\chi_i^{x_1}}{p_2h_{12}p_2\chi_i^{x_1}}_2(x_1) = \scp{\chi_i^{x_1}}{h_{12}\,\chi_i^{x_1}}_2(x_1)
\end{align}
and, for all $i \neq j$,
\begin{align}
\scp{\chi_i^{x_1}}{h_{12}\,\chi_j^{x_1}}_2(x_1) &= \scp{\chi_i^{x_1}}{p_2h_{12}p_2\,\chi_j^{x_1}}_2(x_1) \nonumber \\
&= \scp{\chi_i^{x_1}}{\lambda_j(x_1)\,\chi_j^{x_1}}_2(x_1) \nonumber \\
&= \lambda_j(x_1) \scp{\chi_i^{x_1}}{\chi_j^{x_1}}_2(x_1) \nonumber \\
&= 0.
\end{align}
Furthermore, since $U(x_1)$ is unitary,
\begin{align}
\sum_{i=1}^N \lambda_i(x_1) &= \sum_{i=1}^N \scp{\chi_i^{x_1}}{h_{12} \, \chi_i^{x_1}}_2(x_1) \nonumber \\
&= \sum_{i=1}^N \sum_{j=1}^N \sum_{k=1}^N U^*_{ij}(x_1) U_{ik}(x_1) \scp{\varphi_j}{h_{12} \, \varphi_k}_2(x_1) \nonumber \\
&= \sum_{j=1}^N \sum_{k=1}^N \left( \sum_{i=1}^N U^*_{ij}(x_1) U_{ik}(x_1) \right) \scp{\varphi_j}{h_{12} \, \varphi_k}_2(x_1) \nonumber \\
&= \sum_{j=1}^N \sum_{k=1}^N \delta_{jk} \scp{\varphi_j}{h_{12} \, \varphi_k}_2(x_1) \nonumber \\
&= \sum_{j=1}^N \scp{\varphi_j}{h_{12} \, \varphi_j}_2(x_1) \nonumber \\
&= (h \star \rho_N)(x_1).
\end{align}
\end{proof}

Lemma~\ref{lem:diag_pvp} can now be used to bound scalar products involving expressions like $p_2h_{12}p_2$ by the convolution $h\star\rho_N$. The following three lemmas treat the expressions that we need later in order to bound the time derivative of $\alpha_f(t)$. Recall that $\rho_N(x) := \sum_{i=1}^N |\varphi_i(x)|^2$.

\begin{lemma}\label{lem:psi_pvp_psi}
Let $\varphi_1,\ldots,\varphi_N \in L^2(\RRR^3)$ be orthonormal and $h\geq 0$. Then, for all antisymmetric $\psi_{as} \in L^2(\RRR^{3N})$,
\be\label{psi_pvp_psi}
\SCP{\psi_{as}}{p_2 \, h_{12}\, p_2 \psi_{as}} \leq \frac{1}{N-1} \, \left(\sup_{y \in \RRR^3} (h \star \rho_N)(y) \right) \, \SCP{\psi_{as}}{\psi_{as}}.
\ee
This inequality remains true when $\psi_{as}$ is antisymmetric only in the variables $x_2,\ldots,x_N$. It also remains true with $\frac{1}{N-1}$ replaced by $\frac{1}{N-2}$, when $\psi_{as}$ is antisymmetric only in the variables $x_2,x_4,x_5,\ldots,x_N$.
\end{lemma}

\begin{lemma}\label{lem:psi_ppv2pp_psi}
Let $\varphi_1,\ldots,\varphi_N \in L^2(\RRR^3)$ be orthonormal and $h\geq 0$. Then, for all antisymmetric $\psi_{as} \in L^2(\RRR^{3N})$,
\be\label{psi_ppvpp_psi}
\SCP{\psi_{as}}{p_1p_2 \, h_{12} \, p_1p_2 \psi_{as}} \leq \frac{1}{N(N-1)} \, \left(\int_{\RRR^3} (h \star \rho_N)(y)\, \rho_N(y) \, d^3y\right) \, \SCP{\psi_{as}}{\psi_{as}}.
\ee
This inequality remains true with $\frac{1}{N(N-1)}$ replaced by $\frac{1}{(N-1)(N-2)}$, when $\psi_{as}$ is antisymmetric only in the variables $x_1,x_2,x_4,\ldots,x_N$.
\end{lemma}

\begin{lemma}\label{lem:psi_qppvvppq_psi}
Let $\varphi_1,\ldots,\varphi_N \in L^2(\RRR^3)$ be orthonormal and $h\geq 0$. Then, for all antisymmetric $\psi_{as} \in L^2(\RRR^{3N})$,
\be
\SCP{\psi_{as}}{q_3p_1p_2 \, h_{12} h_{13} \, p_1p_3q_2 \psi_{as}} \leq \frac{1}{(N-1)(N-2)} \, \left(\sup_{y \in \RRR^3} (h \star \rho_N)(y) \right)^2 \, \SCP{\psi_{as}}{q_1 \, \psi_{as}}.
\ee
\end{lemma}

\begin{proof}[Proof of Lemma \ref{lem:psi_pvp_psi}]
First, recall that we denote by $\SCP{\cdot}{\cdot}_{a+1,\ldots,N}(x_1,\ldots,x_a)$ the scalar product only in the variables $x_{a+1},\ldots,x_N$, evaluated at $x_1,\ldots,x_a$, i.e.,
\be
\SCP{\psi}{\chi}_{a+1,\ldots,N}(x_1,\ldots,x_a) := \int d^3x_{a+1} \ldots \int d^3x_N \, \psi^*(x_1,\ldots,x_N) \chi(x_1,\ldots,x_N).
\ee
For all antisymmetric $\psi_{as}$ we find, using Lemmas \ref{lem:diag_pvp} and \ref{lem:projector_norms},
\begin{align}\label{pvp_proof}
\SCP{\psi_{as}}{p_2 \, h_{12}\, p_2 \psi_{as}} &= \sum_{i=1}^N \SCP{\psi_{as}}{\lambda_i(x_1) \, p_2^{\chi_i^{x_1}} \psi_{as}} \nonumber \\
&= \int d^3x_1 \sum_{i=1}^N \underbrace{\lambda_i(x_1)}_{\geq 0} \underbrace{\SCP{\psi_{as}}{p_2^{\chi_i^{x_1}} \psi_{as}}_{2,\ldots,N}(x_1)}_{\geq 0 ~ \forall \, x_1} \nonumber \\
\eqexp{by Lem.~\ref{lem:projector_norms}} &\leq \int d^3x_1 \sum_{i=1}^N \lambda_i(x_1) \frac{1}{N-1}\SCP{\psi_{as}}{\psi_{as}}_{2,\ldots,N}(x_1) \nonumber \\
&\leq \frac{1}{N-1} \left( \sup_{x_1} \sum_{i=1}^N \lambda_i(x_1)\right) \int d^3x_1 \SCP{\psi_{as}}{\psi_{as}}_{2,\ldots,N}(x_1) \nonumber \\
&= \frac{1}{N-1} \left( \sup_{x_1}\, (h\star\rho_N)(x_1) \right) \SCP{\psi_{as}}{\psi_{as}}.
\end{align}
Note that we did not use the antisymmetry in the first variable, so \eqref{pvp_proof} remains true when $\psi_{as}$ is antisymmetric only in $x_2,\ldots,x_N$. For all $\psi_{as}^{1,3}$ that are antisymmetric in all variables except $x_1,x_3$, we find, again using Lemmas \ref{lem:diag_pvp} and \ref{lem:projector_norms},
\begin{align}\label{pvp_proof2}
\SCP{\psi_{as}^{1,3}}{p_2 \, h_{12}\, p_2 \psi_{as}^{1,3}} &= \sum_{i=1}^N \SCP{\psi_{as}^{1,3}}{\lambda_i(x_1) \, p_2^{\chi_i^{x_1}} \psi_{as}^{1,3}} \nonumber \\
&= \int d^3x_1 \sum_{i=1}^N \underbrace{\lambda_i(x_1)}_{\geq 0} \int d^3x_3 \underbrace{\SCP{\psi_{as}^{1,3}}{p_2^{\chi_i^{x_1}} \psi_{as}^{1,3}}_{2,4,\ldots,N}(x_1,x_3)}_{\geq 0 ~ \forall \, x_1,x_3} \nonumber \\
\eqexp{by Lem.~\ref{lem:projector_norms}} &\leq \int d^3x_1 \sum_{i=1}^N \lambda_i(x_1) \frac{1}{N-2} \int d^3x_3 \SCP{\psi_{as}^{1,3}}{\psi_{as}^{1,3}}_{2,4,\ldots,N}(x_1,x_3) \nonumber \\
&\leq \frac{1}{N-2} \left( \sup_{x_1} \sum_{i=1}^N \lambda_i(x_1)\right) \int d^3x_1 \SCP{\psi_{as}^{1,3}}{\psi_{as}^{1,3}}_{2,\ldots,N}(x_1) \nonumber \\
&= \frac{1}{N-2} \left( \sup_{x_1}\, (h\star\rho_N)(x_1) \right) \SCP{\psi_{as}^{1,3}}{\psi_{as}^{1,3}}.
\end{align}
\end{proof}

\begin{proof}[Proof of Lemma \ref{lem:psi_ppv2pp_psi}]
In the following we diagonalize $p_2h_{12}p_2$ as in Lemma~\ref{lem:diag_pvp}. For $g \geq 0$, we also diagonalize $p_1g(x_1)p_1$ as $(N \times N)$-matrix in the first variable. Here the eigenvalues are positive numbers $\mu_j$ and we call the eigenvectors $\tilde{\varphi}_j$, i.e.,
\be
p_1g(x_1)p_1 = \sum_{j=1}^N \mu_j \, p_1^{\tilde{\varphi}_j},
\ee
and
\be
\sum_{j=1}^N \mu_j = \sum_{j=1}^N \scp{\varphi_j}{g (x_1) \varphi_j}_1 = \int_{\RRR^3} g(x) \, \rho_N(x) \, d^3x.
\ee
Using the diagonalizations and Lemma~\ref{lem:projector_norms} we find
\begin{align}
\bigSCP{\psi_{as}}{p_1p_2 \, h_{12} \, p_1p_2 \psi_{as}} &= \bigSCP{p_1\psi_{as}}{p_2 h_{12} p_2 \, p_1\psi_{as}} \nonumber \\
&= \bigSCP{p_1\psi_{as}}{\sum_{i=1}^N \lambda_i(x_1) p_2^{\chi_i^{x_1}} p_1\psi_{as}} \nonumber \\
&= \int d^3x_1 \sum_{i=1}^N \underbrace{\lambda_i(x_1)}_{\geq 0} \underbrace{\bigSCP{p_1\psi_{as}}{ p_2^{\chi_i^{x_1}} p_1\psi_{as}}_{2,\ldots,N}(x_1)}_{\geq 0 ~ \forall x_1} \nonumber \\
\eqexp{by Lem.~\ref{lem:projector_norms}} &\leq \frac{1}{N-1} \int d^3x_1 \sum_{i=1}^N \lambda_i(x_1) \bigSCP{p_1\psi_{as}}{p_1\psi_{as}}_{2,\ldots,N}(x_1) \nonumber \\
&= \frac{1}{N-1} \bigSCP{\psi_{as}}{p_1 \underbrace{\sum_{i=1}^N \lambda_i(x_1)}_{:= g(x_1)} p_1\psi_{as}} \nonumber \\
&= \frac{1}{N-1} \bigSCP{\psi_{as}}{\sum_{j=1}^N \mu_j \, p_1^{\tilde{\varphi}_j} \psi_{as}} \nonumber \\
&= \frac{1}{N-1} \sum_{j=1}^N \underbrace{\mu_j}_{\geq 0} \underbrace{\bigSCP{\psi_{as}}{ p_1^{\tilde{\varphi}_j} \psi_{as}}}_{\geq 0} \nonumber \\
\eqexp{by Lem.~\ref{lem:projector_norms}} &\leq \frac{1}{N(N-1)} \left( \sum_{j=1}^N \mu_j \right) ~ \bigSCP{\psi_{as}}{\psi_{as}} \nonumber \\
&= \frac{1}{N(N-1)} \left( \int g(x) \rho_N(x) \,d^3x \right) ~ \bigSCP{\psi_{as}}{\psi_{as}} \nonumber \\
&= \frac{1}{N(N-1)} \left( \int (h\star\rho_N)(x)\, \rho_N(x) \,d^3x \right) ~ \bigSCP{\psi_{as}}{\psi_{as}}.
\end{align}
If $\psi_{as}$ is antisymmetric in all variables except $x_3$, then, similarly to \eqref{pvp_proof2}, one can only extract factors $\frac{1}{N-2}$ instead of $\frac{1}{N-1}$, and $\frac{1}{N-1}$ instead of $\frac{1}{N}$ from the antisymmetry of $\psi_{as}$, as can be seen from Lemma~\ref{lem:projector_norms}.
\end{proof}

\begin{proof}[Proof of Lemma \ref{lem:psi_qppvvppq_psi}]
We denote by $\phi_{as}^{i_1,\ldots,i_a}$ a normalized function in $L^2(\RRR^{3N})$ that is antisymmetric in all variables except in $x_{i_1}, \ldots, x_{i_a}$. Recall that $h$ is positive. Then, using Cauchy-Schwarz and Lemma~\ref{lem:psi_pvp_psi} in the end,
\begin{align}
& \SCP{\psi_{as}}{q_3p_1p_2 \, h_{12} h_{13} \, p_1p_3q_2 \psi_{as}} \nonumber \\
\quad \quad &\qquad= \SCP{\psi_{as}}{q_3 \, p_1 \sqrt{h_{13}} \, p_2 \sqrt{h_{12}} \, \sqrt{h_{13}} p_3 \, \sqrt{h_{12}} p_1 \, q_2 \psi_{as}} \nonumber \\
&\qquad\leq \norm{q_3\psi_{as}} \left( \sup_{\phi_{as}^3} \norm{\sqrt{h_{13}} p_1 \, \phi_{as}^3} \right) \left( \sup_{\phi_{as}^{1,3}} \norm{\sqrt{h_{12}} p_2 \, \phi_{as}^{1,3}} \right) \times \nonumber \\
&\qquad\quad \times \left( \sup_{\phi_{as}^{1,2}} \norm{\sqrt{h_{13}} p_3 \, \phi_{as}^{1,2}} \right) \left( \sup_{\phi_{as}^2} \norm{\sqrt{h_{12}} p_1 \, \phi_{as}^2} \right) \norm{q_2\psi_{as}} \nonumber \\
&\qquad= \norm{q_1\psi_{as}}^2 \left( \sup_{\phi_{as}^2} \norm{\sqrt{h_{12}} p_1 \, \phi_{as}^2}^2 \right) \left( \sup_{\phi_{as}^{1,3}} \norm{\sqrt{h_{12}} p_2 \, \phi_{as}^{1,3}}^2 \right) \nonumber \\
&\qquad= \norm{q_1\psi_{as}}^2 \left( \sup_{\phi_{as}^2} \, \SCP{\phi_{as}^2}{p_1h_{12} p_1 \phi_{as}^2} \right) \left( \sup_{\phi_{as}^{1,3}} \, \SCP{\phi_{as}^{1,3}}{p_2 h_{12} p_2 \, \phi_{as}^{1,3}} \right) \nonumber \\
\eqexpl{by Lem.~\ref{lem:psi_pvp_psi}} &\qquad\leq \frac{1}{(N-1)(N-2)} \, \left(\sup_{y \in \RRR^3} (h \star \rho_N)(y) \right)^2 \, \SCP{\psi_{as}}{q_1 \, \psi_{as}}.
\end{align}
\end{proof}

\section{Bounds on $\partial_t \alpha_f(t)$}\label{sec:alpha_m_dot_rigorous}
We now give the rigorous bounds for the three terms in the time derivative of $\alpha_f(t)$ given by \eqref{alpha_derivative}. Here, we use the weight function $m^{(\gamma)}(k)$ from \eqref{weight_m_gamma}. This also contains the case where $\gamma=1$, thus the bounds also hold for the weight function $n(k)$. The estimates are collected in the following lemma, which constitutes the heart of the proof of our main results.

We state this lemma only for positive $v^{(N)}$. If $v^{(N)}$ contains both positive and negative parts, we write $v^{(N)} = v^{(N)}_{+} - v^{(N)}_{-}$, with $v^{(N)}_{+},v^{(N)}_{-} \geq 0$, and then estimate the three terms in \eqref{alpha_derivative} separately for $v^{(N)}_{+}$ and $v^{(N)}_{-}$. We denote the direct term by $V^{\dir,(N)}_1 = \big(v^{(N)}\star\rho_N\big)(x_1)$.

\begin{lemma}\label{lem:estimates_terms_alpha_dot_beta}
Let $\varphi_1,\ldots,\varphi_N \in L^2(\RRR^3)$ be orthonormal and $\psi \in L^2(\RRR^{3N})$ be antisymmetric. Let $v^{(N)}$ be positive and set $\rho_N(x) = \sum_{i=1}^N |\varphi_i(x)|^2$. Then,
\begin{enumerate}[(a)]
\item \noindent for the $\boldsymbol{qp}$-$\boldsymbol{pp}$ \textbf{term},
\begin{flalign}\label{term_1_dir}
&\left\lvert 2N \, \Im\, \bigSCP{\psi}{\left(\widehat{m^{(\gamma)}}-\widehat{m^{(\gamma)}}_{-1}\right)q_1 \Big( (N-1)p_2v^{(N)}_{12}p_2 - V_1^{\dir,(N)} \Big) p_1 \psi} \right\rvert & \nonumber \\
&\qquad \leq 4\left( \sup_{x_1\in\RRR^3} \int_{\Omega_N(x_1)}v^{(N)}(x_1-y)^2\rho_N(y)\,d^3y \right)^{\frac{1}{2}} N^{\frac{1}{2}} \, \Big( \alpha_{m^{(\gamma)}}(t) + N^{-\gamma} \Big) & \nonumber \\
&\qquad\quad + 4\sqrt{2} \left( \sup_{y\in\overline{\Omega_N}(0)} v^{(N)}(y) \right) N^{\frac{1}{2}+\frac{\gamma}{2}} \, \Big( \alpha_{m^{(\gamma)}}(t) + N^{-\gamma} \Big), &
\end{flalign}
for any (possibly $N$ dependent) volume $\Omega_N$, with $\Omega_N(x_1)=\Omega_N+x_1$ and $\overline{\Omega_N}:=\RRR^3 \setminus \Omega_N$ (possibly $\Omega_N=\RRR^3$ or $\Omega_N=\emptyset$); also, with $V_1^{(N)}=V_1^{\dir,(N)}$ or $V_1^{(N)}=0$,
\begin{flalign}\label{term_1_dir_2}
&\left\lvert 2N \, \Im\, \bigSCP{\psi}{\left(\widehat{m^{(\gamma)}}-\widehat{m^{(\gamma)}}_{-1}\right)q_1 \Big( (N-1)p_2v^{(N)}_{12}p_2 - V_1^{(N)} \Big) p_1 \psi} \right\rvert & \nonumber \\
&\qquad \leq 4\sqrt{2} \left( \sup_{y\in\RRR^3} (v^{(N)}\star\rho_N)(y) \right) N^{\frac{1}{2}-\frac{\gamma}{2}}, &
\end{flalign}
\item \noindent for the $\boldsymbol{qq}$-$\boldsymbol{pp}$ \textbf{term},
\begin{flalign}\label{term_2}
&\left\lvert N \, \Im\, \bigSCP{\psi}{\left(\widehat{m^{(\gamma)}}-\widehat{m^{(\gamma)}}_{-2}\right)q_1q_2 \Big( (N-1)v^{(N)}_{12} \Big) p_1p_2 \psi} \right\rvert & \nonumber \\
&\qquad \leq \sqrt{12} \Bigg( \left( \sup_{y\in\RRR^3} (v^{(N)}\star\rho_N)(y) \right)^2 \alpha_{m^{(\gamma)}}^2 \nonumber \\
&\qquad\qquad\qquad + \bigg( \int \Big(\big(v^{(N)}\big)^2\star\rho_N\Big)(y)\,\rho_N(y)\,d^3y \bigg) \alpha_{m^{(\gamma)}}(t) N^{-\gamma} \Bigg)^{\frac{1}{2}}, &
\end{flalign}
and also,
\begin{flalign}\label{term_2_alt}
&\left\lvert N \, \Im\, \bigSCP{\psi}{\left(\widehat{m^{(\gamma)}}-\widehat{m^{(\gamma)}}_{-2}\right)q_1q_2 \Big( (N-1)v^{(N)}_{12} \Big) p_1p_2 \psi} \right\rvert & \nonumber \\
&\qquad \leq \sqrt{12} \left( \int \Big(\big(v^{(N)}\big)^2\star\rho_N\Big)(y)\,\rho_N(y)\,d^3y \right)^{\frac{1}{2}} \Big( \alpha_{m^{(\gamma)}} + N^{-\gamma} \Big), &
\end{flalign}
\item \noindent for the $\boldsymbol{qq}$-$\boldsymbol{pq}$ \textbf{term},
\begin{flalign}\label{term_3}
&\left\lvert 2N \, \Im\, \bigSCP{\psi}{\left(\widehat{m^{(\gamma)}}-\widehat{m^{(\gamma)}}_{-1}\right)q_1q_2 \Big( (N-1)v^{(N)}_{12} \Big) p_1q_2 \psi} \right\rvert & \nonumber \\
&\qquad \leq 4 \left( \sup_{y\in\RRR^3} \Big(\big(v^{(N)}\big)^2\star\rho_N\Big)(y) \right)^{\frac{1}{2}} \, N^{\frac{\gamma}{2}} \, \alpha_{m^{(\gamma)}}(t). &
\end{flalign}
\end{enumerate}
\end{lemma}

\begin{proof}
For ease of notation, we often omit subscripts and superscripts $N$ or $(N)$ in this proof; in particular, we abbreviate $v=v^{(N)}$, $V_1^{\dir}=V_1^{\dir,(N)}$ and $\rho=\rho_N$.

\absatz

\textbf{The} $\boldsymbol{qp}$-$\boldsymbol{pp}$ \textbf{term.} Using Lemma~\ref{lem:root_f_hat}, we find that
\begin{align}\label{term_a_estimate}
& \left\lvert 2N \, \Im\, \bigSCP{\psi}{\left(\widehat{m^{(\gamma)}}-\widehat{m^{(\gamma)}}_{-1}\right)q_1 \Big( (N-1)p_2v_{12}p_2 - V_1^{\dir} \Big) p_1 \psi} \right\rvert \nonumber \\
&\qquad= \bigg\lvert 2N \, \Im\, \Big\langle\!\!\Big\langle \underbrace{\psi, \left(\widehat{m^{(\gamma)}}-\widehat{m^{(\gamma)}}_{-1}\right)^{\frac{1}{2}}}_{:= \widetilde{\psi}} q_1 \Big( (N-1)p_2v_{12}p_2 - V_1^{\dir} \Big) \times \nonumber \\
&\qquad \quad \quad \times p_1 \underbrace{\left(\widehat{m^{(\gamma)}}_{1}-\widehat{m^{(\gamma)}} + m^{(\gamma)}(N)P_{N,N} \right)^{\frac{1}{2}} \psi}_{:=\widetilde{\psi}'} \Big\rangle\!\!\Big\rangle \bigg\rvert \nonumber \\
&\qquad= 2N \left\lvert \Im\, \bigSCP{\widetilde{\psi}}{q_1 \Big( (N-1)p_2v_{12}p_2 - V_1^{\dir} \Big) p_1 \widetilde{\psi}'} \right\rvert.
\end{align}
In order to estimate \eqref{term_a_estimate}, we diagonalize $p_2v_{12}p_2$ according to Lemma~\ref{lem:diag_pvp}. We call the eigenvectors $\chi_i^{x_1}$ and the eigenvalues $\lambda_i(x_1)$. Note that $V_1^{\dir}(x_1)=\sum_{i=1}^N \lambda_i(x_1)$. We denote by $\Omega_N \subset \RRR^3$ a possibly $N$ dependent volume. For ease of notation we omit the subscript $N$, i.e., we write $\Omega:=\Omega_N$ for the proof. We split the eigenvalues into two parts:
\begin{align}
\lambda_i(x_1) &= \int_{\RRR^3} v(x_1-y) \left\lvert \chi_i^{x_1}(y) \right\rvert^2 \, d^3y \nonumber \\
&= \underbrace{\int_{\Omega(x_1)} v(x_1-y) \left\lvert \chi_i^{x_1}(y) \right\rvert^2 \, d^3y}_{=:\lambda_{i}^{\Omega}(x_1)} + \underbrace{\int_{\overline{\Omega(x_1)}} v(x_1-y) \left\lvert \chi_i^{x_1}(y) \right\rvert^2 \, d^3y}_{=:\lambda_{i}^{\overline{\Omega}}(x_1)}.
\end{align}
Then, continuing from \eqref{term_a_estimate}, we find 
\begin{align}\label{term_a_estimate2}
& \bigg\lvert \bigSCP{\widetilde{\psi}}{q_1 \Big( (N-1)p_2v_{12}p_2 - V_1^{\dir} \Big) p_1 \widetilde{\psi}'} \bigg\rvert \nonumber \\
&\quad= \bigg\lvert \bigSCP{\widetilde{\psi}}{q_1 \left( \sum_{m=2}^N p_mv_{1m}p_m - V_1^{\dir} \right) p_1 \widetilde{\psi}'} \bigg\rvert \nonumber \\
&\quad= \bigg\lvert \bigSCP{\widetilde{\psi}}{q_1 \Bigg( \sum_{i=1}^N \lambda_i(x_1) \underbrace{\sum_{m=2}^N p_m^{\chi_i^{x_1}}}_{=:p_{\neq 1}^{\chi_i^{x_1}}} - \sum_{i=1}^N \lambda_i(x_1) \Bigg) p_1 \widetilde{\psi}'} \bigg\rvert \nonumber \\
&\quad= \bigg\lvert \bigSCP{\widetilde{\psi}}{q_1 \Bigg( \sum_{i=1}^N \lambda_i(x_1) \underbrace{\left( p_{\neq 1}^{\chi_i^{x_1}} - \id \right)}_{=:-q_{\neq 1}^{\chi_i^{x_1}}} \Bigg) p_1 \widetilde{\psi}'} \bigg\rvert \nonumber \\
&\quad\leq \bigg\lvert \sum_{i=1}^N \bigSCP{\widetilde{\psi}}{q_1 \lambda_{i}^{\Omega}(x_1) q_{\neq 1}^{\chi_i^{x_1}} p_1 \widetilde{\psi}'} \bigg\rvert + \bigg\lvert \sum_{i=1}^N \bigSCP{\widetilde{\psi}}{q_1 \lambda_{i}^{\overline{\Omega}}(x_1) q_{\neq 1}^{\chi_i^{x_1}} p_1 \widetilde{\psi}'} \bigg\rvert.
\end{align}
Here we introduced the projectors
\be
p_{\neq 1}^{\varphi_i} := \sum_{m=2}^N p_m^{\varphi_i}, \quad q_{\neq 1}^{\varphi_i} := \id - p_{\neq 1}^{\varphi_i}
\ee
that act on all but the first variable (see also Chapter~\ref{sec:properties_projectors}). Note that, for all $\psi_{as}^1$ that are antisymmetric in all variables except $x_1$,
\begin{align}\label{p_neq_1}
\bigSCP{\psi_{as}^1}{\sum_{i=1}^N q_{\neq 1}^{\varphi_i} \psi_{as}^1} &= \bigSCP{\psi_{as}^1}{\left( N - \sum_{i=1}^N \sum_{m=2}^N p_m^{\varphi_i} \right) \psi_{as}^1} \nonumber \\
&= \bigSCP{\psi_{as}^1}{\left( N - \sum_{m=2}^N p_m \right) \psi_{as}^1} \nonumber \\
&= \bigSCP{\psi_{as}^1}{\Big( N - (N-1) p_2 \Big) \psi_{as}^1} \nonumber \\
&= (N-1) \bigSCP{\psi_{as}^1}{q_2 \psi_{as}^1} + \bigSCP{\psi_{as}^1}{\psi_{as}^1}.
\end{align}
This remains true if $\ket{\varphi_i}_m=\ket{\chi_i^{x_1}}_m$ ($m\geq 2$). We also have that $\left(q_{\neq 1}^{\chi_i^{x_1}}\right)^2 \psi_{as}^1=q_{\neq 1}^{\chi_i^{x_1}} \psi_{as}^1$. Using this, Cauchy-Schwarz (C.-S.), \eqref{p_neq_1} and Lemma~\ref{lem:q_root_f}, we find for the first summand in \eqref{term_a_estimate2},
\begin{align}\label{term_a_estimate_epsilon}
&\bigg\lvert \sum_{i=1}^N \bigSCP{\widetilde{\psi}}{q_1 \lambda_{i}^{\Omega}(x_1) q_{\neq 1}^{\chi_i^{x_1}} p_1 \widetilde{\psi}'} \bigg\rvert \nonumber \\
\eqexpl{by C.-S.} &\qquad\leq \sum_{i=1}^N \norm{\lambda_{i}^{\Omega}(x_1)q_1\widetilde{\psi}} \norm{q_{\neq 1}^{\chi_i^{x_1}} p_1 \widetilde{\psi}'} \nonumber \\
&\qquad\leq \left( \sum_{i=1}^N \bigSCP{\widetilde{\psi}}{q_1 \left(\lambda_{i}^{\Omega}(x_1)\right)^2 q_1 \widetilde{\psi}} \right)^{\frac{1}{2}} \, \left( \sum_{i=1}^N \bigSCP{\widetilde{\psi}'}{p_1 q_{\neq 1}^{\chi_i^{x_1}} p_1 \widetilde{\psi}'} \right)^{\frac{1}{2}} \nonumber \\
\eqexpl{by \eqref{p_neq_1}} &\qquad\leq \left( \sup_{x_1} \sum_{i=1}^N \left(\lambda_{i}^{\Omega}(x_1)\right)^2 \right)^{\frac{1}{2}} \bigg( \bigSCP{\widetilde{\psi}}{q_1 \widetilde{\psi}} \bigg)^{\frac{1}{2}} \times \nonumber \\
&\qquad\quad \times \, \bigg( (N-1) \bigSCP{\widetilde{\psi}'}{p_1 q_2 \widetilde{\psi}'} + \bigSCP{\widetilde{\psi}'}{p_1 \widetilde{\psi}'} \bigg)^{\frac{1}{2}} \nonumber \\
\eqexpl{by Lem.~\ref{lem:q_root_f}} &\qquad\leq \left( \sup_{x_1} \sum_{i=1}^N \left(\lambda_{i}^{\Omega}(x_1)\right)^2 \right)^{\frac{1}{2}} \bigg( 2 N^{-1} \alpha_{m^{(\gamma)}} \bigg)^{\frac{1}{2}} \, \bigg( 2\alpha_{m^{(\gamma)}} + N^{-\gamma} \bigg)^{\frac{1}{2}} \nonumber \\
&\qquad\leq 2 \left( \sup_{x_1} \sum_{i=1}^N \left(\lambda_{i}^{\Omega}(x_1)\right)^2 \right)^{\frac{1}{2}} N^{-\frac{1}{2}} \bigg( \alpha_{m^{(\gamma)}}^2 + 2N^{-\gamma}\alpha_{m^{(\gamma)}} + N^{-2\gamma} \bigg)^{\frac{1}{2}} \nonumber \\
&\qquad= 2 \left( \sup_{x_1} \sum_{i=1}^N \left(\lambda_{i}^{\Omega}(x_1)\right)^2 \right)^{\frac{1}{2}} N^{-\frac{1}{2}} \bigg( \alpha_{m^{(\gamma)}} + N^{-\gamma} \bigg).
\end{align}
Furthermore, using Cauchy-Schwarz and the fact that $\chi_1^{x_1},\ldots,\chi_N^{x_1}$ are normalized, we find
\begin{align}
\lambda_{i}^{\Omega}(x_1) &= \int_{\Omega(x_1)} v(x_1-y) \left\lvert \chi_i^{x_1}(y) \right\rvert^2 \, d^3y \nonumber \\
&\leq \left( \int_{\Omega(x_1)} v^2(x_1-y) \left\lvert \chi_i^{x_1}(y) \right\rvert^2 \, d^3y \right)^{\frac{1}{2}} \left( \int_{\Omega(x_1)} \left\lvert \chi_i^{x_1}(y) \right\rvert^2 \, d^3y \right)^{\frac{1}{2}} \nonumber \\
&\leq \left( \int_{\Omega(x_1)} v^2(x_1-y) \left\lvert \chi_i^{x_1}(y) \right\rvert^2 \, d^3y \right)^{\frac{1}{2}}.
\end{align}
Now observe that $\sum_{i=1}^N \left\lvert \chi_i^{x_1}(y) \right\rvert^2 = \sum_{i=1}^N \left\lvert \varphi_i(y) \right\rvert^2$, which follows directly from $\chi_i^{x_1}(y) = \sum_{k=1}^N U_{ik}(x_1) \, \varphi_k(y)$ with $U(x_1)$ unitary. Thus,
\begin{align}\label{sup_v2}
\left( \sup_{x_1} \sum_{i=1}^N \left(\lambda_{i}^{\Omega}(x_1)\right)^2 \right)^{\frac{1}{2}} &\leq \left( \sup_{x_1} \sum_{i=1}^N \int_{\Omega(x_1)} v^2(x_1-y) \left\lvert \chi_i^{x_1}(y) \right\rvert^2 \, d^3y \right)^{\frac{1}{2}} \nonumber \\
&= \left( \sup_{x_1} \int_{\Omega(x_1)} v^2(x_1-y) \rho(y) \, d^3y \right)^{\frac{1}{2}}.
\end{align}
Let us turn to the second summand in \eqref{term_a_estimate2}. First, note that due to the normalization of $\chi_1^{x_1},\ldots,\chi_N^{x_1}$,
\begin{align}\label{lambda_omega_bar}
\lambda_{i}^{\overline{\Omega}}(x_1) &= \int_{\overline{\Omega(x_1)}} v(x_1-y) \left\lvert \chi_i^{x_1}(y) \right\rvert^2 \, d^3y \nonumber \\
&\leq \left( \sup_{y \in \overline{\Omega(x_1)}} v(x_1-y) \right) \int_{\RRR^3} \left\lvert \chi_i^{x_1}(y) \right\rvert^2 \, d^3y  \nonumber \\
&= \sup_{y \in \overline{\Omega(0)}} v(y).
\end{align}
Then we find, using Cauchy-Schwarz (C.-S.), $\left(q_{\neq 1}^{\chi_i^{x_1}}\right)^2 \psi_{as}=q_{\neq 1}^{\chi_i^{x_1}} \psi_{as}$ for all antisymmetric $\psi_{as}$, \eqref{p_neq_1}, \eqref{lambda_omega_bar} and Lemma~\ref{lem:q_root_f},
\begin{align}\label{term_a_estimate_epsilon_bar}
&\Big\lvert \sum_{i=1}^N \bigSCP{\widetilde{\psi}}{q_1 \lambda_{i}^{\overline{\Omega}}(x_1) q_{\neq 1}^{\chi_i^{x_1}} p_1 \widetilde{\psi}'} \Big\rvert \nonumber \\
&\qquad= \Big\lvert \sum_{i=1}^N \bigSCP{\widetilde{\psi}}{q_1 q_{\neq 1}^{\chi_i^{x_1}} \lambda_{i}^{\overline{\Omega}}(x_1) q_{\neq 1}^{\chi_i^{x_1}} p_1 \widetilde{\psi}'} \Big\rvert \nonumber \\
\eqexpl{by C.-S.} &\qquad\leq \sum_{i=1}^N \norm{\sqrt{\lambda_{i}^{\overline{\Omega}}(x_1)} q_{\neq 1}^{\chi_i^{x_1}} q_1 \widetilde{\psi}} \norm{\sqrt{\lambda_{i}^{\overline{\Omega}}(x_1)} q_{\neq 1}^{\chi_i^{x_1}} p_1 \widetilde{\psi}'} \nonumber \\
\eqexpl{by C.-S.} &\qquad\leq \left( \sum_{i=1}^N \bigSCP{\widetilde{\psi}}{q_1 q_{\neq 1}^{\chi_i^{x_1}} \lambda_{i}^{\overline{\Omega}}(x_1) q_{\neq 1}^{\chi_i^{x_1}} q_1 \widetilde{\psi}} \right)^{\frac{1}{2}} \left( \sum_{i=1}^N \bigSCP{\widetilde{\psi}'}{p_1 q_{\neq 1}^{\chi_i^{x_1}} \lambda_{i}^{\overline{\Omega}}(x_1) q_{\neq 1}^{\chi_i^{x_1}} p_1 \widetilde{\psi}'} \right)^{\frac{1}{2}} \nonumber \\
&\qquad\leq \left( \left(\sup_{i,x_1} \lambda_{i}^{\overline{\Omega}}(x_1)\right) \sum_{i=1}^N \bigSCP{\widetilde{\psi}}{q_1 q_{\neq 1}^{\chi_i^{x_1}} q_1 \widetilde{\psi}} \right)^{\frac{1}{2}} \times \nonumber \\
&\qquad \quad \times \left( \left(\sup_{i,x_1} \lambda_{i}^{\overline{\Omega}}(x_1)\right) \sum_{i=1}^N \bigSCP{\widetilde{\psi}'}{p_1 q_{\neq 1}^{\chi_i^{x_1}} p_1 \widetilde{\psi}'} \right)^{\frac{1}{2}} \nonumber \\
\eqexpl{by \eqref{p_neq_1}} &\qquad\leq \left(\sup_{i,x_1} \lambda_{i}^{\overline{\Omega}}(x_1)\right) \bigg( (N-1) \bigSCP{\widetilde{\psi}}{q_1 q_2 \widetilde{\psi}} + \bigSCP{\widetilde{\psi}}{q_1 \widetilde{\psi}} \bigg)^{\frac{1}{2}} \times \nonumber \\
&\qquad \quad \times \bigg( (N-1) \bigSCP{\widetilde{\psi}'}{p_1 q_2 \widetilde{\psi}'} + \bigSCP{\widetilde{\psi}'}{p_1 \widetilde{\psi}'} \bigg)^{\frac{1}{2}} \nonumber \\
\eqexpl{by Lem.~\ref{lem:q_root_f}} &\qquad\leq \left(\sup_{i,x_1} \lambda_{i}^{\overline{\Omega}}(x_1)\right) \bigg( 4 N^{\gamma-1} \alpha_{m^{(\gamma)}} + 2 N^{-1} \alpha_{m^{(\gamma)}} \bigg)^{\frac{1}{2}} \bigg( 2 \alpha_{m^{(\gamma)}} + N^{-\gamma} \bigg)^{\frac{1}{2}} \nonumber \\
\eqexpl{by \eqref{lambda_omega_bar}} &\qquad\leq \sqrt{8} \left( \sup_{y \in \overline{\Omega(0)}} v(y) \right) N^{\frac{\gamma}{2}-\frac{1}{2}} \, \times \nonumber \\
&\qquad \quad \times \bigg( \alpha_{m^{(\gamma)}}^2 + \frac{1}{2} N^{-\gamma} \alpha_{m^{(\gamma)}} + \frac{1}{2} N^{-\gamma} \alpha_{m^{(\gamma)}}^2 + \frac{1}{4} N^{-2\gamma} \alpha_{m^{(\gamma)}} \bigg)^{\frac{1}{2}} \nonumber \\
&\qquad\leq \sqrt{8} \left( \sup_{y \in \overline{\Omega(0)}} v(y) \right) N^{\frac{\gamma}{2}-\frac{1}{2}} \bigg( \alpha_{m^{(\gamma)}}^2 + 2N^{-\gamma} \alpha_{m^{(\gamma)}} + N^{-2\gamma} \bigg)^{\frac{1}{2}} \nonumber \\
&\qquad= \sqrt{8} \left( \sup_{y \in \overline{\Omega(0)}} v(y) \right) N^{\frac{\gamma}{2}-\frac{1}{2}} \bigg( \alpha_{m^{(\gamma)}} + N^{-\gamma} \bigg).
\end{align}

The bounds \eqref{term_a_estimate_epsilon} with \eqref{sup_v2}, and \eqref{term_a_estimate_epsilon_bar} give the bound \eqref{term_1_dir} on \eqref{term_a_estimate}.

The alternative estimate \eqref{term_1_dir_2} can be obtained by starting from the second last line of \eqref{term_a_estimate}. Using Cauchy-Schwarz (C.-S.) and Lemmas~\ref{lem:psi_pvp_psi} and \ref{lem:q_root_f} we find
\begin{align}
& 2N \left\lvert \Im\, \bigSCP{\widetilde{\psi}}{q_1 \Big( (N-1)p_2v_{12}p_2 - V_1^{\dir} \Big) p_1 \widetilde{\psi}'} \right\rvert \nonumber \\
&\qquad \leq 2 N(N-1) \left\lvert \bigSCP{\widetilde{\psi}}{q_1p_2 v_{12} p_1p_2 \widetilde{\psi}'} \right\rvert + 2 N \left\lvert \bigSCP{\widetilde{\psi}}{q_1 V_1^{\dir} p_1 \widetilde{\psi}'} \right\rvert \nonumber \\
\eqexpl{by C.-S.}\ &\qquad\leq 2 N(N-1) \norm{\sqrt{v_{12}}q_1p_2\widetilde{\psi}} \norm{\sqrt{v_{12}}p_1p_2\widetilde{\psi}'} + 2 N \norm{q_1\widetilde{\psi}} \norm{V_1^{\dir} p_1\widetilde{\psi}'} \nonumber \\
\eqexpl{by Lem.~\ref{lem:psi_pvp_psi}} &\qquad\leq 2 N \left( \sup_{y} (v\star\rho)(y) \right)^{\frac{1}{2}} \norm{q_1\widetilde{\psi}} \left( \sup_{y} (v\star\rho)(y) \right)^{\frac{1}{2}} \norm{\widetilde{\psi}'} \nonumber \\
&\qquad\quad + 2 N \norm{q_1\widetilde{\psi}} \left(\sup_{y} (v\star\rho)^2(y) \right)^{\frac{1}{2}}  \norm{\widetilde{\psi}'} \nonumber \\
\eqexpl{by Lem.~\ref{lem:q_root_f}} &\qquad \leq 4 N \left( \sup_{y} (v\star\rho)(y) \right) \Big( 2 N^{-1} \alpha_{m^{(\gamma)}} \Big)^{\frac{1}{2}} \Big( N^{-\gamma} \Big)^{\frac{1}{2}} \nonumber \\
&\qquad = 4\sqrt{2} \left( \sup_{y} (v\star\rho)(y) \right) N^{\frac{1}{2}-\frac{\gamma}{2}},
\end{align}
which of course also holds when $V_1^\dir=0$.

\absatz

\textbf{The} $\boldsymbol{qq}$-$\boldsymbol{pp}$ \textbf{term.}  Using Lemma~\ref{lem:root_f_hat}, the antisymmetry of $\psi$ and Cauchy-Schwarz (C.-S.), we find that
\begin{align}\label{term_b_estimate}
&\left\lvert N \, \Im\, \bigSCP{\psi}{\left(\widehat{m^{(\gamma)}}-\widehat{m^{(\gamma)}}_{-2}\right)q_1q_2 \Big( (N-1)v_{12} \Big) p_1p_2 \psi} \right\rvert \nonumber \\
&\qquad = \bigg\lvert N \, \Im\, \Big\langle\!\!\Big\langle \underbrace{\psi, \left(\widehat{m^{(\gamma)}}-\widehat{m^{(\gamma)}}_{-2}\right)^{\frac{1}{2}}}_{:= \widetilde{\psi}} q_1q_2 \Big( (N-1)v_{12} \Big) \times \nonumber \\
&\qquad \quad \quad \times p_1p_2 \underbrace{\left(\widehat{m^{(\gamma)}}_{2}-\widehat{m^{(\gamma)}} + m^{(\gamma)}(N-1)P_{N,N-1} + m^{(\gamma)}(N)P_{N,N} \right)^{\frac{1}{2}} \psi}_{:=\widetilde{\psi}'} \Big\rangle\!\!\Big\rangle \bigg\rvert \nonumber \\
&\qquad = \left\lvert (N-1) N \, \Im\, \bigSCP{\widetilde{\psi}}{q_1q_2 \, v_{12} \, p_1p_2 \widetilde{\psi}'} \right\rvert \nonumber \\
\eqexpl{by antisym.} &\qquad = \bigg\lvert N \, \Im\, \bigSCP{\widetilde{\psi}}{q_1 \sum_{m=2}^N q_m \, v_{1m} \, p_1p_m \widetilde{\psi}'} \bigg\rvert \nonumber \\
\eqexpl{by C.-S.} &\qquad \leq  N \norm{q_1\widetilde{\psi}} \norm{\sum_{m=2}^N q_m \, v_{1m} \, p_1p_m \widetilde{\psi}'}.
\end{align}
From Lemma~\ref{lem:q_root_f} we have
\be\label{term_b_part1}
\norm{q_1\widetilde{\psi}}^2 \leq 6 \, N^{-1} \, \alpha_{m^{(\gamma)}}.
\ee
The trick of shifting one $q$ to the right-hand side of the scalar product is now done in the following calculation. Using Lemmas~\ref{lem:psi_qppvvppq_psi}, \ref{lem:psi_ppv2pp_psi} and \ref{lem:q_root_f} we find
\begin{align}\label{term_b_part2}
 \norm{\sum_{m=2}^N q_m \, v_{1m} \, p_1p_m \widetilde{\psi}'}^2 &= \sum_{\substack{m,n=2 \\ m \neq n}}^N \bigSCP{\widetilde{\psi}'}{p_1 p_m v_{1m} q_m q_n v_{1n} p_n p_1 \widetilde{\psi}'} \nonumber \\
&\quad + \sum_{m=2}^N \bigSCP{\widetilde{\psi}'}{p_1 p_m v_{1m} q_m v_{1m} p_m p_1 \widetilde{\psi}'} \nonumber \\
&= (N-1)(N-2) \bigSCP{\widetilde{\psi}'}{q_3 p_1 p_2 v_{12} v_{13} p_1 p_3 q_2 \widetilde{\psi}'} \nonumber \\
&\quad + (N-1) \bigSCP{\widetilde{\psi}'}{p_1 p_2 v_{12}q_2v_{12} p_1 p_2 \widetilde{\psi}'} \nonumber \\
\eqexp{by Lem.~\ref{lem:psi_qppvvppq_psi}} & \leq \left( \sup_{y} (v\star\rho)(y) \right)^2 \bigSCP{\widetilde{\psi}'}{q_1 \widetilde{\psi}'} \nonumber \\
\eqexp{by Lem.~\ref{lem:psi_ppv2pp_psi}} & \quad + N^{-1} \int (v^2\star\rho)(y)\,\rho(y)\,d^3y ~ \bigSCP{\widetilde{\psi}'}{\widetilde{\psi}'} \nonumber \\
\eqexp{by Lem.~\ref{lem:q_root_f}} & \leq \left( \sup_{y} (v\star\rho)(y) \right)^2 2N^{-1}\alpha_{m^{(\gamma)}} + \int (v^2\star\rho)(y)\,\rho(y)\,d^3y ~ 2N^{-1-\gamma}
\end{align}
With \eqref{term_b_part1} and \eqref{term_b_part2} we can continue the estimate from \eqref{term_b_estimate}:
\begin{align}\label{term_b_part3}
& N \norm{q_1\widetilde{\psi}} \norm{\sum_{m=2}^N q_m \, v_{1m} \, p_1p_m \widetilde{\psi}'} \nonumber \\
&\quad \leq N \left( 6 \, N^{-1} \, \alpha_{m^{(\gamma)}} \right)^{\frac{1}{2}} \left( \left( \sup_{y} (v\star\rho)(y) \right)^2 2N^{-1}\alpha_{m^{(\gamma)}} + \int (v^2\star\rho)(y)\,\rho(y)\,d^3y ~ 2N^{-1-\gamma} \right)^{\frac{1}{2}} \nonumber \\
&\quad \leq \sqrt{12} \left( \left( \sup_{y} (v\star\rho)(y) \right)^2 \alpha_{m^{(\gamma)}}^2 + \int (v^2\star\rho)(y)\,\rho(y)\,d^3y~ \alpha_{m^{(\gamma)}} N^{-\gamma} \right)^{\frac{1}{2}}.
\end{align}
This proves \eqref{term_2}. For the alternative estimate \eqref{term_2_alt}, we use Lemma~\ref{lem:psi_ppv2pp_psi} instead of Lemma~\ref{lem:psi_qppvvppq_psi} in \eqref{term_b_part2}; that is, we bound
\begin{align}\label{term_b_part4}
& (N-1)(N-2) \bigSCP{\widetilde{\psi}'}{q_3 p_1 p_2 v_{12} v_{13} p_1 p_3 q_2 \widetilde{\psi}'} \nonumber \\
\eqexpl{by C.-S.} &\qquad \leq (N-1)(N-2) \norm{v_{12} p_1 p_2 q_3\widetilde{\psi}'} \norm{v_{13} p_1 p_3 q_2 \widetilde{\psi}'} \nonumber \\
\eqexpl{by Lem.~\ref{lem:psi_ppv2pp_psi}} &\qquad \leq \int (v^2\star\rho)(y)\,\rho(y)\,d^3y ~ \bigSCP{\widetilde{\psi}'}{q_1 \widetilde{\psi}'}.
\end{align}
Using that, we derive the bound \eqref{term_b_part3} with $\sup_{y} (v\star\rho)^2(y)$ replaced by $\int (v^2\star\rho)(y)\,\rho(y)\,d^3y$, i.e.,
\begin{align}\label{term_b_part5}
& N \norm{q_1\widetilde{\psi}} \norm{\sum_{m=2}^N q_m \, v_{1m} \, p_1p_m \widetilde{\psi}'} \nonumber \\
&\qquad \leq \sqrt{12} \left( \int (v^2\star\rho)(y)\,\rho(y)\,d^3y \right)^{\frac{1}{2}} \Big( \alpha_{m^{(\gamma)}}^2 + \alpha_{m^{(\gamma)}} N^{-\gamma} \Big)^{\frac{1}{2}} \nonumber \\
&\qquad \leq \sqrt{12} \left( \int (v^2\star\rho)(y)\,\rho(y)\,d^3y \right)^{\frac{1}{2}} \Big( \alpha_{m^{(\gamma)}}^2 + 2\alpha_{m^{(\gamma)}} N^{-\gamma} + N^{-2\gamma} \Big)^{\frac{1}{2}} \nonumber \\
&\qquad = \sqrt{12} \left( \int (v^2\star\rho)(y)\,\rho(y)\,d^3y \right)^{\frac{1}{2}} \Big( \alpha_{m^{(\gamma)}} + N^{-\gamma} \Big),
\end{align}
which is the desired bound \eqref{term_2_alt}.

\absatz

\textbf{The} $\boldsymbol{qq}$-$\boldsymbol{pq}$ \textbf{term.}  Using Lemma~\ref{lem:root_f_hat}, we find that
\begin{align}
& \left\lvert 2N \, \Im\, \bigSCP{\psi}{\left(\widehat{m^{(\gamma)}}-\widehat{m^{(\gamma)}}_{-1}\right)q_1q_2 \Big( (N-1)v_{12} \Big) p_1q_2 \psi} \right\rvert \nonumber \\
&\qquad = \bigg\lvert 2N \, \Im\, \Big\langle\!\!\Big\langle \underbrace{\psi, \left(\widehat{m^{(\gamma)}}-\widehat{m^{(\gamma)}}_{-1}\right)^{\frac{1}{2}}}_{:= \widetilde{\psi}} q_1q_2 \Big( (N-1)v_{12} \Big) \times \nonumber \\
&\qquad \quad \quad \times p_1q_2 \underbrace{\left(\widehat{m^{(\gamma)}}_1-\widehat{m^{(\gamma)}} +m(N)P_{N,N} \right)^{\frac{1}{2}} \psi}_{:=\widetilde{\psi}'} \Big\rangle\!\!\Big\rangle \bigg\rvert \nonumber \\
&\qquad = \left\lvert 2N \, \Im\, \bigSCP{\widetilde{\psi}}{q_1q_2 \Big( (N-1)v_{12} \Big) p_1q_2 \widetilde{\psi}'} \right\rvert.
\end{align}
Using Cauchy-Schwarz (C.-S.) and Lemmas \ref{lem:q_root_f} and \ref{lem:psi_pvp_psi} we find
\begin{align}
& \left\lvert 2N \, \Im\, \bigSCP{\widetilde{\psi}}{q_1q_2 \Big( (N-1)v_{12} \Big) p_1q_2 \widetilde{\psi}'} \right\rvert \nonumber \\
\eqexpl{by C.-S.} &\qquad \leq 2N (N-1) \norm{q_1q_2\widetilde{\psi}} \norm{v_{12}p_1q_2\widetilde{\psi}'} \nonumber \\
\eqexpl{by Lem.~\ref{lem:q_root_f}} &\qquad \leq 2N (N-1) \left( 4N^{\gamma-2} \alpha_{m^{(\gamma)}} \right)^{\frac{1}{2}} \left( \bigSCP{\widetilde{\psi}'}{q_2 \, p_1v_{12}^2p_1 \, q_2 \widetilde{\psi}'}\right)^{\frac{1}{2}} \nonumber \\
\eqexpl{by Lem.~\ref{lem:psi_pvp_psi}} &\qquad \leq 2N (N-1) \left( 4N^{\gamma-2} \alpha_{m^{(\gamma)}} \right)^{\frac{1}{2}} \left( (N-1)^{-1} \left( \sup_{y} (v^2\star\rho)(y) \right) \bigSCP{\widetilde{\psi}'}{q_2 \widetilde{\psi}'}\right)^{\frac{1}{2}} \nonumber \\
\eqexpl{by Lem.~\ref{lem:q_root_f}} &\qquad \leq 2N (N-1) \left( 4N^{\gamma-2} \alpha_{m^{(\gamma)}} \right)^{\frac{1}{2}} \left( (N-1)^{-1} \left( \sup_{y} (v^2\star\rho)(y) \right) N^{-1} \, \alpha_{m^{(\gamma)}} \right)^{\frac{1}{2}} \nonumber \\
&\qquad \leq 4 \left( \sup_{y} (v^2\star\rho)(y) \right)^{\frac{1}{2}} \, N^{\frac{\gamma}{2}} \, \alpha_{m^{(\gamma)}}.
\end{align}
\end{proof}

\section{Proof of the Theorems}\label{sec:proofs_main_theorems_gen}
Let us first state the well-known Gronwall Lemma which we need in the proofs of the main theorems.

\begin{lemma}[Gronwall]\label{lem:gronwall}
Let $t \geq 0$ and let $\eta:\RRR \to \RRR$ be a differentiable function that satisfies the estimate
\be
\partial_t \eta(t) \leq C(t) (\eta(t) + \varepsilon)
\ee
with some real constant $\varepsilon$ and continuous function $C:\RRR \to \RRR$. Then for all $t \geq 0$,
\be
\eta(t) \leq e^{\int_0^t C(s)ds}\eta(0) + \left( e^{\int_0^t C(s)ds} - 1 \right) \varepsilon.
\ee
\end{lemma}

Since different versions of the Gronwall Lemma exist, let us here briefly prove the one we stated above.

\begin{proof}
Let $0 \leq t < \infty$. Let $f:\RRR \to \RRR$ be a differentiable function that satisfies
\be
\partial_t f(t) \leq C(t)f(t)
\ee
and let
\be
g(t) = e^{\int_0^tC(s)ds},
\ee
in particular, $g >0$. Then
\be
\partial_t\left(\frac{f}{g}\right) = \frac{(\partial_t f)g-f(\partial_t g)}{g^2} \leq \frac{Cfg-fCg}{g^2} = 0.
\ee
Since $\frac{f}{g}(0) = f(0)$ it follows that $\frac{f}{g}(t) \leq f(0)$, i.e.,
\be\label{gronwall1}
f(t) \leq g(t)f(0) = f(0)e^{\int_0^tC(s)ds}.
\ee
We now define the function $\xi:\RRR \to \RRR$ as
\be
\xi(t) = e^{\int_0^t C(s)ds}\eta(0) + \left( e^{\int_0^t C(s)ds} - 1 \right) \varepsilon.
\ee
It follows that
\be
\partial_t \xi(t) = C(t) (\xi(t) + \varepsilon)
\ee
and $\xi(0) = \eta(0)$. We then have
\be
\partial_t(\eta-\xi) \leq C (\eta + \varepsilon) - C (\xi + \varepsilon) = C (\eta-\xi)
\ee
and with \eqref{gronwall1}
\be
\eta(t)-\xi(t) \leq (\eta(0)-\xi(0))e^{\int_0^tC(s)ds} = 0,
\ee
i.e., $\eta(t) \leq \xi(t)$.
\end{proof}

Note that the Gronwall Lemma can indeed be applied to $\alpha_f(t)$, since it is differentiable due to the scalar product structure. We now prove the main theorems from Section~\ref{sec:main_theorem_mf_general_v}. 

\begin{proof}[Proof of Theorem~\ref{thm:estimates_terms_alpha_dot_beta_n}]
First, we split $v^{(N)} = v^{(N)}_{+} - v^{(N)}_{-}$, with $v^{(N)}_{+}, v^{(N)}_{-} \geq 0$. Accordingly, we have
\be\label{alpha_plus_minus}
\partial_t \alpha_f(t) = \Term_{+} - \Term_{-},
\ee
with 
\begin{align}
\Term_{\pm} &= 2 \, \Im\, \bigSCP{\psi^t}{N\left(\widehat{f}-\widehat{f}_{-1}\right)q_1 \left( (N-1)p_2v^{(N)}_{\pm,12}\,p_2 - V^{(N)}_{\pm,1} \right) p_1 \psi^t} \nonumber \\
&\quad + \Im\, \bigSCP{\psi^t}{N\left(\widehat{f}-\widehat{f}_{-2}\right)q_1q_2 (N-1)v^{(N)}_{\pm,12}\, p_1p_2 \psi^t} \nonumber \\
&\quad + 2 \, \Im\, \bigSCP{\psi^t}{N\left(\widehat{f}-\widehat{f}_{-1}\right)q_1q_2 (N-1)v^{(N)}_{\pm,12}\, p_1q_2 \psi^t}.
\end{align}
We proceed by estimating $|\Term_{+}|$ and $|\Term_{-}|$ separately. We use the bounds \eqref{term_1_dir} with $\Omega_N=\RRR^3$, \eqref{term_2_alt} and \eqref{term_3} from Lemma~\ref{lem:estimates_terms_alpha_dot_beta} with $\gamma=1$, and, due to $\int \rho_N^t = N$,
\be
\int (v^2\star\rho^t_N)(y)\,\rho^t_N(y)\,d^3y \leq \left( \sup_{y} (v^2\star\rho^t_N)(y) \right) N.
\ee
We then find
\begin{align}
|\Term_{\pm}| &\leq 4 \sqrt{D_1(t)} \, \Big( \alpha_n(t) + N^{-1} \Big) \nonumber \\ 
&\quad + \sqrt{12} \, \sqrt{D_1(t)} \Big( \alpha_n(t) + N^{-1} \Big) \nonumber \\
&\quad + 4 \sqrt{D_1(t)} \, \alpha_n(t) \nonumber \\
&\leq \frac{1}{2} C(t) \Big(\alpha_n(t) + N^{-1} \Big),
\end{align}
with $0<C(t) = 4\big(4+\sqrt{3}\big)\sqrt{D_1(t)} \leq 24\sqrt{D(t)}$. Therefore,
\be
\partial_t \alpha_n(t) \leq |\Term_{+}| + |\Term_{-}| \leq C(t) \, \Big(\alpha_n(t) + N^{-1} \Big).
\ee
Applying the Gronwall Lemma \ref{lem:gronwall} gives \eqref{main_alpha_ineq_n}.
\end{proof}

\noindent\textbf{Remarks.}
\begin{enumerate}
\setcounter{enumi}{\theremarks}
\item\label{itm:exch_term_order} Following up on Remark~\ref{itm:exch} after Theorem~\ref{thm:estimates_terms_alpha_dot_beta_n}, let us consider the size of the error we make by neglecting the exchange term. We suppose that the exchange term is of $O(N^{-\frac{2}{3}})$. It then gives an additional term $C(t) N^{-\frac{2}{3}} \sqrt{\alpha_n(t)}$ in the time derivative of $\alpha_n(t)$. Then (note that $(a+b)^2 \leq 2(a^2+b^2)$)
\begin{align}
\partial_t\alpha_n(t) &\leq C(t) \left( \alpha_n(t) + N^{-1} + N^{-\frac{2}{3}} \sqrt{\alpha_n(t)} \right) \nonumber \\
&\leq C(t) \left( \left( \sqrt{\alpha_n(t)} + N^{-\frac{2}{3}} \right)^2 + N^{-1} \right) \nonumber \\
&\leq 2 C(t) \left( \alpha_n(t) + N^{-\frac{4}{3}} + N^{-1} \right),
\end{align}
so the error in $\partial_t\alpha_n(t)$ is only of $O(N^{-\frac{4}{3}})$.
\end{enumerate}
\setcounter{remarks}{\theenumi}

Now we prove the most general version of our main theorems, using the counting functional $\alpha_{m^{(\gamma)}}(t)$.

\begin{proof}[Proof of Theorem~\ref{thm:estimates_terms_alpha_dot_beta_general}]
Under the assumptions \eqref{alpha_dot_m_ass_1}-\eqref{alpha_dot_m_ass_5} (which hold for $v^{(N)}_{\pm}$ separately), and using the splitting from \eqref{alpha_plus_minus} again, Lemma~\ref{lem:estimates_terms_alpha_dot_beta} gives
\begin{align}
\frac{1}{2} \, \partial_t \alpha_{m^{(\gamma)}}(t) &\leq 4 \sqrt{D_3(t)} \, N^{-\frac{\delta_3}{2}} \, \Big( \alpha_{m^{(\gamma)}}(t) + N^{-\gamma} \Big) + 4 \sqrt{2} D_4(t) \, N^{-\delta_4} \, \Big( \alpha_{m^{(\gamma)}}(t) + N^{-\gamma} \Big) \nonumber \\
&\quad + \sqrt{12} \, \Big( D_0(t)^2 \, \alpha_{m^{(\gamma)}}(t)^2 + D_2(t) \, \alpha_{m^{(\gamma)}}(t) \, N^{\delta_2-\gamma} \Big)^{\frac{1}{2}} \nonumber \\
&\quad + 8 \sqrt{D_1(t)} \, \alpha_{m^{(\gamma)}}(t).
\end{align}
Now note that for any $\alpha,C_N\geq 0$,
\be
\big( \alpha^2 + C_N \alpha \big)^{\frac{1}{2}} \leq \big( \alpha^2 + 2C_N\alpha + C_N^2 \big)^{\frac{1}{2}} = \alpha + C_N.
\ee
Using this we find that
\begin{align}
\partial_t \alpha_{m^{(\gamma)}}(t) &\leq C(t) \Big(\alpha_{m^{(\gamma)}}(t) + N^{-\delta} \Big),
\end{align}
where $0<\delta = \min\left\{ \gamma -\delta_2, \gamma + \frac{\delta_3}{2}, \gamma + \delta_4 \right\}$ and
\be
C(t) = 12 \max\bigg\{ 4 \sqrt{D_3(t)} N^{-\frac{\delta_3}{2}}, 4 \sqrt{2} D_4(t) N^{-\delta_4}, \sqrt{12} D_0(t), \sqrt{12} \frac{D_2(t)}{D_0(t)}, 8\sqrt{D_1(t)} \bigg\}.
\ee
(Note that $D_0(t)>0$ for all $t\geq 0$.) Applying the Gronwall Lemma \ref{lem:gronwall} gives \eqref{main_alpha_ineq_m}.
\end{proof}

\chapter{Proof of Results for $-\Delta$ and Interactions $|x|^{-s}$}\label{sec:mean-field_scalings_general}
Note that in this chapter, for ease of notation, $C$ denotes a constant which can be different from line to line.

\section{Kinetic Energy Inequalities}\label{sec:energy inequalities}
In this section, we state some well-known inequalities, which we use in Chapter~\ref{sec:mean-field_scalings} to show that the conditions of Theorems \ref{thm:estimates_terms_alpha_dot_beta_n} and \ref{thm:estimates_terms_alpha_dot_beta_general} hold if $E_{\kin,\mf}(t) \leq AN$. We denote the kinetic energy of a wave function $\psi$ by $E_{\kin,\psi} = \SCP{\psi}{\sum_{j=1}^N (-\Delta_{x_j}) \psi}$, and the diagonal of the reduced one-particle density matrix, normalized to $N$, by $\rho^{\psi}(x)$. With the notation from \eqref{definition_dens_mat}, $\rho^{\psi}(x) = N\mu_1^{\psi}(x;x)$. Recall that $H^1(\RRR^d) = \left\{ f \in L^2(\RRR^d):  \norm{\nabla f} < \infty \right\}$.

First, we need the kinetic energy inequality due to Lieb and Thirring \cite{lieb:1975}, see also \cite{lieb:2010}. (Note that we state the lemma here in a slightly less general version than in \cite{lieb:2010}.)

\begin{lemma}[Corollary 4.1 in \cite{lieb:2010}]\label{lem:Lieb_5_3}
Let $\psi \in H^1(\RRR^{3N})$ be antisymmetric. Then
\be
\int_{\RRR^3} \rho^{\psi}(x)^{\frac{5}{3}} \, d^3x \leq \frac{5}{9} (2\pi)^{-\frac{2}{3}} ~ E_{\kin,\psi}.
\ee
\end{lemma}

We mostly need this lemma for antisymmetrized product states. Lemma~\ref{lem:Lieb_5_3} then says that for orthonormal  $\varphi_1,\ldots,\varphi_N \in H^1(\RRR^3)$, $\rho_N(x) = \sum_{i=1}^N |\varphi_i(x)|^2$ and $E_{\kin,\mf} = \sum_{i=1}^N \norm{\nabla \varphi_i}^2$, we have
\be
\int_{\RRR^3} \rho_N(x)^{\frac{5}{3}} \, d^3x \leq \frac{5}{9} (2\pi)^{-\frac{2}{3}} ~ E_{\kin,\mf}.
\ee

From Lemma~\ref{lem:Lieb_5_3} it immediately follows a rigorous version of the statement that any fermionic wave function with kinetic energy of $O(N)$ ``occupies a volume'' at least proportional to $N$. This is captured by the following lemma which is similar to Theorem 7.2 in \cite{lieb:2010}.

\begin{lemma}\label{lem:Omega}
Let $\psi \in H^1(\RRR^{3N})$ be antisymmetric. Let $\Omega$ be a measurable subset of $\RRR^3$ with volume $\vol(\Omega) = \int_{\Omega}dx$. Then
\be
\int_{\Omega} \rho^{\psi}(x) \, d^3x \leq C \, (\vol(\Omega))^{\frac{2}{5}} \left( E_{\kin,\psi} \right)^{\frac{3}{5}}.
\ee
\end{lemma}

\begin{proof}
By H\"older's inequality and Lemma~\ref{lem:Lieb_5_3},
\be
\int_{\Omega} \rho^{\psi} \leq \left( \int_{\Omega} 1^{\frac{5}{2}} \right)^{\frac{2}{5}} \left( \int_{\Omega} \left(\rho^\psi\right)^{\frac{5}{3}} \right)^{\frac{3}{5}} \leq (\vol(\Omega))^{\frac{2}{5}} \left( C E_{\kin,\psi} \right)^{\frac{3}{5}}.
\ee
\end{proof}

Thus, for wave functions with $E_{\kin,\psi} \leq CN$, it follows that, e.g., for $\Omega_{N,\varepsilon} = \left[0,N^{\frac{(1-\varepsilon)}{3}}\right]^3$, i.e., $\vol(\Omega_{N,\varepsilon})=N^{1-\varepsilon}$, $\varepsilon>0$,
\be
\int_{\Omega_{N,\varepsilon}} \rho^{\psi}(x) \, d^3x \leq C N^{1-\frac{2}{5}\varepsilon},
\ee
i.e., there are still $C N$ particles outside the volume $\Omega_{N,\varepsilon}$ (recall $\int \rho^{\psi} = N$).

For Section~\ref{sec:mean-field_scalings} we also need the Hardy-Littlewood-Sobolev inequality (see, e.g., \cite[Theorem 4.3]{liebloss:2001}), here stated for three dimensions.

\begin{lemma}[Hardy-Littlewood-Sobolev Inequality]\label{lem:Hardy_Littlewood_Sobolev}
Let $p,r > 1$ and $0<\lambda<3$ with $\frac{1}{p} + \frac{\lambda}{3} + \frac{1}{r} = 2$. Let $f \in L^p(\RRR^3)$ and $h \in L^r(\RRR^3)$. Then there is a constant $C=C(\lambda,p)$, such that
\be\label{Hardy_Littlewood}
\left| \int_{\RRR^3} \int_{\RRR^3} f(x) |x-y|^{-\lambda} h(y) \, d^3x \, d^3y \right| \leq C \norm[p]{f} \norm[r]{h}.
\ee
\end{lemma}

For later reference, we also state Hardy's inequality (see, e.g., \cite{lieb:2010}):
\begin{lemma}[Hardy's inequality]\label{lem:Hardy}
Let $f \in H^1(\RRR^3)$. Then
\be\label{Hardy}
\int_{\RRR^3} \frac{|f(x)|^2}{|x|^2} \, d^3x \leq 4 \int_{\RRR^3} |\nabla f(x)|^2 \, d^3x.
\ee
\end{lemma}

\section{An Estimate Using the Boundedness of Kinetic Energy}\label{sec:mean-field_scalings}
Let us now use the kinetic energy inequality from Lemma~\ref{lem:Lieb_5_3} to estimate the mean-field term $v^{(N)}\star\rho_N$. The lemma we prove in this section is necessary for the proofs of the results from Chapter~\ref{sec:main_theorem_mf_x-s} which we present in Chapter~\ref{sec:proofs_main_theorems}. Recall that we want to consider situations where $E_{\kin,\mf}(t) \leq AN$ and interaction potentials with a long range behavior like
\be\label{v_s}
v_s(x) = |x|^{-s}, ~\text{for}~ 0 < s < \frac{6}{5}.
\ee
For $\epsilon > 0$, let us define the class of interactions with singularity cut off as
\be\label{v_sepsilon}
v_{s,\varepsilon} \in L^{\infty} ~\text{such that}~ v_{s,\varepsilon}(x) = |x|^{-s} ~\text{for}~ |x|\geq\varepsilon,~ 0 \leq v_{s,\varepsilon}(x) \leq |x|^{-s} ~~\forall x\in\RRR^3.
\ee
From the next lemma we can read off that the correct scaling exponents for $v_s$ and $v_{s,\varepsilon}$ are $\beta = 1- \frac{s}{3}$. 

\begin{lemma}\label{lem:scaling_x-s}
Let $\varphi_1,\ldots,\varphi_N \in L^2(\RRR^3)$ be orthonormal. We assume that
\be\label{ass_kin_energy}
E_{\kin,\mf} = \sum_{i=1}^N \norm{\nabla \varphi_i}^2 \leq AN
\ee
for some $A>0$. Let
\be
v^{(N)}(x) = N^{-\beta} \, v(x)
\ee
with either $v=v_s$ (see \eqref{v_s}) or $v=v_{s,\varepsilon}$ (see \eqref{v_sepsilon}), and with $\beta = 1 - \frac{s}{3}$, $0<s<\frac{6}{5}$. We set $\rho_N(x) = \sum_{i=1}^N |\varphi_i(x)|^2$. Then there is a constant $0<C\propto A^{\frac{s}{2}}$ (independent of $N$, dependent on $s$) such that
\be\label{scaling_x-s_term_1}
\left( v^{(N)}\star\rho_N \right)(y) \leq C \quad \forall y \in \RRR^3.
\ee
Furthermore, there is an explicit example of orbitals $\varphi_1,\ldots,\varphi_N$, namely plane waves with constant density, such that also $\left( v^{(N)}\star\rho_N \right)(y) \geq C' ~ \forall y \in \RRR^3$ for some $C'>0$.
\end{lemma}

\noindent\textbf{Remarks.}
\begin{enumerate}
\setcounter{enumi}{\theremarks}
\item An immediate consequence of \eqref{scaling_x-s_term_1} is, that the scaled direct interaction energy of any orbital is $O(1)$,
\be\label{scaling_x-s_term_2}
\int_{\RRR^3} \left( v^{(N)}\star\rho_N \right)(y) \, |\varphi_i(y)|^2 \,d^3y \leq C \quad \forall \, i=1,\ldots,N,
\ee
and the total scaled direct interaction energy is $O(N)$,
\be\label{scaling_x-s_term_3}
\int_{\RRR^3} \left( v^{(N)}\star\rho_N \right)(y) \, \rho_N(y) \,d^3y \leq C N.
\ee
This follows directly from H\"older's inequality, using $\norm[1]{|\varphi_i|^2}=1$ and $\norm[1]{\rho_N}=N$.

\item\label{itm:e_kin_5_3} The proof could easily be generalized for $\varphi_1,\ldots,\varphi_N$ with $E_{\kin,\mf} \leq AN^{\delta}$ for a certain range of $\delta$'s. In the semiclassical case outlined in Chapter~\ref{sec:mf_fermions_semiclassical}, where $E_{\kin,\mf} \leq AN^{\frac{5}{3}}$, we find that for all $s$, the right scaling exponent is $\beta=1$, i.e.,
\be\label{mf_term_order_sc}
N^{-1} \left( v_s \star\rho_N \right)(y) \leq C \quad \forall y \in \RRR^3.
\ee
This follows directly from choosing $R_N$ $N$-independent in \eqref{rho_x-1_calc}. In fact, one can easily show that \eqref{mf_term_order_sc} even holds for any interaction potential in $L^{\frac{5}{2}}_{\loc}(\RRR^3)$, i.e., when the singularities are in $L^{\frac{5}{2}}(\RRR^3)$.
\end{enumerate}
\setcounter{remarks}{\theenumi}

\begin{proof}
From \eqref{lem:Lieb_5_3} with the assumption \eqref{ass_kin_energy} it follows that
\be\label{kin_en_A}
\int_{\RRR^3} \rho_N(x)^{\frac{5}{3}} \, d^3x \leq \frac{5}{9} (2\pi)^{-\frac{2}{3}} A \, N.
\ee
Furthermore, recall that $\int \rho_N = N$. In the following, we show the inequality \eqref{scaling_x-s_term_1} only for $v_s$. Then, since $v_{s,\varepsilon}\leq v_s$, it also follows for $v_{s,\varepsilon}$. Recall that $B_R(y)=\{ x \in \RRR^3: |x-y| < R \}$ and $\overline{B_R(y)} := \RRR^3 \setminus B_R(y)$. First, note that for $0 < s < \frac{6}{5}$,
\begin{align}\label{int_2_5}
\left( \int_{B_R(0)} |x|^{-\frac{5}{2}s} \,d^3x \right)^{\frac{2}{5}} &= \left( 4\pi \int_{B_R(0)} r^{-\frac{5}{2}s} \,r^2 dr \right)^{\frac{2}{5}} = \left( \left(\frac{4\pi}{3-\frac{5}{2}s}\right) R^{3-\frac{5}{2}s} \right)^{\frac{2}{5}} \nonumber \\
&= \left(\frac{4\pi}{3-\frac{5}{2}s}\right)^{\frac{2}{5}} \, R^{\frac{6}{5}-s}.
\end{align}
Then, using H\"older's inequality, \eqref{kin_en_A}, $\int \rho_N = N$ and \eqref{int_2_5}, we find for any possibly $N$-dependent $R_N>0$,
\begin{align}\label{rho_x-1_calc}
\int_{\RRR^3} \frac{\rho_N(x)}{|x-y|^s} \,d^3x &= \int_{B_{R_N}(y)} \frac{\rho_N(x)}{|x-y|^s} \,d^3x + \int_{\overline{B_{R_N}(y)}} \frac{\rho_N(x)}{|x-y|^s} \,d^3x \nonumber \\
\eqexp{by H\"older} &\leq \left( \int_{B_{R_N}(y)} \rho_N(x)^{\frac{5}{3}} \,d^3x \right)^{\frac{3}{5}} \left( \int_{B_{R_N}(y)} |x-y|^{-\frac{5}{2}s} \,d^3x \right)^{\frac{2}{5}} \nonumber \\
&\quad\quad + \left( \int_{\overline{B_{R_N}(y)}} \rho_N(x) \,d^3x \right) \left( \sup_{x \in \overline{B_{R_N}(y)}} |x-y|^{-s} \right) \nonumber \\
&\leq \left( \int_{\RRR^3} \rho_N(x)^{\frac{5}{3}} \,d^3x \right)^{\frac{3}{5}} \left( \int_{B_{R_N}(0)} |x|^{-\frac{5}{2}s} \,d^3x \right)^{\frac{2}{5}} \nonumber \\
&\quad\quad + \left( \int_{\RRR^3} \rho_N(x) \,d^3x \right) \left( \sup_{x \in \overline{B_{R_N}(0)}} |x|^{-s} \right) \nonumber \\
\eqexp{by \eqref{kin_en_A}, \eqref{int_2_5}} &\leq C N^{\frac{3}{5}} R_N^{\frac{6}{5}-s} + N R_N^{-s}.
\end{align}
Setting $R_N=N^{\frac{1}{3}}$ (if we set $R_N=N^{\delta}$ and then optimize \eqref{rho_x-1_calc} with respect to $\delta$ we find $\delta=\frac{1}{3}$) we find
\be\label{rho_x-1}
\int_{\RRR^3} \frac{\rho_N(x)}{|x-y|^s} \,d^3x \leq C N^{1-\frac{s}{3}}.
\ee
Using the explicit values of the constants from \eqref{kin_en_A} and \eqref{int_2_5}, setting $R_N=rN^{\frac{1}{3}}$, with $N$-independent $r>0$, and minimizing the resulting expression \eqref{rho_x-1_calc} with respect to $r$ gives an explicit value for the constant of \eqref{rho_x-1}:
\be\label{rho_x-1_constant}
C = \left( \frac{6}{5}-s \right)^{\frac{s}{2}-1} s^{-\frac{5}{6}s} \left(\frac{6}{5}\right) 2^{\frac{2s}{3}} \, 3^{-s} \, 5^{\frac{s}{6}} ~ A^{\frac{s}{2}}.
\ee
Note that \eqref{rho_x-1} does \emph{not} follow directly from the generalized Young inequality
\be
\norm[r]{f \star g} \leq C \norm[q,w]{g} \norm[p]{f}
\ee
with $\frac{1}{p}+\frac{1}{q}=1+\frac{1}{r}$ and where $\norm[q,w]{\cdot}$ denotes the weak $L^q$ norm, since it only holds for $1 < p,q,r < \infty$. 

It remains to show that when the $\varphi_1,\ldots,\varphi_N$ are plane waves with constant kinetic energy per particle, the inequality \eqref{scaling_x-s_term_1} holds also in the other direction. Let $\id_V(x)$ denote the characteristic function of the set $\left[ -\frac{L}{2},\frac{L}{2} \right]^3 = V \subset\RRR^3$ and let $N=cL^3$ with constant $c$. The orthonormal functions
\be
\varphi_{k}(x) = L^{-\frac{3}{2}} e^{i\frac{2\pi}{L} k \cdot x} \, \id_V(x),
\ee
$k\in \ZZZ^3$, have the density $\rho_0(x) = \sum_{|k|,\#|k|=N} |\varphi_{k}(x)|^2 = \frac{N}{L^3} \id_V(x) = c \, \id_V(x)$. In the ground state, the kinetic energy is proportional to $N$, since
\be\label{plane_wave_E_kin}
\sum_{|k|,\#|k|=N} \scp{\varphi_k}{(-\Delta)\varphi_k} = \sum_{|k|,\#|k|=N} \left( \frac{2\pi}{L} k \right)^2 \leq C \left(\frac{2\pi}{L}\right)^2 \int_0^{N^{\frac{1}{3}}} r^2~ r^2 dr  \propto \frac{N^{\frac{5}{3}}}{L^2} = c^{\frac{2}{3}} N.
\ee
Then (as we already know from \eqref{scaling_x-s_term_1}), we find
\begin{align}
(v_{s,\varepsilon}\star\rho_0)(y) \leq (v_s\star\rho_0)(y) &= \int_{\RRR^3} \frac{\rho_0(x)}{|x-y|^s} \,d^3x = c \int_{\left[ -\frac{L}{2},\frac{L}{2} \right]^3} |x-y|^{-s} \,d^3x \nonumber \\
& \leq c \int_{\left[ -\frac{L}{2},\frac{L}{2} \right]^3} |x|^{-s} \,d^3x \propto \int_{0}^{N^{\frac{1}{3}}} r^{-s}\, r^2 dr \propto N^{1-\frac{s}{3}}.
\end{align}
Also, for example for $R \geq |y|$,
\begin{align}\label{calc_plane_waves}
\int_{\RRR^3} \frac{\rho_0(x)}{|x-y|^s} \,d^3x &= c \int_{\left[-\frac{L}{2},\frac{L}{2}\right]^3} |x-y|^{-s} \,d^3x = c \int_{\left[-\frac{L}{2},\frac{L}{2}\right]^3-y} |x|^{-s} \,d^3x \nonumber \\
&\geq C \int_{R}^{R+\frac{L}{2}} r^{-s} \, r^2 dr \propto \left(R+\frac{L}{2}\right)^{3-s} - R^{3-s} \geq \left(\frac{L}{2}\right)^{3-s} \propto N^{1-\frac{s}{3}}.
\end{align}
In the same way, for example for $R \geq |y|+\varepsilon$ and for $N$ big enough,
\be\label{calc_plane_waves2}
\int_{\RRR^3} \rho_0(x) v_{s,\varepsilon}(x) \,d^3x \geq C \int_{R}^{R+\frac{L}{2}} r^{-s} \,r^2 dr \propto N^{1-\frac{s}{3}}.
\ee
From this, it follows directly that \eqref{scaling_x-s_term_1} holds with $\geq$ for the example of plane waves.
\end{proof}

\section{Proof of the Results}\label{sec:proofs_main_theorems}
\begin{proof}[Proof of Corollary~\ref{cor:estimates_terms_alpha_dot_Coulomb_rho_infty}]
We use Theorem~\ref{thm:estimates_terms_alpha_dot_beta_n}. We thus have to show that for $v^{(N)}(x) = N^{-\frac{2}{3}} \, |x|^{-1}$,
\be
\left( \Big(v^{(N)}\Big)^2\star\rho_N^t \right)(y) \leq C N^{-1} \quad \forall y \in \RRR^3,
\ee
if $E_{\kin,\mf} \leq AN$ and $\norm[\infty]{\rho_N^t} < D$. Let us write $\Omega_N = B_{N^{\frac{1}{3}}}(0)$, i.e., the ball with radius $N^{\frac{1}{3}}$ around $0$, and define $\Omega_N(y) = \Omega_N + y$ and $\overline{S} = \RRR^3\setminus S$ for any set $S\subset\RRR^3$. By splitting the convolution into two integrals and then using H\"older's inequality, we find (similar to the calculation \eqref{rho_x-1_calc})
\begin{align}
\left( \Big(v^{(N)}\Big)^2\star\rho_N^t \right)(y) &= N^{-\frac{4}{3}} \left( \int_{\Omega_N(y)} |x-y|^{-2} \rho_N^t(x)\,d^3x + \int_{\overline{\Omega}_N(y)} |x-y|^{-2} \rho_N^t(x)\,d^3x \right) \nonumber \\
\eqexp{by H\"older} &\leq N^{-\frac{4}{3}} \left( \int_{\Omega_N} |x|^{-2} \,d^3x  \, \norm[\infty]{\rho_N^t} + \left( \sup_{x\in\overline{\Omega}_N} |x|^{-2} \right) \norm[1]{\rho_N^t} \right) \nonumber \\
&\leq C N^{-\frac{4}{3}} \left( N^{\frac{1}{3}}\norm[\infty]{\rho_N^t} + N^{-\frac{2}{3}} N \right) \nonumber \\
&\leq C N^{-1}.
\end{align}
\end{proof}

\begin{proof}[Proof of Theorem~\ref{thm:E_kin_only}]
We consider the three different interactions separately.
\begin{itemize}
\item Let
\be
v_s(x) = \pm |x|^{-s}, \text{with } 0<s<\frac{3}{5} \text{ and } \beta=1-\frac{s}{3}.
\ee
We use Theorem~\ref{thm:estimates_terms_alpha_dot_beta_n}, i.e., we have to show that $\Big(\big(v_s^{(N)}\big)^2\star\rho_N^t\Big)(y) \leq C N^{-1} ~ \forall y \in \RRR^3$. Since we consider only $0<s<\frac{3}{5}$, we can use \eqref{scaling_x-s_term_1} with $s$ replaced by $2s$ (recall that \eqref{scaling_x-s_term_1} holds for $0<s<\frac{6}{5}$). When we write out the scaling explicitly, \eqref{scaling_x-s_term_1} becomes
\be
\left( v_{2s}\star\rho_N^t \right)(y) \leq C N^{\left(1-\frac{2s}{3}\right)}.
\ee
Thus,
\begin{align}\label{s_2s}
\Big(\big(v_s^{(N)}\big)^2\star\rho_N^t\Big)(y) &= N^{-2\left(1-\frac{s}{3}\right)} \Big(v_s^2\star\rho_N^t\Big)(y) \nonumber \\
&= N^{-2\left(1-\frac{s}{3}\right)} \Big(v_{2s}\star\rho_N^t\Big)(y) \nonumber \\
&\leq C N^{-2\left(1-\frac{s}{3}\right)} N^{\left(1-\frac{2s}{3}\right)} \nonumber \\
&= C N^{-1}.
\end{align}
If we use that the constant in \eqref{scaling_x-s_term_1} is proportional to $A^{\frac{s}{2}}$, we find that the constant in \eqref{s_2s} is proportional to $A^{s}$ and thus the $C$ appearing in the $\alpha_n$-estimate \eqref{main_alpha_ineq_n_applied1} is proportional to $A^{\frac{s}{2}}$.

\item Let
\begin{align}
&v_{s,\varepsilon} \in L^{\infty} \text{ with } 0 \leq v_{s,\varepsilon}(x) \left\{\begin{array}{cl} \leq |x|^{-s} &, \, \text{for } |x|\leq \varepsilon \\ =|x|^{-s} &, \, \text{for } |x|>\varepsilon , \end{array}\right., \text{with } \varepsilon>0, \nonumber \\
&\text{with } 0<s<\frac{6}{5} \text{ and } \beta = 1 - \frac{s}{3}.
\end{align}
We use Theorem~\ref{thm:estimates_terms_alpha_dot_beta_general} with $\Omega_N = \emptyset$. The assumption \eqref{alpha_dot_m_ass_1} is satisfied according to Lemma~\ref{lem:scaling_x-s}. Using $v_{s,\varepsilon} \in L^{\infty}$ and \eqref{alpha_dot_m_ass_1} we find
\begin{align}
\Big(\big(v_{s,\varepsilon}^{(N)}\big)^2\star\rho_N^t\Big)(y) &\leq N^{-\left(1-\frac{s}{3}\right)} \Big( \sup_y v_{s,\varepsilon}(y) \Big) \Big(v_{s,\varepsilon}^{(N)}\star\rho_N^t\Big)(y) \nonumber \\
&\leq C N^{-\left(1-\frac{s}{3}\right)},
\end{align}
i.e., \eqref{alpha_dot_m_ass_2} holds if $\gamma \leq 1 -\frac{s}{3}$. In order to show that \eqref{alpha_dot_m_ass_3} holds, we use the Hardy-Littlewood-Sobolev inequality \eqref{Hardy_Littlewood}. Note that from $\int \rho_N^t = N$ and $\int (\rho_N^t)^{\frac{5}{3}} \leq C N$ it follows that $\int (\rho_N^t)^p \leq C N$ for all $1 \leq p \leq \frac{5}{3}$. For $\lambda=2s$ we have $p=(1-\frac{s}{3})^{-1}$ and, since $0<s<\frac{6}{5}$, we find $1 < p < \frac{5}{3}$, so that
\begin{align}\label{HLS_applied}
\int \int_{\RRR^3} \int_{\RRR^3} \big(v_{s,\varepsilon}(x-y)\big)^2 \rho_N^t(x)\rho_N^t(y) \,d^3x \,d^3y &\leq \int_{\RRR^3} \int_{\RRR^3} \frac{\rho_N^t(x)\rho_N^t(y)}{|x-y|^{2s}} \,d^3x \,d^3y \nonumber \\
&\leq C \norm[p]{\rho_N^t}^2 = C \left( \int (\rho_N^t)^p \right)^{\frac{2}{p}} \nonumber \\
&\leq C N^{\frac{2}{p}} = C N^{2\left(1-\frac{s}{3}\right)},
\end{align}
i.e., \eqref{alpha_dot_m_ass_3} is satisfied with $\delta_2=0$. Furthermore,
\be
\sup_{y \in \RRR^3} v_{s,\varepsilon}^{(N)}(y) \leq C \, N^{-\left(1-\frac{s}{3}\right)},
\ee
i.e., \eqref{alpha_dot_m_ass_5} holds if $\frac{1}{2}+\frac{\gamma}{2}+\delta_4=1-\frac{s}{3}$; that is, since $\delta_4\geq 0$, $\gamma\leq 1-\frac{2s}{3}$. Thus, $\delta = \min\left\{ \gamma -\delta_2, \gamma + \delta_4 \right\} = \gamma$. The conditions on $\gamma$ are $\gamma \leq 1 -\frac{s}{3}$ and $\gamma\leq 1-\frac{2s}{3}$, therefore, the theorem holds for all $\gamma \leq 1-\frac{2s}{3}$.

\item Let
\be
v(x) = \left\{\begin{array}{cl} |x|^{-s} &, \, \text{for } |x|\leq 1 \\ |x|^{-1} &, \, \text{for } |x|>1 , \end{array}\right., 
\ee
with $0<s<\frac{1}{3}$ and $\beta = \frac{2}{3}$. First, note that Lemma~\ref{lem:scaling_x-s} applies, since $|v(x)| \leq |x|^{-1}$. Using this, we find that
\be
\sup_{y\in\RRR^3} \Big(\big|v^{(N)}\big|\star\rho_N^t\Big)(y) \leq D_0,
\ee
and, from \eqref{HLS_applied} with $s=1$,
\be
\int \bigg(\big(v^{(N)}\big)^2\star\rho^t_N\bigg)(y)\,\rho^t_N(y)\,d^3y \leq D_2.
\ee
With Hardy's inequality \eqref{Hardy} we find
\be
\Big(\big(v^{(N)}\big)^2\star\rho_N^t\Big)(y) \leq N^{-\frac{4}{3}} \Big(|\cdot|^{-2}\star\rho_N^t\Big)(y) \leq C N^{-\frac{4}{3}} \sum_{i=1}^N \norm{\nabla\varphi_i^t}^2 \leq D_1 \, N^{-\frac{1}{3}}.
\ee
Therefore, since $\gamma\leq\frac{1}{3}$ was assumed, the assumptions \eqref{alpha_dot_m_ass_1}, \eqref{alpha_dot_m_ass_2} and \eqref{alpha_dot_m_ass_3} hold, with $\delta_2=0$. Let us now turn to assumption \eqref{alpha_dot_m_ass_5}. Let us write $\Omega_N(y)=B_{N^{-\tilde{\delta}}}(y)$ for the ball with radius $N^{-\tilde{\delta}}$ around $y$, with $\tilde{\delta}=\frac{2}{9-15s}$. Using H\"older's inequality and Lemma~\ref{lem:Lieb_5_3}, we find
\begin{align}
\int_{\Omega_N(y)} \Big(v^{(N)}(y-x)\Big)^2\rho_N^t(x)\,d^3x &= N^{-\frac{4}{3}} \int_{\Omega_N(y)} |x-y|^{-2s} \rho_N^t(x)\,d^3x \nonumber \\
\eqexp{by H\"older} &\leq  N^{-\frac{4}{3}} \left( \int_{\Omega_N(0)} \big(|x|^{-2s}\big)^{\frac{5}{2}} \,d^3x \right)^{\frac{2}{5}} \left( \int_{\RRR^3} \rho_N^t(x)^{\frac{5}{3}} \,d^3x \right)^{\frac{3}{5}} \nonumber \\
\eqexp{by Lem.~\ref{lem:Lieb_5_3}} &\leq C\, N^{-\frac{4}{3}} \left( \int_{0}^{N^{-\tilde{\delta}}} r^{-5s} \, r^2 dr \right)^{\frac{2}{5}} N^{\frac{3}{5}} \nonumber \\
&\leq C\, N^{-\frac{11}{15}} \left( \left( N^{-\tilde{\delta}} \right)^{3-5s} \right)^{\frac{2}{5}} \nonumber \\
&= C\, N^{-\frac{11}{15}-\frac{6}{5}\tilde{\delta}+2s\tilde{\delta}}.
\end{align}
Inserting $\tilde{\delta}=\frac{2}{9-15s}$, we find
\be
\sup_{y\in\RRR^3} \int_{\Omega_N(y)} \Big(v^{(N)}(y-x)\Big)^2\rho^t_N(x)\,d^3x \leq D_3 \, N^{-1},
\ee
so assumption \eqref{alpha_dot_m_ass_4} holds, with $\delta_3=0$. Then it follows that
\be
\sup_{y \in \RRR^3 \setminus \Omega_N} \big|v^{(N)}(y)\big| = N^{-\frac{2}{3}} \sup_{y \in \RRR^3 \setminus B_{N^{-\tilde{\delta}}}(0)} \big|v(y)\big| = N^{-\frac{2}{3}+s\tilde{\delta}}.
\ee
Thus, for assumption \eqref{alpha_dot_m_ass_5} to hold we need that ($\delta_4=0$)
\be
-\frac{2}{3} + s\tilde{\delta} \leq -\frac{1}{2} - \frac{\gamma}{2},
\ee
i.e.,
\be
\gamma \leq \frac{1}{3} - 2s\tilde{\delta} = \frac{1}{3} - \frac{4s}{9-15s}.
\ee
The convergence rate is $\delta = \min\left\{ \gamma -\delta_2, \gamma + \frac{\delta_3}{2}, \gamma + \delta_4 \right\} = \gamma$. \hfill $\qedhere$
\end{itemize}
\end{proof}

\begin{proof}[Proof of Proposition~\ref{pro:coulombN1}]
Note that the equations \eqref{free_equations} conserve the total energy which in this case is the total kinetic energy. Therefore \eqref{condition_E_kin_0} holds for all times $t\geq 0$, i.e., $E_{\kin,\mf}(t) \leq AN$. In the following, we write $v^{(N)}(x)=N^{-1}v(x)$. From Lemma~\ref{lem:scaling_x-s} we find
\be
\Big(\big|v^{(N)}\big|\star\rho_N^t\Big)(y) = N^{-\frac{1}{3}} \left( N^{-\frac{2}{3}} \Big(\big|v\big|\star\rho_N^t\Big)(y) \right) \leq C N^{-\frac{1}{3}},
\ee
and from the Hardy-Littlewood-Sobolev inequality, i.e., \eqref{HLS_applied} with $s=1$,
\be
\int_{\RRR^3} \left( \left(v^{(N)}\right)^2\star\rho_N \right)(y) \, \rho_N(y) \,d^3y = N^{-2} \int_{\RRR^3} \left( v^2\star\rho_N \right)(y) \, \rho_N(y) \,d^3y \leq C N^{-\frac{2}{3}}.
\ee
Using Hardy's inequality \eqref{Hardy} we find
\be
\Big(\big(v^{(N)}\big)^2\star\rho_N^t\Big)(y) = N^{-2} \int_{\RRR^3} \frac{\rho_N^t(x-y)}{|x|^2} \, d^3x \leq C N^{-2} E_{\kin,\mf}(t) \leq C N^{-1}.
\ee
Now we apply Lemma~\ref{lem:estimates_terms_alpha_dot_beta} and use the alternative estimate \eqref{term_1_dir_2} for the $qp$-$pp$ term and \eqref{term_2_alt} for the $qq$-$pp$ term. Then (using $\alpha_{m^{(\gamma)}}(t) \leq 1$)
\begin{align}\label{alpha_dot_trivial}
\partial_t \alpha_{m^{(\gamma)}}(t) &\leq C N^{-\frac{1}{3}+\frac{1}{2}-\frac{\gamma}{2}} + C N^{-\frac{1}{3}} \left(1 + N^{-\gamma} \right) + C N^{\frac{\gamma}{2}-\frac{1}{2}} \nonumber \\
&\leq C N^{-\delta},
\end{align}
where $\delta = \min\left\{ \frac{\gamma}{2}-\frac{1}{6}, \frac{1}{3}, -\frac{\gamma}{2}+\frac{1}{2} \right\} = \min\left\{ \frac{\gamma}{2}-\frac{1}{6}, -\frac{\gamma}{2}+\frac{1}{2} \right\}$. For $\delta>0$ we need $\frac{1}{3} < \gamma < 1$. Now \eqref{main_alpha_ineq_coulombN1_gamma} immediately follows by integrating \eqref{alpha_dot_trivial}.
\end{proof}

\appendix
\chapter{Proof for Semiclassical Scaling}\label{sec:proof_sc_scaling}
We here prove Theorem~\ref{thm:sc_main_thm}. We need some auxiliary lemmas. First, we state the result about the propagation of the semiclassical initial data that was obtained in \cite{benedikter:2013}. (Recall that $\hat{v}$ is the Fourier transform of the interaction potential $v$.)

\begin{lemma}[Proposition 3.4 in \cite{benedikter:2013}]\label{lem:sc_prop_sc}
Let $v \in L^1(\RRR^3)$ be such that
\be
\int d^3k \, (1+|k|^2)\, |\hat{v}(k)| < \infty.
\ee
Let $\omega_N$ be a non-negative trace class operator on $L^2(\RRR^3)$, with $\tr(\omega_N) = N$, $\norm[\op]{\omega_N} \leq 1$ and such that
\be
\sup_{k\in\RRR^3} \frac{1}{1+|k|} \norm[\tr]{\left[ \omega_N, e^{ik\cdot x} \right]} \leq CN^{\frac{2}{3}},
\ee
\be
\norm[\tr]{\left[ \omega_N, \nabla \right]} \leq C N.
\ee
Let $\omega_{N,t}$ be the solution to the Hartree-Fock equation \eqref{outline_hartree_fock_scaled_sc_app} (or the Hartree equation \eqref{outline_hartree_scaled_sc_app}) with initial data $\omega_N$. Then, there exist constants $c_1,c_2>0$, only depending on $v$, such that
\be
\sup_{k\in\RRR^3} \frac{1}{1+|k|} \norm[\tr]{\left[ \omega_{N,t}, e^{ik\cdot x} \right]} \leq c_1N^{\frac{2}{3}} \exp(c_2|t|),
\ee
\be
\norm[\tr]{\left[ \omega_{N,t}, \nabla \right]} \leq c_1 N \exp(c_2|t|),
\ee
for all $t \in \RRR$.
\end{lemma}

We apply this lemma to the case $\omega_N=p_1(0)$ and $\omega_{N,t} = p_1(t)$. What we refer to later in the proof of Theorem~\ref{thm:sc_main_thm} is (using $p_1q_1=0$)
\be\label{sc_prop}
\norm[\tr]{q_1 e^{ikx} p_1} = \norm[\tr]{q_1 \left[ p_1, e^{ikx} \right]} \leq \norm[\tr]{\left[ p_1, e^{ikx} \right]} \leq N^{\frac{2}{3}} \, C(t) \, (1 + \lvert k \rvert). 
\ee

We also make use of the singular value decomposition for compact operators. We state this decomposition for later reference in a separate lemma (for the proof, see, e.g., \cite[Thm.\ VI.17]{reedsimon1:1980}).

\begin{lemma}[Singular value decomposition]\label{lem:sing_value}
Let $A$ be a compact operator on a Hilbert space $\Hilbert$. Then there exist (not necessarily complete) orthonormal sets $\{ \phi_{\ell}\}_{\ell \in \NNN}$ and $\{ \tilde{\phi}_{\ell}\}_{\ell \in \NNN}$ and positive real numbers $\mu_{\ell}$ such that
\be
A = \sum_{\ell} \mu_{\ell} \ketbra{\phi_{\ell}}{\tilde{\phi}_{\ell}}.
\ee
The singular values $\mu_{\ell}$ are the eigenvalues of $|A|$.
\end{lemma}

It follows in particular that $\norm[\tr]{A} = \sum_{\ell} \mu_{\ell}$. Another lemma we use in the proof of Theorem~\ref{thm:sc_main_thm} is the following.

\begin{lemma}\label{lem:sc_trace_est}
Let $A_i: L^2(\RRR^3) \to L^2(\RRR^3)$ be positive self-adjoint compact operators that act only on $x_i$. Then, for all antisymmetric $\psi \in L^2(\RRR^{3N})$,
\begin{align}
\left|\bigSCP{\psi}{A_1A_2A_3\psi}\right| &\leq \frac{1}{(N-2)} \norm[\tr]{A_1} \left|\bigSCP{\psi}{A_2A_3 \psi}\right|.
\end{align}
\end{lemma}

\begin{proof}
We use the eigenvalue decomposition $A_1=\sum_j \lambda_j \ketbr{\phi_j}_1$ and Lemma~\ref{lem:projector_norms} which immediately yields
\begin{align}
\left|\bigSCP{\psi}{A_1A_2A_3\psi}\right| &= \left| \sum_j \lambda_j \bigSCP{\psi}{\ketbr{\phi_j}_1 A_2A_3\psi} \right| \nonumber \\
\eqexp{by Lem.~\ref{lem:projector_norms}} &\leq \sum_j \lambda_j \frac{1}{(N-2)} \left|\bigSCP{\psi}{A_2A_3 \psi}\right| \nonumber \\
&= \frac{1}{(N-2)} \norm[\tr]{A_1} \left|\bigSCP{\psi}{A_2A_3 \psi}\right|.
\end{align}
\end{proof}

\begin{proof}[Proof of Theorem~\ref{thm:sc_main_thm}]
The strategy of the proof is again to calculate the time derivative of $\alpha_n(t)$ and proving the bound
\be
\partial_t \alpha_n(t) \leq C(t) \big(\alpha_n(t) + N^{-1}\big).
\ee
Then the desired bound \eqref{main_alpha_ineq_sc} follows from the Gronwall Lemma~\ref{lem:gronwall}.

We calculate the time derivative of $\alpha_n(t)$ as in Lemma~\ref{lem:alpha_derivative}; in particular, we can use \eqref{alpha_derivative_n_remark}, since the weight function is $n(k)=\frac{k}{N}$ here. Due to the $N^{-\frac{1}{3}}$ in front of the time derivatives in the Schr\"odinger and mean-field equations, we get an additional factor of $N^{\frac{1}{3}}$. We find, using the scaling $v^{(N)}=N^{-1}v$,
\begin{align}\label{alpha_derivative_sc}
\partial_t \alpha_n(t) =&~ 2 N^{\frac{1}{3}} \, \Im\, \bigSCP{\psi^t}{q_1\Big( (N-1)p_2v^{(N)}_{12}p_2 - V^{(N)}_1 \Big) p_1 \psi^t} \nonumber \\
& + 2 N^{\frac{1}{3}} \, \Im\, \bigSCP{\psi^t}{q_1q_2\Big( (N-1)v^{(N)}_{12} \Big)p_1p_2 \psi^t} \nonumber \\
& + 2 N^{\frac{1}{3}} \, \Im\, \bigSCP{\psi^t}{q_1q_2\Big( (N-1)v^{(N)}_{12} \Big)p_1q_2 \psi^t} \nonumber \\
=&~ 2 N^{-\frac{2}{3}} \, \Im\, \bigSCP{\psi^t}{q_1\Big( (N-1)p_2v_{12}p_2 - V_1 \Big) p_1 \psi^t} \nonumber \\
& + 2 N^{-\frac{2}{3}} \, \Im\, \bigSCP{\psi^t}{q_1q_2\, (N-1)v_{12} \,p_1p_2 \psi^t} \nonumber \\
& + 2 N^{-\frac{2}{3}} \, \Im\, \bigSCP{\psi^t}{q_1q_2\, (N-1)v_{12} \,p_1q_2 \psi^t},
\end{align}
with $V_1 = V_1^{\dir}$ in the case of the fermionic Hartree equations, and $V_1 = V_1^{\dir} + V_1^{\exch}$ in the case of the Hartree-Fock equations. We now estimate the three terms separately. For ease of notation we omit the $t$-dependence in the following, i.e., we write $\psi^t = \psi$, $\varphi_i^t = \varphi_i$ and the constants $C$ could be $t$-dependent. For the estimates, we decompose $v$ in its Fourier components, i.e., we write $v(x) = \int d^3k \, \hat{v}(k) e^{ikx}$. Note that the assumption $\int d^3k \, (1+|k|^2)\, |\hat{v}(k)| < \infty$ in particular implies that $\int d^3k \, |\hat{v}(k)| < \infty$ and $\int d^3k \, |k| \, |\hat{v}(k)| < \infty$.

\absatz

\textbf{The} $\boldsymbol{qp}$-$\boldsymbol{pp}$ \textbf{term.} Let us first simplify the contribution from the exchange term. We find
\begin{align}\label{V_1_exch}
V_1^{\exch}p_1 &= - \sum_{j,\ell=1}^N \left(v_{12} \star (\varphi_{\ell}^*\varphi_j)\right)(x_1) \ketbra{\varphi_{\ell}}{\varphi_j}_1 \nonumber \\
&= - \int d^3k \, \hat{v}(k) \sum_{j,\ell=1}^N \left(e^{ik(x_1-\cdot)}  \star (\varphi_{\ell}^*\varphi_j)\right)(x_1) \ketbra{\varphi_{\ell}}{\varphi_j}_1 \nonumber \\
&= - \int d^3k \, \hat{v}(k)e^{ikx_1} \sum_{j,\ell=1}^N \scp{\varphi_{\ell}}{e^{-ikx} \varphi_j} \ketbra{\varphi_{\ell}}{\varphi_j}_1 \nonumber \\
&= - \int d^3k \, \hat{v}(k)e^{ikx_1} \sum_{j,\ell=1}^N \ket{\varphi_{\ell}}_1 \scp{\varphi_{\ell}}{e^{-ikx} \varphi_j} \bra{\varphi_j}_1 \nonumber \\
&= - \int d^3k \, \hat{v}(k)e^{ikx_1} p_1 e^{-ikx_1} p_1.
\end{align}
Thus, by Cauchy-Schwarz,
\begin{align}\label{sc_term1_exch}
N^{-\frac{2}{3}} \, \Big\lvert \bigSCP{\psi}{q_1 V^{\exch}_1 p_1 \psi} \Big\rvert &= N^{-\frac{2}{3}} \, \Big\lvert \int d^3k \, \hat{v}(k) \bigSCP{\psi}{q_1 e^{ikx_1} p_1 e^{-ikx_1} p_1 \psi} \Big\rvert \nonumber \\
&\leq N^{-\frac{2}{3}} \int d^3k \, \lvert\hat{v}(k)\rvert  \, \norm{q_1\psi} \nonumber \\
&\leq C N^{-\frac{2}{3}} \sqrt{\alpha_n}.
\end{align}
Here we see explicitly that the contribution from the exchange term is of lower order in $N$. Now follows the main part of the proof, namely to estimate the part of the $qp$-$pp$ term, where the difference between the Schr\"odinger and the mean-field interaction appears. Let us first rewrite this term using the Fourier decomposition of $v$. We find
\begin{align}\label{sc_term1_1}
&N^{-\frac{2}{3}} \bigSCP{\psi}{q_1\Big( (N-1)p_2v_{12}p_2 - V^{\dir}_1 \Big) p_1 \psi} \nonumber \\
&\quad =  N^{-\frac{2}{3}} \int d^3k \, \hat{v}(k) \bigSCP{\psi}{q_1\Big( (N-1) p_2e^{ik(x_1-x_2)}p_2 - \sum_{j=1}^N \scp{\varphi_j}{e^{ik(x_1-\cdot)}\varphi_j} \Big) p_1 \psi} \nonumber \\
&\quad = N^{-\frac{2}{3}} \int d^3k \, \hat{v}(k) \bigSCP{\psi}{\Big( (N-1) p_2 e^{-ikx_2} p_2 - \sum_{j=1}^N \scp{\varphi_j}{e^{-ikx}\varphi_j} \Big) q_1 e^{ikx_1} p_1 \psi}.
\end{align}
To deal with this expression we would like to diagonalize the operator $p_2e^{-ikx_2}p_2$, similar to as we did in the proof of Lemma~\ref{lem:estimates_terms_alpha_dot_beta}. However, since it is not self-adjoint, we decompose $e^{-ikx} = \cos(kx) - i \sin(kx)$ and diagonalize the self-adjoint operators
\be\label{sc_diag}
p_2 \cos(kx_2) p_2 = \sum_{j=1}^N \lambda_j p_2^{\chi_j}, \quad\quad p_2 \sin(kx_2) p_2 = \sum_{j=1}^N \tilde{\lambda}_j p_2^{\tilde{\chi}_j},
\ee
where the real eigenvalues $\lambda_j, \tilde{\lambda}_j$ and orthonormal eigenvectors $\chi_j, \tilde{\chi}_j$ depend on $k$. The eigenvalues have the property that $\lambda_j = \scp{\chi_j}{\cos(kx)\chi_j}$, so $|\lambda_j| \leq \norm{\chi_j}^2 = 1$, and furthermore $\sum_{j=1}^N \lambda_j  = \sum_{j=1}^N \scp{\chi_j}{\cos(kx) \chi_j} = \sum_{j=1}^N \scp{\varphi_j}{\cos(kx) \varphi_j}$, and analogous for $\tilde{\lambda}_j$. In the following, we use the projector $q^{\chi_j}_{\neq 1} = 1 - \sum_{m=2}^N p_m^{\chi_j}$ (it is a projector on antisymmetric functions). Let us now decompose the expression \eqref{sc_term1_1} by using $e^{-ikx} = \cos(kx) - i \sin(kx)$. We also use the singular value decomposition $q_1e^{ikx_1}p_1 = \sum_{\ell} \mu_{\ell}\ketbra{\phi_{\ell}}{\tilde{\phi}_{\ell}}_1$. Note that according to Lemma~\ref{lem:sing_value}, $\sum_{\ell} \mu_{\ell} = \norm[\tr]{q_1e^{ikx_1}p_1}$. For the $\cos$-term, we find, using the diagonalization from \eqref{sc_diag}, Cauchy-Schwarz (C.-S.), \eqref{p_neq_1} and Lemma~\ref{lem:projector_norms},
\begin{align}\label{sc_term1_part1}
&N^{-\frac{2}{3}} \, \bigg\lvert \int d^3k \, \hat{v}(k) \bigSCP{\psi}{\Big( (N-1) p_2 \cos(kx_2) p_2 - \sum_{j=1}^N \scp{\varphi_j}{\cos(kx) \varphi_j} \Big) q_1 e^{ikx_1} p_1 \psi} \bigg\rvert \nonumber \\
\eqexpl{by \eqref{sc_diag}} &\qquad = N^{-\frac{2}{3}} \, \bigg\lvert \int d^3k \, \hat{v}(k) \sum_{j=1}^N \lambda_j \bigSCP{\psi}{\Big( (N-1) p_2^{\chi_j} - 1 \Big) q_1 e^{ikx_1} p_1 \psi} \bigg\rvert \nonumber \\
\eqexpl{by antisym.} &\qquad = N^{-\frac{2}{3}} \, \bigg\lvert \int d^3k \, \hat{v}(k) \sum_{j=1}^N \lambda_j \bigSCP{\psi}{ \left( \sum_{m=2}^N p_m^{\chi_j} - 1 \right) q_1 e^{ikx_1} p_1 \psi} \bigg\rvert \nonumber \\
&\qquad = N^{-\frac{2}{3}} \, \bigg\lvert \int d^3k \, \hat{v}(k) \sum_{j=1}^N \lambda_j \bigSCP{\psi}{  q^{\chi_j}_{\neq 1} q_1 e^{ikx_1} p_1 \psi} \bigg\rvert \nonumber \\
&\qquad = N^{-\frac{2}{3}} \, \bigg\lvert \int d^3k \, \hat{v}(k) \sum_{j=1}^N \lambda_j \sum_{\ell} \mu_{\ell} \bigSCP{\psi}{ q^{\chi_j}_{\neq 1} \ketbra{\phi_{\ell}}{\tilde{\phi}_{\ell}}_1 q^{\chi_j}_{\neq 1} \psi} \bigg\rvert \nonumber \\
\eqexpl{by C.-S.} &\qquad \leq N^{-\frac{2}{3}} \, \int d^3k \, |\hat{v}(k)| \sum_{j=1}^N \underbrace{|\lambda_j|}_{\leq 1} \sum_{\ell} \mu_{\ell} \norm{\bra{\phi_{\ell}}_1q^{\chi_j}_{\neq 1}\psi} \norm{\bra{\tilde{\phi}_{\ell}}_1q^{\chi_j}_{\neq 1}\psi} \nonumber \\
\eqexpl{by C.-S.} &\qquad \leq N^{-\frac{2}{3}} \, \int d^3k \, |\hat{v}(k)| \sum_{\ell} \mu_{\ell} \sqrt{\sum_{j=1}^N \bigSCP{\psi}{ q^{\chi_j}_{\neq 1} \ketbra{\phi_{\ell}}{\phi_{\ell}}_1 \psi}} \,\times \nonumber \\ 
&\qquad \quad \quad \times \sqrt{\sum_{j=1}^N \bigSCP{\psi}{ q^{\chi_j}_{\neq 1} \ketbra{\tilde{\phi}_{\ell}}{\tilde{\phi}_{\ell}}_1 \psi}} \nonumber \\
\eqexpl{by \eqref{p_neq_1}} &\qquad = N^{-\frac{2}{3}} \, \int d^3k \, |\hat{v}(k)| \sum_{\ell} \mu_{\ell} \sqrt{\bigSCP{\psi}{\ketbra{\phi_{\ell}}{\phi_{\ell}}_1 (Nq_2+p_2) \psi}} \,\times \nonumber \\ 
&\qquad \quad \quad \times \sqrt{\bigSCP{\psi}{ \ketbra{\tilde{\phi}_{\ell}}{\tilde{\phi}_{\ell}}_1 (Nq_2+p_2)  \psi}} \nonumber \\
\eqexpl{by Lem.~\ref{lem:projector_norms}} &\qquad \leq N^{-\frac{2}{3}} \, \int d^3k \, |\hat{v}(k)| \norm[\tr]{q_1e^{ikx_1}p_1} \bigSCP{\psi}{(q_2+N^{-1}p_2) \psi} \nonumber \\ 
\eqexpl{by \eqref{sc_prop}} &\qquad \leq C \int d^3k \, |\hat{v}(k)| (1+|k|) \big( \alpha_n + N^{-1} \big) \nonumber \\
&\qquad \leq C \big( \alpha_n + N^{-1} \big).
\end{align}
The $\sin$-term goes exactly analogous. Let us summarize. If we consider only the fermionic Hartree equations (i.e., without exchange term), we find from \eqref{sc_term1_part1} that
\be
|qp\text{-}pp \text{ term}| \leq C \Big(\alpha_n(t) + N^{-1}\Big).
\ee
If we consider the Hartree-Fock equations, we use additionally \eqref{sc_term1_exch}, and thus find the same bound,
\be
|qp\text{-}pp \text{ term}| \leq C \Big( \alpha_n(t) + N^{-1} + \sqrt{\alpha_n(t)} N^{-\frac{2}{3}} \Big) \leq C \Big( \alpha_n(t) + N^{-1} \Big).
\ee

\absatz

\textbf{The} $\boldsymbol{qq}$-$\boldsymbol{pp}$ \textbf{term.} Similarly to Lemma~\ref{lem:estimates_terms_alpha_dot_beta}, we use the antisymmetry of $\psi$ to shift one $q$ to the right-hand side of the scalar product. In the following, we use the short-hand notation $A_m^k := q_me^{ikx_m}p_m$ and the decomposition $A_m^k = \sum_{\ell} \mu_{\ell} \ketbra{\phi_{\ell}}{\tilde{\phi}_{\ell}}_m$. Note that in particular $|A^k_m| = \sqrt{A^{k*}_mA^k_m} = \sum_{\ell} \mu_{\ell} \ketbr{\tilde{\phi}_{\ell}}_m$ and $|A_m^{k*}| = \sum_{\ell} \mu_{\ell} \ketbr{\phi_{\ell}}_m$. We then find
\begin{align}\label{sc_term2_prep}
&\Big\lvert \bigSCP{\psi}{q_1q_2 (N-1)v_{12} p_1p_2 \psi} \Big\rvert \nonumber \\ 
&\qquad = (N-1) \, \Big\lvert \int d^3k \,\hat{v}(k) \bigSCP{\psi}{q_1e^{ikx_1}p_1 q_2e^{-ikx_2}p_2 \psi} \Big\rvert \nonumber \\
\eqexpl{by antisym.} &\qquad= \Big\lvert \int d^3k \,\hat{v}(k) \sum_{\ell} \mu_{\ell} \bigSCP{\psi}{q_1\ketbra{\phi_{\ell}}{\tilde{\phi}_{\ell}}_1 \sum_{m=2}^N A_m^{-k} \psi} \Big\rvert \nonumber \\
\eqexpl{by C.-S.} &\qquad\leq \int d^3k\, |\hat{v}(k)| \sum_{\ell} \mu_{\ell} \Big|\Big|\bra{\phi_{\ell}}_1 q_1\psi\Big|\Big| \norm{\bra{\tilde{\phi}_{\ell}}_1 \sum_{m=2}^N A_m^{-k} \psi} \nonumber \\
\eqexpl{by C.-S.} &\qquad\leq \int d^3k\, |\hat{v}(k)| \sqrt{\bigSCP{\psi}{q_1 |A_1^{k*}| q_1\psi}} \times \nonumber \\
&\qquad\quad \times \sqrt{N^2 \bigSCP{\psi}{A_2^{-k*} |A^k_1| A_3^{-k}\psi} + N \bigSCP{\psi}{|A^k_1| \big|A_2^{-k}\big|^2\psi}}.
\end{align}
The appearing terms are bounded as follows. First, since $\norm[\op]{|A_1^{k*}|} \leq 1$,
\be\label{sc_term2_part1}
\bigSCP{\psi}{q_1 |A_1^{k*}| q_1\psi} \leq \norm{q_1\psi}^2 \norm[\op]{|A_1^{k*}|} \leq \alpha_n.
\ee
Furthermore, using Lemma~\ref{lem:sc_trace_est} twice, $\big|\big| |A_3^{-k*}| \big|\big|_{\op} \leq 1$ and the decomposition $A_m^{-k} = \sum_{\ell} \tilde{\mu}_{\ell} \ketbra{\chi_{\ell}}{\tilde{\chi}_{\ell}}_m$, we find
\begin{align}\label{sc_term2_part2}
\bigSCP{\psi}{A_2^{-k*} |A^k_1| A_3^{-k}\psi} &= \sum_{j,\ell} \tilde{\mu}_j\tilde{\mu}_{\ell} \bigSCP{\psi}{q_3 \sqrt{|A_1^k|} \ketbra{\tilde{\chi}_j}{\chi_j}_2 \ketbra{\chi_{\ell}}{\tilde{\chi}_{\ell}}_3 \sqrt{|A_1^k|} q_2 \psi} \nonumber \\
&\leq \bigSCP{\psi}{q_3 |A_1^k| \, |A_2^{-k}| \, |A_3^{-k*}| q_3 \psi} \nonumber \\
\eqexp{by Lem.~\ref{lem:sc_trace_est} twice} &\leq CN^{-2} \norm[\tr]{A_1^k} \norm[\tr]{A_2^{-k}} \alpha_n.
\end{align}
Finally, using $\sup_{\ell} \mu_{\ell} = \norm[\op]{|A^k_1|} \leq 1$, Lemma~\ref{lem:projector_norms} and Lemma~\ref{lem:sc_trace_est} for the first inequality, we find
\begin{align}\label{sc_term2_part3}
\bigSCP{\psi}{\big|A_1^k\big| \big|A_2^{-k}\big|^2\psi} &\leq C N^{-1} \norm[\tr]{A_1^k} \bigSCP{\psi}{\big|A_2^{-k}\big|^2\psi} \nonumber \\
&= C N^{-1} \norm[\tr]{A_1^k} \sum_{\ell} \mu_{\ell}^2 \bigSCP{\psi}{\ketbr{\tilde{\phi}_{\ell}}_2\psi} \nonumber \\
\eqexp{by Lem.~\ref{lem:projector_norms}} &\leq C N^{-2} \norm[\tr]{A_1^k} \left(\sup_{\ell}\mu_{\ell}\right) \sum_{\ell} \mu_{\ell} \nonumber \\
&\leq C N^{-2} \norm[\tr]{A_1^k} \norm[\tr]{A_2^{-k}}.
\end{align}
Using the estimates \eqref{sc_term2_part1}, \eqref{sc_term2_part2} and \eqref{sc_term2_part3} for the terms in \eqref{sc_term2_prep} and using \eqref{sc_prop} we find for the $qq$-$pp$ term, that
\begin{align}
N^{-\frac{2}{3}} \Big\lvert \bigSCP{\psi}{q_1q_2 (N-1)v_{12} p_1p_2 \psi} \Big\rvert &\leq C \int d^3k \, \hat{v}(k) (1+|k|) \sqrt{\alpha_n} \sqrt{\alpha_n + N^{-1}} \nonumber \\
&\leq C \big( \alpha_n + N^{-1} \big).
\end{align}

\absatz

\textbf{The} $\boldsymbol{qq}$-$\boldsymbol{pq}$ \textbf{term.} 
Using the Fourier decomposition of $v$, the singular value decomposition of $q_1e^{ikx_1}p_1$, Cauchy-Schwarz, \eqref{sc_prop} and Lemma~\ref{lem:projector_norms} we find
\begin{align}
N^{\frac{1}{3}} \, \Big\lvert \bigSCP{\psi}{q_1q_2 v_{12} p_1q_2 \psi} \Big\rvert 
&= N^{\frac{1}{3}} \, \Big\lvert \int d^3k \, \hat{v}(k) \bigSCP{\psi}{q_1e^{ikx_1}p_1q_2 e^{-ikx_2}q_2 \psi} \Big\rvert \nonumber \\
&= N^{\frac{1}{3}} \Big\lvert \int d^3k \, \hat{v}(k) \sum_{\ell} \mu_{\ell} \bigSCP{\psi}{q_2\ketbra{\phi_{\ell}}{\tilde{\phi}_{\ell}}_1 e^{-ikx_2}q_2 \psi} \Big\rvert \nonumber \\
\eqexp{by C.-S.} &\leq N^{\frac{1}{3}} \int d^3k \, |\hat{v}(k)| \sum_{\ell} \mu_{\ell} \, \Big|\Big|\bra{\phi_{\ell}}_1 q_2\psi\Big|\Big| \norm{\bra{\tilde{\phi}_{\ell}}_1 q_2 \psi} \nonumber \\
\eqexp{by Lem.~\ref{lem:projector_norms}} &\leq N^{\frac{1}{3}} \int d^3k \, |\hat{v}(k)| \norm[\tr]{q_1e^{ikx_1}p_1} N^{-1} \norm{q_2\psi}^2 \nonumber \\
\eqexp{by \eqref{sc_prop}} &\leq C \alpha_n.
\end{align}

\end{proof}

\bibliographystyle{plain}
\bibliography{references}

\begin{thebibliography}{10}

\bibitem{bach:1992}
V.~Bach.
\newblock Error bound for the {H}artree-{F}ock energy of atoms and molecules.
\newblock {\em Commun. Math. Phys.}, 147(3):527--548, 1992.

\bibitem{bach:1993}
V.~Bach.
\newblock Accuracy of mean field approximations for atoms and molecules.
\newblock {\em Commun. Math. Phys.}, 155(2):295--310, 1993.

\bibitem{bardos:2007}
C.~Bardos, B.~Ducomet, F.~Golse, A.~D. Gottlieb, and N.~J. Mauser.
\newblock The {TDHF} approximation for {H}amiltonians with m-particle
  interaction potentials.
\newblock {\em Commun. Math. Sci.}, 5:1--9, 2007.

\bibitem{bardos:2003}
C.~Bardos, F.~Golse, A.~D. Gottlieb, and N.~J. Mauser.
\newblock Mean field dynamics of fermions and the time-dependent {Hartree-Fock}
  equation.
\newblock {\em J. Math. Pures Appl.}, 82(6):665–--683, 2003.

\bibitem{bardos:2004}
C.~Bardos, F.~Golse, A.~D. Gottlieb, and N.~J. Mauser.
\newblock Accuracy of the time-dependent {Hartree-Fock} approximation for
  uncorrelated initial states.
\newblock {\em J. Stat. Phys.}, 115(3--4):1037--1055, 2004.

\bibitem{benedikter:2012}
N.~Benedikter, G.~de~Oliveira, and B.~Schlein.
\newblock Quantitative derivation of the {G}ross-{P}itaevskii equation.
\newblock {\em Preprint}, 2012.
\newblock
  \href{http://arxiv.org/abs/1208.0373v2}{http://arxiv.org/abs/1208.0373v2}.

\bibitem{benedikter:2014}
N.~Benedikter, M.~Porta, and B.~Schlein.
\newblock Mean-field dynamics of fermions with relativistic dispersion.
\newblock {\em J. Math. Phys.}, 55(2), 2014.

\bibitem{benedikter:2013}
N.~Benedikter, M.~Porta, and B.~Schlein.
\newblock Mean-field evolution of fermionic systems.
\newblock {\em Commun. Math. Phys., to appear}, 2014.
\newblock
  \href{http://arxiv.org/abs/1305.2768v3}{http://arxiv.org/abs/1305.2768v3}.

\bibitem{bove:1974}
A.~Bove, G.~{Da Prato}, and G.~Fano.
\newblock An existence proof for the {H}artree-{F}ock time-dependent problem
  with bounded two-body interaction.
\newblock {\em Commun. Math. Phys.}, 37(3):183--191, 1974.

\bibitem{bove:1976}
A.~Bove, G.~{Da Prato}, and G.~Fano.
\newblock On the {H}artree-{F}ock time-dependent problem.
\newblock {\em Commun. Math. Phys.}, 49(1):25--33, 1976.

\bibitem{hepp:1977}
W.~Braun and K.~Hepp.
\newblock The {V}lasov dynamics and its fluctuations in the {1/N} limit of
  interacting classical particles.
\newblock {\em Commun. Math. Phys.}, 56(2):101--113, 1977.

\bibitem{chadam:1976}
J.~M. Chadam.
\newblock The time-dependent {H}artree-{F}ock equations with {C}oulomb two-body
  interaction.
\newblock {\em Commun. Math. Phys.}, 46(2):99--104, 1976.

\bibitem{chadam:1975}
J.~M. Chadam and R.~T. Glassey.
\newblock Global existence of solutions to the {C}auchy problem for
  time-dependent {H}artree equations.
\newblock {\em J. Math. Phys.}, 16(5):1122--1130, 1975.

\bibitem{dirac:1930}
P.~A.~M. Dirac.
\newblock Note on exchange phenomena in the {T}homas atom.
\newblock {\em Math. Proc. Cambridge Philos. Soc.}, 26(3):376--385, 1930.

\bibitem{erdoes:2004}
A.~Elgart, L.~Erd{\"o}s, B.~Schlein, and H.-T. Yau.
\newblock Nonlinear {H}artree equation as the mean field limit of weakly
  coupled fermions.
\newblock {\em J. Math. Pures Appl.}, 83(10):1241–--1273, 2004.

\bibitem{erdoes:2006}
L.~Erd{\"o}s, B.~Schlein, and H.-T. Yau.
\newblock Derivation of the {G}ross-{P}itaevskii hierarchy for the dynamics of
  {B}ose-{E}instein condensate.
\newblock {\em Commun. Pure Appl. Math.}, 59(12):1659--1741, 2006.

\bibitem{erdoes:2007}
L.~Erd{\"o}s, B.~Schlein, and H.-T. Yau.
\newblock Derivation of the cubic non-linear {S}chr\"odinger equation from
  quantum dynamics of many-body systems.
\newblock {\em Invent. Math.}, 167(3):515--614, 2007.

\bibitem{erdoes:2007_2}
L.~Erd{\"o}s, B.~Schlein, and H.-T. Yau.
\newblock Rigorous derivation of the {G}ross-{P}itaevskii equation.
\newblock {\em Phys. Rev. Lett.}, 98(4), 2007.

\bibitem{erdoes:2009}
L.~Erd{\"o}s, B.~Schlein, and H.-T. Yau.
\newblock Rigorous derivation of the {G}ross-{P}itaevskii equation with a large
  interaction potential.
\newblock {\em J. Am. Math. Soc.}, 22(4):1099--1156, 2009.

\bibitem{erdoes:2010}
L.~Erd{\"o}s, B.~Schlein, and H.-T. Yau.
\newblock Derivation of the {G}ross-{P}itaevskii equation for the dynamics of
  {B}ose-{E}instein condensate.
\newblock {\em Ann. Math.}, 172(1):291--370, 2010.

\bibitem{erdoes:2001}
L.~Erd{\"o}s and H.-T. Yau.
\newblock Derivation of the nonlinear {S}chr\"odinger equation from a many body
  {C}oulomb system.
\newblock {\em Adv. Theor. Math. Phys.}, 5(6):1169--1205, 2001.

\bibitem{fefferman:1990}
C.~L. Fefferman and L.~A. Seco.
\newblock On the energy of a large atom.
\newblock {\em Bull. Amer. Math. Soc. (N.S.)}, 23(2):251--633, 1990.

\bibitem{fefferman:1994}
C.~L. Fefferman and L.~A. Seco.
\newblock On the {D}irac and {S}chwinger corrections to the ground-state energy
  of an atom.
\newblock {\em Adv. Math.}, 107(1):1--185, 1994.

\bibitem{fock:1930}
V.~Fock.
\newblock {N}{\"a}herungsmethode zur {L\"o}sung des quantenmechanischen
  {M}ehrk{\"o}rperproblems.
\newblock {\em Z. Phys.}, 61(1--2):126--148, 1930.

\bibitem{froehlich:2011}
J.~Fr{\"o}hlich and A.~Knowles.
\newblock A microscopic derivation of the time-dependent {H}artree-{F}ock
  equation with {C}oulomb two-body interaction.
\newblock {\em J. Stat. Phys.}, 145(1):23--50, 2011.

\bibitem{froehlich:2009}
J.~Fr{\"o}hlich, A.~Knowles, and S.~Schwarz.
\newblock On the mean-field limit of bosons with {C}oulomb two-body
  interaction.
\newblock {\em Commun. Math. Phys.}, 288(3):1023--1059, 2009.

\bibitem{froehlich:2007}
J.~Fr{\"o}hlich and E.~Lenzmann.
\newblock Dynamical collapse of white dwarfs in {H}artree- and {H}artree-{F}ock
  theory.
\newblock {\em Commun. Math. Phys.}, 274(3):737--750, 2007.

\bibitem{graf_solovej:1994}
G.~M. Graf and J.~P. Solovej.
\newblock A correlation estimate with applications to quantum systems with
  {C}oulomb interaction.
\newblock {\em Rev. Math. Phys.}, 6(5a):977--997, 1994.

\bibitem{hainzl:2010}
C.~Hainzl, E.~Lenzmann, M.~Lewin, and B.~Schlein.
\newblock On blowup for time-dependent generalized {H}artree-{F}ock equations.
\newblock {\em Ann. Henri Poincar{\'{e}}}, 11(6):1023--1052, 2010.

\bibitem{hainzl:2005}
C.~Hainzl, M.~Lewin, and C.~Sparber.
\newblock Existence of global-in-time solutions to a generalized {Dirac-Fock}
  type evolution equation.
\newblock {\em Lett. Math. Phys.}, 72(2):99--113, 2005.

\bibitem{hainzl:2009}
C.~Hainzl and B.~Schlein.
\newblock Stellar collapse in the time dependent {H}artree-{F}ock
  approximation.
\newblock {\em Commun. Math. Phys.}, 287(2):705--717, 2009.

\bibitem{hartree:1928}
D.~R. Hartree.
\newblock The wave mechanics of an atom with a non-{C}oulomb central field.
  {P}art {I}. {T}heory and methods.
\newblock {\em Math. Proc. Cambridge}, 24(1):89--110, 1928.

\bibitem{hepp:1974}
K.~Hepp.
\newblock The classical limit for quantum mechanical correlation functions.
\newblock {\em Commun. Math. Phys.}, 35(4):265--277, 1974.

\bibitem{mitrouskas:2013}
M.~Jeblick, D.~Mitrouskas, S.~Petrat, and P.~Pickl.
\newblock Effective dynamics of a tracer particle in a dense homogeneous
  quantum gas.
\newblock {\em Preprint}, 2014.
\newblock
  \href{http://arxiv.org/abs/1402.1696}{http://arxiv.org/abs/1402.1696}.

\bibitem{pickl:2010hartree}
A.~Knowles and P.~Pickl.
\newblock Mean-field dynamics: Singular potentials and rate of convergence.
\newblock {\em Commun. Math. Phys.}, 298(1):101--138, 2010.

\bibitem{lanford:1975}
O.~E. Lanford.
\newblock Time evolution of large classical systems.
\newblock In J.~Moser, editor, {\em Dynamical systems, theory and
  applications}, volume~38 of {\em Lecture Notes in Physics}, pages 1--111.
  Springer, 1975.

\bibitem{liebloss:2001}
E.~Lieb and M.~Loss.
\newblock {\em Analysis}.
\newblock American Mathematical Society, second edition, 2001.

\bibitem{lieb:2010}
E.~Lieb and R.~Seiringer.
\newblock {\em The Stability of Matter in Quantum Mechanics}.
\newblock Cambridge University Press, 2010.

\bibitem{lieb:1974}
E.~Lieb and B.~Simon.
\newblock On solutions to the {H}artree-{F}ock problem for atoms and molecules.
\newblock {\em J. Chem. Phys.}, 61(2):735--736, 1974.

\bibitem{lieb:1977}
E.~Lieb and B.~Simon.
\newblock The {H}artree-{F}ock theory for {C}oulomb systems.
\newblock {\em Commun. Math. Phys.}, 53(3):185--194, 1977.

\bibitem{lieb:1975}
E.~Lieb and W.~E. Thirring.
\newblock Bound for the kinetic energy of fermions which proves the stability
  of matter.
\newblock {\em Phys. Rev. Lett.}, 35(11):687–--689, 1975.

\bibitem{lions:1987}
P.~L. Lions.
\newblock Solutions of {H}artree-{F}ock equations for {C}oulomb systems.
\newblock {\em Commun. Math. Phys.}, 109(1):33--97, 1987.

\bibitem{lukkarinen:2009}
J.~Lukkarinen and H.~Spohn.
\newblock Not to normal order---notes on the kinetic limit for weakly
  interacting quantum fluids.
\newblock {\em J. Stat. Phys.}, 134(5--6):1133--1172, 2009.

\bibitem{narnhofer:1981}
H.~Narnhofer and G.~L. Sewell.
\newblock Vlasov hydrodynamics of a quantum mechanical model.
\newblock {\em Commun. Math. Phys.}, 79(1):9--24, 1981.

\bibitem{pickl:2010gp_ext}
P.~Pickl.
\newblock Derivation of the time dependent {G}ross {P}itaevskii equation with
  external fields.
\newblock {\em Preprint}, 2010.
\newblock
  \href{http://arxiv.org/abs/1001.4894v2}{http://arxiv.org/abs/1001.4894v2}.

\bibitem{pickl:2010gp_pos}
P.~Pickl.
\newblock Derivation of the time dependent {G}ross-{P}itaevskii equation
  without positivity condition on the interaction.
\newblock {\em J. Stat. Phys.}, 140(1):76--89, 2010.

\bibitem{pickl:2011method}
P.~Pickl.
\newblock A simple derivation of mean field limits for quantum systems.
\newblock {\em Lett. Math. Phys.}, 97(2):151--164, 2011.

\bibitem{reedsimon1:1980}
M.~Reed and B.~Simon.
\newblock {\em Methods of Modern Mathematical Physics. I: Functional Analysis}.
\newblock Academic Press, Inc., first edition, 1980.

\bibitem{rodnianski:2009}
I.~Rodnianski and B.~Schlein.
\newblock Quantum fluctuations and rate of convergence towards mean field
  dynamics.
\newblock {\em Commun. Math. Phys.}, 291(1):31--61, 2009.

\bibitem{slater:1930}
J.~C. Slater.
\newblock Note on {H}artree's method.
\newblock {\em Phys. Rev.}, 35(2):210--211, 1930.

\bibitem{spohn:1980}
H.~Spohn.
\newblock Kinetic equations from {H}amiltonian dynamics: {M}arkovian limits.
\newblock {\em Rev. Mod. Phys.}, 53(3):569–--615, 1980.

\bibitem{spohn:1981}
H.~Spohn.
\newblock On the {V}lasov hierarchy.
\newblock {\em Math. Methods Appl. Sci.}, 3(1):445--455, 1981.

\bibitem{spohn:1991}
H.~Spohn.
\newblock {\em Large Scale Dynamics of Interacting Particles}.
\newblock Springer, first edition, 1991.

\end{thebibliography}

\newpage
\pagestyle{empty}
\cleardoublepage

\selectlanguage{ngerman}

\begin{center}
{\Large\textbf{Eidesstattliche Versicherung}\\
\small (Siehe Promotionsordnung vom 12.07.11, $\S 8$, Abs.\ 2 Pkt.\ 5)}
\end{center}

\vspace{0.1cm}

\begin{flushleft}
Hiermit erkl\"are ich an Eidesstatt, dass die Dissertation von mir selbstst\"andig, ohne unerlaubte Beihilfe angefertigt ist.

\vspace{2cm}

M\"unchen, den 15.05.2014 \hspace*{3cm} S\"oren Petrat

\end{flushleft}

\end{document}